\documentclass[11pt]{article}
\usepackage{authblk}
\usepackage{hyperref}
\usepackage{url}
\usepackage{centernot}
\usepackage{amsmath,amsfonts,amssymb,amsthm}
\usepackage{graphicx}
\usepackage[round]{natbib}
\usepackage{booktabs} 
\usepackage{subcaption}
\usepackage{bbm}
\usepackage{bm}
\usepackage{mwe}
\usepackage{paralist} 
\usepackage{indentfirst} 
 \usepackage{booktabs, makecell, pifont, threeparttable} 
\usepackage{xcolor}
\usepackage[toc,page,header]{appendix}
\usepackage{minitoc}
\usepackage[most]{tcolorbox}
\usepackage{longtable} 
\usepackage{booktabs}

\usepackage{setspace}

\usepackage[margin=1in]{geometry}

\DeclareMathOperator*{\argmax}{arg\,max}
\DeclareMathOperator*{\argmin}{arg\,min}

\renewcommand{\emptyset}{\varnothing}

\newcommand{\cmark}{\ding{51}}          
\newcommand{\xmark}{\ding{55}}          
\newcommand{\naentry}{\textrm{NA}}      

\newcommand{\rd}{\,\mathrm{d}}

\newcommand{\e}{\mathbb{E}}
\newcommand{\var}{\mathrm{Var}}
\newcommand{\cov}{\mathrm{Cov}}
\newcommand{\corr}{\mathrm{Corr}}

\newcommand{\real}{\mathbb{R}}

\newcommand{\tran}{\mathsf{T}}
\newcommand{\indep}{\raisebox{0.05em}{\rotatebox[origin=c]{90}{$\models$}}}

\newcommand{\cd}{\mathcal{D}}

\newcommand{\ci}{\mathcal{I}}

\newcommand{\cw}{\mathcal{W}}

\newcommand{\cy}{\mathcal{Y}}

\newcommand{\bfs}{\mathbb{S}}
\newcommand{\bfts}{\tilde{\mathbb{S}}}
\newcommand{\bftp}{\tilde{\mathbb{P}}}

\newcommand{\completeSampNotation}{\textnormal{II}}
\newcommand{\htc}{\hat{\theta}^{\completeSampNotation}}
\newcommand{\hgc}{\hat{\gamma}^{\completeSampNotation}}
\newcommand{\hga}{\hat{\gamma}^{\textnormal{I}}}

\newcommand{\PTDSuperScriptAcr}{\textnormal{MPD}}
\newcommand{\htPTD}{\hat{\theta}^{\PTDSuperScriptAcr}}

\newcommand{\xobs}{X^{\mathrm{c}}}
\newcommand{\xmiss}{X^{\mathrm{e}}}
\newcommand{\txmiss}{\tilde{X}^{\mathrm{e}}}

\newcommand{\goodX}{X}
\newcommand{\proxyX}{\tilde{X}}
\newcommand{\argx}{x}
\newcommand{\argxProxy}{\tilde{x}}
\newcommand{\xSpace}{\mathcal{X}}
\newcommand{\xProxySpace}{\tilde{\mathcal{X}}}

\newcommand{\datvecraw}{V}
\newcommand{\datvecrawDist}{\mathbb{P}_V}
\newcommand{\nSuper}{N_{\textnormal{Super}}}

\newcommand{\LabelStrategy}[1]{\mathcal{A}_{\pi}^{({#1})}}
\newcommand{\LabelRulePi}[1]{\pi_{\cd_{{#1}-1}}^{({#1})}}
\newcommand{\LabelRulePiDoubleArg}[2]{\pi_{\cd_{{#2}}}^{({#1})}}

\newcommand{\LimitingLabelRulePi}[1]{\bar{\pi}^{({#1})}}

\newcommand{\thetaTarg}{\theta_0}
\newcommand{\gammaTarg}{\gamma_0}

\newcommand{\giv}{\!\mid\!} 

\newcommand{\samplingSchemeName}[1]{two-phase proxy-assisted multiwave sampling{#1}}
\newcommand{\samplingSchemeNameCap}[1]{Two-phase proxy-assisted multiwave sampling{#1}}

\newcommand{\singlespacematrix}{\renewcommand{\arraystretch}{0.7}}

\usepackage{etoolbox}

\theoremstyle{plain}  
\newtheorem{theorem}{Theorem}
\newtheorem{lemma}[theorem]{Lemma}
\newtheorem{corollary}[theorem]{Corollary}
\newtheorem{proposition}[theorem]{Proposition}

\theoremstyle{remark}

\newtheorem{remark}{Remark}
\newtheorem{fact}{Fact}

\newtheorem{assumption}{Assumption}

\title{M-estimation under Two-Phase Multiwave Sampling \\
with Applications to Prediction-Powered Inference}

 \author[1,2]{\large Dan M. Kluger \thanks{Corresponding author: dkluger@mit.edu}}
 \author[2]{Stephen Bates}
 \affil[1]{\footnotesize Institute for Data, Systems, and Society, Massachusetts Institute of Technology}
 \affil[2]{\footnotesize Department of Electrical Engineering and Computer Science, Massachusetts Institute of Technology}

\date{}

\begin{document}
    
\maketitle
\vspace{-6pt}
\begin{abstract}

In two-phase multiwave sampling, inexpensive measurements are collected on a large sample and expensive, more informative measurements are adaptively obtained on subsets of units across multiple waves. Adaptively collecting the expensive measurements can increase efficiency but complicates statistical inference. We give valid estimators and confidence intervals for M-estimation under adaptive two-phase multiwave sampling. We focus on the case where proxies for the expensive variables---such as predictions from pretrained machine learning models---are available for all units and propose a \textit{Multiwave Predict-Then-Debias} estimator that combines proxy information with the expensive, higher-quality measurements to improve efficiency while removing bias. We establish asymptotic linearity and normality and propose asymptotically valid confidence intervals. We also develop an approximately greedy sampling strategy that improves efficiency relative to uniform sampling. Data-based simulation studies support the theoretical results and demonstrate efficiency gains.
  
\end{abstract}

\noindent  {\bf Keywords:} active inference, adaptive designs, two-phase sampling, Neyman allocation

\section{Introduction}

With recent advances in machine learning and artificial intelligence, researchers are increasingly assembling and analyzing large datasets in which some variables are algorithmic outputs rather than direct measurements of the quantity of interest. For example, a study may use a predicted protein structure from a protein language model rather than a structure measured with crystallography, because the latter is costly and time intensive. In this situation, naively trusting the inexpensive predictions as if they are real data in a statistical analysis will result in biased estimators and invalid confidence intervals. Nonetheless, valid statistical inference is possible using methods from the missing data, measurement error, and survey sampling literatures~\citep[e.g.,][]{RobinsEtAl94JASA,ChenAndChen2000,sarndal2003model,TsiatisMissingDataSemiparametricChapter} and the subsequent Prediction-Powered Inference literature~\citep[e.g.,][]{OriginalPPI,DemistifyingPPI}. Such methods require a small amount of \textit{gold-standard} direct measurements (or an internal
validation sample) to complement the error-prone algorithmic predictions.

In this paper, we consider the version of this problem where the researcher can adaptively collect the gold-standard measurements. This has the promise of increasing sample efficiency, resulting in narrower confidence intervals for the parameter of interest with the same amount of data. For example, when estimating the 0.95-quantile of crop yield (a common measure of attainable yield in agronomy) from machine learning-based crop yield estimates and a limited number of field measurements (Section~\ref{sec:AttainableYieldSims}), the most informative units to measure are those whose crop yields are near the quantile being estimated. Identifying the most informative units requires knowledge about the joint distribution of the crop yield measurements and their machine learning-based estimates. This joint distribution can be learned from earlier measurements so that later measurements can focus on the most informative units. Despite the promise of efficiency gains, a major technical challenge is that adaptive sampling schemes can introduce statistical dependencies across the samples, rendering it substantially more challenging to conduct valid statistical inference. In this work, we introduce a new estimator and confidence intervals for M-estimation tasks (e.g., mean, quantile, or regression coefficient estimation tasks) in the adaptive sampling setting, and prove their validity.




Our work can be viewed as an instance of two-phase multiwave sampling~\citep{AdaptiveSamplingInTwoPhaseDesignsMcIsaacAndCook,ChenLumleyMultiwaveInTwoPhase} in which expensive variables are adaptively collected across multiple measurement waves. While the literature on two-phase multiwave sampling also studies practical sampling strategies and estimators in the regime we study, to our knowledge this literature has not established asymptotic normality of M-estimators using theory that accounts for statistical dependencies induced by the proposed sampling strategies. 

Simultaneously, our work can also be viewed as a part of the broader adaptive sampling and experimental design literature, and it is closely related to recent research on Active Statistical Inference~\citep{ActiveInferencePaper}.  While a number of works in this literature provide theoretical guarantees for M-estimation (or Z-estimation) tasks under adaptive sampling~\citep[e.g.,][]{zhangJansonMurphyBanditMestimator2021mestimator,ActiveInferencePaper,guo2026StatisticalInferenceForMisspecifiedContextualBandits,leinerRamdas2026MestimatorsAdaptive}, these works focus on online or contextual bandit sampling schemes (rather than two-phase multiwave sampling schemes) and require a smoother loss (or score) function than we do. To our knowledge, the theoretical guarantees in this literature do not extend to quantile estimation under adaptive sampling. We discuss related work in detail in Section~\ref{sec:relatedWork}.


\paragraph{Our contribution.} We introduce an estimator for the two-phase multiwave setting. We prove that this estimator is asymptotically linear and asymptotically normal, and use this to provide asymptotic confidence intervals. To our knowledge, these are the first theoretical guarantees of either asymptotic linearity or asymptotic normality of a general M-estimation approach under two-phase multiwave sampling. 
Moreover, our theoretical guarantees of asymptotic linearity and normality involve relatively mild smoothness assumptions on the loss function and thus account for quantile estimation tasks. We also discuss how the user should choose the sampling strategy for increased statistical efficiency.



\section{Setting, point estimator, and confidence intervals}\label{sec:settingAndEstimator}

In this section, we formally introduce our notation and setting and describe \samplingSchemeName{.} We then introduce appropriate inverse probability weights that can be used in these settings. For M-estimation tasks, we propose a Predict-Then-Debias type estimator \citep{ChenAndChen2000,PTDBootstrapPaper} that leverages all available data and present corresponding confidence intervals.

Throughout the text, we suppose $\datvecraw \equiv (\xobs,\txmiss,\xmiss) \sim \datvecrawDist$ is a random vector in $\mathbb{R}^q$ in which $\xmiss$ denotes a vector of expensive-to-measure variables, $\txmiss$ denotes a cheap-to-measure estimate of $\xmiss$, and $\xobs$ is a vector of other variables of interest or auxiliary variables that are also cheap-to-measure. We use the shorthand notation $\goodX \equiv (\xobs,\xmiss) \in \mathbb{R}^p$ to denote the vector of gold-standard measurements for all variables of interest. It will also be convenient to let $\proxyX \equiv (\xobs, \txmiss) \in \mathbb{R}^p$ denote the cheap-to-measure vector of estimates for all variables of interest. We use $\xSpace \subseteq \mathbb{R}^p$ and $\xProxySpace \subseteq \mathbb{R}^p$ to denote the supports of $\goodX$ and $\proxyX$, respectively. 
We use $\mathcal{P} \equiv \{ \pi : \xProxySpace \to (0,1) \}$ to denote the space of possible labelling rules (a labelling rule maps $\proxyX$ observations to a probability that the corresponding $\xmiss$ will be measured). For positive integers $n$ we use $[n]=\{1,\dots,n\}$ to denote the set of the first $n$ positive integers. We use $\mathsf{S}_N= \{\tau : [N] \to [N] \text{ such that } \tau \text{ is bijective} \}$ to denote the collection of permutations of the first $N$ integers. $I_{d \times d}$ denotes the $d \times d$ identity matrix.
We use $o_p(1)$ and $O_p(1)$ to denote sequences that converge in probability and are bounded in probability, respectively, as $N \to \infty$. 
Unless otherwise specified, sums and products over ranges of indices in which the lower limit exceeds the upper limit are defined to be $0$ and $1$, respectively.

 We focus on M-estimation settings where there is a prespecified loss function $l_{\theta}: \mathbb{R}^p \to \mathbb{R}$ parameterized by $\theta \in \Theta \subseteq \mathbb{R}^d $ and the goal is to estimate the well-defined quantity of interest \begin{equation}\label{eq:EstimandThetaStar}
    \thetaTarg \equiv \argmin_{\theta \in \Theta} \e[l_{\theta} (\goodX)].
\end{equation} As examples, the loss function $l_{\theta}(\cdot)$ could be chosen such that $\thetaTarg$ is a population mean, a population quantile, or a population regression coefficient in a GLM or robust regression model that regresses one component of $\goodX$ on other components of $\goodX$.

\subsection{\samplingSchemeNameCap{}}\label{sec:DescriptionOfSamplingScheme}

In \samplingSchemeName{,} Phase I involves collecting a large sample of size $N$ in which the less expensive variables $\proxyX=(\xobs,\txmiss)$ are measured for each sample but measurement of the expensive variable $\xmiss$ is reserved for Phase II. 
Phase II involves $K$ waves, where $K$ is a small, prespecified, positive integer that does not grow with $N$ asymptotically. In each wave, $\xmiss_i$ measurements are collected on a subset of the $N$ samples using independent Bernoulli sampling. The probability that an $\xmiss_i$ measurement is collected is a function of the available data for the $i$th sample $\proxyX_i=(\xobs_i,\txmiss_i)$ and a labelling rule that is learned on data from previous waves. The procedure, which we refer to as \samplingSchemeName{}, is formally stated in the box below, with elaborations and assumptions about each step subsequently provided.

\begin{tcolorbox}[
    breakable, 
    colback=gray!5,
    colframe=gray!60,
    title={\samplingSchemeNameCap{} scheme},
    fonttitle=\bfseries,
    boxrule=0.6pt,
    arc=2pt
]\label{box:SamplingRegime}

\vspace{6pt}
\textbf{Phase I (Inexpensive variable collection):}
Collect an i.i.d. sample of size $N$ from the superpopulation.  
For each $i \in [N]$, observe the inexpensive variables $\proxyX_i=(\xobs_i,\txmiss_i)$, while the expensive variable $\xmiss_i$ remains unobserved. Data available after Phase I is denoted by $\cd_0 \equiv ( \proxyX_i )_{i=1}^N$. \newline

\vspace{1pt}

\textbf{Phase II (Adaptive multiwave measurements):}
For waves $k = 1,\dots,K$:

\begin{enumerate}
    \item \textbf{Learn labelling rule:}  
    Use all data collected prior to wave $k$, denoted by $\cd_{k-1}$, and the labelling strategy $\LabelStrategy{k}$ to obtain a labelling rule
    $\LabelRulePi{k} \equiv \LabelStrategy{k}(\cd_{k-1}) \in \mathcal{P}.$

    \item \textbf{Select units for measurement:}  
    For each $i \in [N]$, compute the labelling probability
    $ \LabelRulePi{k}(\proxyX_i).
    $
    Draw $U_i^{(k)} \stackrel{\text{i.i.d.}}{\sim} \mathrm{Unif}[0,1]$ independently of $\cd_{k-1}$ and set
    $$I_i^{(k)} \equiv \mathbbm{1}\{U_i^{(k)} \leq \LabelRulePi{k}(\proxyX_i) \}.
    $$

    \item \textbf{Collect measurements:}  
    Observe $\xmiss_i$ if both $I_i^{(k)}=1$ and $\xmiss_i$ has not been previously measured. Data available after wave $k$ is denoted by $$\cd_k \equiv \big( (I_i^{(j)},I_i^{(j)} \xmiss_i )_{j=1}^k, \proxyX_i \big)_{i=1}^N.$$
\end{enumerate}

\vspace{3pt}

\textbf{Observed data after Phase II:} Let $I_i \in \{0,1\}$ be an indicator of whether $\xmiss_i$ was measured in any of the $K$ waves. Store the sampling information and data denoted by 
$$ \cd_{\text{out}} \equiv
 \Big( \big(I_i^{(j)},U_i^{(j)},\LabelRulePi{j}(\proxyX_i) \big)_{j=1}^K,\; \proxyX_i,\; I_i \xmiss_i \Big)_{i=1}^N.
$$

\end{tcolorbox}

We now elaborate on assumptions and implementation requirements for specific steps within \samplingSchemeName{}.
 In Phase I, we assume that the $N$ samples collected are i.i.d. from the superpopulation of interest. While the $\xmiss_i$ values are unobserved in Phase I, the following assumption states that the data would be an i.i.d. sample if $\xmiss_i$ were collected for all Phase I samples.

\begin{assumption}[i.i.d. Phase I data] \label{assump:IIDUnderlyingData}
    $\datvecraw_1,\dots,\datvecraw_N \stackrel{\text{i.i.d.}}{\sim} \datvecrawDist$, where $\datvecraw_i \equiv (\xobs_i,\txmiss_i,\xmiss_i)$ for $i \in [N]$.
\end{assumption}

In each wave $k$ in Phase II, the labelling strategy $\LabelStrategy{k}$ is a prespecified function that maps all previously observed data (and sampling indicators), denoted by $\cd_{k-1}$, to a labelling rule in $\mathcal{P} \equiv \{ \pi : \xProxySpace \to (0,1) \}$. If the labelling rule $\pi \in \mathcal{P}$ is selected, then for each $i \in [N]$, $\xmiss_i$ is to be measured with probability $\pi(\proxyX_i)$ according to independent Bernoulli sampling. If $\xmiss_i$ has already been measured in a previous wave, a measurement of it is not collected again (for simplicity, we consider settings where $\xmiss_i$ can be measured without noise so repeated measurements are of no value). We also suppose that the prespecified labelling strategy $\LabelStrategy{k}$ is sufficiently regular to not introduce measurability concerns (specifically, for each $i \in [N]$ and $k \in [K]$ we assume that the wave $k$ labelling probabilities for sample $i$, given by $[\LabelStrategy{k}(\cd_{k-1})](\proxyX_i)$ can be expressed as a measurable function of $\cd_{k-1}$ and $\proxyX_i$). See Section \ref{sec:ChoicesToIncreaseEfficiency} for specific labelling strategies $\LabelStrategy{k}$ designed for efficient parameter estimation. 

\subsection{Multiwave inverse probability weights}\label{sec:IntroduceMIPWWeights}

We next introduce inverse probability-type weights that are appropriate for two-phase multiwave sampling settings. For each $i \in [N]$, define $W_i^{(1)} \equiv I_i^{(1)}/{\LabelRulePiDoubleArg{1}{0}(\proxyX_i)}$ and define \begin{equation}\label{eq:InverseProbabilityWeights} 
W_i^{(k)} \equiv \Big( \prod_{j=1}^{k-1}  \frac{1-I_i^{(j)}}{1-\LabelRulePi{j} (\proxyX_i)} \Big) \frac{I_i^{(k)}}{\LabelRulePi{k} (\proxyX_i)},\end{equation} for $k \in \{2,\dots,K\}$. 
We next aggregate the weights across the $K$ waves. In particular, fix $c_1,c_2,\dots,c_K \in [0,1]$ such that $\sum_{k=1}^K c_k=1$, and define the following multiwave inverse probability weights 
\begin{equation}\label{eq:aggregated_Wk}
    W_i \equiv \sum_{k=1}^K c_k W_i^{(k)} \quad \text{for all } i \in [N].
\end{equation} The prespecified $c_k$ determine how much weight should be given to each wave, and as a starting point can be made proportional to the expected size of each wave. 

We remark that these weights are not the same as those seen in some other works on two-phase multiwave sampling \citep{ChenRecentSurrogatePoweredInferenceMultiwave,YangDiaoCook22}. In those works the inverse probability weights are given by calculating the total probability of a sample being labelled over the course of all waves.
Our construction of the weights enables us to establish theoretical guarantees by recursive applications of the tower property where we condition on data from previous waves. We expect that some properties of the multiwave inverse probability weights $W_i$, such as their lack of pairwise correlations, may be useful in other contexts, and we record them in Appendix \ref{sec:PropertiesOfWeights}.

\subsection{Multiwave Predict-Then-Debias estimator}\label{sec:IntroduceMestimator}

Recall that the goal is to estimate $\thetaTarg \equiv \argmin_{\theta \in \Theta} \e[l_{\theta}(\goodX)]$ defined at \eqref{eq:EstimandThetaStar}. Consider three simple estimators (the first two of which use Phase II information):
\begin{equation}\label{eq:PTDComponentEstimatorsDEF}
    \htc \equiv \argmin_{\theta \in \Theta} \frac{1}{N} \sum_{i=1}^N W_i l_{\theta}(\goodX_i), \ \  \hgc \equiv \argmin_{\theta \in \Theta} \frac{1}{N} \sum_{i=1}^N W_i l_{\theta}(\proxyX_i), \ \  \text{and }  \hga \equiv \argmin_{\theta \in \Theta} \frac{1}{N} \sum_{i=1}^N  l_{\theta}(\proxyX_i). 
\end{equation} 
We combine these estimators into our proposed \emph{Multiwave Predict-Then-Debias estimator}
\begin{equation}\label{eq:PTD_estimator} \htPTD = \hat{\Omega} \hga +  (\htc -  \hat{\Omega} \hgc),\end{equation}
where $\hat{\Omega} \in \mathbb{R}^{d \times d}$ is a tuning parameter for improved efficiency that we will discuss in due course. The reader can keep in mind the case where $\hat \Omega = I_{d \times d}$ as an intuitive special case.

The idea behind the estimator is as follows. The estimator
$\htc$ minimizes an empirical weighted loss whose expected value is $\e[l_{\theta}(\goodX)]$ (see Proposition \ref{prop:WeightsEnableUnbiasedMeanEst} for details), so $\htc$ is an estimator targeting the estimand of interest $\thetaTarg$. Likewise, $\hgc$ and $\hga$ minimize empirical loss functions whose expected values are $\e[l_{\theta}(\proxyX)]$, so these estimators target the quantity $$\gammaTarg \equiv \argmin_{\theta \in \Theta} \e[l_{\theta} (\proxyX)].$$ Importantly, $\gammaTarg$ is generally not equal to $\thetaTarg$, since the distribution of $\proxyX$ is not the same as that of $X$. Still, $\hga$ has low variance because it is based on all $N$ samples, so it is useful to anchor on this quantity. Then, we add a bias-correction term such that the resulting estimator targets $\thetaTarg$.

In this manuscript, we will show that this estimator is consistent and asymptotically normal. Moreover, it results in improved efficiency compared to baseline approaches such as $\htc$ and other nonadaptive strategies. We will also give a consistent variance estimator, which leads to valid confidence intervals.

\subsection{Asymptotic variance estimator and confidence intervals}\label{sec:IntroduceCovarianceEstimatorsAndCIs}

A formula for the asymptotic variance of $\htPTD$ will subsequently be given in Theorem \ref{theorem:CLT_PTDEstimator}, and we state a consistent variance estimator here. Define the following:
\begin{equation}\label{eq:HessianAndSigmaEstimators} 
\begin{split}
   \hat{\Sigma}_{11} & \equiv \frac{1}{N} \sum_{i=1}^N W_i \dot{l}_{\htc} (\goodX_i) [ W_i \dot{l}_{\htc} (\goodX_i)]^\tran, \quad \hat{\Sigma}_{12} \equiv \frac{1}{N} \sum_{i=1}^N W_i \dot{l}_{\htc} (\goodX_i) [ W_i \dot{l}_{\hga} (\proxyX_i)]^\tran, \\   \hat{\Sigma}_{22} & \equiv \frac{1}{N} \sum_{i=1}^N W_i  \dot{l}_{\hga} (\proxyX_i) [ W_i \dot{l}_{\hga} (\proxyX_i) ]^\tran, \quad \hat{\Sigma}_{13} \equiv \frac{1}{N} \sum_{i=1}^N W_i \dot{l}_{\htc} (\goodX_i) [ \dot{l}_{\hga} (\proxyX_i) ]^\tran, \\ \hat{\Sigma}_{33} & \equiv \frac{1}{N} \sum_{i=1}^N  \dot{l}_{\hga} (\proxyX_i) [ \dot{l}_{\hga} (\proxyX_i) ]^\tran, \quad \hat{H}_{\thetaTarg} \equiv \frac{1}{N} \sum_{i=1}^N W_i \ddot{l}_{\htc}(\goodX_i), \quad \text{and} \quad \hat{H}_{\gammaTarg} \equiv \frac{1}{N} \sum_{i=1}^N \ddot{l}_{\hga}(\proxyX_i).
\end{split}
\end{equation}
Above for each $\argx \in \xSpace \cup \xProxySpace$ and $\theta' \in \Theta$, $\dot{l}_{\theta'}(\argx)$ and $\ddot{l}_{\theta'}(\argx)$ denote the gradient and Hessian, respectively, of the map $\theta \mapsto l_{\theta}(\argx)$ evaluated at $\theta=\theta'$. (We remark that in some cases $\theta \mapsto l_{\theta}(\argx)$ is not differentiable or twice differentiable at $\theta=\htc$ or $\theta=\hga$. In these cases, 
we define $\dot{l}_{\theta'}(\argx)$ and $\ddot{l}_{\theta'}(\argx)$ in terms of the first and second order upper right-hand Dini partial derivatives to ensure that the estimators in \eqref{eq:HessianAndSigmaEstimators} are well-defined). Letting $\hat{\Omega} \in \mathbb{R}^{d \times d}$ be the (possibly data-dependent) tuning matrix used to construct $\htPTD$, an estimator for the asymptotic variance of $\htPTD$ is given by \begin{equation}\label{eq:PTDAsympVarEstimator} 
\begin{split}
\hat{\Sigma}^{\PTDSuperScriptAcr} & \equiv \hat{H}_{\thetaTarg}^{-1} \hat{\Sigma}_{11} \hat{H}_{\thetaTarg}^{-1} + \hat{\Omega} \hat{H}_{\gammaTarg}^{-1} \big( \hat{\Sigma}_{22}-\hat{\Sigma}_{33} \big) \hat{H}_{\gammaTarg}^{-1} \hat{\Omega}^\tran 
 \\ & \quad + \hat{H}_{\thetaTarg}^{-1} (\hat{\Sigma}_{13}-\hat{\Sigma}_{12}) \hat{H}_{\gammaTarg}^{-1} \hat{\Omega}^\tran +  \big(  \hat{H}_{\thetaTarg}^{-1} (\hat{\Sigma}_{13}-\hat{\Sigma}_{12}) \hat{H}_{\gammaTarg}^{-1} \hat{\Omega}^\tran \big)^\tran. 
 \end{split}
 \end{equation} Two-sided $(1-\alpha)$-confidence intervals for the $j$th component of $\thetaTarg$ are then given by \begin{equation}\label{eq:CIDef}
    \mathcal{C}_j^{(1-\alpha)} \equiv \Big[ \htPTD_j - z_{1-\alpha/2} \sqrt{\hat{\Sigma}_{jj}^{\PTDSuperScriptAcr}/N},\htPTD_j + z_{1-\alpha/2} \sqrt{\hat{\Sigma}_{jj}^{\PTDSuperScriptAcr} /N} \Big], 
\end{equation} for each $j \in [d]$ and $\alpha \in (0,1)$, where $z_{1-\alpha/2}$ denotes the $(1-\alpha/2)$-th quantile of a standard normal distribution. Under certain assumptions, this variance estimator is consistent and these confidence intervals are asymptotically valid; we turn to the technical details next.

\section{Asymptotic theory}\label{sec:AsympTheory}

In this section, we study the asymptotic properties of the estimator $\htPTD$ under \samplingSchemeName{.} In particular, we show that under relatively mild regularity conditions, $\htPTD$ is consistent and asymptotically linear. Here the asymptotic linear expansion involves statistically dependent weights $(W_i)_{i=1}^N$. Under additional assumptions, we establish that $\htPTD$ is asymptotically normal. 
Finally, under additional regularity conditions we show that the confidence intervals defined at \eqref{eq:CIDef} are asymptotically valid. 

In Appendix \ref{sec:HowToGeneralizeToTwoPhaseMultiwave}, we discuss how the theoretical results in this section can be used to obtain theoretical guarantees for M-estimators under other two-phase multiwave sampling designs, not just proxy-assisted ones. In that appendix, we also extend our theory to prove asymptotic linearity and normality of generalized raking-type estimators~\citep[e.g.,][]{DelvilleSarndalSautory93GeneralizedRaking,LumleyGeneralizedRaking2011} under two-phase multiwave sampling and suitable conditions.

\subsection{Consistency, \texorpdfstring{$\sqrt{N}$}{Root-N}-consistency, and asymptotic linearity}\label{sec:StateMestimationResults}

We begin by requiring that the labelling probabilities $\LabelRulePi{k}(\proxyX_i)$ are bounded away from zero and one, as is commonplace.

\begin{assumption}[Overlap condition]\label{assump:LabellingRuleOverlap}
    There exists a constant $b \in (0,1/2)$ that does not depend on $N$ such that almost surely $\LabelRulePi{k}(\proxyX_i) \in [b,1-b]$ for each $k \in [K]$ and $i \in [N]$.
\end{assumption} The constant $b$ may be arbitrarily close to $0$, but may not decrease as $N \to \infty$. In our setting, we have the capability to ensure that this assumption holds, since we control the labelling probabilities.

Next, we require regularity conditions of the loss function to enable M-estimation.
In particular, for each $\theta \in \Theta$ define $L(\theta) \equiv \e[l_{\theta}(\goodX)]$ and $\tilde{L}(\theta) \equiv \e[l_{\theta}(\proxyX)]$ to be the population losses and we suppose that the loss $l_{\theta}(\cdot)$ satisfies the following conditions.

\begin{assumption}[Regularity conditions for M-estimation]\label{assump:SmoothEnoughForAsymptoticLineariaty} \ 

\begin{enumerate}[(i)] 
    \item $\theta \mapsto l_{\theta}(\argx)$ is convex for every $\argx \in \xSpace \cup \xProxySpace$.
    \item Across the domain $\Theta$, $\thetaTarg$ is the unique minimizer of $L(\theta)$ and $\gammaTarg$ is the unique minimizer of $\tilde{L}(\theta)$, with $\thetaTarg$ and $\gammaTarg$ being in the interior of the set $\Theta$.
    \item $\theta \mapsto l_{\theta}(\goodX)$ is differentiable at $\theta=\thetaTarg$ almost surely and $\theta \mapsto l_{\theta}(\proxyX)$ is differentiable at $\theta=\gammaTarg$ almost surely.
    \item $\theta \mapsto l_{\theta}(\goodX)$ is locally Lipschitz around $\theta=\thetaTarg$ and $\theta \mapsto l_{\theta}(\proxyX)$ is locally Lipschitz around $\theta=\gammaTarg$. In particular, there exist neighborhoods $\mathcal{L}_{\thetaTarg}$ of $\thetaTarg$ and $\mathcal{L}_{\gammaTarg}$ of $\gammaTarg$, and there exist functions $M,\tilde{M} : \mathbb{R}^p \to (0,\infty)$ such that $\e[M^2(\goodX)] < \infty$, $\e[\tilde{M}^2(\proxyX)]< \infty$, and such that for all $\argx \in \xSpace$ and $\theta,\theta' \in \mathcal{L}_{\thetaTarg}$, $\vert l_{\theta}(\argx) - l_{\theta'}(\argx)\vert \leq M(\argx) \vert \vert  \theta-\theta' \vert \vert $ while for all $\argxProxy \in \xProxySpace$ and $\theta,\theta' \in \mathcal{L}_{\gammaTarg}$, $\vert l_{\theta}(\argxProxy) - l_{\theta'}(\argxProxy)\vert \leq \tilde{M}(\argxProxy) \vert \vert  \theta- \theta' \vert \vert $. 
    \item The population losses given by $L(\theta)$ and 
     $\tilde{L}(\theta)$ admit 2nd-order Taylor expansions about $\thetaTarg$ and $\gammaTarg$, respectively, and the Hessians $\nabla^2 L(\thetaTarg)$ and $\nabla^2 \tilde{L}(\gammaTarg)$ are nonsingular.
    \item $\e[l_{\thetaTarg}^2(\goodX)] < \infty$ and $\e[l_{\gammaTarg}^2(\proxyX)] < \infty$.
\end{enumerate}
\end{assumption}

The above assumptions are fairly standard in M-estimation theory \citep[e.g.,][]{VanderVaartTextbook,PPI++}, even when the data are an i.i.d. sample rather than our more challenging adaptive sampling setting. The above assumptions on the loss and expected loss will typically hold when the estimand is a population-level mean, quantile, or a regression coefficient in a GLM. For example, if $p=1$ and the goal is to estimate the $q$th quantile of $\goodX$, the loss function will always satisfy Assumptions \ref{assump:SmoothEnoughForAsymptoticLineariaty}(i) and (iv), Assumptions \ref{assump:SmoothEnoughForAsymptoticLineariaty}(ii), (iii) and (v) above will be met if both $\goodX$ and $\proxyX$ have cumulative distribution functions that cross $q$ at a unique point and have positive probability densities at their $q$th quantile, and Assumption  \ref{assump:SmoothEnoughForAsymptoticLineariaty}(vi) will be met if $\e[\goodX^2] <\infty $ and $\e[\proxyX^2] <\infty$.

Define $H_{\thetaTarg} \equiv \nabla^2 L(\thetaTarg)= \nabla_{\theta}^2\e[l_{\theta}(\goodX)] \big|_{\theta=\thetaTarg}$ and $H_{\gammaTarg} \equiv \nabla^2 \tilde{L}(\gammaTarg)= \nabla_{\theta}^2\e[l_{\theta}(\proxyX)] \big|_{\theta=\gammaTarg}$.  $H_{\thetaTarg}$ and $H_{\gammaTarg}$ exist and are invertible under Assumption \ref{assump:SmoothEnoughForAsymptoticLineariaty}(v). Meanwhile, under Assumption \ref{assump:SmoothEnoughForAsymptoticLineariaty}(iii), $\dot{l}_{\thetaTarg}(\goodX)=  \nabla_{\theta} l_{\theta}(\goodX) \big|_{\theta=\thetaTarg}$ and $\dot{l}_{\gammaTarg}(\proxyX)=\nabla_{\theta} l_{\theta}(\proxyX) \big|_{\theta=\gammaTarg}$ almost surely.


 Assumptions \ref{assump:IIDUnderlyingData}--\ref{assump:SmoothEnoughForAsymptoticLineariaty} are sufficient to ensure that $\htc$, $\hgc$, and $\hga$ are each $\sqrt{N}$-consistent estimators for $\thetaTarg$, $\gammaTarg$, and $\gammaTarg$ respectively, and that they each admit asymptotic linear expansions. We defer all proofs of the theoretical results from the main text to the appendix. 

\begin{theorem}\label{theorem:AsymptoticLInearMestsStacked}
Under \samplingSchemeName{} and Assumptions \ref{assump:IIDUnderlyingData}--\ref{assump:SmoothEnoughForAsymptoticLineariaty}, $${ \singlespacematrix \sqrt{N} \Bigg( \begin{bmatrix} 
    \htc \\ \hgc \\ \hga
\end{bmatrix} - \begin{bmatrix}
    \thetaTarg \\ \gammaTarg \\ \gammaTarg
\end{bmatrix} \Bigg) = - \frac{1}{\sqrt{N}} \sum_{i=1}^N \begin{bmatrix}
    H_{\thetaTarg}^{-1} & 0 & 0 \\
    0 & H_{\gammaTarg}^{-1} & 0 \\ 
    0 & 0 & H_{\gammaTarg}^{-1}
\end{bmatrix}\begin{bmatrix}
     W_i  \dot{l}_{\thetaTarg}(\goodX_i) \\ W_i  \dot{l}_{\gammaTarg}(\proxyX_i) \\  \dot{l}_{\gammaTarg}(\proxyX_i) 
\end{bmatrix} }+o_p(1).$$
Moreover, the above are $O_p(1)$.
\end{theorem}

 As a corollary, in the setting of Theorem \ref{theorem:AsymptoticLInearMestsStacked}, $\htPTD$ is consistent for $\thetaTarg$ and asymptotically linear, provided that $\hat{\Omega}$ converges in probability as $N \to \infty$. \begin{corollary}\label{cor:PTDEstAsymptoticallyLinear}
    Under \samplingSchemeName{} and Assumptions \ref{assump:IIDUnderlyingData}--\ref{assump:SmoothEnoughForAsymptoticLineariaty}, and if the tuning matrix $\hat{\Omega} \xrightarrow{p} \Omega$ as $N \to \infty$, $\htPTD \xrightarrow{p} \thetaTarg$ as $N \to \infty$ and \begin{equation}\label{eq:AsympLinearExpansion}
    \sqrt{N} \big( \htPTD -\thetaTarg \big) = -\frac{1}{\sqrt{N}} \sum_{i=1}^N \Big( W_i  H_{\thetaTarg}^{-1} \dot{l}_{\thetaTarg}(\goodX_i) +(1-W_i) \Omega H_{\gammaTarg}^{-1} \dot{l}_{\gammaTarg}(\proxyX_i) \Big) + o_p(1).\end{equation}
\end{corollary}
This conclusion about $\htPTD$ is the aspect of this subsection of most methodological interest. We will next leverage this result to prove asymptotic normality, but first we pause to comment on the technical challenges behind the above result.

Theorem~\ref{theorem:AsymptoticLInearMestsStacked} and its proof make up a major technical contribution of this paper. Asymptotic linearity is common in i.i.d. settings, but extending it to the adaptive setting is delicate.
In particular, in Appendix \ref{sec:ControlLocalEmpiricalProcessResults} we use similar symmetrization and chaining arguments for empirical processes as seen in canonical texts such as \cite{Vershynin_2018} and \cite{VDVAndWellnerTextbook}. In contrast to these, however, we decompose empirical weighted processes into multiple terms, each of which can be controlled by particular symmetrization arguments that condition on data from prior waves. 
Our modified symmetrization and chaining arguments also leverage the boundedness of the weights (Assumption \ref{assump:LabellingRuleOverlap}; Equation \eqref{eq:InverseProbabilityWeights}) to establish that expected fluctuations of the relevant local empirical processes are still controlled at a sufficiently fast rate in our dependency regime. 

\begin{remark}\label{remark:TaylorExpansionVersionOfAsymptoticLinearity} Another approach to proving asymptotic linearity seen in the adaptive sampling literature involves taking a Taylor expansion of the gradient of the empirical loss function evaluated at a proposed estimator $\hat{\theta}$ about $\thetaTarg$. Unfortunately, such a step would require smoothness assumptions on the loss function that are not met for quantile estimation tasks. For example, in Theorem 2 of arXiv version 3 of \cite{ActiveInferencePaper}, the Hessian of the loss function is assumed to be locally Lipschitz in a neighborhood of $\thetaTarg$, whereas our local Lipschitz assumption (Assumption \ref{assump:SmoothEnoughForAsymptoticLineariaty}(iv)) is on the loss itself and holds for quantile estimation tasks. Similarly, \cite{zhangJansonMurphyBanditMestimator2021mestimator} assume the third derivative of the loss function is locally dominated in a neighborhood of $\thetaTarg$ by an integrable function. Analogously, in the nonadaptive sampling setting, simpler proofs of the asymptotic linearity of M-estimators are possible when assuming stronger smoothness conditions on the loss function (see the ``Classical Conditions" discussed in Chapter 5.6 of \cite{VanderVaartTextbook}). Notably, our theoretical results enable statistical inference under adaptive sampling for quantile estimation tasks and other M-estimation tasks with less smooth loss functions.
    
\end{remark}

\subsection{Asymptotic normality}

We next turn to asymptotic normality.
Under nonadaptive sampling, an asymptotic linear expansion such as the one in Corollary \ref{cor:PTDEstAsymptoticallyLinear} together with the central limit theorem (CLT) and Slutsky's lemma immediately establish asymptotic normality. However, in \samplingSchemeName{,} the standard multivariate CLT cannot be applied because the terms being averaged are not statistically independent. We introduce additional assumptions that are sufficient to ensure that $\htPTD$ is asymptotically normal.

\begin{assumption}[Regularity conditions for establishing asymptotic normality]\label{assump:CLTRgularityConditions} \
\begin{enumerate}[(i)] 

    \item \textit{Symmetric labelling strategies.} The labelling rule in each wave does not depend on the indexing of the $N$ samples. In particular, for each $k \in [K]$ and permutation $\tau \in \mathsf{S}_N$, letting $\cd_{k-1}^{(\tau)} \equiv \big( (I_{\tau(i)}^{(j)},I_{\tau(i)}^{(j)} \cdot \xmiss_{\tau(i)})_{j=1}^{k-1},\proxyX_{\tau(i)} \big)_{i=1}^N$ denote a re-indexing of the dataset $\cd_{k-1}$ according to the permutation $\tau$, $\LabelStrategy{k}(\cd_{k-1})= \LabelStrategy{k}(\cd_{k-1}^{(\tau)})$. Here, both $\LabelStrategy{k}(\cd_{k-1}) \in \mathcal{P}$ and $\LabelStrategy{k}(\cd_{k-1}^{(\tau)}) \in \mathcal{P}$ are labelling rules (i.e., mappings $\xProxySpace \to (0,1)$), which are assumed to be the same for any permutation $\tau \in \mathsf{S}_N$.

    \item \textit{$L^1$ convergence of labelling rules}. For each $k \in [K]$ there exists a measurable function $\LimitingLabelRulePi{k} : \xProxySpace \to [b,1-b]$, such that $\lim\limits_{N \to \infty} \e \big[ \vert \LabelRulePi{k}(\proxyX_1) -\LimitingLabelRulePi{k}(\proxyX_1) \vert \big]=0.$ 
    \item \textit{Bounded moments of order greater than 2.} There exists an $\eta_*>0$, such that for each $j \in [d]$, $\e \big[ \big| [\dot{l}_{\thetaTarg}(\goodX)]_j \big|^{2+\eta_*} \big] < \infty$ and $\e \big[ \big| [\dot{l}_{\gammaTarg}(\proxyX)]_j \big|^{2+\eta_*} \big] < \infty$.
\end{enumerate}
\end{assumption}

 Assumption \ref{assump:CLTRgularityConditions}(i) can be ensured by an investigator who is choosing a labelling strategy, and will hold if the adaptive labelling strategy only gives preference based on the values of the data that was previously observed rather than the particular index of each sample. We remark that Assumption \ref{assump:CLTRgularityConditions}(i) is not strictly necessary and can be removed if Assumption \ref{assump:CLTRgularityConditions}(ii) is strengthened to state that $\lim_{N \to \infty} \sup_{i \in [N]} \e \big[ \vert \LabelRulePi{k}(\proxyX_i) -\LimitingLabelRulePi{k}(\proxyX_i) \vert \big]=0$. Assumption \ref{assump:CLTRgularityConditions}(ii) is a fairly common condition requiring that the labelling rules from each wave converge as $N \to \infty$. If we are in an asymptotic regime where the number of waves remains fixed as $N \to \infty$, the labelling strategies can be carefully chosen so that this assumption holds (e.g., using parametric or consistent nonparametric approaches to learn a good labelling rule). Notably, the $L^1$ convergence of the labelling rule can happen at an arbitrarily slow rate, while other theoretical results in the adaptive experiment literature \citep[e.g.,][]{Hahn2011AdaptiveExperimentalDesign,HarrisonBatchAdaptiveDMLCausal,LikeHarrisonsPaperAndCausalButAboutOutcomeCollection} assume particular rates of convergence. We remark that in contrast to our work, \cite{zhangJansonMurphyBanditMestimator2021mestimator} does not assume labelling rule (or policy) convergence, and instead assumes that the data comes from a correctly specified parametric or semiparametric model.

To present a formula for the asymptotic variance of $\htPTD$, it is convenient to define \begin{equation}\label{eq:PiBarProd1tok}
    \LimitingLabelRulePi{1:k}(\argxProxy) \equiv \LimitingLabelRulePi{k}(\argxProxy) \prod_{j=1}^{k-1} \big( 1-  \LimitingLabelRulePi{j}(\argxProxy) \big) \quad \text{for } k \in [K], \argxProxy \in \xProxySpace. 
\end{equation} The above quantity can be thought of as a limiting probability of a sample with cheap-to-measure data $\argxProxy$ being selected for labelling in the $k$th wave, but not in previous waves. Define also \begin{equation}\label{eq:GradientAsympCovFormulae}
\begin{aligned}
    \Sigma_{11} & \equiv  \sum_{k=1}^K c_k^2 \e \Big[ \frac{ \dot{l}_{\thetaTarg}(\goodX) [  \dot{l}_{\thetaTarg}(\goodX) ]^\tran}{\LimitingLabelRulePi{1:k}(\proxyX)} \Big], \quad  \Sigma_{12} \equiv  \sum_{k=1}^K c_k^2 \e \Big[ \frac{ \dot{l}_{\thetaTarg}(\goodX) [  \dot{l}_{\gammaTarg}(\proxyX) ]^\tran}{\LimitingLabelRulePi{1:k}(\proxyX)} \Big], \quad \Sigma_{13} \equiv  \e \big[ \dot{l}_{\thetaTarg}(\goodX) [\dot{l}_{\gammaTarg}(\proxyX)]^\tran \big],  \\  
    \Sigma_{22} & \equiv \sum_{k=1}^K c_k^2 \e \Big[ \frac{ \dot{l}_{\gammaTarg}(\proxyX) [  \dot{l}_{\gammaTarg}(\proxyX) ]^\tran}{\LimitingLabelRulePi{1:k}(\proxyX)} \Big], \quad  
      \text{and} \quad \Sigma_{33} \equiv \e \big[ \dot{l}_{\gammaTarg}(\proxyX) [\dot{l}_{\gammaTarg}(\proxyX)]^\tran \big]. 
\end{aligned}
\end{equation} 
For any fixed tuning matrix $\Omega \in \mathbb{R}^{d \times d}$,
the asymptotic variance is then\begin{equation}\label{eq:AsympVarPTDFormula}
\begin{split}
    \Sigma^{\PTDSuperScriptAcr}(\Omega) & \equiv H_{\thetaTarg}^{-1} \Sigma_{11} H_{\thetaTarg}^{-1} + \Omega H_{\gammaTarg}^{-1} \big( \Sigma_{22}-\Sigma_{33} \big) H_{\gammaTarg}^{-1} \Omega^\tran 
\\ & \quad + H_{\thetaTarg}^{-1} (\Sigma_{13}-\Sigma_{12}) H_{\gammaTarg}^{-1} \Omega^\tran +  \big(  H_{\thetaTarg}^{-1} (\Sigma_{13}-\Sigma_{12}) H_{\gammaTarg}^{-1} \Omega^\tran \big)^\tran.
\end{split}
\end{equation}

\begin{theorem}\label{theorem:CLT_PTDEstimator}
Under \samplingSchemeName{} and Assumptions \ref{assump:IIDUnderlyingData}--\ref{assump:CLTRgularityConditions}, $$\text{if} \quad \hat{\Omega} \xrightarrow{p} \Omega \quad \text{then} \quad \sqrt{N} (\htPTD-\thetaTarg) \xrightarrow{d} \mathcal{N} \big(0,\Sigma^{\PTDSuperScriptAcr}(\Omega) \big) \quad \text{as } N \to \infty.$$
\end{theorem} The main proof strategy is to introduce weights that are both (i) statistically independent of each other and (ii) asymptotically close enough to the statistically dependent weights $(W_i)_{i=1}^N$. Up to $o_p(1)$ terms, the asymptotic linear expansion in Corollary \ref{cor:PTDEstAsymptoticallyLinear} is then rewritten in terms of these independent weights, enabling the use of the standard multivariate CLT.

\subsection{Confidence interval validity}

To establish asymptotic validity of the confidence intervals using the previous asymptotic normality result, it remains to establish that the asymptotic covariance estimator at \eqref{eq:PTDAsympVarEstimator} is consistent. 
To ensure consistent asymptotic covariance matrix estimation, we introduce the following assumptions on the loss function $l_{\theta}(\cdot)$.

\begin{assumption}[Regularity conditions for consistent variance estimation]\label{assump:SmoothEnoughForConsitentVarEst} \

\begin{enumerate}[(i)] 
    \item $\theta \mapsto l_{\theta}(\goodX)$ is continuously twice differentiable at $\theta=\thetaTarg$ almost surely, while $\theta \mapsto l_{\theta}(\proxyX)$ is continuously twice differentiable at $\theta=\gammaTarg$ almost surely. Moreover, the loss is smooth enough such that 2nd derivatives and expectations can be swapped so that $\e[\ddot{l}_{\thetaTarg}(\goodX)]=H_{\thetaTarg}$ and $\e[\ddot{l}_{\gammaTarg}(\proxyX)]=H_{\gammaTarg}$.
    \item For each $j,j' \in [d]$, there exist functions $L_{jj'}, \tilde{L}_{jj'} : \mathbb{R}^p \to [0,\infty]$ and neighborhoods $\mathcal{B}_{jj'}$ and $\tilde{\mathcal{B}}_{jj'}$ of $\thetaTarg$ and $\gammaTarg$, respectively, such that $\vert [\ddot{l}_{\theta}(\argx)]_{jj'} \vert \leq L_{jj'}(\argx)$ for all $\theta \in \mathcal{B}_{jj'}$ and $\argx \in \xSpace$, $\vert [\ddot{l}_{\theta}(\argxProxy)]_{jj'} \vert \leq \tilde{L}_{jj'}(\argxProxy)$ for all $\theta \in \tilde{\mathcal{B}}_{jj'}$ and $\argxProxy \in \xProxySpace$, and $\e[L_{jj'}(\goodX)+\tilde{L}_{jj'}(\proxyX)] <\infty$. 
    \item For each $j,j' \in [d]$, $\e \big[ \big( [\ddot{l}_{\thetaTarg}(\goodX)]_{jj'} \big)^2 + \big( [\dot{l}_{\thetaTarg}(\goodX)]_{j} \big)^4 \big] < \infty$ 
    and $\e \big[ \big( [\dot{l}_{\gammaTarg}(\proxyX)]_{j} \big)^4 \big] < \infty$.
    
\end{enumerate}
\end{assumption} The above assumptions are all smoothness and bounded moment conditions on the first and second derivatives of $\theta \mapsto l_{\theta}(X)$ and $\theta \mapsto l_{\theta}(\tilde{X})$ in neighborhoods of $\thetaTarg$ and $\gammaTarg$, respectively. In Appendix \ref{sec:ConsistencyOfCovarianceMatrixComponents}, we show that they are sufficient (alongside Assumptions \ref{assump:IIDUnderlyingData}--\ref{assump:CLTRgularityConditions}) for proving $(\hat{\Sigma}_{11}, \hat{\Sigma}_{12}, \hat{\Sigma}_{22}, \hat{\Sigma}_{13}, \hat{\Sigma}_{33}, \hat{H}_{\thetaTarg},\hat{H}_{\gammaTarg}) \xrightarrow{p} (\Sigma_{11}, \Sigma_{12}, \Sigma_{22}, \Sigma_{13}, \Sigma_{33}, H_{\thetaTarg},H_{\gammaTarg})$ as $N \to \infty$. 


To be in the setting of Theorem \ref{theorem:CLT_PTDEstimator}, we need to choose a tuning matrix $\hat{\Omega}$ such that $\hat{\Omega} \xrightarrow{p} \Omega$ for some $\Omega \in \mathbb{R}^{d \times d}$. 
A fairly general way to choose a converging tuning matrix is to take it to be some function of the empirical matrices defined at \eqref{eq:HessianAndSigmaEstimators}:
\begin{equation}\label{eq:HatOmegaContinuousFunctionChoice}
    \hat{\Omega} = f \big(\hat{\Sigma}_{11}, \hat{\Sigma}_{12}, \hat{\Sigma}_{22}, \hat{\Sigma}_{13}, \hat{\Sigma}_{33}, \hat{H}_{\thetaTarg},\hat{H}_{\gammaTarg} \big) \quad \text{ for some } f : (\mathbb{R}^{d \times d})^7 \to \mathbb{R}^{d \times d}.
 \end{equation} As an example, our simulations use the optimal diagonal tuning matrix (derived in Appendix \ref{sec:AsymptoticallyOptimalDiagTuningMatrix}) whose $j$th diagonal entry is given by $[ \hat{H}_{\thetaTarg}^{-1} (\hat{\Sigma}_{12}-\hat{\Sigma}_{13}) \hat{H}_{\gammaTarg}^{-1} ]_{jj}/[ \hat{H}_{\gammaTarg}^{-1} ( \hat{\Sigma}_{22}-\hat{\Sigma}_{33} ) \hat{H}_{\gammaTarg}^{-1} ]_{jj}$ for each $j \in [d]$. 

 
 We combine Theorem \ref{theorem:CLT_PTDEstimator} and a consistency result for $(\hat{\Sigma}_{11}, \hat{\Sigma}_{12}, \hat{\Sigma}_{22}, \hat{\Sigma}_{13}, \hat{\Sigma}_{33}, \hat{H}_{\thetaTarg},\hat{H}_{\gammaTarg})$ to get the following result, which establishes asymptotically valid confidence intervals for $\thetaTarg$.


\begin{proposition}\label{prop:AsymptoticallyValidCIs}
    Suppose the data are collected via \samplingSchemeName{,} that Assumptions \ref{assump:IIDUnderlyingData}--\ref{assump:SmoothEnoughForConsitentVarEst} hold and that $\htPTD$ is tuned using a tuning matrix $\hat{\Omega}$ given by \eqref{eq:HatOmegaContinuousFunctionChoice} for some $f: (\mathbb{R}^{d \times d})^7 \to \mathbb{R}^{d \times d}$ that is continuous at $(\Sigma_{11}, \Sigma_{12}, \Sigma_{22}, \Sigma_{13}, \Sigma_{33}, H_{\thetaTarg},H_{\gammaTarg} )$. Then, for some matrix $\Sigma_f^{\PTDSuperScriptAcr} \in \mathbb{R}^{d \times d}$, $\hat{\Sigma}^{\PTDSuperScriptAcr} \xrightarrow{p}  \Sigma_f^{\PTDSuperScriptAcr}$ as $N \to \infty$, where $\hat{\Sigma}^{\PTDSuperScriptAcr}$ is defined at Equation \eqref{eq:PTDAsympVarEstimator}.
     Moreover, for each $\alpha \in (0,1)$ and each $j \in [d]$ such that $[\Sigma_f^{\PTDSuperScriptAcr}]_{jj} \neq 0$, $\lim_{N \to \infty} \mathbb{P} \big( [\thetaTarg]_j \in \mathcal{C}_j^{(1-\alpha)} \big) = 1-\alpha,$ where $\mathcal{C}_j^{(1-\alpha)}$ denotes the $(1-\alpha)$-confidence interval for the $j$th component of $\thetaTarg$ defined at \eqref{eq:CIDef}. 
 
\end{proposition}

\section{Choosing sampling rules to increase efficiency}\label{sec:ChoicesToIncreaseEfficiency}

The appeal of two-phase multiwave sampling is that we can choose to sample points that are most informative, increasing efficiency. In particular, sampling in multiple waves allows us to adaptively update our sampling approach as we acquire more data and better understand which future data is likely to be most beneficial.
We now turn our attention to labelling strategies $\LabelStrategy{k}$ for each $k \in [K]$, discussing choices that lead to improved efficiency. 

\subsection{An approximate greedy optimal strategy}\label{sec:DescribeApproximateGreedyOptimalStrategy}

In this subsection, we present a strategy for sampling in each wave that is designed to increase precision. We focus on settings where the investigator is primarily interested in a design that will lead to narrow confidence intervals for the $j$th component of $\thetaTarg$ for some fixed $j \in [d]$. (Other objectives can be considered in our framework but are omitted due to space constraints; see \cite{Recent2PhaseMultiwaveGeneralizedRaking} and \cite{HarrisonBatchAdaptiveDMLCausal} for A-optimality and more general objectives in related settings).  
For simplicity, we develop efficient designs for the case where $c_1,\dots,c_K$ are fixed and prespecified and the tuning matrix $\hat{\Omega}=I_{d \times d}$. 

Fix $j \in [d]$. With the choice of tuning matrix $\hat{\Omega}=I_{d \times d}$, we consider labelling strategies for minimizing the asymptotic variance of $\htPTD_j$ in a greedy manner. We start with some sampling rule for the first wave $\LabelRulePiDoubleArg{1}{0}(\cdot)$, which does not depend on the Phase II data. 
Next, for each fixed $k^* \in [K] \setminus \{1\}$ we define the greedy optimal sampling rule recursively. To do this, note that in the setting of Theorem \ref{theorem:CLT_PTDEstimator}, the asymptotic variance of a version of the estimator $\htPTD_j$ that \textit{only uses the first $k^*$ waves} of data is given by \begin{equation}\label{eq:AsympVarUntunedFirstKstarWaves} \big[\Sigma^{\PTDSuperScriptAcr,(1:k^*)} (I_{d \times d}) \big]_{jj} = \kappa_1^2 \cdot \e \Bigg[ \frac{\big( \big[  H_{\thetaTarg}^{-1} \dot{l}_{\thetaTarg}(\goodX) - H_{\gammaTarg}^{-1} \dot{l}_{\gammaTarg}(\proxyX) \big]_j \big)^2 }{ \LimitingLabelRulePi{k^*}(\proxyX) \prod_{k=1}^{k^*-1} \big(1-\LimitingLabelRulePi{k}(\proxyX) \big)
} \Bigg] + \kappa_2,\end{equation} where $\kappa_1$ and $\kappa_2$ are quantities that do not depend on the function $\LimitingLabelRulePi{k^*}(\cdot)$ (see Appendix \ref{sec:FormulaForGreedyOptimalObjective} for details). After wave $k^*-1$ and prior to wave $k^*$, the labelling strategy $\LabelStrategy{k^*}$ (i.e., the approach for choosing the labelling rule $\LabelRulePi{k^*}$) will impact $\LimitingLabelRulePi{k^*}$ but not $\kappa_1, \kappa_2$ or $\{\LimitingLabelRulePi{k}  \}_{k=1}^{k^*-1}$. As a greedy approach, suppose the investigator chooses a labelling rule $\LabelRulePi{k^*}$ with the goal of approximating the function $\LimitingLabelRulePi{k^*}$ that minimizes the above asymptotic variance expression, subject to budget and overlap constraints. Defining $\varrho_j : \mathbb{R}^p \to [0,\infty)$ given by \begin{equation}\label{eq:varrhoDef}
    \varrho_j(\proxyX) \equiv \e \Big[ \big( \big[  H_{\thetaTarg}^{-1} \dot{l}_{\thetaTarg}(\goodX) - H_{\gammaTarg}^{-1} \dot{l}_{\gammaTarg}(\proxyX)  \big]_j \big)^2  \Big| \proxyX \Big],
\end{equation} and fixing $b_l$, $b_u$, and $B^{(k^*)}$ such that $0<b_l < B^{(k^*)} < b_u < 1$ to be overlap and budget constraint constants, 
the function $\bar{\pi}^{(k^*)}: \mathbb{R}^p \to [b_l,b_u]$ minimizing the asymptotic variance at \eqref{eq:AsympVarUntunedFirstKstarWaves}, subject to the budget constraint $\e[ \bar{\pi}^{(k^*)}(\proxyX)] \leq B^{(k^*)}$, satisfies \begin{equation}\label{eq:OptimalGreedySolutionWithBounds}   
    \bar{\pi}_{\text{opt}}^{(k^*)} (\proxyX)  = \min \Biggl\{ b_u, \max \Bigl\{ b_l, \sqrt{\frac{\varrho_j(\proxyX)}{ \lambda^{(k^*)} \prod_{k=1}^{k^*-1} \big(1-\LimitingLabelRulePi{k}(\proxyX) \big)} } \Bigr\} \Biggr\} , \  \text{ for some }   \lambda^{(k^*)}>0.
\end{equation} See Appendix \ref{sec:OptimizationProblem} for a derivation. 

In our numerical experiments, we use a labelling rule that approximates $\bar{\pi}_{\text{opt}}^{(k^*)}$, given by \begin{equation}\label{eq:EstimatedGreedyOptLabellingRule} \LabelRulePi{k^*}(\argxProxy) =  \min \Biggl\{ b_u, \max \Bigl\{ b_l, \sqrt{\frac{\hat{\varrho}^{(k^*)}_j(\argxProxy)}{ \hat{\lambda}^{(k^*)} \prod_{k=1}^{k^*-1} \big(1-\LabelRulePi{k}(\argxProxy) \big)} } \Bigr\} \Biggr\} \quad \text{ for each } \argxProxy \in \xProxySpace,\end{equation} where $\hat{\varrho}^{(k^*)}_j(\cdot)$ is some estimate for the function $\varrho_j(\cdot)$ that is learned using the data $\cd_{k^*-1}$ and $\hat{\lambda}^{(k^*)}$ is computed numerically to enforce the budget constraint. (In the next subsection, we discuss approaches for estimating $\varrho_j(\cdot)$, and in Appendix \ref{sec:EnforceBudgetConstraints} we describe how we calculate $\hat{\lambda}^{(k^*)}$ numerically in our experiments).


For $k^*=1$, direct estimation of the function $\varrho_j(\cdot)$ at \eqref{eq:varrhoDef} is not feasible prior to the first wave of Phase II because of lack of $\xmiss$ observations. Instead, for the first wave we propose setting $ \LabelRulePiDoubleArg{1}{0}(\argxProxy) =n_{\text{targ}}^{(1)}/N$ for each $\argxProxy \in \xProxySpace$ so that the first stage is a uniform random sample that broadly explores the covariate space. We remark that strategic choices for the initial (first wave) labelling rule that leverage both $\cd_0$ and prior information have been considered elsewhere in the literature \citep[e.g.,][]{ChenLumleyMultiwaveInTwoPhase,SkylerYashEmmanuelRecentArXiv_LLMEvalMeans}. 

\subsection{Plug-in estimates}

We next discuss a way to estimate $\varrho_j(\cdot)$ after each wave of Phase II. Throughout this subsection we will fix $k^* \in [K] \setminus \{1\}$, and discuss estimating $\varrho_j(\cdot)$ after wave $k^*-1$ has been completed.  
While estimates of $\gammaTarg$ and $H_{\gammaTarg}$ are available from the Phase I data (see $\hga$ and $\hat{H}_{\gammaTarg}$ defined at \eqref{eq:PTDComponentEstimatorsDEF} and \eqref{eq:HessianAndSigmaEstimators}), $\thetaTarg$ and $H_{\thetaTarg}$ are not known precisely and must be estimated from samples collected in earlier waves of Phase II. To do this define \begin{equation}\label{eq:WeightsAggFirstKstar}\cw_i^{(k^*-1)} \equiv \frac{1}{ \sum_{k'=1}^{k^*-1} c_{k'} }  \sum_{k=1}^{k^*-1} c_k W_i^{(k)} \quad \text{for each } i \in [N]\end{equation}
to be the multiwave inverse probability weights (analogous to those in \eqref{eq:aggregated_Wk}), which only consider the first $k^*-1$ waves. Prior to the start of wave $k^*$ we use estimators of $\thetaTarg$ and $H_{\thetaTarg}$ given by $\hat{\theta}^{\completeSampNotation,(k^*-1)} = \argmin_{\theta \in \Theta} \bigl\{ \sum_{i=1}^N  \cw_i^{(k^*-1)}  l_{\theta}(\goodX_i) \bigr\}$ and $\hat{H}_{\thetaTarg}^{(k^*-1)} = N^{-1} \sum_{i=1}^N  \cw_i^{(k^*-1)} \ddot{l}_{\hat{\theta}^{\completeSampNotation,(k^*-1)}}(\goodX_i)$. Under $k^*-1$ Phase II waves, these estimators correspond to the consistent estimators $\htc$ and $\hat{H}_{\thetaTarg}$ studied in Sections \ref{sec:settingAndEstimator} and \ref{sec:AsympTheory}. Thus, defining  $$\cy_i^{(k^*-1)} = \Big( \Big[  (\hat{H}_{\thetaTarg}^{(k^*-1)})^{-1} \dot{l}_{\hat{\theta}^{\completeSampNotation,(k^*-1)}}(\goodX_i) - \hat{H}_{\gammaTarg}^{-1} \dot{l}_{\hga}(\proxyX_i) \Big]_j \Big)^2 \quad \text{for} \quad i \in [N],$$ note that $\cy_i^{(k^*-1)}$ approximates the quantity inside the conditional expectation in the formula for $\varrho_j$ (see \eqref{eq:varrhoDef}) and that $\cy_i^{(k^*-1)}$ can be evaluated for each $i \in \mathcal{I}^{(k^*-1)}$ where $\mathcal{I}^{(k^*-1)} \subseteq [N]$ denotes the samples in which the label $\xmiss_i$ was collected in a wave prior to wave $k^*$. 

We now consider some approaches to estimating $\varrho_j(\cdot)$.

\subsubsection{Estimating optimal labelling rule with machine learning}\label{sec:EstimateOptimalStrategyWithML}

First, we use machine learning to estimate the function $\varrho_j(\cdot)$, which is defined as a conditional mean. Specifically, we train a machine learning model on the sample $(\cy_i^{(k^*-1)}, \proxyX_i)_{i \in \mathcal{I}^{(k^*-1)}}$ to predict $\cy_i^{(k^*-1)}$ values from $\proxyX_i$ values, resulting in a learned function that we define as $\hat{\varrho}^{(k^*)}_j(\cdot)$. In this manuscript, we consider using k-nearest neighbors to learn a function that predicts $\cy_i^{(k^*-1)}$ from $\proxyX_i$, since it is a universal function approximator and our theory allows for the use of methods even with a slow rate of convergence.

\subsubsection{Estimating the greedy optimal labelling probabilities in each stratum}\label{sec:EstimateOptimalLabellingProbabilityInEachStrata}

The second approach we consider is to stratify the space $\xProxySpace$ and assign our estimate of $\varrho_j(\argxProxy)$ to be constant within each stratum. In particular,
partition $\xProxySpace$ into $L$ prespecified, disjoint strata $\mathcal{S}_1,\dots,\mathcal{S}_L$. We use a coarse estimator for $\varrho_j(\cdot)$ given by the piecewise function $$\hat{\varrho}_j^{(\text{strat},k^*)}(\argxProxy) = \sum_{r=1}^L \mathbbm{1} \{ \argxProxy \in \mathcal{S}_r \} \cdot \Bigg(\frac{ N^{-1} \sum_{i=1}^N \cw_i^{(k^*-1)} \cdot \mathbbm{1} \{ \proxyX_i \in \mathcal{S}_r \} \cdot  \cy_i^{(k^*-1)}   }{ N^{-1} \sum_{i=1}^N \mathbbm{1} \{ \proxyX_i \in \mathcal{S}_r \} } \Bigg) \quad \text{for } \argxProxy \in \xProxySpace,$$ that is derived in Appendix \ref{sec:FormulaForEstimatedStrataSpecificLabellingRule}. In Appendix \ref{sec:ProofOfStrataLabRuleConvergence}, we show that when using the estimated greedy optimal labelling approach described in Section \ref{sec:DescribeApproximateGreedyOptimalStrategy} and plugging in this coarse estimator of $\varrho_j(\cdot)$, the labelling rules for each wave converge in $L^1$ (and hence satisfy Assumption \ref{assump:CLTRgularityConditions}(ii)) under suitable conditions. Using $\hat{\varrho}_j^{(\text{strat},k^*)}$ has parallels to a strategy sometimes seen in the two-phase multiwave sampling literature in which strata are preselected and the number of samples in each stratum is determined using Neyman allocation to minimize the variance of estimated influence function terms~\citep[e.g.,][]{ChenLumleyMultiwaveInTwoPhase}.

\subsubsection{Quantile-specific labelling strategies}

For quantile estimation tasks, estimating $\varrho_j(\cdot)$ reduces to estimating the function $\mathbb{P} ( \goodX \leq \thetaTarg \giv \proxyX = \cdot)$. In Appendix \ref{sec:QuantileEstimationImplementationApproaches} (particularly, Appendix \ref{sec:EstimateQuantileSpecificGreedyOpt}), we develop labelling strategies that exploit this fact and involve estimating $\mathbb{P} ( \goodX \leq \thetaTarg \giv \proxyX = \cdot)$ with either isotonic regression or with a bivariate Gaussian copula model.

\section{Numerical experiments}\label{sec:Simulations}

In this section we present five experiments using three different datasets and one synthetic dataset to validate the coverage guarantees in Proposition \ref{prop:AsymptoticallyValidCIs} and to empirically study the efficiency gains that can be obtained by the strategies discussed in Section \ref{sec:ChoicesToIncreaseEfficiency}. All experiments focus on M-estimation tasks: either the estimand is a quantile or a population regression coefficient from a linear or logistic regression model. In all experiments we consider the efficiency gains from the approximately greedy optimal labelling strategy discussed in Section \ref{sec:ChoicesToIncreaseEfficiency}.  We compare to the baseline of a single-wave Predict-Then-Debias approach where Phase II is a uniform random sample. In these experiments, we present the results using two different approaches for estimating $\varrho_j(\cdot)$ defined at \eqref{eq:varrhoDef}: nearest-neighbors estimation (Section~\ref{sec:EstimateOptimalStrategyWithML}) and stratification (Section~\ref{sec:EstimateOptimalLabellingProbabilityInEachStrata}). We also vary the number of waves $K$. We describe the experimental setup and datasets in more detail next.

\subsection{Overview of experiments}

In each experiment, we treat the empirical distribution of a large fully observed dataset of size $\nSuper$ as the superpopulation distribution $\datvecrawDist$. The target parameter $\thetaTarg$ is defined as the empirical risk minimizer computed using all $\nSuper$ observations.
For each experiment and for each $K \in \{2,6,26\}$, we conduct $1{,}000$ Monte Carlo simulations of \samplingSchemeName{}. In each Monte Carlo simulation, a new Phase I sample of size $N$ is drawn with replacement from the $\nSuper$ observations. Phase II uses a total expected labelling budget of $n_{\text{targ}}$. The first (“explore”) wave uses i.i.d. Bernoulli sampling with expected size $n_{\text{targ}}^{(1)}$. The remaining budget is split evenly across the subsequent $K-1$ adaptive waves: for $k \geq 2$, the expected number of labels collected in wave $k$ is set to
$n_{\text{targ}}^{(k)} = ( n_{\text{targ}} - n_{\text{targ}}^{(1)})/(K-1).$  For $k \in [K]$, we set $c_k = n_{\text{targ}}^{(k)}/ \big(\sum_{k'=1}^K n_{\text{targ}}^{(k')} \big)$ (see Appendix \ref{sec:ChoiceOfCk} for discussion on the choice of $c_k$).

After Phase II, we compute our estimator $\htPTD$ and construct 90\% confidence intervals using \eqref{eq:CIDef}. We compare results from adaptive sampling to a baseline in which Phase II consists of an i.i.d. Bernoulli sample with the same expected total labelling budget $n_{\text{targ}}$. In the baseline, the Predict-Then-Debias estimator is used. In all cases, we use the asymptotically optimal tuning matrix among the class of diagonal matrices,
derived in Appendix \ref{sec:AsymptoticallyOptimalDiagTuningMatrix}.
Table \ref{table:ExperimentSummary} summarizes the datasets, estimands, and values of $\nSuper$, $N$, $n_{\text{targ}}$, and $n_{\text{targ}}^{(1)}/n_{\text{targ}}$ used in each experiment. Additional implementation details are provided in Appendices \ref{sec:MonteCarloProcedure} and \ref{sec:DatasetsAdditionalDetails}.

\paragraph{Evaluation metrics.}
For each method and value of $K$, we summarize performance across the $1{,}000$ Monte Carlo simulations using three metrics. First, we compute the root mean squared error (RMSE) of the estimator relative to the superpopulation value $\thetaTarg$. Second, we report empirical coverage of the nominal 90\% confidence intervals. Third, to quantify efficiency gains relative to uniform sampling with the same expected number of labels $n_{\text{targ}}$, we compute an effective sample size ratio based on squared confidence interval width ratios, adjusting for stochastic differences in the number of labelled samples. Results aggregated across simulations are in Figure \ref{fig:ResultsAggregated}. To assess the normal approximation in Theorem \ref{theorem:CLT_PTDEstimator}, Figure \ref{fig:TestNormality} displays the empirical distribution of $\htPTD$ across simulations. In experiments with $d>1$, evaluation metrics focus on the component of interest indicated by Column 3 in Table \ref{table:ExperimentSummary}.

\begin{table}[hbt!]
\caption{\small Summary of experiments. The penultimate column gives the expected number of samples collected in Phase II. The final column gives the proportion of the labelling budget that was allocated to the first wave, with the remaining budget split evenly among the subsequent waves.
}
\label{table:ExperimentSummary}
\centering
\scriptsize
\begin{tabular}{c c c c c c r r c} 
\toprule
Exp \#  & Dataset & Estimand & $\nSuper$ & $N$   & $n_{\text{targ}}$ & $n_{\text{targ}}^{(1)}/n_{\text{targ}}$ \\ 
\midrule
1  & Synthetic &  Linear regression coeff. for treatment & $3 \times 10^6$ &  $20{,}000$ & $2{,}000$ &  1/4 \\ [1ex] 
2   & Housing Price &  Linear regression coeff. for nightlight & $46{,}418$ &  $10{,}000$ & $1{,}000$ & 1/4 \\ [1ex] 
3  & AlphaFold &  Logistic regression coeff. for interaction & $10{,}802$ &  $20{,}000$ & $5{,}000$ & 1/3 \\ [1ex] 
4   & Tree cover &  Logistic regression coeff. for population & $67{,}968$  & $10{,}000$ & $1{,}000$ &   1/4 \\ [1ex] 
5  & Tree cover &  $0.75$th quantile of tree cover & $67{,}968$ & $10{,}000$ & $1{,}000$ & 1/4 
\\ [1ex] 
 6 & Maize yield$^\dag$ &  $0.95$th quantile of maize yield & $4{,}173$ & $10{,}000$ & $500$ & 1/4  \\ [1ex] 
 \bottomrule 
\end{tabular}
\\ [1ex]
 \dag Deferred to Section \ref{sec:AttainableYieldSims} where the scientific context of this experiment is discussed.
\smallskip

\end{table}

\subsection{Experiments and Datasets}\label{sec:Datasets}

We consider one synthetic experiment and four data-based experiments spanning linear regression, logistic regression, and quantile M-estimation tasks. Variables not explicitly stated to be measured in Phase II are simulated to be observed in Phase I. Additional details on the datasets and the construction of Phase I strata are provided in Appendix \ref{sec:DatasetsAdditionalDetails}.


\paragraph{Synthetic experiment.}
We generated a large synthetic superpopulation of size $\nSuper = 3 \times 10^6$ consisting of an outcome $Y$, a continuous covariate $Z_{\text{cov}}$, a binary treatment $Z_{\text{trt}}$, and an error-prone proxy $\tilde{Z}_{\text{trt}}$. The estimand is the population regression coefficient for $Z_{\text{trt}}$ in a linear regression of $Y$ on $(Z_{\text{cov}}, Z_{\text{trt}})$. The data-generating process was constructed so that (i) regression residuals were heteroskedastic with variance depending strongly on the Phase I covariates, and (ii) prediction errors in $\tilde{Z}_{\text{trt}}$ were differential (i.e., $Z_{\text{trt}} \centernot{\indep} Y \giv (Z_{\text{cov}}, \tilde{Z}_{\text{trt}})$). In Phase I, only $(Y, Z_{\text{cov}}, \tilde{Z}_{\text{trt}})$ were observed, while $Z_{\text{trt}}$ was measured in Phase II.


\paragraph{Housing price experiment.}
We used a dataset of $46{,}418$ grid cells containing housing price, income, nightlight intensity, and road length from \cite{MOSAIKSPaper}. The estimand is the population regression coefficient for nightlight intensity in a linear regression of housing price on income, nightlight intensity, and road length. In Phase I, proxy measures of nightlight intensity and road length derived from daytime satellite imagery were observed; corresponding gold-standard measurements were only collected in Phase II.  


\paragraph{AlphaFold experiment.}
We used $10{,}802$ protein regions with indicators for acetylation and ubiquitination and an outcome indicating whether the region is intrinsically disordered \citep{bludau2022structural}. The estimand is the interaction coefficient in a logistic regression of disorder status on acetylation and ubiquitination. In Phase I, predictions of disorder status were observed; corresponding gold-standard labels were only collected in Phase II. 


\paragraph{Forest cover experiment.}
We used $67{,}968$ grid cells containing percent tree cover, population, and elevation taken from the previously mentioned data source \citep{MOSAIKSPaper}. The estimand is the logistic regression coefficient for population when regressing a binary forest cover indicator (based on a 10\% tree cover threshold) on elevation and population. In Phase I, proxy measures of forest cover and population were observed; corresponding gold-standard measurements were only collected in Phase II. 


\paragraph{Tree cover quantile experiment.}
Using the same dataset, we estimate the 0.75-quantile of percent tree cover. In Phase I, only proxy predictions of tree cover were observed; gold-standard measurements were collected in Phase II. 
Quantile estimation required a different approach for estimating the Hessians than direct use of Equation \eqref{eq:HessianAndSigmaEstimators} (see Appendix \ref{sec:EstimatingHessiansForQuantileEst}).

\subsection{Results}

\begin{figure}[!hbt]
    \centering
    \includegraphics[width=0.975\linewidth]{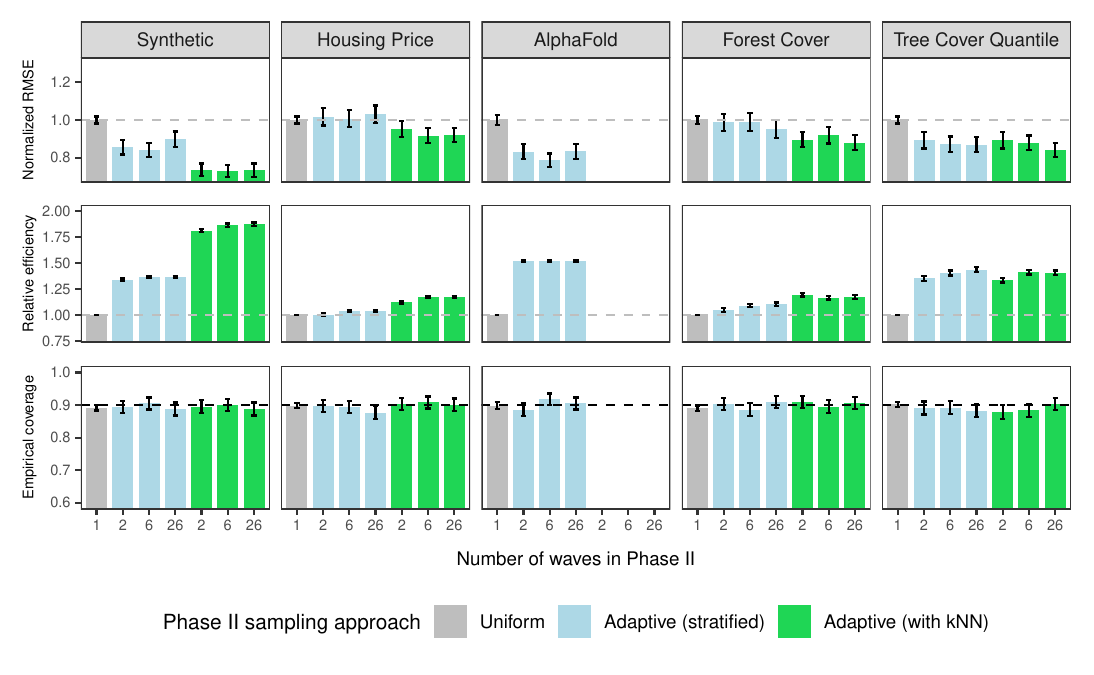}
    \caption{\small Comparison of two-phase multiwave sampling strategies. The baseline strategy (grey) involves one wave of uniform random sampling and is compared to adaptive sampling with either $2$, $6$, or $26$ waves in Phase II and with either the stratified approach described in Section \ref{sec:EstimateOptimalLabellingProbabilityInEachStrata} (blue) or a kNN-based approach described in Section \ref{sec:EstimateOptimalStrategyWithML} (green) for approximating a greedy optimal labelling rule. For all sampling strategies considered, the (Multiwave) Predict-Then-Debias estimator with the optimal diagonal tuning matrix (Appendix \ref{sec:AsymptoticallyOptimalDiagTuningMatrix}) was used. Each column corresponds to a different experiment. The first row shows the RMSEs calculated across the $1{,}000$ simulations after being rescaled by a constant (separately for each experiment) so that the uniform sampling baseline would have an RMSE of $1$. The second row shows the efficiency relative to the uniform sampling baseline averaged across the $1{,}000$ simulations. The third row gives the empirical coverage of the 90\% confidence intervals across the $1{,}000$ simulations with the dashed line giving the nominal coverage. The error bars give $\pm 2$ standard errors for the evaluation metric being plotted. Note that the uniform sampling baseline has smaller standard errors because $6{,}000$ simulations of the baseline were conducted. In the AlphaFold experiment, we do not consider kNN approaches due to the feature space being discrete.
    }\label{fig:ResultsAggregated}
\end{figure}

Across experiments, adaptive implementations of \samplingSchemeName{} consistently reduced RMSE relative to uniform Phase II sampling and yielded effective sample size gains of up to approximately 1.8 (Figure \ref{fig:ResultsAggregated}). In most settings, the machine learning-based approximation to the greedy optimal labelling rule outperformed the stratified approximation, reflecting information loss induced by coarse or suboptimal discretizations of $\xProxySpace$. The primary exception was the AlphaFold experiment, where $\xProxySpace$ is discrete and stratification does not incur approximation error. In the tree cover quantile experiment, stratification also performed comparably to the machine learning approach, as the Phase I variables were bounded and one-dimensional.

Increasing the number of waves from $2$ to $6$ yielded modest efficiency gains in some experiments, while further increasing to $26$ waves provided no discernible additional benefit. This pattern is consistent with prior findings \citep{YangDiaoCook22} suggesting limited returns beyond a small number of adaptive waves. 

Efficiency gains were largest in the synthetic, AlphaFold interaction, and tree cover quantile experiments. In the synthetic experiment, heteroskedastic residual variance and differential proxy error induced substantial heterogeneity in the informativeness of the samples, favoring adaptive labelling schemes. In the AlphaFold experiment, imbalance in one binary covariate made weighted sampling particularly beneficial for estimating the interaction coefficient. In the quantile estimation task, adaptive prioritization of observations likely to lie near the target quantile could have substantially improved precision.

Across experiments and adaptive sampling strategies considered, empirical coverage was close to the nominal coverage of 90\% (Figure \ref{fig:ResultsAggregated}). 
Consistent with Theorem \ref{theorem:CLT_PTDEstimator}, the sampling distributions of $\htPTD$ were approximately normal across experiments (Figure \ref{fig:TestNormality}). 

\subsection{Comparison to other estimators and sampling schemes}

\paragraph{Comparison to raking estimators.} In the four experiments involving linear or logistic regression, we compared our estimator $\htPTD$ to generalized raking estimators computed on data from the same adaptive sampling scheme (Section \ref{sec:EmpiricalComparisonToGeneralizedRaking}; Figure \ref{fig:ComparisonToRaking}). Raking estimators tended to be more efficient than ours, although the efficiency gap disappeared when using a different tuning matrix $\hat{\Omega}$ (Equation \eqref{eq:EstimatedFullOptimalTuningMatrix}) and when restricting our attention to raking estimators with known theoretical guarantees of asymptotic normality under adaptive two-phase multiwave sampling (Theorem \ref{theorem:asymptoticNormalityRaking}).

\paragraph{Comparison to sampling schemes that do not induce statistical dependencies.} We modified the experiments involving data from \cite{MOSAIKSPaper} and considered three additional sampling schemes as baselines (Section \ref{sec:ComparisonToSchemesThatDoNotInduceDependency}). In one baseline sampling scheme, Phase II is a single wave of independent Bernoulli draws, with sampling probabilities selected using an estimate of $\varrho_j(\cdot)$ based on already-available external data (in our simulations the external data was from a different geographical region than the target region of interest). Using an uncertainty model fit on the same external data, we also considered the Batch Active Inference \citep{ActiveInferencePaper} approach as a baseline in the subset of the modified experiments where it could be readily implemented. In the third baseline sampling scheme, an internal pilot sample was collected and was used to determine the sampling design but was discarded for the purposes of constructing the estimator $\htPTD$ and its confidence intervals. These baseline sampling schemes avoid inducing dependencies between the Phase II sampling indicators; however, in our numerical experiments they were typically less efficient than two-phase multiwave sampling (Figures \ref{fig:ComparisonToBaselineSamplingSchemes} and \ref{fig:ComparisonToBatchActiveInference}). 



\section{Estimating attainable maize yields}\label{sec:AttainableYieldSims}

Attainable yield is a metric that agronomists use to understand the productivity of the best-performing farms in a given region and to inform policy decisions. Specifically, the attainable yield for a given region, crop type, and time period is the 95th percentile of crop yield among farms growing that crop in the region and time period of interest. The attainable yield is sometimes estimated using crop yield estimates from satellite imagery \citep{LobellYieldGapWithSatellite} or using publicly available mean crop yield datasets at county or state resolution \citep{AttainableYieldUsesQuantiles,AttainableYieldAndYieldGapUsesQuantileReg}. However, taking the 95th percentile of either type of data could give a biased estimate of the attainable yield, and statistical approaches for estimating the attainable yield have been questioned \citep{NatureFoodCritiquesSatisticalApproachesForYieldGapEst}. A reliable alternative approach would be to survey farmers to collect ground truth data, but such surveys can be time consuming or expensive, and estimating the 95th percentile with high confidence would require a fairly large survey. In principle, the cost of such surveys could be reduced using a \samplingSchemeName{} scheme, so we investigated the efficiency gains of such an approach using a real dataset.

We downloaded maize yield measurements paired with image-based estimates (from a convolutional neural network) of maize yield from $\nSuper=4{,}173$ fields in Zimbabwe from \cite{HarrisonCropYieldPPI}. We conducted $1{,}000$ Monte Carlo simulations in which we compared a two-phase sampling scheme where the Phase II sample comes from a uniform random sample to our proposed two-phase multiwave sampling scheme (with $K=6$ waves and either the kNN-based Phase II labelling strategy or the preferred quantile-specific labelling strategy identified in the numerical experiments of Section \ref{sec:ComparisonOfQuantileSpecificLabellingStrategies}). The Phase I, Phase II, and Wave 1 budgets in the simulations are shown in Table \ref{table:ExperimentSummary}, and other parameter choices for the two-phase multiwave sampling scheme were the same as those for the previously described experiments. In order to empirically investigate whether our sampling scheme allowing revisiting confers efficiency benefits, we also considered a version of two-phase multiwave sampling where each field is assigned to exactly one wave in which crop yield from that field can be measured, and fields corresponding to an early wave cannot be revisited in later waves. For more details on this experiment see Section \ref{sec:AttainableYieldSimsAdditionalDetails}.

The results (Figure \ref{fig:AttainableYieldSims}), indicate that in attainable yield estimation tasks where Phase I data is an image-based crop yield estimate, using Phase I data alone ($\hga$, which is near the red line) can severely underestimate the attainable yield. Using Phase II data improves the reliability of the estimates; however, collecting Phase II data uniformly at random is notably less efficient than adaptive sampling approaches. Allowing revisiting of fields which had been considered for labelling in prior waves confers modest efficiency advantages. Using a nonuniform first wave (that leverages prior knowledge about the Spearman's correlation coefficient) and a bivariate Gaussian copula model to estimate $\varrho_j(\cdot)$ after each wave was found to result in the most efficient adaptive sampling (Appendix \ref{sec:EstimateQuantileSpecificGreedyOpt}; Figure \ref{fig:QuantileEstimationStrategiesSims}).

\begin{figure}[!hbt]
    \centering
    \includegraphics[width=0.975\linewidth]{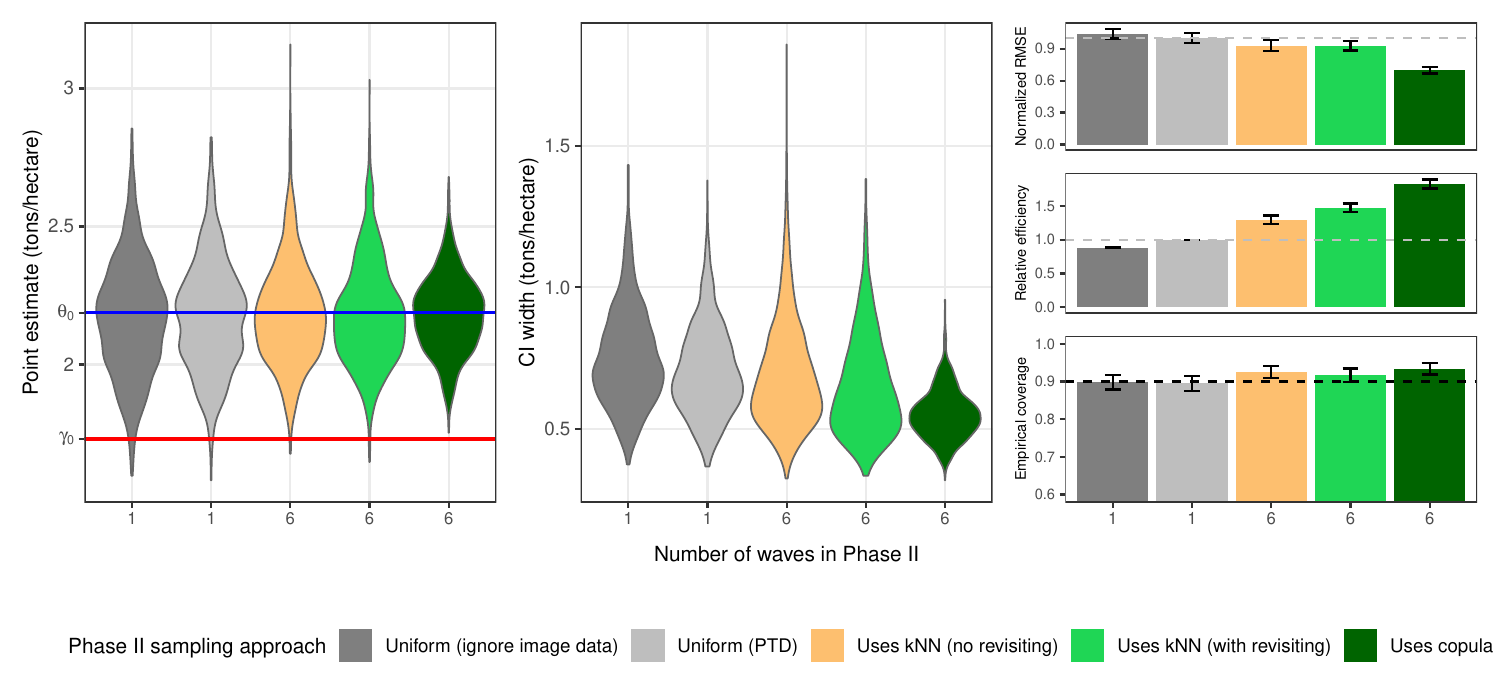}
    \caption{\small Comparison of two-phase sampling strategies for estimating the attainable maize yield. The distribution of point estimates of attainable yield (left) and the width of the corresponding 90\% confidence intervals (middle) across $1{,}000$ simulations for different approaches are depicted with violin plots. The reference lines give $\thetaTarg$ (blue) and $\gammaTarg$ (red), calculated using all available 2024 Zimbabwe data from \cite{HarrisonCropYieldPPI}. Bar plots of evaluation metrics described in Section \ref{sec:Simulations} are shown with error bars that give $\pm$2 standard errors (right). The baseline nonadaptive strategies involve a uniform random sample in Phase II and a point estimator that either ignores the Phase I image-based yield estimates ($\htc$; dark grey) or uses the Predict-Then-Debias estimator from \eqref{eq:PTD_estimator} (light grey). Adaptive sampling strategies that use kNN to estimate $\varrho_j(\cdot)$ and either allow (light green) or do not allow (orange) revisiting samples which had been considered for labelling in prior waves are compared. The adaptive strategy that was found to be the preferred one for quantile estimation tasks in Appendix \ref{sec:ComparisonOfQuantileSpecificLabellingStrategies} (dark green) uses a nonuniform first wave (Appendix \ref{sec:FirstWaveStrategyPrior}) and a bivariate Gaussian copula model to estimate $\varrho_j(\cdot)$ (Appendix \ref{sec:FittedBivariateCopulaStrategy}). We caution that scientific conclusions about the attainable yield in Zimbabwe in 2024 should not be drawn from this figure given aspects of the underlying dataset (e.g., measurement noise and representativeness) that are unaccounted for in these proof-of-concept simulations.}\label{fig:AttainableYieldSims}
\end{figure}

\section{Related work}\label{sec:relatedWork}

This paper is most closely related to the literature on two-phase multiwave sampling. It also closely relates to the Active Inference literature, as our focus is on settings where estimates of expensive-to-measure variables are initially available and gold-standard labels are subsequently collected adaptively. Our work also has some similarities with the literature on adaptive experiments and bandits, where some papers provide theoretical inference guarantees that account for the statistical dependencies induced by these adaptive designs.

\subsection{Two-phase sampling literature}

In two-phase sampling designs, an investigator collects cheap-to-measure variables $X^{\textnormal{(I)}}$ from $N$ samples (or subjects) in Phase I, and then in Phase II measures a more expensive variable $X^{\textnormal{(II)}}$ on a subset of the $N$ samples. Two-phase sampling results in a dataset of the form $(I_i,X_i^{\textnormal{(I)}},I_i X_i^{\textnormal{(II)}})_{i=1}^N$ where $I_i \in \{0,1 \}$ is an indicator of whether a measurement of $X_i^{\textnormal{(II)}}$ was collected in Phase II. In principle, the second phase of sampling can be conducted in order to minimize the variance of an estimator that will ultimately be deployed. However, as we discuss next, a central challenge in the two-phase sampling literature is that the optimal sampling rule is not known a priori and estimating it can induce statistical dependencies across the samples $(I_i,X_i^{\textnormal{(I)}},I_i X_i^{\textnormal{(II)}})_{i=1}^N$ that complicate statistical inference. 

\paragraph{Optimal sampling design in Phase II:}

Theoretically optimal sampling strategies in Phase II have been derived for a variety of estimands~\citep[e.g.,][]{NeymanAllocation1934,ReillyTheoryForOptimalSamplingAssumesKnownParamsOrPilot,GilbertOptimalDesignInTwoPhaseSamplingUsesPilot,OptimalDesign2Phase_TaoZengLin2020JASA,wang2025maximinoptimalapproachsampling}. 
While these works derive theoretically optimal Phase II sampling schemes, a key practical challenge is that the formulas for optimal Phase II sampling probabilities include unknown parameters (sometimes infinite-dimensional ones) that can only be estimated using jointly measured $X^{\textnormal{(I)}}$ and $X^{\textnormal{(II)}}$ data. To address this issue, some papers advocate that investigators estimate these unknown parameters using historical data or a small pilot sample of $(X^{\textnormal{(I)}},X^{\textnormal{(II)}})$ data that is used for determining an optimal Phase II sampling scheme but is discarded for statistical inference purposes (to avoid distribution shift or statistical dependency issues). Other works suggest merging internal pilot data with adaptively collected data from a subsequent optimal sample in order to increase power, while arguing that ignoring the introduced statistical dependencies for simplicity leads to only minor statistical inference issues \citep{InternalPilotThenFuse,wang2025maximinoptimalapproachsampling}. The latter study notes the complications involved in rigorously accounting for i.i.d. violations, stating that it constitutes future work.

\paragraph{Breaking Phase II into multiple waves (two-phase multiwave sampling):} An alternative to using pilot or historical data to estimate an optimal Phase II sampling rule is to break Phase II into multiple waves. Before each wave, all previously collected data is used to estimate unknown parameters in an efficient sampling rule for the upcoming wave (notably, after the first wave, paired internal samples of $(X^{\textnormal{(I)}},X^{\textnormal{(II)}})$ are available for determining sampling probabilities in subsequent waves). This sampling procedure is often referred to as two-phase multiwave sampling. Optimal designs for this sampling procedure have been studied for a variety of two-phase estimators of regression coefficients \citep{AdaptiveSamplingInTwoPhaseDesignsMcIsaacAndCook,ChenLumleyMultiwaveInTwoPhase,ChenAndLumleyPaper2,YangDiaoCook22,Recent2PhaseMultiwaveGeneralizedRaking}. 
Two-phase multiwave sampling designs have been implemented in electronic health record studies \citep{ShepherdTwoPhaseMultiWave} and in a recent R package \citep{RPackagePaperFor2PhaseMultiwave}, suggesting the practical promise and appeal of these designs. Simulations in \cite{ChenLumleyMultiwaveInTwoPhase} and concurrent work \citep{ChenRecentSurrogatePoweredInferenceMultiwave} study two-phase multiwave sampling designs in our motivating setting of interest in which the Phase I data includes an inexpensive proxy (or estimate) of the expensive variable measured in Phase II.

While the two-phase multiwave sampling literature has developed practically appealing approaches and has demonstrated efficiency gains, we note a major theoretical gap in this literature. 
In two-phase multiwave sampling, the resulting sample $(I_i,X_i^{\textnormal{(I)}},I_i X_i^{\textnormal{(II)}})_{i=1}^N$ is not i.i.d. To our knowledge, none of the works in the two-phase multiwave sampling literature provide theoretical guarantees of asymptotic normality or confidence interval validity that explicitly account for statistical dependencies (see Appendix \ref{sec:LitSearchTwoPhaseMultiwaveTheory}). Instead, these works derive optimal sampling strategies under i.i.d. settings, and verify in simulations that the confidence intervals for the resulting weighted estimators attain the nominal coverage. 
We believe adding theoretical guarantees to this literature can broaden the appeal of two-phase multiwave sampling designs and facilitate further research on efficient sampling strategies beyond those already proposed in the two-phase multiwave literature.

\paragraph{Asymptotic theory for dependency in two-phase sampling:}

Some work in the two-phase sampling literature establishes asymptotic theory for M- or Z-estimators that accounts for certain types of statistical dependence between the Phase II inclusion indicators $I_i$, but such works do not apply to our setting. For example, \cite{BreslowAndWellnerZtheoremCorrection} and \cite{SaegusaAndWellner2013AOS} develop asymptotic theory for two-phase stratified sampling designs where in Phase II a fixed number of samples within each stratum are collected randomly without replacement, inducing particular correlation structures between the indicators $I_i$. More generally, \cite{HanAndWellner2021AOS_AsymptoticsForIPWTypeMestimationUnderVariousSamplingDesigns} develops asymptotic theory under a broad class of sampling designs which may have complex dependency structures. They assume that $(I_1,\dots,I_N) \indep  (X_1^{\textnormal{(II)}},\dots,X_N^{\textnormal{(II)}}) \giv (X_1^{\textnormal{(I)}},\dots,X_N^{\textnormal{(I)}})$.
However, the dependency structure induced in two-phase multiwave sampling does not fit this criterion as observations from earlier waves determine data collection in subsequent waves. Meanwhile, \cite{ZhouEtAl14} proves asymptotic normality of regression estimators for a special case of two-phase two-wave sampling designs in which the second wave of Phase II involves collecting a simple random subsample among the samples that are expected to have extreme values of the missing covariate.

\subsection{Prediction-Powered and Active Inference literature}

Driven by the growing use of pretrained machine learning models, a rapidly growing literature is investigating methods for using machine learning predictions to impute missing data to increase the power of statistical analyses, while using complete, labelled samples to maintain reliability. This literature, which we refer to as the Prediction-Powered Inference literature \citep{OriginalPPI}, develops a family of methods that have origins in the semiparametric literature for missing data~\citep[e.g.,][]{RobinsEtAl94JASA,CalibeatingPPI} and the survey sampling literature \citep{sarndal2003model,ChenAndChen2000}. See \cite{DemistifyingPPI} for an overview of the Prediction-Powered Inference literature. We now comment in detail on two strands most directly related to our investigation.

\paragraph{Active Inference literature:} In the Active Inference setting~\citep{ActiveInferencePaper}, an investigator has access to a large unlabelled dataset and a pretrained prediction model.
They have a limited budget for collecting labels and want to do so in a strategic manner. Under independent Bernoulli labelling, a formula for the (approximately) optimal labelling probabilities can be derived as a function of the features; however, these optimal labelling probabilities must be estimated from the data which can induce complex statistical dependencies that complicate inference. 

This literature has considered a number of approaches to circumvent this challenge of statistical dependency. \cite{ActiveInferencePaper} and \cite{TijanaLLMActiveInferencePaper} propose ordering the unlabelled data, deciding whether or not to collect labels for the current sample based on previous samples in the ordering, and conducting inference using the martingale CLT. In contrast to our work, they do not consider settings in which the decision of whether or not to label early samples in the ordering can be revisited. For mean estimation tasks, \cite{MCLTforAdaptiveMeanEstimationAo24}, \cite{hamiltonRecentActiveInferencePaper}, and \cite{SkylerYashEmmanuelRecentArXiv_LLMEvalMeans} develop approaches that allow for revisiting samples with low indices, however, their theoretical guarantees cannot be readily modified for quantile estimation settings (see Remark \ref{remark:TaylorExpansionVersionOfAsymptoticLinearity}).

\cite{ActiveInferencePaper} also considers a setting where optimal labelling probabilities can be roughly estimated using historical data and are not estimated adaptively (preserving statistical independence). In particular, for some settings and estimators of interest, the uncertainty in the machine learning model is the primary unknown quantity that must be estimated from the data. In LLM prediction settings, \cite{StratifiedPPIFisch24} and \cite{ActiveInferenceMixWithUnifromLiZrnicCandes} suggest that estimates of the uncertainty in each prediction can be directly queried from the LLM. 
Focusing on mean estimation tasks and settings with available uncertainty scores from historical models, \cite{chen2025balancedActiveInference} consider balanced sampling constraints, which induce statistical dependencies but allows for further efficiency gains (they also consider more general M-estimation tasks, but do not provide any theory, variance estimates, or confidence intervals for more general M-estimators). 

To our knowledge, the adaptive labelling schemes from this literature that are most similar to our two-phase multiwave sampling setting can be found in \cite{LikeHarrisonsPaperAndCausalButAboutOutcomeCollection} and \cite{ChenRecentSurrogatePoweredInferenceMultiwave}. \cite{LikeHarrisonsPaperAndCausalButAboutOutcomeCollection} develops an approach for estimating average treatment effects in settings where ground truth measurements of the outcome variable are collected adaptively in Phase II. They account for the statistical dependencies by assuming that the labelling probabilities converge to some limiting value at a fast enough rate. \cite{ChenRecentSurrogatePoweredInferenceMultiwave} consider Z-estimation tasks in two-phase multiwave sampling settings where the ground truth labels of a prediction are collected adaptively across multiple waves in Phase II, but they do not present theoretical guarantees under this adaptive sampling regime.

\paragraph{Estimator debiasing approaches:} Much of the Prediction-Powered Inference literature, including the Active Inference literature, focuses on point estimators that minimize a modified, debiased loss function 
\citep{OriginalPPI,PPI++,ActiveInferencePaper}. In this paper, we instead focus on a different class of estimators originating in \cite{ChenAndChen2000}, which involve direct debiasing of the estimators and have been studied in the Prediction-Powered Inference literature \citep{gronsbell2024PromotesCC,PPBootNote,PTDBootstrapPaper,KerriRSEPaper,zou2Witten2026EstimatorDebiasing}. As discussed in \cite{PTDBootstrapPaper}, loss debiasing and estimator debiasing approaches result in fundamentally different estimators (although they align for mean estimation tasks). Estimator debiasing affords more flexibility (loss debiasing sometimes results in nonconvexity challenges
) and the ability to lean on existing statistical software. To our knowledge, within the Active Inference literature, only \cite{ChenRecentSurrogatePoweredInferenceMultiwave} consider estimator debiasing approaches; however, their estimator differs from ours in using different inverse probability weights and a sparse tuning approach to accommodate multiple surrogates.

\subsection{Adaptive experiment and multi-armed bandit literatures}

A large literature has developed inference guarantees for online adaptive experiment and multi-armed bandit settings, where martingale techniques naturally apply (see \cite{KallusReviewInferenceAfterAdaptiveExperiments} for a review). This literature focuses on different estimators and adaptive sampling schemes than we do, and many works in this literature focus on causal estimands \citep[e.g.,][]{Hadad_RenormalizeToEnableMCLT,DesignStabilitySengupta26}. While some works in this literature also focus on inference for Z- or M-estimators with adaptively collected data \citep{zhangJansonMurphyBanditMestimator2021mestimator,ZhangMurphy25StatSci,guo2026StatisticalInferenceForMisspecifiedContextualBandits,leinerRamdas2026MestimatorsAdaptive}, they assume smoother loss functions than we do, and their theoretical guarantees do not extend to quantile estimation tasks. In particular, \cite{zhangJansonMurphyBanditMestimator2021mestimator}, \cite{leinerRamdas2026MestimatorsAdaptive}, and \cite{guo2026StatisticalInferenceForMisspecifiedContextualBandits} each study M-estimation (or Z-estimation) assuming the loss function (or score function) is three (two) times differentiable with a third (second) derivative that is locally dominated by an integrable function. Meanwhile, \cite{ZhangMurphy25StatSci} studies M-estimators for which the loss function is a least squares criterion aggregated across time periods.





Also related to our setting are adaptive experiments with a small number of batches. 
\cite{Hahn2011AdaptiveExperimentalDesign} establish asymptotic normality of causal effect estimators in two-batch adaptive experiments where the propensity scores are a function of covariates that are coarsened to lie in a finite space. \cite{HarrisonBatchAdaptiveDMLCausal} establish asymptotic normality of a more general class of causal estimators in adaptive experiments with multiple batches, while dropping the coarsened propensity score restrictions. Both of these works assume that the propensity scores converge at particular rates (roughly $o_p(N^{-1/2})$ and $O_p(N^{-1/4})$, respectively) while our theoretical results hold under arbitrarily slow convergence of the analogous labelling rule (Assumption \ref{assump:CLTRgularityConditions}(ii)). Other works study IPW estimators for (differences of) means in two-batch adaptive experiments \citep{ZhimeiTwoWaveDOMExperemintWIPW} and regression estimators in learn-as-you-go trials in which (possibly continuous) treatments are adaptively allocated \citep{NevoLAGO_AOS,LAGOBiometricsRecentGLMExtension}. While these works provide inferential guarantees that account for statistical dependencies, none of them establish asymptotic normality for general M-estimation tasks under the batch adaptive sampling design we study.

\section{Discussion}

We developed an approach to M-estimation under two-phase multiwave sampling and provided theoretical guarantees that were further supported by our simulation studies. We discuss some possible directions for further efficiency gains. First, the greedy optimal approaches considered in Section \ref{sec:ChoicesToIncreaseEfficiency} may not be optimal. Second, the number of labels collected in each wave is stochastic. While fixing the number of labels per wave or using stratified sampling would likely increase statistical efficiency \citep[e.g.,][]{JasperOptimalSubsamplingViaStratification}, extensions of our theory are needed in order to provide guarantees of confidence interval validity in such settings. Third, many works in the Active Statistical Inference literature consider the setting where a prediction model can be fine-tuned such that as more $\xmiss$ labels are collected, $\txmiss$ becomes a more accurate estimate of $\xmiss$ \citep[e.g.,][]{ActiveInferencePaper,MCLTforAdaptiveMeanEstimationAo24,hamiltonRecentActiveInferencePaper}. Our theoretical results do not account for the possibility of fine-tuning a prediction model, but in applications where improving the prediction model is possible, doing so can lead to efficiency gains.

\subsubsection*{Acknowledgments}

\noindent D.M.K. was supported by the MIT IDSS Michael Hammer Postdoctoral Fellowship. D.M.K. and S.B. were supported by a research gift from Generali Group through its research partnership with the Laboratory for Information and Decision Systems at MIT. When revising the manuscript, Claude tools were used to prove Lemma F.1 and Fact 5, for certain coding tasks such as organizing the code and generating code to format the figures, and for proofreading. The authors have carefully checked the proofs and code and take full responsibility. The authors wish to thank Alexandra Ferrante, David Lobell, Jae-Kwang Kim, Alessandro Sinai, Sherrie Wang, and Jasper Yang, and anonymous reviewers for helpful conversations and comments.

\subsubsection*{Data and code availability}

\noindent Data and code that can be used to reproduce this study are available at \url{https://github.com/DanKluger/MultiwavePTD}.







\bibliographystyle{apalike}
\bibliography{ActivePTD}

\begingroup
\hypersetup{bookmarksdepth=-2}
\doparttoc 
\faketableofcontents 
\part{} 
\endgroup

\newpage



\appendix

\addcontentsline{toc}{section}{Appendix} 
\part{Appendices} 
\parttoc 
\counterwithin{theorem}{section} 

\renewcommand{\thefigure}{S\arabic{figure}}
\setcounter{figure}{0}
\renewcommand{\thetable}{S\arabic{table}}
\setcounter{table}{0}

\section*{Guide to supplement}
\addcontentsline{toc}{section}{Guide to supplement}

Below we present an overview of the appendices and how they relate to each other. We then state notation conventions that are used throughout the main text and the appendices. We also present a table of notation that can be used as a navigational aid.

\subsection*{Overview of appendices}
\addcontentsline{toc}{subsection}{Overview of appendices}

\paragraph{Theory and proofs.} Appendices \ref{sec:PropertiesOfWeights}--\ref{sec:AppendixProovingConsistentCovarianceAndValidCIs} prove the theoretical results from Section \ref{sec:AsympTheory} of the main text. They establish asymptotic linearity and normality of the Multiwave Predict-Then-Debias estimator $\htPTD$, as well as asymptotic validity of the corresponding confidence intervals. Specifically, the first half of Appendix \ref{sec:PropertiesOfWeights} (through Appendix \ref{sec:WeightPropertiesFirstOrder}) establishes properties of the multiwave inverse probability weights $W_i$ under Assumption \ref{assump:IIDUnderlyingData}, such as their impact on expectations and covariances. These properties are subsequently used in Appendix \ref{sec:AsymptoticsForMestimators} to prove $\sqrt{N}$-consistency of $(\htc,\hgc, \hga)$ and asymptotic linearity of both $(\htc,\hgc, \hga)$ and $\htPTD$ under Assumptions \ref{assump:IIDUnderlyingData}--\ref{assump:SmoothEnoughForAsymptoticLineariaty} (Theorem \ref{theorem:AsymptoticLInearMestsStacked}; Corollary \ref{cor:PTDEstAsymptoticallyLinear}). Under an additional assumption (Assumption \ref{assump:CLTRgularityConditions}), Appendices \ref{sec:ExchangeabilityPropertiesAppendix} and \ref{sec:L1ConvergenceImplicationsAppendix} establish exchangeability properties of the weights and an asymptotic convergence property of the inverse labelling probabilities. These additional properties of the weights are used alongside the asymptotic linearity result (shown in Appendix \ref{sec:AsymptoticsForMestimators}) to prove asymptotic normality of $\htPTD$ (Theorem \ref{theorem:CLT_PTDEstimator}) in Appendix \ref{sec:CLTProof}. The final subsection of Appendix \ref{sec:PropertiesOfWeights} establishes properties of the \textit{squared} multiwave inverse probability weights. These properties of the squared weights are used alongside additional regularity conditions  (Assumption \ref{assump:SmoothEnoughForConsitentVarEst}) to prove that the asymptotic covariance matrix estimators are consistent in Appendix \ref{sec:AppendixProovingConsistentCovarianceAndValidCIs}. Appendix \ref{sec:AppendixProovingConsistentCovarianceAndValidCIs} concludes with a proof that the confidence intervals are asymptotically valid (Proposition \ref{prop:AsymptoticallyValidCIs}), using consistency of the covariance matrix estimators and asymptotic normality (shown in Appendix \ref{sec:CLTProof}). 
Appendix \ref{sec:HowToGeneralizeToTwoPhaseMultiwave} discusses how these results can be used to provide theoretical guarantees for M-estimation tasks in more general two-phase multiwave sampling settings (not just proxy-assisted ones). In that appendix we also provide theoretical guarantees for a generalized raking-type estimator.

\paragraph{Additional details on efficiency, related work, and simulations.} Appendices \ref{sec:EfficientLabellingStrategyAppendix}--\ref{sec:QuantileEstimationImplementationApproaches} give additional details on strategies that were used to increase efficiency and on the implementation of our approach that were omitted from the main text for brevity. In particular, Appendix \ref{sec:EfficientLabellingStrategyAppendix} gives additional details and derivations for the approximate greedy optimal sampling strategy. Appendix \ref{sec:OtherWaysToIncreaseEfficiencyAppendix} discusses choices that can be used to increase efficiency after Phase II is complete and all data has been collected. Appendix \ref{sec:QuantileEstimationImplementationApproaches} develops additional labelling strategies that can be used in quantile estimation tasks and also discusses the Hessian estimation in such settings. Appendix \ref{sec:LitSearchTwoPhaseMultiwaveTheory} describes a search of the two-phase multiwave sampling literature, providing evidence in support of our theoretical novelty claims. Appendix \ref{sec:SimulationDetails} gives additional details on the simulations and presents the results from additional simulations that were omitted from the main text.

\subsection*{Notation and conventions}
\addcontentsline{toc}{subsection}{Notation and  conventions}

Throughout the appendices we will use the following conventions. We define $\text{sgn}: \mathbb{Z} \to \{-1,0,1\}$ to be a function that satisfies $\text{sgn}(0)=0$, $\text{sgn}(z)=1$ for $z >0$, and $\text{sgn}(z)=-1$ for $z<0$ (this function encodes the sign of the input integer). With slight abuse of notation, for each positive integer $j$ we let $e_j$ denote the $j$th standard basis vector in $\mathbb{R}^{d^*}$ whose dimension $d^*$ depends on the context. For any vector $v$ and matrix $A$, $[v]_j=e_j^\tran v$ denotes the $j$th component of $v$ and $[A]_{jj'}=e_j^\tran A e_{j'}$ denotes the $j$th entry in the $j'$th column of $A$ (provided that $j$ and $j'$ are small enough positive integers for those quantities to be defined). 

As in the main text, unless otherwise specified, sums and products over ranges of indices in which the lower limit exceeds the upper limit are defined to be $0$ and $1$, respectively. For all $\theta' \in \Theta$ and $\argx \in \xSpace \cup \xProxySpace$, $\dot{l}_{\theta'}(\argx)$ denotes the vector of upper right-hand Dini partial derivatives of $\theta \mapsto l_{\theta}(\argx)$ evaluated at $\theta=\theta'$. Similarly, for all $\theta' \in \Theta$ and $\argx \in \xSpace \cup \xProxySpace$, $\ddot{l}_{\theta'}(\argx)$ denotes the matrix of upper right-hand Dini partial derivatives of $\theta \mapsto \dot{l}_{\theta}(\argx)$ evaluated at $\theta=\theta'$. (When $\theta \mapsto l_{\theta}(\argx)$ is twice differentiable, $\dot{l}_{\theta'}(\argx)$ and $\ddot{l}_{\theta'}(\argx)$ correspond to the gradient and Hessian of $\theta \mapsto l_{\theta}(\argx)$ evaluated at $\theta=\theta'$, respectively).

\paragraph{Probability space and expectations.}\label{paragraph:ProbailitySpaceClarifications}
 The random variables in this manuscript include the full-data vectors $\datvecraw_1, \ldots, \datvecraw_N$, independent copies of $\datvecraw_i=(\xobs_i,\txmiss_i,\xmiss_i)$ denoted by $\datvecraw=(\xobs,\txmiss,\xmiss)$, the independent uniform draws $U^{(k)}_i$, and quantities derived from
them, such as the labelling indicators $I^{(k)}_i$, the data $\cd_k$, and the
weights $W_i^{(k)}$ and $W_i$. These random variables as well as all other random variables in this manuscript are defined on a single probability space, with probability measure $\mathbb{P}$. Unless otherwise noted, $\e[\cdot]$ denotes expectation
under $\mathbb{P}$, and conditional expectations such as
$\e[\cdot \giv \cd_k]$ are also taken under $\mathbb{P}$.

\subsection*{Notation glossary (reference)}
\addcontentsline{toc}{subsection}{Notation glossary (reference)}

An overview of the notation used throughout the main text and appendices is given in Table \ref{tab:notation}. Table  \ref{tab:notation} is intended as a navigational aid; all notation is formally defined in the main text and appendices, which should be treated as the authoritative source.

{\small
\begin{longtable}{lll}
\caption{Summary of notation. The first column gives the symbol denoting commonly seen mathematical quantities, and when possible provides a defining equation or the space in which the quantity exists. The second column gives a heuristic description of the quantity. The third column gives the corresponding section or equation where the quantity is formally defined. This table is intended as a navigational aid. Some descriptions are heuristic; see the referenced sections and equations for formal definitions.}
\label{tab:notation} \\
\toprule
\textbf{Notation} & \textbf{Description} & \textbf{Reference} \\
\midrule
\endfirsthead
\toprule
\textbf{Symbol} & \textbf{Description} & \textbf{Reference} \\
\midrule
\endhead
\bottomrule
\endfoot
\multicolumn{3}{l}{\textit{Sample and wave size parameters}} \\
$N$ & Phase I sample size & Section~\ref{sec:settingAndEstimator} \\
$K$ & Number of waves in Phase II & Section~\ref{sec:settingAndEstimator} \\ 
$n_{\text{targ}}^{(k)}$ & Expected number of labels collected in wave $k$ & Section~\ref{sec:Simulations} \\
$n_{\text{targ}}=\sum_{k=1}^K n_{\text{targ}}^{(k)}$ & Expected number of labels collected in Phase II & Section~\ref{sec:Simulations} \\
\addlinespace
\addlinespace
\multicolumn{3}{l}{\textit{Data and variables}} \\
$\datvecraw \equiv (\xobs, \txmiss, \xmiss) \in \real^q$ & Full data vector for single sample & Section~\ref{sec:settingAndEstimator} \\
$\goodX \equiv (\xobs, \xmiss) \in \real^p$ & Gold-standard measurements & Section~\ref{sec:settingAndEstimator} \\
$\proxyX \equiv (\xobs, \txmiss) \in \real^p$ & Cheap-to-measure proxy vector & Section~\ref{sec:settingAndEstimator} \\
$\xSpace, \xProxySpace \subseteq \real^p$ & Supports of $\goodX$ and $\proxyX$ & Section~\ref{sec:settingAndEstimator} \\
$\cd_k \equiv  \big( (I_i^{(j)},I_i^{(j)} \xmiss_i )_{j=1}^k, \proxyX_i \big)_{i=1}^N$ & Data available and sampling indicators after wave $k$ & Section~\ref{sec:DescriptionOfSamplingScheme} \\
\addlinespace
\addlinespace
\multicolumn{3}{l}{\textit{Labelling rules and indicators}} \\
$\mathcal{P} \equiv \{\pi: \xProxySpace \to (0,1)\}$ & Space of labelling rules & Section~\ref{sec:settingAndEstimator} \\
$\LabelStrategy{k}$ & Labelling strategy for wave $k$ & Section~\ref{sec:DescriptionOfSamplingScheme} \\
$\LabelRulePi{k} \equiv \LabelStrategy{k}(\cd_{k-1}) \in \mathcal{P}$ & Labelling rule for wave $k$ & Section~\ref{sec:DescriptionOfSamplingScheme} \\
$\LabelRulePiDoubleArg{k}{k-1}(\proxyX_i) \in (0,1)$ & Labelling probability for sample $i$ in wave $k$ & Section~\ref{sec:DescriptionOfSamplingScheme} \\
$U_i^{(k)} \stackrel{\text{i.i.d.}}{\sim} \mathrm{Unif}[0,1]$ & Uniform draws for labelling decision & Section~\ref{sec:DescriptionOfSamplingScheme} \\
$I_i^{(k)} \equiv \mathbbm{1}\{U_i^{(k)} \leq \LabelRulePiDoubleArg{k}{k-1}(\proxyX_i)\}$ & Labelling indicator for sample $i$ in wave $k$ & Section~\ref{sec:DescriptionOfSamplingScheme} \\
$I_i \in \{0,1\}$ & Indicator that $\xmiss_i$ was labelled in Phase II & Section~\ref{sec:DescriptionOfSamplingScheme} \\
$b \in (0,1/2)$ & Overlap constant (Assumption~\ref{assump:LabellingRuleOverlap}) & Assump.~\ref{assump:LabellingRuleOverlap} \\
$\LimitingLabelRulePi{k} \in \mathcal{P}$ & Limiting labelling rule $\LabelRulePi{k}$ for wave $k$ & Assump.~\ref{assump:CLTRgularityConditions} \\
$\LimitingLabelRulePi{1:k}(\argxProxy)$ & Limiting probability of labelling sample $i$ & Eq.~\eqref{eq:PiBarProd1tok} \\
& \quad in wave $k$ but not earlier given $\proxyX_i=\argxProxy$ &  \\
\addlinespace

\multicolumn{3}{l}{\textit{Weights}} \\
$c_1, \dots, c_K \in [0,1]$ & Prespecified constants; $\sum_{k=1}^K c_k = 1$ & Section \ref{sec:IntroduceMIPWWeights} \\
$\phi_{-1}(t)=1-t,\; \phi_0(t)=t,\; \phi_1(t)=1$ & Auxiliary functions for weight decomposition & Eq.~\eqref{eq:ComplementOr1Operator} \\
\addlinespace
$W_i^{(k,j)} \equiv \frac{\phi_{\text{sgn}(j-k)}( I_i^{(j)})}{\phi_{\text{sgn}(j-k)} ( \LabelRulePi{j}(\proxyX_i) )}$ & Factor of $W_i^{(k)}$ corresponding to wave $j$ & Eq.~\eqref{eq:IPWsingleSingle} \\
\addlinespace
$W_i^{(k)} =\prod_{j=1}^K W_i^{(k,j)}$ & Wave-$k$ inverse probability weight & Eqs.~\eqref{eq:InverseProbabilityWeights};\eqref{eq:WkFormulaAlt} \\
$W_i \equiv \sum_{k=1}^K c_k W_i^{(k)}$ & (Aggregated) Multiwave IPW weights & Eq.~\eqref{eq:aggregated_Wk} \\
$\bar{W}_i^{(k,j)},\bar{W}_i^{(k)}, \bar{W}_i$ & i.i.d.\ approximations to $W_i^{(k,j)},W_i^{(k)},W_i$ & Eqs.~\eqref{eq:IndepAggregatedWeightsMainText};\eqref{eq:WkIndepVersionDecomp} \\
& (using $\{\LimitingLabelRulePi{j}\}_{j=1}^K$ and $\{U_i^{(j)}\}_{j=1}^K$) &  \\
\addlinespace
\multicolumn{3}{l}{\textit{Loss, estimands, and Hessians}} \\
$l_\theta(\cdot)$ & Loss function parameterized by $\theta \in \Theta \subseteq \real^d$ & Section~\ref{sec:IntroduceMestimator} \\
$\thetaTarg \equiv \argmin_{\theta \in \Theta} \e[l_\theta(\goodX)]$ & Estimand of interest & Eq.~\eqref{eq:EstimandThetaStar} \\
$\gammaTarg \equiv \argmin_{\theta \in \Theta} \e[l_\theta(\proxyX)]$ & Proxy estimand & Section~\ref{sec:IntroduceMestimator} \\
$L(\theta) \equiv \e[l_\theta(\goodX)]$ & Population loss & Section~\ref{sec:StateMestimationResults} \\
$\tilde{L}(\theta) \equiv \e[l_\theta(\proxyX)]$ & Proxy population loss & Section~\ref{sec:StateMestimationResults} \\
$H_{\thetaTarg} \equiv \nabla^2 L(\thetaTarg)$ & Hessian of population loss at $\thetaTarg$ & Section~\ref{sec:StateMestimationResults} \\
$H_{\gammaTarg} \equiv \nabla^2 \tilde{L}(\gammaTarg)$ & Hessian of proxy loss at $\gammaTarg$ & Section~\ref{sec:StateMestimationResults} \\
\addlinespace
\addlinespace

\multicolumn{3}{l}{\textit{Estimators}} \\
$\htc= \argmin_{\theta \in \Theta} \{ \sum_{i=1}^N W_i l_{\theta}(\goodX_i) \}$ & Estimator of $\thetaTarg$ using samples selected in Phase II &  Eq.~\eqref{eq:PTDComponentEstimatorsDEF} \\
$\hgc = \argmin_{\theta \in \Theta} \{ \sum_{i=1}^N W_i l_{\theta}(\proxyX_i) \} $ & Estimator of $\gammaTarg$ using samples selected in Phase II & Eq.~\eqref{eq:PTDComponentEstimatorsDEF} \\
$\hga= \argmin_{\theta \in \Theta} \{ \sum_{i=1}^N  l_{\theta}(\proxyX_i) \} $ & Estimator of $\gammaTarg$ using all Phase I samples & Eq.~\eqref{eq:PTDComponentEstimatorsDEF} \\
$\hat{\Omega} \in \mathbb{R}^{d \times d}$ & Tuning matrix (possibly data-dependent) & Section~\ref{sec:IntroduceMestimator} \\
$\htPTD \equiv \hat{\Omega} \hga + (\htc -\hat{\Omega} \hgc)$ & Multiwave Predict-Then-Debias estimator & Eq.~\eqref{eq:PTD_estimator} \\
\addlinespace
\multicolumn{3}{l}{\textit{Variance estimation}} \\
$\hat{\Sigma}_{11}, \hat{\Sigma}_{12}, \hat{\Sigma}_{22}, \hat{\Sigma}_{13}, \hat{\Sigma}_{33}$ & Covariance component estimators & Eq.~\eqref{eq:HessianAndSigmaEstimators} \\
$\Sigma_{11}, \Sigma_{12}, \Sigma_{22}, \Sigma_{13}, \Sigma_{33}$ & Population covariance components & Eq.~\eqref{eq:GradientAsympCovFormulae} \\
$\hat{H}_{\thetaTarg}, \hat{H}_{\gammaTarg}$ & Estimator for Hessians & Eq.~\eqref{eq:HessianAndSigmaEstimators} \\
$\hat{\Sigma}^{\PTDSuperScriptAcr}$ & Estimator for asymptotic variance of $\htPTD$ & Eq.~\eqref{eq:PTDAsympVarEstimator} \\
$\Sigma^{\PTDSuperScriptAcr}(\Omega)$ & Asymptotic variance of $\htPTD$ if $\hat{\Omega} \xrightarrow{p} \Omega$ & Eq.~\eqref{eq:AsympVarPTDFormula}  \\
\addlinespace
\addlinespace
\multicolumn{3}{l}{\textit{Empirical Process Theory (Appendix \ref{sec:AsymptoticsForMestimators})}} \\
$\bfs_N f \equiv N^{-1} \sum_{i=1}^N W_i f(\goodX_i)$ & Weighted empirical average (gold-standard) & Eq.~\eqref{eq:AveragingOperatorDefs} \\
$\bfts_N f \equiv N^{-1} \sum_{i=1}^N W_i f(\proxyX_i)$ & Weighted empirical average (proxy, labelled) & Eq.~\eqref{eq:AveragingOperatorDefs} \\
$\bftp_N f \equiv N^{-1} \sum_{i=1}^N f(\proxyX_i)$ & Unweighted empirical average (proxy, all) & Eq.~\eqref{eq:AveragingOperatorDefs} \\
$\mathbb{G}_N f \equiv \sqrt{N}(\bfs_N f - \e[f(\goodX)])$ & Centered and scaled empirical weighted average & Eq.~\eqref{eq:G_operator} \\
$\tilde{\mathbb{G}}_N f \equiv \sqrt{N}(\bfts_N f - \e[f(\proxyX)])$ & Same but for proxy & Eq.~\eqref{eq:G_operator} \\
$\Delta_N(\theta), \tilde{\Delta}_N(\theta)$ & Centered empirical processes & Eq.~\eqref{eq:DeltaDef} \\
$\omega_{N,\delta_0}(\delta), \tilde{\omega}_{N,\delta_0}(\delta)$ & Moduli of continuity of $\Delta_N, \tilde{\Delta}_N$ & Eq.~\eqref{eq:ModuliOfContinuityDef} \\

\addlinespace
\multicolumn{3}{l}{\textit{Approximate greedy optimal sampling strategy }} \\
$\varrho_j(\proxyX)$ & Conditional mean of squared differences  & Eq.~\eqref{eq:varrhoDef} \\
& \quad in influence functions given $\proxyX$ & \\
\bottomrule
\end{longtable}}

\newpage

\section{Properties of multiwave inverse probability weights}\label{sec:PropertiesOfWeights}

In this appendix we establish useful properties of the multiwave inverse probability weights $W_i$ under \samplingSchemeName{.} We briefly summarize some of the more notable properties, which are then proven in subsequent subsections.

Using the weights $W_i$ enables unbiased estimation of the mean of functions of $\datvecraw$, even though $\datvecraw$ is only fully observed on some (generally nonuniform) random subset of the $N$ samples. Further the multiwave inverse probability weights $W_i$ do not induce pairwise correlations between samples (although, notably, they do induce statistical dependencies). More specifically, under \samplingSchemeName{} and Assumption \ref{assump:IIDUnderlyingData}, for any fixed, measurable functions $f,g : \mathbb{R}^q \to \mathbb{R}$ and for any $i,i' \in [N]$ such that $i \neq i'$, \begin{enumerate}[(i)]
    \item $\e[W_i f(\datvecraw_i)]=\e[f(\datvecraw)]$, and 
    \item $\cov \big( W_i f(\datvecraw_i), W_{i'} g(\datvecraw_{i'}) \big)=\cov \big( W_i f(\datvecraw_i),  g(\datvecraw_{i'}) \big) =0$.
\end{enumerate} Properties (i) and (ii) above are formally established in Propositions \ref{prop:WeightsEnableUnbiasedMeanEst} and \ref{prop:NoCovarianceMultiwaveIPW}, and can also be checked using recursive applications of the tower property (in which data from earlier and earlier waves are conditioned upon). Under further assumptions, we prove that covariances of the form $\cov \big( W_i^2 f(\datvecraw_i), W_{i'}^2 g(\datvecraw_{i'}) \big)$ converge to $0$ as $N \to \infty$ (see Proposition \ref{prop:Vanishing2ndOrderCov}).  

Another notable property of the weights $W_i$ is that in our motivating setting they are exchangeable. In particular, under \samplingSchemeName{} and Assumptions \ref{assump:IIDUnderlyingData} and \ref{assump:CLTRgularityConditions}(i), $(I_i, W_i,\goodX_i,\proxyX_i)_{i=1}^N$ is an exchangeable sequence of $N$ random vectors (see Proposition \ref{prop:ExchangeableData}).

\subsection{Helpful notation for studying weights}

 Recall that for each $i \in [N]$ and $k \in [K]$, the multiwave inverse probability weights defined in Equations \eqref{eq:InverseProbabilityWeights} and \eqref{eq:aggregated_Wk} are given by $$W_i^{(k)} \equiv \Big( \prod_{j=1}^{k-1}  \frac{1-I_i^{(j)}}{1-\LabelRulePi{j} (\proxyX_i)} \Big) \frac{I_i^{(k)}}{\LabelRulePi{k} (\proxyX_i)} \quad \text{and} \quad W_i \equiv \sum_{k=1}^K c_k W_i^{(k)},$$ where $c_1,\dots,c_K \in [0,1]$ are prespecified constants satisfying $\sum_{k=1}^K c_k=1$. It will also be convenient to define $$W_i^{(0)} \equiv 1 \quad \text{ for each } i \in [N].$$ 
 
To study properties of the weights $W_i^{(k)}$, it helps to express $W_i^{(k)}$ as a product of $K$ terms. In particular, define $\phi_{-1},\phi_0, \phi_1 : [0,1] \to [0,1]$ by \begin{equation}\label{eq:ComplementOr1Operator}
    \phi_{-1}(t) \equiv 1-t, \quad \phi_0(t) \equiv t, \quad \text{and} \quad  \phi_1(t) \equiv 1 \quad \text{ for all } t \in [0,1],
\end{equation} and recall that $\text{sgn}: \mathbb{Z} \to \{-1,0,1\}$ gives the sign of an input integer (i.e., $\text{sgn}(0)=0$, $\text{sgn}(z)=1$ for positive $z$, and $\text{sgn}(z)=-1$ for negative $z$).
Next define \begin{equation}\label{eq:IPWsingleSingle}
     W_i^{(k,j)} \equiv  \frac{\phi_{\text{sgn}(j-k)} \big( I_i^{(j)} \big)}{\phi_{\text{sgn}(j-k)} \big( \LabelRulePi{j}(\proxyX_i) \big)} \quad \text{ for each } i \in [N], j \in [K], k \in \{0\} \cup [K],
 \end{equation} and observe that for $k < j$, $W_i^{(k,j)}=1$ and that 
 
 \begin{equation}\label{eq:WkFormulaAlt}
W_i^{(k)} = \prod_{j=1}^k  W_i^{(k,j)} = \prod_{j=1}^K  W_i^{(k,j)} \quad \text{ for each } i \in [N], k \in \{0\} \cup [K]. \end{equation}

This notation can be used to immediately establish the following two facts that are used throughout the appendices. The notation is also used in the next subsections to prove properties that simplify expectations of quantities multiplied by $W_i^{(k)}$.

\begin{fact}\label{fact:WsqFormula}
       For $i \in [N]$, $W_i^2= \sum_{k=1}^K c_k^2 (W_i^{(k)})^2$.
\end{fact}

\begin{proof}
    Fix $i \in [N]$. Also fix $k,k' \in [K]$ such that $k \neq k'$. Note that $ W_i^{(k,j)} W_i^{(k',j)}=0$ when $j=\text{min}\{k,k'\}$. Hence, by \eqref{eq:WkFormulaAlt}, $W_i^{(k)}W_i^{(k')}=\prod_{j=1}^K    W_i^{(k,j)} W_i^{(k',j)}=0$. Since this holds for any $k,k' \in [K]$ such that $k \neq k'$,  $W_i^2=\big( \sum_{k=1}^K c_k W_i^{(k)} \big)^2=\sum_{k=1}^K c_k^2 (W_i^{(k)})^2$.
\end{proof}

\begin{fact}\label{fact:WeightBound}
    Under Assumption \ref{assump:LabellingRuleOverlap}, for each $i \in [N]$ $0 \leq  W_i   \leq b^{-K}$ almost surely. 
\end{fact}

\begin{proof}
Fix $i \in [N]$. For any $j,k \in [K]$, by Assumption \ref{assump:LabellingRuleOverlap} and Definition \eqref{eq:ComplementOr1Operator}, $\phi_{\text{sgn}(j-k)} \big( I_i^{(j)} \big) \in \{0,1 \}$ and $\phi_{\text{sgn}(j-k)} \big( \LabelRulePi{j}(\proxyX_i) \big) \in [b,1]$ almost surely. Hence by \eqref{eq:IPWsingleSingle}, $0 \leq W_i^{(k,j)} \leq b^{-1}$ almost surely for any $j,k \in [K]$. Since by Equations \eqref{eq:aggregated_Wk} and \eqref{eq:WkFormulaAlt}, $W_i=\sum_{k=1}^K c_k \prod_{j=1}^K W_i^{(k,j)}$, where $c_k \geq 0$ and $\sum_{k=1}^K c_k =1$, $W_i \in [0,b^{-K}]$ almost surely.
\end{proof}

\subsection{Helpful lemmas for simplifying weighted expectations}

To study expectations of quantities when multiplied by $W_i^{(k)}$ we start by explicitly stating some properties of \samplingSchemeName{.} Intuitively, the following fact holds because for $j \in [K]$, $I_i^{(j)}$ for $i=1,\dots,N$ are from Bernoulli draws that are independent conditionally on $\cd_{j-1}$, each with success probability $\LabelRulePi{j}(\proxyX_i)$.

\begin{fact}\label{fact:CondMeanAndIndep}
    Under \samplingSchemeName{,} for all $j \in [K]$ and $i,i' \in [N]$, \begin{equation}\label{eq:PropertiesPhaseIISamplingCondMeanInd}
    \e[I_i^{(j)} \giv \cd_{j-1}]=\LabelRulePi{j}(\proxyX_i), \quad  \text{while both}
\end{equation}  

\begin{equation}\label{eq:PropertiesPhaseIISamplingCondIndep}
   \big( I_{i}^{(j)},I_{i'}^{(j)} \big) \indep (\datvecraw_i, \datvecraw_{i'} ) \giv \cd_{j-1} \quad \text{and} \quad I_{i}^{(j)} \indep  I_{i'}^{(j)} \giv \cd_{j-1} \quad \text{if } i \neq i'.
\end{equation} 
\end{fact}

\begin{proof} Fix $j \in [K]$, and $i,i' \in [N]$ such that $i \neq i'$. Recall that under \samplingSchemeName{,} $I_i^{(j)} \equiv \mathbbm{1}\{U_i^{(j)} \leq \LabelRulePi{j}(\proxyX_i) \}$ and $I_{i'}^{(j)} \equiv \mathbbm{1}\{U_{i'}^{(j)} \leq \LabelRulePi{j}(\proxyX_{i'}) \}$, that $\LabelRulePi{j}(\proxyX_i)$ and $\LabelRulePi{j}(\proxyX_{i'})$ can both be expressed as fixed, measurable functions of $\cd_{j-1}$, and that $U_i^{(j)}, U_{i'}^{(j)} \stackrel{\text{i.i.d.}}{\sim} \text{Unif}[0,1]$ are generated independently of $\cd_{j-1}$ (see Section \ref{sec:DescriptionOfSamplingScheme}). As immediate consequences, $U_i^{(j)}, U_{i'}^{(j)} \giv \cd_{j-1} \stackrel{\text{i.i.d.}}{\sim} \text{Unif}[0,1]$, and moreover \eqref{eq:PropertiesPhaseIISamplingCondMeanInd} and the second claim in \eqref{eq:PropertiesPhaseIISamplingCondIndep} both hold. 

To prove the first claim in \eqref{eq:PropertiesPhaseIISamplingCondIndep} note that under \samplingSchemeName{,} prior to the $j$th wave of data collection, $U_i^{(j)}, U_{i'}^{(j)} \stackrel{\text{i.i.d.}}{\sim} \text{Unif}[0,1]$ are generated independently of both $\cd_{j-1}$ and $\xmiss$ data that has not yet been observed. Thus regardless of whether or not $\xmiss_i$ or $\xmiss_{i'}$ was observed prior to the start of wave $j$, $U_i^{(j)}, U_{i'}^{(j)} \stackrel{\text{i.i.d.}}{\sim} \text{Unif}[0,1]$ are generated independently of both $\cd_{j-1}$ and $(\datvecraw_i,\datvecraw_{i'})=(\proxyX_i,\xmiss_i,\proxyX_{i'},\xmiss_{i'})$. Hence $(U_{i}^{(j)},U_{i'}^{(j)}) \indep (\datvecraw_i,\datvecraw_{i'}) \giv \cd_{j-1}$. Recalling that $\LabelRulePi{j}(\proxyX_{i})$ and
$\LabelRulePi{j}(\proxyX_{i'})$ can be written as fixed, measurable functions of $\cd_{j-1}$ and noting the formulas for $I_i^{(j)}$ and $I_{i'}^{(j)}$ establishes the first claim in \eqref{eq:PropertiesPhaseIISamplingCondIndep}.

\end{proof}



The above properties and the earlier definitions for $W_i^{(k,j)}$ can be used to prove the following auxiliary lemma, which enables simplifications of weighted expectations when conditioning on $\cd_{j-1}$. 

\begin{lemma}\label{lemma:SimplifyCondExpectations}
    Under \samplingSchemeName{,}  for any fixed, measurable $f,g: \mathbb{R}^q  \to \mathbb{R}$, and
     for all $k,k' \in \{0\} \cup [K]$, $j \in [K]$, $s,s' \in \{1,2\}$, and $i,i' \in [N]$ such that $i \neq i'$,

    $$ \e \big[ \big( W_i^{(k,j)} \big)^s  \big( W_{i'}^{(k',j)} \big)^{s'} f(\datvecraw_i) g(\datvecraw_{i'}) \giv \cd_{j-1} \big]=\frac{\e \big[ f(\datvecraw_i) g(\datvecraw_{i'}) \giv \cd_{j-1} \big]}{\big[\phi_{\textnormal{sgn}(j-k)} \big( \LabelRulePi{j}(\proxyX_i) \big) \big]^{s-1}  \big[\phi_{\textnormal{sgn}(j-k')} \big( \LabelRulePi{j}(\proxyX_{i'}) \big) \big]^{s'-1}}.$$
\end{lemma}

\begin{proof} 

     Fix measurable $f,g: \mathbb{R}^q  \to \mathbb{R}$. Further fix $k,k' \in \{0 \} \cup [K]$, $j \in [K]$, $s,s' \in \{1,2\}$, and $i,i' \in [N]$ such that $i \neq i'$. For convenience, let $r = \text{sgn}(j-k)$ and $r' = \text{sgn}(j-k')$. 
     
     Note that by considering all 3 possible cases in \eqref{eq:ComplementOr1Operator} and applying \eqref{eq:PropertiesPhaseIISamplingCondMeanInd}, $$\e[\phi_{r} (I_i^{(j)}) \giv \cd_{j-1} ] = \phi_{r} \big( \LabelRulePi{j}(\proxyX_i) \big), \quad  \text{ and similarly, } \quad  \e[\phi_{r'} (I_{i'}^{(j)}) \giv \cd_{j-1} ] = \phi_{r'} \big( \LabelRulePi{j}(\proxyX_{i'}) \big).$$ Also note that by examining \eqref{eq:ComplementOr1Operator}, regardless of the values of $I_i^{(j)},I_{i'}^{(j)} \in \{0,1\} $, $ s,s' \in \{1,2\}$ and $r,r' \in \{-1,0,1\}$, $$\big[ \phi_r (I_i^{(j)} ) \big]^s=\phi_r (I_i^{(j)} ) \quad \text{and} \quad \big[ \phi_{r'} (I_{i'}^{(j)} ) \big]^{s'}=\phi_{r'} (I_{i'}^{(j)} ),$$ because $\phi_{r^*}(I) \in \{0,1\}$ for any $I \in \{0,1\}$ and $r^* \in \{-1,0,1\}$. Next note that $\LabelRulePi{j}(\proxyX_i)$ and $\LabelRulePi{j}(\proxyX_{i'})$ are measurable functions of $\cd_{j-1}$ and $\phi_r$, $\phi_{r'}$, $f$, and $g$ are all measurable functions. Hence, by definition \eqref{eq:IPWsingleSingle}, the above result, and the conditional independence properties of \samplingSchemeName{} given in \eqref{eq:PropertiesPhaseIISamplingCondIndep}, $$\begin{aligned} \e \big[ \big( W_i^{(k,j)} \big)^s  \big( W_{i'}^{(k',j)} \big)^{s'} f(\datvecraw_i) g(\datvecraw_{i'}) \giv \cd_{j-1} \big] & =  \e \Bigg[ \frac{\phi_r (I_i^{(j)} )  \cdot \phi_{r'} ( I_{i'}^{(j)} ) f(\datvecraw_i) g(\datvecraw_{i'})}{\big[ \phi_r \big( \LabelRulePi{j}(\proxyX_i) \big) \big]^s \cdot \big[\phi_{r'} \big( \LabelRulePi{j}(\proxyX_{i'}) \big) \big]^{s'}}  \Bigg| \cd_{j-1} \Bigg]
    \\ & = \frac{\e \big[ \phi_r (I_i^{(j)} ) \cdot \phi_{r'} ( I_{i'}^{(j)} ) f(\datvecraw_i) g(\datvecraw_{i'}) \giv \cd_{j-1} \big] }{\big[ \phi_r \big( \LabelRulePi{j}(\proxyX_i) \big) \big]^s \cdot \big[\phi_{r'} \big( \LabelRulePi{j}(\proxyX_{i'}) \big) \big]^{s'}}
    \\ & = \frac{\e \big[ \phi_r (I_i^{(j)} ) \cdot \phi_{r'} ( I_{i'}^{(j)} ) \giv \cd_{j-1} \big] \cdot  \e \big[ f(\datvecraw_i) g(\datvecraw_{i'}) \giv \cd_{j-1} \big] }{\big[ \phi_r \big( \LabelRulePi{j}(\proxyX_i) \big) \big]^s \cdot \big[\phi_{r'} \big( \LabelRulePi{j}(\proxyX_{i'}) \big) \big]^{s'}}
    \\ & = \frac{\e \big[ \phi_r (I_i^{(j)} ) \giv \cd_{j-1} \big]  \cdot \e \big[ \phi_{r'} ( I_{i'}^{(j)} ) \giv \cd_{j-1} \big] \cdot  \e \big[ f(\datvecraw_i) g(\datvecraw_{i'}) \giv \cd_{j-1} \big] }{\big[ \phi_r \big( \LabelRulePi{j}(\proxyX_i) \big) \big]^s \cdot \big[\phi_{r'} \big( \LabelRulePi{j}(\proxyX_{i'}) \big) \big]^{s'}}
    \\ & = \frac{\e \big[ f(\datvecraw_i) g(\datvecraw_{i'}) \giv \cd_{j-1} \big]}{\big[ \phi_r \big( \LabelRulePi{j}(\proxyX_i) \big) \big]^{s-1} \cdot \big[\phi_{r'} \big( \LabelRulePi{j}(\proxyX_{i'}) \big) \big]^{s'-1}}, \end{aligned}$$ where the last equality follows from a previously established result. Recalling that we had set $r = \text{sgn}(j-k)$ and $r' = \text{sgn}(j-k')$, this proves the desired result.
\end{proof}

The following lemma shows that multiplying quantities by the multiwave inverse probability weights $W_i^{(k)}$ and $W_{i'}^{(k')}$ does not change their expected value, and readily sets up the proofs of Propositions \ref{prop:WeightsEnableUnbiasedMeanEst} and \ref{prop:NoCovarianceMultiwaveIPW}. The proof of the lemma involves a recursive application of the tower property as well as Lemma \ref{lemma:SimplifyCondExpectations} in the case where $s=s'=1$.

\begin{lemma}\label{lemma:SimplifyWeightedExpectations2Terms}
     Under \samplingSchemeName{,}  for any fixed, measurable $f,g: \mathbb{R}^q  \to \mathbb{R}$, $$\e \big[ W_i^{(k)} W_{i'}^{(k')} f(\datvecraw_i) g(\datvecraw_{i'})  \big]=\e \big[ f(\datvecraw_i) g(\datvecraw_{i'}) \big]$$
     for all $k,k' \in \{0\} \cup [K]$, and all $i,i' \in [N]$ such that $i \neq i'$.
\end{lemma}

\begin{proof}
    Fix measurable $f,g: \mathbb{R}^q  \to \mathbb{R}$. Further fix $k,k' \in \{0 \} \cup [K]$ and $i,i' \in [N]$ such that $i \neq i'$. By applying Lemma \ref{lemma:SimplifyCondExpectations} in the special case where $s=s'=1$, note that
    \begin{equation}\label{eq:SimplifyCondExpectations}
        \e \big[ W_i^{(k,j)} W_{i'}^{(k',j)} f(\datvecraw_i) g(\datvecraw_{i'}) \giv \cd_{j-1} \big]=\e \big[ f(\datvecraw_i) g(\datvecraw_{i'}) \giv \cd_{j-1} \big] \quad \text{for each } j \in [K].
    \end{equation} The proof of this lemma will follow by noting that by \eqref{eq:WkFormulaAlt}, $$\e \big[ W_i^{(k)} W_{i'}^{(k')} f(\datvecraw_i) g(\datvecraw_{i'})  \big]= \e \Big[ \prod_{j=1}^K W_i^{(k,j)} W_{i'}^{(k',j)} f(\datvecraw_i) g(\datvecraw_{i'})  \Big]$$ and recursively applying the tower property conditioning on $\cd_{j-1}$ for decreasing $j$.

    Formally, we will prove by induction that for each $j^* \in [K]$, \begin{equation}\label{eq:InductiveStep1stOrderExpectations}
   \e \Big[ \Big( \prod_{j=1}^{j^*} W_i^{(k,j)} W_{i'}^{(k',j)} \Big) f(\datvecraw_i) g(\datvecraw_{i'})  \Big] = \e [ f(\datvecraw_i) g(\datvecraw_{i'})]. \end{equation} In the case where $j^*=1$, \eqref{eq:InductiveStep1stOrderExpectations} holds because by the tower property and \eqref{eq:SimplifyCondExpectations}, $$\e \big[  W_i^{(k,1)} W_{i'}^{(k',1)} f(\datvecraw_i) g(\datvecraw_{i'})  \big] = \e \big[  \e[ W_i^{(k,1)} W_{i'}^{(k',1)} f(\datvecraw_i) g(\datvecraw_{i'}) \giv \cd_0]  \big] = \e \big[  \e[f(\datvecraw_i) g(\datvecraw_{i'}) \giv \cd_0]   \big]=\e [ f(\datvecraw_i) g(\datvecraw_{i'})].$$ Next fix some $j^* \in [K-1]$ and assume that \eqref{eq:InductiveStep1stOrderExpectations} holds for $j^*$. Since for $j \leq j^*$, $ W_i^{(k,j)}$ and $W_{i'}^{(k',j)}$ are measurable functions of $\cd_{j^*}$, by the tower property and \eqref{eq:SimplifyCondExpectations}, $$\begin{aligned} \e \Big[ \prod_{j=1}^{j^* +1} W_i^{(k,j)} W_{i'}^{(k',j)} f(\datvecraw_i) g(\datvecraw_{i'})  \Big] & = \e \Big[ \Big( \prod_{j=1}^{j^* } W_i^{(k,j)} W_{i'}^{(k',j)}  \Big) \e[ W_i^{(k,j^*+1)} W_{i'}^{(k',j^*+1)} f(\datvecraw_i) g(\datvecraw_{i'}) \giv \cd_{j^*}]  \Big] 
   \\ & =   \e \Big[ \Big( \prod_{j=1}^{j^* } W_i^{(k,j)} W_{i'}^{(k',j)} \Big) \e[ f(\datvecraw_i) g(\datvecraw_{i'}) \giv \cd_{j^*}]  \Big] 
   \\ & = \e \Big[ \Big( \prod_{j=1}^{j^* } W_i^{(k,j)} W_{i'}^{(k',j)} \Big)  f(\datvecraw_i) g(\datvecraw_{i'})   \Big]
   \\ & = \e [ f(\datvecraw_i) g(\datvecraw_{i'})], \end{aligned}$$ where the last step holds provided that \eqref{eq:InductiveStep1stOrderExpectations} holds for $j^*$. Hence we have shown that for any $j^* \in [K-1]$, if \eqref{eq:InductiveStep1stOrderExpectations} holds for $j^*$, then \eqref{eq:InductiveStep1stOrderExpectations} also holds for $j^*+1$, and additionally, \eqref{eq:InductiveStep1stOrderExpectations} holds in the case where $j^*=1$. Thus, by induction, \eqref{eq:InductiveStep1stOrderExpectations} holds for all $j^* \in [K]$. 
   
   By recalling the alternative formula for $W_i^{(k)}$ and $W_{i'}^{(k')}$ at \eqref{eq:WkFormulaAlt}, and applying Equation \eqref{eq:InductiveStep1stOrderExpectations} in the case where $j^*=K$, $$\e \big[ W_i^{(k)} W_{i'}^{(k')} f(\datvecraw_i) g(\datvecraw_{i'})  \big]= \e \Big[ \prod_{j=1}^K W_i^{(k,j)} W_{i'}^{(k',j)} f(\datvecraw_i) g(\datvecraw_{i'})  \Big]=\e[f(\datvecraw_i) g(\datvecraw_{i'})].$$ 
\end{proof} 

Recalling that for $i \in [N]$, by definition $W_i^{(0)} =1$, we obtain the following corollary. This corollary is subsequently used to prove Proposition \ref{prop:WeightsEnableUnbiasedMeanEst}.

\begin{corollary}\label{cor:SimplifyWeightedExpectations1Term}
      Under \samplingSchemeName{,}  for any fixed, measurable $f: \mathbb{R}^q  \to \mathbb{R}$, $$\e \big[ W_i^{(k)}f(\datvecraw_i)  \big]=\e \big[ f(\datvecraw_i)  \big] \quad \text{ for all } i \in [N], k \in [K].$$
\end{corollary}

\begin{proof}
    Fix $f : \mathbb{R}^q \to \mathbb{R}$ to be a measurable function. Take $g_{1} : \mathbb{R}^q \to \mathbb{R}$ to be a constant function that satisfies $g_{1}(v)=1$ for all $v \in \mathbb{R}^q$ and recall that by definition for any $i' \in [N]$, $W_{i'}^{(0)}=1$. By applying Lemma \ref{lemma:SimplifyWeightedExpectations2Terms} in the case where $g=g_{1}$ and $k'=0$, it follows that for each $k \in [K]$ and $i,i' \in [N]$ such that $i \neq i'$, $$\e[ W_i^{(k)} f(\datvecraw_i) ]=\e \big[ W_i^{(k)} W_{i'}^{(0)} f(\datvecraw_i) g_{1}(\datvecraw_{i'})  \big]=\e[f(\datvecraw_i) g_{1}(\datvecraw_{i'})]=\e[f(\datvecraw_i)].$$
\end{proof}

\begin{remark}\label{remark:IdentificationAndNotation}

Recall that all random variables in this manuscript, including the uniform random variables $U_i^{(k)}$ generated for sampling and the weights $W_i^{(k)}$ that depend on them, are defined on a single probability space with probability measure $\mathbb{P}$ (see the notation conventions on Page \pageref{paragraph:ProbailitySpaceClarifications}). In the above lemma and corollary $\e[\cdot]$ denotes an expectation under $\mathbb{P}$. An expectation such as $\e[f(\datvecraw_i)]$ in Corollary \ref{cor:SimplifyWeightedExpectations1Term} simplifies and depends only on the full-data vector distribution $\mathbb{P}_{\datvecraw}$. In contrast, an expectation involving the weights, such as $\e[W_i^{(k)} f(\datvecraw_i)]$, is the expectation of quantities that depend on the randomly generated uniforms $U_i^{(k)}$ (via $W_i^{(k)}$) as well as $\datvecraw_i$ (which has a marginal distribution $\mathbb{P}_{\datvecraw}$). Results such as Corollary \ref{cor:SimplifyWeightedExpectations1Term} (or Lemma \ref{lemma:SimplifyWeightedExpectations2Terms}) are identification results, since $W_i^{(k)} f(\datvecraw_i) = 0$ for all $k \in [K]$ whenever $\xmiss_i$ is unobserved. In particular, they express expectations of functions of full-data vectors in terms of the distribution of the observed data and sampling information. Because both sides are expectations under the same probability measure $\mathbb{P}$, we do not attach subscripts to $\e$.
\end{remark}

\subsection{Impact of weights on expectations and pairwise covariances}\label{sec:WeightPropertiesFirstOrder}

The following propositions establish that weighted averages using $W_i$ are unbiased and that the weights do not induce pairwise correlations. They focus on the setting of Assumption \ref{assump:IIDUnderlyingData}, although more general results are possible because their proofs use Lemma \ref{lemma:SimplifyWeightedExpectations2Terms} and Corollary \ref{cor:SimplifyWeightedExpectations1Term}. 

\begin{proposition}\label{prop:WeightsEnableUnbiasedMeanEst}
    Under \samplingSchemeName{} and Assumption \ref{assump:IIDUnderlyingData}, for any measurable function $f: \mathbb{R}^q \to \mathbb{R}$, $\e[W_i f(\datvecraw_i)] = \e[f(\datvecraw)]$ for each $i \in [N]$ and hence $$\e \Big[ \frac{1}{N} \sum_{i=1}^N W_i f(\datvecraw_i) \Big] = \e[f(\datvecraw)].$$
\end{proposition}

\begin{proof}
    
Fix $f : \mathbb{R}^q \to \mathbb{R}$ to be a measurable function, and suppose that Assumption \ref{assump:IIDUnderlyingData} holds. By linearity of expectation, the definition for $W_i$ in \eqref{eq:aggregated_Wk}, and Corollary \ref{cor:SimplifyWeightedExpectations1Term}, for any $i \in [N]$,
$$\e[W_i f(\datvecraw_i)] = \sum_{k=1}^K c_k \e[ W_i^{(k)} f(\datvecraw_i) ]=\sum_{k=1}^K c_k \e[f(\datvecraw_i)]=\e[f(\datvecraw)],$$ 
where the last step follows from Assumption \ref{assump:IIDUnderlyingData} and because $\sum_{k=1}^K c_k=1$. Hence, by linearity of expectation, $\e \big[N^{-1} \sum_{i=1}^N W_i f(\datvecraw_i) \big] = N^{-1} \sum_{i=1}^N  \e[ W_i f(\datvecraw_i)]= \e[f(V)].$
  
\end{proof}

\begin{proposition}\label{prop:NoCovarianceMultiwaveIPW} Under \samplingSchemeName{} and Assumption \ref{assump:IIDUnderlyingData}, for any fixed, measurable functions $f,g: \mathbb{R}^q \to \mathbb{R}$ and $i,i' \in [N]$ such that $i \neq i'$, $$\cov \big( W_i  f(\datvecraw_i), W_{i'} g(\datvecraw_{i'}) \big)= \cov \big( W_i  f(\datvecraw_i),  g(\datvecraw_{i'}) \big)=0.$$
\end{proposition}

\begin{proof}
    
 Fix $f,g : \mathbb{R}^q \to \mathbb{R}$ to be measurable functions, and suppose that Assumption \ref{assump:IIDUnderlyingData} holds. Next fix $i,i' \in[N]$ such that $i \neq i'$. Note that by definition of $W_i$ in \eqref{eq:aggregated_Wk} and by applying Lemma \ref{lemma:SimplifyWeightedExpectations2Terms} and Corollary \ref{cor:SimplifyWeightedExpectations1Term}, $$\begin{aligned}
 \cov \big( W_i  f(\datvecraw_i), W_{i'} g(\datvecraw_{i'}) \big) & = \sum_{k=1}^K \sum_{k'=1}^K c_k c_{k'}  \cov \big( W_i^{(k)}  f(\datvecraw_i), W_{i'}^{(k')} g(\datvecraw_{i'}) \big)
 \\ & = \sum_{k=1}^K \sum_{k'=1}^K c_k c_{k'} \Big( \e \big[  W_i^{(k)}  W_{i'}^{(k')}  f(\datvecraw_i) g(\datvecraw_{i'}) \big] -\e \big[  W_i^{(k)} f(\datvecraw_i) \big] \e \big[ W_{i'}^{(k')} g(\datvecraw_{i'}) \big] \Big)
 \\ & = \sum_{k=1}^K \sum_{k'=1}^K c_k c_{k'} \Big( \e \big[  f(\datvecraw_i) g(\datvecraw_{i'}) \big] -\e \big[   f(\datvecraw_i) \big] \e \big[  g(\datvecraw_{i'}) \big] \Big)
 \\ & = \sum_{k=1}^K \sum_{k'=1}^K c_k c_{k'} \cov \big(   f(\datvecraw_i) , g(\datvecraw_{i'}) \big)
 \\ & =0.
 \end{aligned}$$ Above, the last step holds because by Assumption \ref{assump:IIDUnderlyingData}, $\datvecraw_i \indep \datvecraw_{i'}$, so $\cov \big(   f(\datvecraw_i) , g(\datvecraw_{i'}) \big)=0$. By a similar argument, and noting that by definition $W_{i'}^{(0)} =1$, $$\begin{aligned}
 \cov \big( W_i  f(\datvecraw_i), g(\datvecraw_{i'}) \big) & = \sum_{k=1}^K  c_k   \cov \big( W_i^{(k)}  f(\datvecraw_i), g(\datvecraw_{i'}) \big)
 \\ & = \sum_{k=1}^K  c_k \Big( \e \big[  W_i^{(k)}  W_{i'}^{(0)}  f(\datvecraw_i) g(\datvecraw_{i'}) \big] -\e \big[  W_i^{(k)} f(\datvecraw_i) \big] \e \big[ g(\datvecraw_{i'}) \big] \Big)
 \\ & = \sum_{k=1}^K  c_k  \Big( \e \big[  f(\datvecraw_i) g(\datvecraw_{i'}) \big] -\e \big[   f(\datvecraw_i) \big] \e \big[  g(\datvecraw_{i'}) \big] \Big) 
 \\ & = \sum_{k=1}^K  c_k  \cov \big(   f(\datvecraw_i) , g(\datvecraw_{i'}) \big) =0.
 \end{aligned}$$ 

 \end{proof}

\subsection{Exchangeability properties}\label{sec:ExchangeabilityPropertiesAppendix}

In this subsection we formally prove exchangeability results that hold under Assumption \ref{assump:IIDUnderlyingData} that $V_1,\dots,V_N$ are i.i.d. when the labelling strategy is symmetric (Assumption \ref{assump:CLTRgularityConditions}(i)). Some readers may find these results intuitive and prefer to skip the formal proofs in this subsection.

We first prove a more general lemma, and then present some implications of the lemma that are cited in later parts of the appendix. For the more general lemma it is convenient to define \begin{equation}\label{eq:xiDef_ObsDatVec}
\xi_i \equiv \Big(\datvecraw_i, \big(\LabelRulePi{j}(\proxyX_i) ,U_i^{(j)},I_i^{(j)} \big)_{j=1}^K \Big) \in \mathbb{R}^{q+3K} \quad \text{for each } i \in [N]. \end{equation} It is worth noting that for each $i \in [N]$, $W_i$ can be expressed as a fixed, measurable function of the corresponding $\xi_i$ value (see Equations \eqref{eq:InverseProbabilityWeights} and \eqref{eq:aggregated_Wk}).

\begin{lemma}\label{lemma:StrongerExchangeability}
   Under \samplingSchemeName{} and Assumptions \ref{assump:IIDUnderlyingData} and  \ref{assump:CLTRgularityConditions}(i), $(\xi_i)_{i=1}^N$ is an exchangeable sequence of random vectors.
\end{lemma}

\begin{proof}

For each $i \in [N]$, recursively define $$\xi_i^{(0)} \equiv  \datvecraw_i \quad \text{and} \quad  \xi_i^{(k)}  \equiv \Big(\xi_i^{(k-1)}, \LabelRulePi{k}(\proxyX_i),U_i^{(k)} ,I_i^{(k)}  \Big) \quad \text{for } k=1,\dots,K.$$ We will prove by induction that for each $k \in \{0\} \cup [K]$, $(\xi_i^{(k)})_{i=1}^N$ is an exchangeable sequence of $(q+3k)$-dimensional random vectors. By Assumption \ref{assump:IIDUnderlyingData}, $\datvecraw_1,\dots,\datvecraw_N$ are i.i.d. (and hence exchangeable), implying that $(\xi_i^{(k)})_{i=1}^N$ is an exchangeable sequence of random vectors in the case where $k=0$. 

Next fix $k \in [K]$ and assume the inductive hypothesis that $(\xi_i^{(k-1)})_{i=1}^N$ is an exchangeable sequence of random vectors; we will show that $(\xi_i^{(k)})_{i=1}^N$ is an exchangeable sequence of random vectors. To do this, first fix a permutation $\tau \in \mathsf{S}_N$. Next note that by Assumption \ref{assump:CLTRgularityConditions}(i), $$\Big( \xi_{\tau(i)}^{(k-1)}, \LabelRulePi{k}(\proxyX_{\tau(i)}) \Big)_{i=1}^N =  \Big( \xi_{\tau(i)}^{(k-1)},\big[ \LabelStrategy{k} (\cd_{k-1}) \big](\proxyX_{\tau(i)}) \Big)_{i=1}^N=\Big( \xi_{\tau(i)}^{(k-1)},\big[ \LabelStrategy{k} (\cd_{k-1}^{(\tau)}) \big](\proxyX_{\tau(i)}) \Big)_{i=1}^N.$$ Next let $h(\cdot)$ and $h^*(\cdot)$ be fixed measurable functions such that $h(\xi_i^{(k-1)}) \equiv \big( (I_{i}^{(j)},I_{i}^{(j)} \cdot \xmiss_{i})_{j=1}^{k-1},\proxyX_{i} \big)$ and $h^*(\xi_i^{(k-1)}) \equiv \proxyX_i$ for each $i \in [N]$. Recalling the definition of $\cd_{k-1}^{(\tau)}$ and $\cd_{k-1}$, observe that $\cd_{k-1}^{(\tau)}=\big( h(\xi_{\tau(i')}^{(k-1)}) \big)_{i'=1}^N$ and $\cd_{k-1}=\big( h(\xi_{i'}^{(k-1)}) \big)_{i'=1}^N$. By the inductive hypothesis $(\xi_i^{(k-1)})_{i=1}^N$ is exchangeable, and hence the sequence of $N$ (long) random vectors (in $\mathbb{R}^{(N+1)(q+3k-3)}$) indexed by $i$ satisfies
$$\Big( \xi_i^{(k-1)}, (\xi_{i'}^{(k-1)})_{i'=1}^N \Big)_{i=1}^N \stackrel{\text{dist}}{=} \Big( \xi_{\tau(i)}^{(k-1)}, (\xi_{\tau(i')}^{(k-1)})_{i'=1}^N \Big)_{i=1}^N.$$ Combining previous results, since $\LabelStrategy{k}$ is a fixed labelling strategy, $$\begin{aligned} \Big( \xi_{\tau(i)}^{(k-1)}, \LabelRulePi{k}(\proxyX_{\tau(i)}) \Big)_{i=1}^N & =  \Big( \xi_{\tau(i)}^{(k-1)},\big[ \LabelStrategy{k} (\cd_{k-1}^{(\tau)}) \big](\proxyX_{\tau(i)}) \Big)_{i=1}^N 
\\ & = \Bigg( \xi_{\tau(i)}^{(k-1)},\Big[ \LabelStrategy{k} \Big( \big( h(\xi_{\tau(i')}^{(k-1)}) \big)_{i'=1}^N \Big) \Big] \Big( h^* ( \xi_{\tau(i)}^{(k-1)} ) \Big) \Bigg)_{i=1}^N
\\ & \stackrel{\text{dist}}{=}  \Bigg( \xi_{i}^{(k-1)},\Big[ \LabelStrategy{k} \Big( \big( h(\xi_{i'}^{(k-1)}) \big)_{i'=1}^N \Big) \Big] \Big( h^* ( \xi_i^{(k-1)} ) \Big) \Bigg)_{i=1}^N 
\\ & = \Big( \xi_i^{(k-1)},\big[ \LabelStrategy{k} (\cd_{k-1}) \big](\proxyX_i) \Big)_{i=1}^N 
\\ & = \Big( \xi_i^{(k-1)}, \LabelRulePi{k} (\proxyX_i) \Big)_{i=1}^N.\end{aligned}$$ Next recall that $U_1^{(k)},\dots,U_N^{(k)} \stackrel{\text{i.i.d.}}{\sim} \text{Unif}[0,1]$ are the quantities that are generated independently of $\big( \xi_i^{(k-1)}, \LabelRulePi{k} (\proxyX_i) \big)_{i=1}^N$ in \samplingSchemeName{} and hence by the previous result $$\Big( \xi_{\tau(i)}^{(k-1)}, \LabelRulePi{k}(\proxyX_{\tau(i)}),U_{\tau(i)}^{(k)} \Big)_{i=1}^N \stackrel{\text{dist}}{=} \Big( \xi_i^{(k-1)}, \LabelRulePi{k} (\proxyX_i),U_i^{(k)} \Big)_{i=1}^N.$$ This further implies $$\Big( \xi_{\tau(i)}^{(k-1)}, \LabelRulePi{k}(\proxyX_{\tau(i)}),U_{\tau(i)}^{(k)},\mathbbm{1} \bigl\{ U_{\tau(i)}^{(k)} \leq \LabelRulePi{k} (\proxyX_{\tau(i)}) \bigr\} \Big)_{i=1}^N \stackrel{\text{dist}}{=} \Big( \xi_i^{(k-1)}, \LabelRulePi{k} (\proxyX_i),U_i^{(k)},\mathbbm{1} \bigl\{ U_i^{(k)} \leq \LabelRulePi{k} (\proxyX_i) \bigr\} \Big)_{i=1}^N.$$ Recalling that under \samplingSchemeName{,} for $i \in [N]$, $I_i^{(k)} = \mathbbm{1}\bigl\{ U_i^{(k)} \leq \LabelRulePi{k} (\proxyX_i) \bigr\}$ and that by definition for $i \in [N]$, $\xi_i^{(k)}=\big( \xi_i^{(k-1)}, \LabelRulePi{k} (\proxyX_i),U_i^{(k)} ,I_i^{(k)} \big)$, we can simplify each side of the above expression to get that $\big( \xi_{\tau(i)}^{(k)} \big)_{i=1}^N \stackrel{\text{dist}}{=} \big( \xi_i^{(k)} \big)_{i=1}^N$. Because this argument that $\big( \xi_{\tau(i)}^{(k)} \big)_{i=1}^N \stackrel{\text{dist}}{=} \big( \xi_i^{(k)} \big)_{i=1}^N$ holds for any fixed permutation $\tau \in \mathsf{S}_N$, $\big( \xi_i^{(k)} \big)_{i=1}^N$ is an exchangeable sequence of random vectors. 

Thus we have shown that when we assume the inductive hypothesis that for some $k \in [K]$, $\big( \xi_i^{(k-1)} \big)_{i=1}^N$ is an exchangeable sequence of random vectors it follows that $\big( \xi_i^{(k)} \big)_{i=1}^N$ is an exchangeable sequence of random vectors. Recalling that $\big( \xi_i^{(0)} \big)_{i=1}^N$ is an exchangeable sequence of random vectors, we have thus shown by induction that for each $k \in \{0\} \cup [K]$, $(\xi_i^{(k)})_{i=1}^N$ is an exchangeable sequence of random vectors. To complete the proof observe that by definition $\xi_i=\xi_i^{(K)}$ for each $i \in [N]$, and hence $(\xi_i)_{i=1}^N$ is an exchangeable sequence of random vectors.
\end{proof}

The following result is a consequence of this lemma. Throughout the text it enables us to alternate between expectations that are taken with respect to the $i$th (and $i'$th) sample and expectations that are taken with respect to the 1st (and 2nd) sample.

\begin{proposition}\label{prop:ExchangeableData}
    Under \samplingSchemeName{} and Assumptions \ref{assump:IIDUnderlyingData} and  \ref{assump:CLTRgularityConditions}(i), $(I_i, W_i,\goodX_i,\proxyX_i)_{i=1}^N$ is an exchangeable sequence of random vectors.
\end{proposition}

\begin{proof}
    
Recall that for each $i \in [N]$, $\datvecraw_i \equiv ( \xobs_i,\txmiss_i,\xmiss_i)$, $\proxyX_i \equiv (\xobs_i,\txmiss_i)$, $\goodX_i \equiv ( \xobs_i,\xmiss_i)$, and by Equations \eqref{eq:InverseProbabilityWeights} and \eqref{eq:aggregated_Wk} and the definition of $I_i$, $$W_i = \sum_{k=1}^K c_k \Big( \prod_{j=1}^{k-1}  \frac{1-I_i^{(j)}}{1-\LabelRulePi{j} (\proxyX_i)} \Big) \frac{I_i^{(k)}}{\LabelRulePi{k} (\proxyX_i)} \quad \text{and} \quad I_i= \mathbbm{1} \bigl\{   I_i^{(k)} =1 \text{ for some } k \in [K] \bigr\}.$$ Because by Definition \eqref{eq:xiDef_ObsDatVec}, $$\xi_i \equiv \Big(\datvecraw_i, \big(\LabelRulePi{j}(\proxyX_i),U_i^{(j)}, I_i^{(j)} \big)_{j=1}^K \Big) \quad \text{for each } i \in [N],$$ it follows that there exists a fixed, measurable function $h: \mathbb{R}^{q+3K} \to \mathbb{R}^{2+2p}$ such that $h(\xi_i)=(I_i,W_i,\goodX_i,\proxyX_i)$ for each $i \in [N]$. Letting $h$ be such a fixed function, and noting that by Lemma \ref{lemma:StrongerExchangeability} $(\xi_i)_{i=1}^N$ is an exchangeable sequence of random vectors, it follows that $$\big(h(\xi_i) \big)_{i=1}^N= (I_i,W_i,\goodX_i,\proxyX_i)_{i=1}^N$$ is an exchangeable sequence of random vectors.

\end{proof}

We next state a corollary to Proposition \ref{prop:ExchangeableData}. This corollary is used when establishing consistency of the covariance estimators in Appendix \ref{sec:AppendixProovingConsistentCovarianceAndValidCIs}. 

\begin{corollary}\label{cor:SameDistWithEstimators}
    In the setting of Proposition \ref{prop:ExchangeableData}, for each fixed $N \in \mathbb{Z}_+$ and $i \in [N]$, $(\htc, \hga, \goodX_i, \proxyX_i)$ and $(\htc, \hga, \goodX_1, \proxyX_1)$ have the same joint distribution.
\end{corollary}

\begin{proof}
    Fix $N \in \mathbb{Z}_+$ and $i \in [N]$. Fix $\tau \in \mathsf{S}_N$ to be some permutation such that $\tau(i)=1$. Next, define $$\cd_{\text{all}} \equiv  \big( (W_{i'},\goodX_{i'},\proxyX_{i'}) \big)_{i'=1}^N \quad \text{and} \quad \cd_{\text{all}}^{(\tau)} \equiv  \big( (W_{\tau(i')},\goodX_{\tau(i')},\proxyX_{\tau(i')}) \big)_{i'=1}^N,$$ to be a stacking of $N$ $\mathbb{R}^{1+2p}$-valued random vectors and a permuted version of this stacking, respectively.  Next observe that as a consequence of Proposition \ref{prop:ExchangeableData}, $\cd_{\text{all}}$ and $\cd_{\text{all}}^{(\tau)}$ have the same joint distribution. Since $(\goodX_i,\proxyX_i)$ are components of $\cd_{\text{all}}$ and $(\goodX_{\tau(i)},\proxyX_{\tau(i)})$ are correspondingly indexed components of $\cd_{\text{all}}^{(\tau)}$, $$(\cd_{\text{all}},\goodX_i,\proxyX_i) \stackrel{\text{dist}}{=}(\cd_{\text{all}}^{(\tau)},\goodX_{\tau(i)},\proxyX_{\tau(i)}).$$ 

    Now recall from \eqref{eq:PTDComponentEstimatorsDEF} that $\htc= \argmin_{\theta \in \Theta} \sum_{i'=1}^N W_{i'} l_{\theta}(\goodX_{i'})$ and $\hga = \argmin_{\theta \in \Theta}  \sum_{i'=1}^N  l_{\theta}(\proxyX_{i'}).$ Note that the values of $\htc$ and $\hga$ do not change under permutations of $\big( (W_{i'},\goodX_{i'},\proxyX_{i'}) \big)_{i'=1}^N$. In particular, there exists a fixed function $h_N$ such that for all $\sigma \in \mathsf{S}_N$, $$h_N \Big( \big( (W_{\sigma(i')},\goodX_{\sigma(i')},\proxyX_{\sigma(i')}) \big)_{i'=1}^N \Big)=(\htc,\hga),$$ regardless of the permutation $\sigma \in \mathsf{S}_N$. Choosing such a function $h_N$, note that $(\htc,\hga)=h_N( \cd_{\text{all}} )$ and $h_N( \cd_{\text{all}}^{(\tau)} )=(\htc,\hga)$. Combining this with a result displayed above and recalling that $\tau(i)=1$, $$\big( \htc,\hga, \goodX_i,\proxyX_i \big) =\big( h_N( \cd_{\text{all}} ), \goodX_i,\proxyX_i \big) \stackrel{\text{dist}}{=} \big( h_N( \cd_{\text{all}}^{(\tau)} ), \goodX_{\tau(i)},\proxyX_{\tau(i)} \big)=\big( \htc,\hga, \goodX_1,\proxyX_1 \big).$$

\end{proof}

\subsection{Implications of \texorpdfstring{$L^1$}{L1} convergence assumption}\label{sec:L1ConvergenceImplicationsAppendix}

The following lemma enables us to switch from a statement about $L^1$ convergence of the labelling rule assumed by Assumption \ref{assump:CLTRgularityConditions}(ii), to other notions of $L$-type convergence of the labelling rule. The proof leverages Assumption \ref{assump:LabellingRuleOverlap} that the labelling probabilities are bounded away from $0$ and $1$.

\begin{lemma}\label{lemma:InverseLpConvergence}
    Under \samplingSchemeName{} and Assumptions \ref{assump:IIDUnderlyingData}, \ref{assump:LabellingRuleOverlap}, \ref{assump:CLTRgularityConditions}(i) and \ref{assump:CLTRgularityConditions}(ii), for any $j \in [K]$, $r \in \{-1,0,1\} $, $l \geq 1$, $$\lim_{ N \to \infty } \e \Big[ \Big| \frac{1}{\phi_r \big(\LabelRulePi{j}(\proxyX_i) \big)} - \frac{1}{\phi_r \big( \LimitingLabelRulePi{j}(\proxyX_i) \big)} \Big|^l \Big]=0,$$ where $i$ is a fixed positive integer that does not increase with $N$, and $\phi_r$ is defined at \eqref{eq:ComplementOr1Operator}.
\end{lemma}

\begin{proof}
   Fix $j \in [K]$, $r \in \{-1,0,1\}$, $l \geq 1$, and $i$ to be a positive integer that does not increase as $N \to \infty$. Next consider the three functions $h_{-1},h_0,h_1 : [b,1-b] \to \mathbb{R}$ given by $$h_{-1}(t) \equiv \frac{1}{\phi_{-1}(t)} = \frac{1}{1-t}, \quad h_{0}(t) \equiv \frac{1}{\phi_{0}(t)}=\frac{1}{t}, \quad \text{and} \quad h_1(t) \equiv \frac{1}{\phi_1(t)}=1 \quad \text{for } t \in [b,1-b],$$ where $b \in (0,1/2)$. Note that regardless of the value of $r \in \{-1,0,1\}$, the maximum derivative of $h_r(\cdot)$ satisfies $\sup_{t \in [b,1-b]} \vert h_r'(t) \vert \leq b^{-2}$. Since by Assumptions \ref{assump:LabellingRuleOverlap} and \ref{assump:CLTRgularityConditions}(ii), $\LabelRulePi{j}(\proxyX_i)  \in [b,1-b]$ and $\LimitingLabelRulePi{j}(\proxyX_i) \in [b,1-b]$ almost surely, $$\begin{aligned}
\e \Big[ \Big| \frac{1}{\phi_r \big( \LabelRulePi{j}(\proxyX_i) \big) } - \frac{1}{ \phi_r \big( \LimitingLabelRulePi{j}(\proxyX_i) \big) } \Big|^l \Big] & = \e \big[ \big| h_r \big( \LabelRulePi{j}(\proxyX_i) \big) - h_r \big( \LimitingLabelRulePi{j}(\proxyX_i) \big) \big|^l \big] 
\\ & \leq  \e \big[ \big( \sup_{t \in [b,1-b]} \vert h_r'(t) \vert \cdot \vert  \LabelRulePi{j}(\proxyX_i) -  \LimitingLabelRulePi{j}(\proxyX_i) \vert \big)^l \big] 
\\ & \leq b^{-2l} \e \big[  \big| \LabelRulePi{j}(\proxyX_i)- \LimitingLabelRulePi{j}(\proxyX_i) \big|^l  \big]
\\ & \leq b^{-2l} \e \big[  \big| \LabelRulePi{j}(\proxyX_i)- \LimitingLabelRulePi{j}(\proxyX_i) \big|  \big]
\\ & \leq b^{-2l} \e \big[  \big| \LabelRulePi{j}(\proxyX_1)- \LimitingLabelRulePi{j}(\proxyX_1) \big|  \big].
\end{aligned}$$ Above the penultimate step holds because $\big| \LabelRulePi{j}(\proxyX_i)- \LimitingLabelRulePi{j}(\proxyX_i) \big| \in [0,1]$ almost surely and $l \geq 1$, while the final step follows because $\LimitingLabelRulePi{j}$ is a fixed function while $\big(\proxyX_i,\LabelRulePi{j}(\proxyX_i) \big)_{i=1}^N$ is an exchangeable sequence of random vectors (as a direct consequence of Lemma \ref{lemma:StrongerExchangeability}). Considering the limit as $N \to \infty$, by the above inequality and the labelling rule convergence assumption (Assumption  \ref{assump:CLTRgularityConditions}(ii)), $$0 \leq \limsup_{N \to \infty} \e \Big[ \Big| \frac{1}{\phi_r \big( \LabelRulePi{j}(\proxyX_i) \big) } - \frac{1}{ \phi_r \big( \LimitingLabelRulePi{j}(\proxyX_i) \big) } \Big|^l \Big] \leq  b^{-2l} \limsup_{N \to \infty} \e \big[  \big| \LabelRulePi{j}(\proxyX_1)- \LimitingLabelRulePi{j}(\proxyX_1) \big|  \big] =0,$$ implying the desired result.

\end{proof}

\subsection{Asymptotic simplifications for expectations with squared weights}\label{sec:SquaredWeightsPropertiesAppendix}
In this subsection, we study the asymptotic properties of expectations of quantities multiplied by $(W_1^{(k)})^2 (W_2^{(k')})^2$ for each $k,k' \in \{0\} \cup [K]$. The results are ultimately used to prove Proposition \ref{prop:Vanishing2ndOrderCov} which establishes that covariances of the form $\cov\big(W_1^2 f(\datvecraw_1), W_2^2 g(\datvecraw_2) \big)$ decay asymptotically. The only place in this manuscript where results from this subsection are used is in Proposition \ref{prop:ConsistencyOfSigmaUpperLeft} when proving the consistency of the covariance matrix estimators $\hat{\Sigma}_{11}$, $\hat{\Sigma}_{12}$, and $\hat{\Sigma}_{22}$ defined at \eqref{eq:HessianAndSigmaEstimators}, which have $W_i^2$ terms in their formulas. Therefore, readers less interested in the theoretical development of the consistency results for $\hat{\Sigma}_{11}$, $\hat{\Sigma}_{12}$, and $\hat{\Sigma}_{22}$ (and the confidence interval validity result they help establish (Proposition \ref{prop:AsymptoticallyValidCIs})) may prefer to skim this subsection, looking only at the Proposition \ref{prop:SimplifyHigherOrderExpectations}, Corollary \ref{cor:SimplifyHigherOrderExpectationsOneTerm}, and Proposition \ref{prop:Vanishing2ndOrderCov} statements while skipping the proofs and lemmas below.


For each $k,k',\kappa \in \{0 \} \cup [K]$, it helps to define $$\Xi^{(k,k',\kappa)} \equiv \frac{ \prod_{j=1}^{\kappa}  \big( W_1^{(k,j)} \big)^2 \big( W_2^{(k',j)} \big)^2  }{ \prod_{j=\kappa+1}^{K}   \phi_{\text{sgn}(j-k)}\big( \LimitingLabelRulePi{j}(\proxyX_1) \big)  \cdot  \phi_{\text{sgn}(j-k')} \big( \LimitingLabelRulePi{j}(\proxyX_2) \big)  }$$ and a variant without the $\big(W_1^{(k,\kappa)} \big)^2$ and $\big(W_2^{(k',\kappa)}\big)^2$ terms in the numerator by $$\Xi_{-}^{(k,k',\kappa)} \equiv \frac{ \prod_{j=1}^{\kappa-1}  \big( W_1^{(k,j)} \big)^2 \big( W_2^{(k',j)} \big)^2  }{ \prod_{j=\kappa+1}^{K}   \phi_{\text{sgn}(j-k)}\big( \LimitingLabelRulePi{j}(\proxyX_1) \big) \cdot \phi_{\text{sgn}(j-k')} \big(\LimitingLabelRulePi{j}(\proxyX_2) \big)  }.$$ In the above definitions, when the lower limit is larger than the upper limit of a product, that product is defined to be $1$. Next recall definition \eqref{eq:PiBarProd1tok} that
    $$\LimitingLabelRulePi{1:k}(\argxProxy) = \LimitingLabelRulePi{k}(\argxProxy) \prod_{j=1}^{k-1} \big( 1-  \LimitingLabelRulePi{j}(\argxProxy) \big) \quad \text{for } k \in [K], \argxProxy \in \xProxySpace, \quad \text{and define }  \LimitingLabelRulePi{1:0}(\argxProxy) \equiv 1 \quad \text{for } \argxProxy \in \xProxySpace.$$ Note that by formula \eqref{eq:ComplementOr1Operator} for $\phi_{-1},\phi_0,\phi_1$, \begin{equation}\label{eq:XiAt0}
        \Xi^{(k,k',0)} = \frac{1}{\LimitingLabelRulePi{1:k}(\proxyX_1)  \cdot \LimitingLabelRulePi{1:k'}(\proxyX_2) }.
    \end{equation} Meanwhile recalling the alternative formulas for $W_1^{(k)}$ and $W_2^{(k')}$ at \eqref{eq:WkFormulaAlt}, 
    
    \begin{equation}\label{eq:XiAtK}
        \Xi^{(k,k',K)}=   \big( \prod_{j=1}^{K} W_1^{(k,j)} \big)^2 \big( \prod_{j=1}^{K} W_2^{(k',j)} \big)^{2} =\big( W_1^{(k)} \big)^2 \big( W_2^{(k')} \big)^{2}.
    \end{equation}

The next lemma allows us to relate $\Xi^{(k,k',\kappa)}$ to $\Xi_-^{(k,k',\kappa)}$ quantities while the lemma after it allows us to relate $\Xi_-^{(k,k',\kappa)}$ to $\Xi^{(k,k',\kappa-1)}$ quantities. These lemmas are then applied in an alternating fashion to relate expectations with $\Xi^{(k,k',K)}$ to those with $\Xi^{(k,k',0)}$.

\begin{lemma}\label{lemma:DropSquares} Under \samplingSchemeName{,} for any fixed, measurable $f,g : \mathbb{R}^q \to \mathbb{R}$, and for any $k,k' \in \{0 \} \cup [K]$ and $\kappa \in [K]$, 
    $$\e \Big[ \Xi^{(k,k',\kappa)} f(\datvecraw_1) g(\datvecraw_2) \Big] = \e \Bigg[ \frac{\Xi_-^{(k,k',\kappa)} f(\datvecraw_1) g(\datvecraw_2)}{ \phi_{\textnormal{sgn}(\kappa-k)} \big(\LabelRulePi{\kappa}(\proxyX_1) \big) \cdot  \phi_{\textnormal{sgn}(\kappa-k')} \big(\LabelRulePi{\kappa}(\proxyX_2) \big)}  \Bigg].$$
\end{lemma}

\begin{proof} Fix measurable functions $f,g : \mathbb{R}^q \to \mathbb{R}$ and further fix $k,k'\in \{0 \} \cup [K]$ and $\kappa \in [K]$. Next note that for $j \leq \kappa-1$, $W_1^{(k,j)}$,$W_2^{(k',j)}$, $\tilde{X}_1$ and $\tilde{X}_2$ are all measurable functions of $\cd_{\kappa-1}$, while $\phi_r$ and $\bar{\pi}^{(j)}$ are nonrandom, measurable functions for $r \in \{-1,0,1\}$ and $j \in [K]$. Hence by definition of $\Xi^{(k,k',\kappa)}$ and $\Xi_{-}^{(k,k',\kappa)}$ and by the tower property, 

$$\begin{aligned}
\e \Big[ \Xi^{(k,k',\kappa)} f(\datvecraw_1) g(\datvecraw_2) \Big] & = \e \Bigg[ \frac{ \prod_{j=1}^{\kappa}  \big( W_1^{(k,j)} \big)^2 \big( W_2^{(k',j)} \big)^{2} f(\datvecraw_1) g(\datvecraw_2) }{ \prod_{j=\kappa+1}^{K}  \phi_{\text{sgn}(j-k)} \big(\LimitingLabelRulePi{j}(\proxyX_1) \big) \cdot \phi_{\text{sgn}(j-k')} \big(\LimitingLabelRulePi{j}(\proxyX_2) \big) } \Bigg]
\\ & = \e \Bigg[ \frac{ \prod_{j=1}^{\kappa-1}  \big( W_1^{(k,j)} \big)^2 \big( W_2^{(k',j)} \big)^{2} \e \big[ \big( W_1^{(k,\kappa)} \big)^2 \big( W_2^{(k',\kappa)} \big)^{2}  f(\datvecraw_1) g(\datvecraw_2) \giv \cd_{\kappa-1} \big] }{ \prod_{j=\kappa+1}^{K}   \phi_{\text{sgn}(j-k)} \big(\LimitingLabelRulePi{j}(\proxyX_1) \big) \cdot \phi_{\text{sgn}(j-k')} \big(\LimitingLabelRulePi{j}(\proxyX_2) \big)  }  \Bigg]
\\ & = \e \Big[ \Xi_{-}^{ (k,k',\kappa )} \cdot \e \big[ \big( W_1^{(k,\kappa)} \big)^2 \big( W_2^{(k',\kappa)} \big)^{2}  f(\datvecraw_1) g(\datvecraw_2) \giv \cd_{\kappa-1} \big] \Big]
\\ & = \e \Bigg[ \frac{\Xi_{-}^{ (k,k',\kappa )} \cdot \e[ f(\datvecraw_1) g(\datvecraw_2) \giv \cd_{\kappa-1}]}{\phi_{\textnormal{sgn}(\kappa-k)} \big( \LabelRulePi{\kappa}(\proxyX_1) \big) \cdot \phi_{\textnormal{sgn}(\kappa-k')} \big( \LabelRulePi{\kappa}(\proxyX_2) \big) } \Bigg]
\\ & = \e \Bigg[ \frac{\Xi_{-}^{ (k,k',\kappa )}  f(\datvecraw_1) g(\datvecraw_2) }{\phi_{\textnormal{sgn}(\kappa-k)} \big( \LabelRulePi{\kappa}(\proxyX_1) \big) \cdot \phi_{\textnormal{sgn}(\kappa-k')} \big( \LabelRulePi{\kappa}(\proxyX_2) \big) } \Bigg].
\end{aligned}$$ Above, the penultimate step follows from a direct application of Lemma \ref{lemma:SimplifyCondExpectations} in the case where $i=1$, $i'=2$, and $s=s'=2$.
\end{proof}

\begin{lemma}\label{lemma:DropOrder1TermsViaBoundedness} Suppose \samplingSchemeName{} is conducted in such a way that Assumptions \ref{assump:IIDUnderlyingData}, \ref{assump:LabellingRuleOverlap}, \ref{assump:CLTRgularityConditions}(i), and \ref{assump:CLTRgularityConditions}(ii) hold. For any fixed, measurable $f,g : \mathbb{R}^q \to \mathbb{R}$ such that for some $\eta > 0$, $\e \big[ \vert f(\datvecraw) \vert ^{1+\eta} \big]<\infty$ and $\e \big[ \vert g(\datvecraw) \vert^{1+\eta} \big]<\infty$, and for any $k,k' \in \{0 \} \cup [K]$, $\kappa \in [K]$,
    $$\e \Bigg[ \frac{\Xi_-^{(k,k',\kappa)} f(\datvecraw_1) g(\datvecraw_2)}{ \phi_{\textnormal{sgn}(\kappa-k)} \big(\LabelRulePi{\kappa}(\proxyX_1) \big) \cdot  \phi_{\textnormal{sgn}(\kappa-k')} \big(\LabelRulePi{\kappa}(\proxyX_2) \big) } \Bigg] =\e \Big[ \Xi^{(k,k',\kappa-1 )} f(\datvecraw_1) g(\datvecraw_2) \Big]+o(1),$$ where $o(1)$ denotes a term that converges to $0$ as $N \to \infty$.
\end{lemma}

\begin{proof}

Fix $f,g : \mathbb{R}^q \to \mathbb{R}$ such that for some $\eta > 0$, $\e \big[ \vert f(\datvecraw) \vert ^{1+\eta} \big]<\infty$ and $\e \big[ \vert g(\datvecraw) \vert^{1+\eta} \big]<\infty$, and then let $\eta$ be sufficiently small to satisfy this condition on $f,g$. Further fix $k,k' \in \{0 \} \cup [K]$, and $\kappa \in [K]$. After fixing these quantities, for the purposes of this proof, it will be convenient to define $R \equiv \Xi_-^{(k,k',\kappa)} f(\datvecraw_1) g(\datvecraw_2)$, $$T \equiv \big[ \phi_{\textnormal{sgn}(\kappa-k)} \big(\LabelRulePi{\kappa}(\proxyX_1) \big) \big]^{-1}, \quad \text{and} \quad  T' \equiv \big[ \phi_{\textnormal{sgn}(\kappa-k')} \big(\LabelRulePi{\kappa}(\proxyX_2) \big) \big]^{-1},$$ and similarly using the limiting labelling rules define$$\bar{T} \equiv \big[ \phi_{\textnormal{sgn}(\kappa-k)} \big(\LimitingLabelRulePi{\kappa}(\proxyX_1) \big) \big]^{-1} \quad \text{and} \quad \bar{T}' \equiv \big[ \phi_{\textnormal{sgn}(\kappa-k')} \big(\LimitingLabelRulePi{\kappa}(\proxyX_2) \big) \big]^{-1}.$$

 Now note that $$\begin{aligned}
    \e [ R T T' ] &  = \e [ R \bar{T} \bar{T}' ] + \e \big[ R \bar{T} (T'-\bar{T}' ) ]+ \e \big[ R T' ( T -\bar{T} )  \big]
    \\ & = \e \big[ \Xi^{(k,k',\kappa-1 )} f(\datvecraw_1) g(\datvecraw_2) \big]+  \e \big[ R \bar{T} (T'-\bar{T}' ) ]+ \e \big[ R T' ( T -\bar{T} )  \big]. 
\end{aligned}$$ Above the second step follows from plugging in the formulas for $R,\bar{T},\bar{T'}$, and $\Xi_-^{(k,k',\kappa)}$ into the first term $\e[R \bar{T} \bar{T}']$ and by recalling the formula for $\Xi^{(k,k',\kappa-1 )}$. 

We will show that the second and third terms in the above equation are $o(1)$ in the sense that they converge to $0$ as $N \to \infty$, and the proof will be completed by plugging in the values for $R$, $T$, and $T'$ into the left hand side of the above equation.
To do this, note that by applying Jensen's inequality and subsequently Hölder's inequality for the $1+\eta$ and $(1+\eta)/\eta$ pair, $$\begin{aligned}
   \big| \e \big[ R \bar{T} ( T' -\bar{T}')  \big] \big| & \leq    \e \big[  \vert R \bar{T}  \vert \cdot \vert T' -\bar{T}' \vert \big] 
   \\ & \leq \big(\e \big[ \vert R \bar{T} \vert^{1+\eta} \big] \big)^{1/(1+\eta)} \cdot \big( \e \big[ \vert T' - \bar{T}'\vert^{1+1/\eta} \big]  \big)^{\eta/(1+\eta)}
   \\ & \leq \Big(\e \Big[ \Big| \frac{\Xi_-^{(k,k',\kappa)} f(\datvecraw_1) g(\datvecraw_2)}{ \phi_{\textnormal{sgn}(\kappa-k)} \big(\LimitingLabelRulePi{\kappa}(\proxyX_1) \big) } \Big|^{1+\eta} \Big] \Big)^{1/(1+\eta)} \cdot \big( \e \big[ \vert T' - \bar{T}'\vert^{1+1/\eta} \big]  \big)^{\eta/(1+\eta)}
    \\ & \leq \Big( \big( b^{-4K} \big)^{1+\eta} \cdot \e \big[  \big|  f(\datvecraw_1) g(\datvecraw_2)  \big|^{1+\eta} \big] \Big)^{1/(1+\eta)} \cdot \big( \e \big[ \vert T' - \bar{T}'\vert^{1+1/\eta} \big]  \big)^{\eta/(1+\eta)}
    \\ & \leq  b^{-4K}   \big( \e [  \vert  f(\datvecraw) \vert^{1+\eta} ] \cdot \e[ \vert g(\datvecraw) \vert^{1+\eta} ] \big)^{1/(1+\eta)} \cdot \big( \e \big[ \vert T' - \bar{T}'\vert^{1+1/\eta} \big]  \big)^{\eta/(1+\eta)}.
\end{aligned}$$ Above the last step holds by Assumption \ref{assump:IIDUnderlyingData}, while the penultimate step follows from Assumption \ref{assump:LabellingRuleOverlap} that labelling probabilities lie within $[b,1-b]$ for some $b \in (0,1/2)$. In particular, under Assumption \ref{assump:LabellingRuleOverlap}, for any $j \in [K]$ and  $r \in \{-1,0,1\}$, $\vert W_1^{(k,j)} \vert \leq b^{-1}$ and $\vert W_2^{(k',j)} \vert \leq b^{-1}$ almost surely,  $\big[ \phi_{r} \big(\LimitingLabelRulePi{\kappa}(\argxProxy) \big) \big]^{-1} \leq b^{-1}$ for all $\argxProxy \in \xProxySpace$, and thus almost surely $$ \Bigg| \frac{\Xi_-^{(k,k',\kappa)} }{\phi_{\textnormal{sgn}(\kappa-k)} \big(\LimitingLabelRulePi{\kappa}(\proxyX_1) \big) } \Bigg| =  \Bigg| \frac{  \phi_{\text{sgn}(\kappa-k')} \big(\LimitingLabelRulePi{\kappa}(\proxyX_2) \big)  \prod_{j=1}^{\kappa-1}  \big( W_1^{(k,j)} \big)^2 \big( W_2^{(k',j)} \big)^2  }{ \prod_{j=\kappa}^{K}  \phi_{\text{sgn}(j-k)}\big( \LimitingLabelRulePi{j}(\proxyX_1) \big) \cdot \phi_{\text{sgn}(j-k')} \big(\LimitingLabelRulePi{j}(\proxyX_2) \big)  } \Bigg| \leq b^{-4K}.$$ 

Next, note that by applying Lemma \ref{lemma:InverseLpConvergence} in the case where $i=2$ and $l=1+1/\eta>1$, $$\lim_{N \to \infty} \e \big[ \vert T' - \bar{T}'\vert^{1+1/\eta} \big] = \lim_{N \to \infty} \e \Big[ \Big| \frac{1}{\phi_{\textnormal{sgn}(\kappa-k')} \big(\LabelRulePi{\kappa}(\proxyX_2) \big)} - \frac{1}{\phi_{\textnormal{sgn}(\kappa-k')} \big(\LimitingLabelRulePi{\kappa}(\proxyX_2) \big)} \Big|^{1+1/\eta} \Big]=0.$$ 
Since $b^{-4K} < \infty$, $\e \big[ \vert f(\datvecraw) \vert^{1+\eta} \big]<\infty$ and $\e \big[ \vert g(\datvecraw) \vert^{1+\eta} \big]<\infty$, we can combine this result with a previous inequality to get that $$\limsup_{N \to \infty} \big| \e \big[ R \bar{T} ( T' -\bar{T}')  \big] \big| \leq  b^{-4K}   \big( \e [  \vert  f(\datvecraw) \vert^{1+\eta} ] \cdot \e[ \vert g(\datvecraw) \vert^{1+\eta} ] \big)^{1/(1+\eta)} \cdot \big( \limsup_{N \to \infty} \e \big[ \vert T' - \bar{T}'\vert^{1+1/\eta} \big]  \big)^{\eta/(1+\eta)}=0.$$
We have thus shown that $\e \big[ R \bar{T} ( T' -\bar{T}')  \big]=o(1)$. A similar argument (which establishes and leverages the facts that $\e \big[ \vert R T' \vert^{1+\eta} \big] < \infty$ and $\lim_{N \to \infty} \e \big[ \vert T - \bar{T} \vert^{1+1/\eta}  \big] = 0$) shows that $\e \big[ R T' ( T -\bar{T} )  \big]=o(1)$.

To complete the proof, recalling an earlier formula for $ \e [ R T T' ] $, $$\begin{aligned}
    \e [ R T T' ]   & = \e \big[ \Xi^{(k,k',\kappa-1 )} f(\datvecraw_1) g(\datvecraw_2) \big]+  \e \big[ R \bar{T} (T'-\bar{T}' ) ]+ \e \big[ R T' ( T -\bar{T} )  \big]
    \\ & = \e \big[ \Xi^{(k,k',\kappa-1 )} f(\datvecraw_1) g(\datvecraw_2) \big]+o(1). 
\end{aligned}$$ Plugging in the formulas for $R$, $T$, and $T'$ into the left hand side of the above expression gives the desired result.
\end{proof}

\begin{proposition}\label{prop:SimplifyHigherOrderExpectations}
    Suppose \samplingSchemeName{} is conducted in such a way that Assumptions \ref{assump:IIDUnderlyingData}, \ref{assump:LabellingRuleOverlap}, \ref{assump:CLTRgularityConditions}(i), and \ref{assump:CLTRgularityConditions}(ii) hold. For any fixed, measurable $f,g : \mathbb{R}^q \to \mathbb{R}$ such that for some $\eta > 0$, $\e \big[ \vert f(\datvecraw) \vert ^{1+\eta} \big]<\infty$ and $\e \big[ \vert g(\datvecraw) \vert^{1+\eta} \big]<\infty$, and for any $k,k' \in \{0 \} \cup [K]$,
    $$\lim_{N \to \infty} \e \Big[ \big( W_1^{(k)} \big)^2 \big( W_2^{(k')} \big)^2 f(\datvecraw_1) g(\datvecraw_2) \Big] =\e \Bigg[ \frac{f(\datvecraw_1) }{\LimitingLabelRulePi{1:k}(\proxyX_1) }  \Bigg] \e \Bigg[ \frac{g(\datvecraw_2)}{ \LimitingLabelRulePi{1:k'}(\proxyX_2) } \Bigg].$$
\end{proposition}

\begin{proof}
    Fix  $f,g : \mathbb{R}^q \to \mathbb{R}$ to be measurable functions such that for some $\eta > 0$, $\e \big[ \vert f(\datvecraw) \vert ^{1+\eta} \big]<\infty$ and $\e \big[ \vert g(\datvecraw) \vert^{1+\eta} \big]<\infty$. Next fix $k, k' \in \{0 \} \cup [K]$. Note that by Equations \eqref{eq:XiAt0} and \eqref{eq:XiAtK} and because $(\datvecraw_1,\proxyX_1) \indep (\datvecraw_2,\proxyX_2)$ under Assumption \ref{assump:IIDUnderlyingData}, it will suffice to show that $$\lim_{N \to \infty} \e \big[\Xi^{(k,k',K)} f(\datvecraw_1) g(\datvecraw_2) \big] = \e \big[\Xi^{(k,k',0)} f(\datvecraw_1) g(\datvecraw_2) \big].$$
    
    To prove the above claim we will show that
    \begin{equation}\label{eq:RecursiveStep2ndOrderExpectations}
        \e \big[\Xi^{(k,k',\kappa)} f(\datvecraw_1) g(\datvecraw_2) \big] = \e \big[\Xi^{(k,k',\kappa-1 )} f(\datvecraw_1) g(\datvecraw_2) \big] +o(1) \quad \text{for all } \kappa \in [K],
    \end{equation} where $o(1)$ is a term that converges to $0$ as $N \ \to \infty$, and then recursively apply \eqref{eq:RecursiveStep2ndOrderExpectations}. To verify the claim in \eqref{eq:RecursiveStep2ndOrderExpectations}, fix $\kappa \in [K]$ and note that by a direct application of Lemma \ref{lemma:DropSquares} and then Lemma \ref{lemma:DropOrder1TermsViaBoundedness}, 
$$\begin{aligned} \e \big[ \Xi^{(k,k',\kappa)} f(\datvecraw_1) g(\datvecraw_2) \big] & =  \e \Bigg[ \frac{\Xi_-^{(k,k',\kappa)} f(\datvecraw_1) g(\datvecraw_2)}{ \phi_{\textnormal{sgn}(\kappa-k)} \big(\LabelRulePi{\kappa}(\proxyX_1) \big) \cdot \phi_{\textnormal{sgn}(\kappa-k')} \big(\LabelRulePi{\kappa}(\proxyX_2) \big) } \Bigg]
\\ & =  \e \Big[ \Xi^{(k,k',\kappa-1 )} f(\datvecraw_1) g(\datvecraw_2) \Big]+o(1).
\end{aligned}$$ The above argument holds for any $\kappa \in [K]$, proving \eqref{eq:RecursiveStep2ndOrderExpectations}. 

As a consequence of  \eqref{eq:RecursiveStep2ndOrderExpectations}, $$\lim_{N \to \infty} \sum_{\kappa=1}^K \Bigg( \e \big[\Xi^{(k,k', \kappa )} f(\datvecraw_1) g(\datvecraw_2) \big] - \e \big[\Xi^{(k,k', \kappa -1)} f(\datvecraw_1) g(\datvecraw_2) \big] \Bigg)=0.$$ By simplifying the telescoping sum in the above expression and adding $\e \big[\Xi^{(k,k',0)} f(\datvecraw_1) g(\datvecraw_2) \big]$ (a quantity that does not vary with $N$) to each side of the above equation, $$\begin{aligned}
    \lim_{N \to \infty}  \e \big[\Xi^{(k,k',K)} f(\datvecraw_1) g(\datvecraw_2) \big] & = \e \big[\Xi^{(k,k',0)} f(\datvecraw_1) g(\datvecraw_2) \big] 
    \\ & = \e \Bigg[ \frac{ f(\datvecraw_1) g(\datvecraw_2)}{ \LimitingLabelRulePi{1:k}(\proxyX_1) \cdot \LimitingLabelRulePi{1:k'}(\proxyX_2) } \Bigg] 
     \\ & = \e \Bigg[ \frac{ f(\datvecraw_1)}{\LimitingLabelRulePi{1:k}(\proxyX_1) } \Bigg] \cdot \e \Bigg[ \frac{ g(\datvecraw_2)}{ \LimitingLabelRulePi{1:k'}(\proxyX_2) } \Bigg].
\end{aligned}$$ Above the 2nd equality follows from plugging in the formula for $\Xi^{(k,k',0)}$ at \eqref{eq:XiAt0}. The final equality above follows because by Assumption \ref{assump:IIDUnderlyingData}, $\datvecraw_1 \indep \datvecraw_2$, where $\datvecraw_1=(\proxyX_1,\xmiss_1)$ and $\datvecraw_2=(\proxyX_2,\xmiss_2)$, while $\LimitingLabelRulePi{1:k}$ and $\LimitingLabelRulePi{1:k'}$ are fixed functions that do not depend on the observed data. The proof is completed by plugging in the formula at \eqref{eq:XiAtK} that $\Xi^{(k,k',K)}=\big(W_1^{(k)} \big)^2 \big(W_2^{(k')} \big)^2$ into the left hand side of the above equation.

\end{proof}

\begin{corollary}
    \label{cor:SimplifyHigherOrderExpectationsOneTerm}
    Suppose \samplingSchemeName{} is conducted in such a way that Assumptions \ref{assump:IIDUnderlyingData}, \ref{assump:LabellingRuleOverlap}, \ref{assump:CLTRgularityConditions}(i), and \ref{assump:CLTRgularityConditions}(ii) hold. For any fixed, measurable $f: \mathbb{R}^q \to \mathbb{R}$ such that for some $\eta > 0$, $\e \big[ \vert f(\datvecraw) \vert ^{1+\eta} \big]<\infty$ and for any $k \in  [K]$ and $i \in \{1,2\}$,
    $$\lim_{N \to \infty} \e \Big[ \big( W_i^{(k)} \big)^2 f(\datvecraw_i) \Big] =\e \Bigg[ \frac{f(\datvecraw) }{\LimitingLabelRulePi{1:k}(\proxyX) }  \Bigg].$$
\end{corollary}

\begin{proof}
    Fix $k \in [K]$ and a measurable $f: \mathbb{R}^q \to \mathbb{R}$ such that for some $\eta > 0$, $\e \big[ \vert f(\datvecraw) \vert ^{1+\eta} \big]<\infty$. Next recall that $W_2^{(0)}=1$ by definition and let $g_1 : \mathbb{R}^q \to \mathbb{R}$ be a constant function defined by $g_1(v)=1$ for all $v \in \mathbb{R}^q$. Applying Proposition \ref{prop:SimplifyHigherOrderExpectations} in the case where $g=g_1$, $k'=0$, $k=k$, and $f=f$, $$\lim_{N \to \infty} \e \Big[ \big( W_1^{(k)} \big)^2  f(\datvecraw_1)  \Big] =\e \Bigg[ \frac{f(\datvecraw_1) }{\LimitingLabelRulePi{1:k}(\proxyX_1) }  \Bigg] \cdot 1 = \e \Bigg[ \frac{f(\datvecraw) }{\LimitingLabelRulePi{1:k}(\proxyX) }  \Bigg],$$ where the last step follows from Assumption \ref{assump:IIDUnderlyingData}. To complete the proof note that by Proposition \ref{prop:ExchangeableData} establishing exchangeability, $\big( W_1^{(k)} \big)^2  f(\datvecraw_1)$ and $\big( W_2^{(k)} \big)^2  f(\datvecraw_2)$ have the same distribution. Hence $$\lim_{N \to \infty} \e \Big[ \big( W_2^{(k)} \big)^2  f(\datvecraw_2)  \Big]=\lim_{N \to \infty} \e \Big[ \big( W_1^{(k)} \big)^2  f(\datvecraw_1)  \Big]=\e \Bigg[ \frac{f(\datvecraw) }{\LimitingLabelRulePi{1:k}(\proxyX) }  \Bigg].$$
\end{proof}

The following proposition is a consequence of previous results. It is eventually used when establishing the consistency of the covariance matrix estimators $\hat{\Sigma}_{11}$, $\hat{\Sigma}_{12}$, and $\hat{\Sigma}_{22}$ defined at \eqref{eq:HessianAndSigmaEstimators}, which have $W_i^2$ terms in their formulas.
\begin{proposition}\label{prop:Vanishing2ndOrderCov} Suppose \samplingSchemeName{} is conducted in such a way that Assumptions \ref{assump:IIDUnderlyingData}, \ref{assump:LabellingRuleOverlap}, \ref{assump:CLTRgularityConditions}(i), and \ref{assump:CLTRgularityConditions}(ii) hold. For any fixed, measurable $f,g : \mathbb{R}^q \to \mathbb{R}$ such that for some $\eta>0$, $\e \big[\vert f(\datvecraw) \vert^{1+\eta}]<\infty$ and $\e \big[\vert g(\datvecraw) \vert^{1+\eta}]<\infty$,

\begin{equation}\label{eq:DecayingPairwiseCov2ndOrder}
     \lim_{N \to \infty} \cov\big(W_1^2 f(\datvecraw_1), W_2^2 g(\datvecraw_2) \big)=0.
 \end{equation} 
    
\end{proposition}

\begin{proof}

 Fix measurable $f,g : \mathbb{R}^q \to \mathbb{R}$ such that for some $\eta>0$, $\e \big[\vert f(\datvecraw) \vert^{1+\eta}]<\infty$ and $\e \big[\vert g(\datvecraw) \vert^{1+\eta}]<\infty$. 
To establish \eqref{eq:DecayingPairwiseCov2ndOrder}, it helps to first show that
 \begin{equation}\label{eq:DecayingPairwiseCov2ndOrder1Term}
     \lim_{N \to \infty} \cov\big( \big(W_1^{(k)} \big)^2 f(\datvecraw_1), \big( W_2^{(k')} \big)^2 g(\datvecraw_2) \big) = 0 \quad \text{ for any} \quad k,k' \in [K].
 \end{equation} To do this fix $k,k' \in [K]$ and let $o(1)$ denote terms that converge to $0$ as $N \to \infty$. Note that by applying Proposition \ref{prop:SimplifyHigherOrderExpectations} and Corollary \ref{cor:SimplifyHigherOrderExpectationsOneTerm} $$\begin{aligned} \cov\big( \big(W_1^{(k)} \big)^2 f(\datvecraw_1), \big( W_2^{(k')} \big)^2 g(\datvecraw_2) \big) & =  \e \big[ \big(W_1^{(k)} \big)^2 \big( W_2^{(k')} \big)^2 f(\datvecraw_1) g(\datvecraw_2) \big] -  \e \big[ \big(W_1^{(k)} \big)^2 f(\datvecraw_1) \big] \e \big[ \big( W_2^{(k')} \big)^2 g(\datvecraw_2) \big]
 \\ & = \e \Bigg[ \frac{f(\datvecraw_1) }{\LimitingLabelRulePi{1:k}(\proxyX_1) }  \Bigg] \e \Bigg[ \frac{g(\datvecraw_2)}{ \LimitingLabelRulePi{1:k'}(\proxyX_2) } \Bigg] 
+o(1)
\\ & \quad - \Bigg( \e \Bigg[ \frac{f(\datvecraw) }{\LimitingLabelRulePi{1:k}(\proxyX) }  \Bigg] +o(1) \Bigg)  \Bigg(\e \Bigg[ \frac{g(\datvecraw)}{ \LimitingLabelRulePi{1:k'}(\proxyX) } \Bigg]+o(1) \Bigg)
\\ & = o(1),
\end{aligned}$$ where the last step uses the fact that $\LimitingLabelRulePi{1:k}$, $\LimitingLabelRulePi{1:k'}$, $f$, and $g$ are all fixed functions while $\datvecraw_1$ and $\datvecraw_2$ have the same distribution as $V$. This confirms \eqref{eq:DecayingPairwiseCov2ndOrder1Term} holds.

 To establish \eqref{eq:DecayingPairwiseCov2ndOrder} note that by Fact \ref{fact:WsqFormula}, $W_1^2=\sum_{k=1}^K c_k^2 \big( W_1^{(k)} \big)^2$ and $W_2^{2} =\sum_{k'=1}^K c_{k'}^2 \big( W_2^{(k')} \big)^{2}$. 
 Thus by \eqref{eq:DecayingPairwiseCov2ndOrder1Term}, $$\lim_{N \to \infty} \cov\big(W_1^2 f(\datvecraw_1), W_2^{2} g(\datvecraw_2) \big)  = \sum_{k=1}^K \sum_{k'=1}^K c_k^2 c_{k'}^{2}  \lim_{N \to \infty} \cov\big( \big(W_1^{(k)} \big)^2 f(\datvecraw_1), \big( W_2^{(k')} \big)^{2} g(\datvecraw_2) \big) =0.$$ 

 \end{proof}
 
\section{Asymptotics for M-estimators}\label{sec:AsymptoticsForMestimators}

In this appendix we prove the consistency, $\sqrt{N}$-consistency, and asymptotic linearity of M-estimators under \samplingSchemeName{.} Taken together, these results establish Theorem \ref{theorem:AsymptoticLInearMestsStacked} and Corollary \ref{cor:PTDEstAsymptoticallyLinear} from the main text.

Throughout this appendix, it will be convenient to define empirical (weighted) averaging operators $\bfs_N$, $\bfts_N$ and $\bftp_N$, such that for any $p^* \in \mathbb{Z}_+$ and fixed function $f : \mathbb{R}^p \to \mathbb{R}^{p^*}$,\begin{equation}\label{eq:AveragingOperatorDefs} \bfs_N f = \frac{1}{N} \sum_{i=1}^N W_i f(\goodX_i), \quad \bfts_N f= \frac{1}{N} \sum_{i=1}^N W_i f(\proxyX_i), \quad \text{and} \quad \bftp_N f = \frac{1}{N} \sum_{i=1}^N f(\proxyX_i). \end{equation} 
Note that using this notation, $\htPTD$ consists of the following 3 component estimators given in Equation \eqref{eq:PTDComponentEstimatorsDEF}, \begin{equation}\label{eq:PTDComponentEstimatorsOperatorNotation}
    \htc = \argmin_{\theta \in \Theta} \bfs_N l_{\theta}, \quad \hgc = \argmin_{\theta \in \Theta} \bfts_N l_{\theta}, 
    \quad \text{and} \quad \hga = \argmin_{\theta \in \Theta} \bftp_N l_{\theta}.
\end{equation} In this appendix we study the asymptotic properties of these 3 component estimators. 

In the next subsection we present a helpful lemma, and some properties of the empirical weighted averaging operators $\bfs_N$, $\bfts_N$ and $\bftp_N$. These properties of averages are subsequently used to prove consistency, $\sqrt{N}$-consistency and asymptotic linearity of the 3 estimators in \eqref{eq:PTDComponentEstimatorsOperatorNotation}.

\subsection{Helpful properties of weighted averages}

\begin{lemma}\label{lemma:VarianceOfWeightedAveragesGeneric}
    Under \samplingSchemeName{} and Assumptions \ref{assump:IIDUnderlyingData} and \ref{assump:LabellingRuleOverlap}, for any fixed, measurable function $g: \mathbb{R}^q \to \mathbb{R}$, $$\var \Big( \frac{1}{N} \sum_{i=1}^N W_i g(\datvecraw_i) \Big) \leq \frac{1}{N b^{2K}} \cdot \e[g^2(\datvecraw)].$$  If in addition $\e[g^2(\datvecraw)] < \infty$, $$\frac{1}{N} \sum_{i=1}^N W_i g(\datvecraw_i) \xrightarrow{p} \e[g(\datvecraw)] \quad \text{as} \quad N \to \infty.$$ 
\end{lemma}

\begin{proof}
    Fix a measurable function $g : \mathbb{R}^q \to \mathbb{R}$. Recall that by Proposition \ref{prop:NoCovarianceMultiwaveIPW}, for each $i,i' \in [N]$ such that $i \neq i'$, $\cov \big(  W_i g(\datvecraw_i),  W_{i'} g(\datvecraw_{i'}) \big) =0$  and that by Fact \ref{fact:WeightBound}, $\vert W_i \vert \leq b^{-K}$ almost surely. Hence,
    
    $$\var \Big( \frac{1}{N} \sum_{i=1}^N W_i g(\datvecraw_i) \Big) 
         = \frac{1}{N^2} \sum_{i=1}^N \var \big(   W_i g(\datvecraw_i) \big) 
       \leq \frac{1}{N^2} \sum_{i=1}^N \e \big[  W_i^2 g^2(\datvecraw_i) \big] \leq \frac{1}{N^2 b^{2K}} \sum_{i=1}^N \e \big[  g^2(\datvecraw_i) \big].$$ Because $\e[g^2(\datvecraw_i)]=\e[g^2(\datvecraw)]$ for each $i \in [N]$ by Assumption \ref{assump:IIDUnderlyingData}, the above inequality simplifies to $$\var \Big( \frac{1}{N} \sum_{i=1}^N W_i g(\datvecraw_i) \Big) \leq \frac{1 }{N b^{2K}} \cdot \e[g^2(\datvecraw)].$$


    To complete the proof, further suppose that $\e[g^2(\datvecraw)]<\infty$. Note that under this constraint the above inequality implies that $$\lim_{N \to \infty} \var \Big( \frac{1}{N} \sum_{i=1}^N W_i g(\datvecraw_i) \Big)=0.$$ By Proposition \ref{prop:WeightsEnableUnbiasedMeanEst}, $\e[\frac{1}{N} \sum_{i=1}^N W_i g(\datvecraw_i)]=\e[g(\datvecraw)]$, so by Chebyshev's inequality, for any $\epsilon>0$, $$\limsup_{N \to \infty} \mathbb{P} \Big( \Big| \frac{1}{N} \sum_{i=1}^N W_i g(\datvecraw_i)  - \e[g(\datvecraw)] \Big| > \epsilon \Big) \leq \limsup_{N \to \infty} \frac{\var \big( \frac{1}{N} \sum_{i=1}^N W_i g(\datvecraw_i) \big)}{\epsilon^2 } =0.$$ Thus $\frac{1}{N} \sum_{i=1}^N W_i g(\datvecraw_i) \xrightarrow{p} \e[g(\datvecraw)]$ as $N \to \infty$ when $ \e[g^2(\datvecraw)]<\infty.$

\end{proof}

The following corollary gives some notable properties of the empirical weighted averaging operators $\bfs_N$ and $\bfts_N$. 

\begin{corollary}\label{cor:PropertiesOfWeightedAveragees} Under \samplingSchemeName{} and Assumptions \ref{assump:IIDUnderlyingData} and \ref{assump:LabellingRuleOverlap}, for any $p^* \in \mathbb{Z}_+$ and fixed, measurable $f,\tilde{f}: \mathbb{R}^p \to \mathbb{R}^{p^*}$ the following three properties hold.

\begin{enumerate}[(I)]
    \item Unbiasedness: $\e[\bfs_N f]=\e[f(\goodX)]$ and $\e[\bfts_N \tilde{f}]=\e[\tilde{f}(\proxyX)]$.
 
    \item Consistency: if we further suppose that $\e \big[[f(\goodX)]_j^2 \big] < \infty$ and $\e \big[[\tilde{f}(\proxyX)]_j^2 \big] < \infty$ for each $j \in [p^*]$, then as $N \to \infty$, $\bfs_N f \xrightarrow{p} \e[f(\goodX)]$ and $\bfts_N \tilde{f} \xrightarrow{p} \e[\tilde{f}(\proxyX)]$.

    \item Variance upper bound: If $p^*=1$, $$\var(\bfs_N f) \leq \frac{\e  [ f^2(\goodX)  ]}{N b^{2K}}  , \quad \text{and} \quad \var(\bfts_N \tilde{f}) \leq \frac{\e  [ \tilde{f}^2(\proxyX)  ]}{N b^{2K}}.$$
\end{enumerate}

\end{corollary}

\begin{proof}

    To prove (I) and (II), fix $p^* \in \mathbb{Z}_+$ and measurable functions $f,\tilde{f}: \mathbb{R}^p \to \mathbb{R}^{p^*}$. 
    For each $j \in [p^*]$ define $f_j,\tilde{f}_j : \mathbb{R}^p \to \mathbb{R}$ to be functions that give the $j$th component of $f$ and $\tilde{f}$, given by $f_j(x)=e_j^\tran f(x)$ and $\tilde{f}_j(x)=e_j^\tran \tilde{f}(x)$ for $x \in \mathbb{R}^p$.
    
    To prove (I), note that 
    for each $j \in [p^*]$, by Proposition \ref{prop:WeightsEnableUnbiasedMeanEst} and because $\proxyX_i$ and $\goodX_i$ are each given by a subset of entries of $\datvecraw_i$, $$e_j^\tran \e[\bfs_N f]=\e \Big[ \frac{1}{N} \sum_{i=1}^N W_i f_j(\goodX_i) \Big]=\e[f_j(\goodX)] \quad \text{and} \quad e_j^\tran \e[ \bfts_N \tilde{f}]=\e \Big[ \frac{1}{N} \sum_{i=1}^N W_i \tilde{f}_j(\proxyX_i) \Big]=\e[\tilde{f}_j(\proxyX)].$$ Since the above expression holds for each $j \in [p^*]$, $\e[\bfs_N f]=\e[f(\goodX)]$ and $\e[ \bfts_N \tilde{f}]=\e[\tilde{f}(\proxyX)]$, proving (I).

    To prove (II), further suppose that for each $j \in [p^*]$, $\e[f_j^2(\goodX)]< \infty$ and $\e[\tilde{f}_j^2(\proxyX)]< \infty$. Next fix $j \in [p^*]$. 
    Note that because $\proxyX_i$ and $\goodX_i$ are each given by a subset of entries of $\datvecraw_i$, by Lemma \ref{lemma:VarianceOfWeightedAveragesGeneric}, as $N \to \infty$, $$e_j^\tran \bfs_N f = \frac{1}{N} \sum_{i=1}^N W_i f_j(\goodX_i) \xrightarrow{p} \e[f_j(\goodX)] \quad \text{and} \quad e_j^\tran \bfts_N \tilde{f} = \frac{1}{N} \sum_{i=1}^N W_i \tilde{f}_j(\proxyX_i) \xrightarrow{p} \e[\tilde{f}_j(\proxyX)].$$ Thus we have shown that for each $j \in [p^*]$, as  $N \to \infty$, $e_j^\tran \bfs_N f \xrightarrow{p} \e[f_j(\goodX)]$ and $e_j^\tran \bfts_N \tilde{f} \xrightarrow{p} \e[\tilde{f}_j(\proxyX)]$. Because entrywise convergence in probability implies convergence in probability (e.g., Theorem 2.7 in \cite{VanderVaartTextbook}), it follows that $\bfs_N f \xrightarrow{p} \e[f(\goodX)]$ and $\bfts_N \tilde{f} \xrightarrow{p} \e[\tilde{f}(\proxyX)]$ as $N \to \infty$, proving (II).

    To prove (III), fix $f,\tilde{f}: \mathbb{R}^p \to \mathbb{R}$.  
    Because $\proxyX_i$ and $\goodX_i$ are each given by a subset of entries of $\datvecraw_i$, by Lemma \ref{lemma:VarianceOfWeightedAveragesGeneric},  $$\var(\bfs_N f) =\var \Big( \frac{1}{N} \sum_{i=1}^N W_i f(\goodX_i) \Big) \leq \frac{ \e  [ f^2(\goodX)  ]}{N b^{2K}},   \quad \text{and}$$ $$\var(\bfts_N \tilde{f}) =\var \Big( \frac{1}{N} \sum_{i=1}^N W_i \tilde{f}(\proxyX_i) \Big)\leq \frac{ \e  [ \tilde{f}^2(\proxyX)  ]}{N b^{2K}},$$ completing the proof of (III).
\end{proof}

\subsection{Proof of point estimator consistency}

In the following proposition we establish consistency of $\htc$, $\hgc$, and $\hga$. The proof leverages the convexity of the loss function and uses an argument similar to that seen in \cite{PPI++}.

\begin{proposition}\label{prop:ComponentEstimatorConsistency}
    Under \samplingSchemeName{} and Assumptions \ref{assump:IIDUnderlyingData}-- \ref{assump:SmoothEnoughForAsymptoticLineariaty}, $$\htc \xrightarrow{p} \thetaTarg, \quad \hgc \xrightarrow{p} \gammaTarg,\quad  \text{and} \quad \hga \xrightarrow{p} \gammaTarg \quad \text{as } N \to \infty.$$
\end{proposition}

\begin{proof}

First note that under Assumptions \ref{assump:IIDUnderlyingData} and \ref{assump:SmoothEnoughForAsymptoticLineariaty}, $\hga \xrightarrow{p} \gammaTarg$, by standard consistency arguments for M-estimators based on i.i.d. samples (e.g., one can show this by applying Proposition 1 of \cite{PPI++} for the case where $\lambda=0$). We thus will show that $\htc \xrightarrow{p} \thetaTarg$, and we remark that the proof that $\hgc \xrightarrow{p} \gammaTarg$ follows from an analogous argument.

To do this fix $\epsilon_1>0$. For any $\epsilon>0$, define $B(\thetaTarg;\epsilon) \equiv \{ \theta \in \mathbb{R}^d \ : \ \vert \vert \theta-\thetaTarg \vert \vert_2 \leq \epsilon \}$ to be a ball of radius $\epsilon$ about $\thetaTarg$. Next fix $\epsilon \in (0,\epsilon_1]$ to be small enough such that $B(\thetaTarg;\epsilon) \subset \Theta$ (which is possible by Assumption \ref{assump:SmoothEnoughForAsymptoticLineariaty}(ii)) and such that $B(\thetaTarg;\epsilon) \subset \mathcal{L}_{\thetaTarg}$, where $\mathcal{L}_{\thetaTarg}$ is the neighborhood about $\thetaTarg$ guaranteed by Assumption \ref{assump:SmoothEnoughForAsymptoticLineariaty}(iv). Recall $M: \mathbb{R}^p \to (0,\infty)$ is a function guaranteed by Assumption \ref{assump:SmoothEnoughForAsymptoticLineariaty}(iv) that satisfies $\e[M^2(\goodX)]<\infty$ and $\vert l_{\theta}(\argx)-l_{\theta'}(\argx) \vert < M(\argx) \vert \vert \theta-\theta' \vert \vert$ for all $\argx \in \xSpace$ and $\theta,\theta' \in \mathcal{L}_{\thetaTarg}$. 

We start by using a standard covering argument to show that \begin{equation}\label{eq:UniformEmiricalLossConvergence}
\sup\limits_{\theta \in B(\thetaTarg;\epsilon)} \vert \bfs_N l_{\theta} - L(\theta) \vert \xrightarrow{p} 0, \quad \text{as } N \to \infty.\end{equation} To do this fix $\eta>0$. Let $\delta = \eta/(3 \e[M(\goodX)])  \in (0,\infty)$ and take $\mathcal{C}_{\delta} \subset B(\thetaTarg;\epsilon)$ to be a finite $\delta$-covering of $B(\thetaTarg;\epsilon)$. Observe that for any $\theta \in B(\thetaTarg; \epsilon)$ there exists a $\theta' \in \mathcal{C}_{\delta}$ such that $\vert \vert \theta-\theta' \vert \vert_2 < \delta$, and hence for such a choice of $\theta' \in \mathcal{C}_{\delta}$, $$\begin{aligned} \vert \bfs_N l_{\theta} -L(\theta) \vert & \leq \vert \bfs_N l_{\theta} -\bfs_N l_{\theta'} \vert + \vert \bfs_N l_{\theta'} -L(\theta') \vert + \vert L(\theta')-L(\theta) \vert
\\ & \leq \vert \bfs_N  (l_{\theta}-l_{\theta'}) \vert + \vert \e[l_{\theta'}(\goodX)-l_{\theta}(\goodX)] \vert  + \vert \bfs_N l_{\theta'} -L(\theta') \vert
\\ & \leq \bfs_N \vert l_{\theta}-l_{\theta'} \vert +  \e[ \vert l_{\theta'}(\goodX)-l_{\theta}(\goodX) \vert ]   + \vert \bfs_N l_{\theta'} -L(\theta') \vert
\\ & \leq \vert \vert \theta - \theta'  \vert \vert_2 \cdot \bfs_N M  +  \e[ M(\goodX) \cdot \vert \vert \theta' - \theta  \vert \vert_2  ]   + \vert \bfs_N l_{\theta'} -L(\theta') \vert
\\ &  \leq \big( \bfs_N M + \e[M(\goodX)] \big) \delta + \sup_{\theta^* \in \mathcal{C}_{\delta}} \vert \bfs_N l_{\theta^*} -L(\theta^*) \vert
\\ & =  \frac{2 \eta}{3} + (\bfs_N M-\e[M(\goodX)])\delta +\sup_{\theta^* \in \mathcal{C}_{\delta}} \vert \bfs_N l_{\theta^*} -L(\theta^*) \vert.
\end{aligned}$$
Defining $T_1 (\delta) \equiv (\bfs_N M-\e[M(\goodX)])\delta$ and $T_2 (\delta) \equiv \sup_{\theta^* \in \mathcal{C}_{\delta}} \vert \bfs_N l_{\theta^*} -L(\theta^*) \vert$, since the above inequality holds for any $\theta \in B(\thetaTarg; \epsilon)$, it follows that $$\sup_{\theta \in B(\thetaTarg; \epsilon)} \vert \bfs_N l_{\theta} -L(\theta) \vert \leq  \frac{2 \eta}{3} + T_1(\delta) +T_2(\delta),$$ and hence $$\mathbb{P} \Big( \sup_{\theta \in B(\thetaTarg; \epsilon)} \vert \bfs_N l_{\theta} -L(\theta) \vert > \eta \Big) \leq \mathbb{P}\Big(T_1(\delta)+T_2(\delta) > \frac{\eta}{3} \Big) .$$ Since $\e[M^2(\goodX)]< \infty$, by Corollary  \ref{cor:PropertiesOfWeightedAveragees}, $\bfs_N M \xrightarrow{p} \e[M(\goodX)]$ and thus $T_1(\delta) \xrightarrow{p}0$. In addition for any $\theta^* \in \mathcal{C}_{\delta}$, $\vert \theta^* -\thetaTarg \vert < \epsilon$, so by the Cauchy-Schwarz inequality and Assumption \ref{assump:SmoothEnoughForAsymptoticLineariaty}(iv), $$\begin{aligned} \e[l_{\theta^*}^2(\goodX)] & = \e[l_{\thetaTarg}^2(\goodX)]+\e[\vert l_{\theta^*}(\goodX) -l_{\thetaTarg}(\goodX) \vert^2] + 2 \e \big[l_{\thetaTarg}(\goodX) \big(l_{\theta^*}(\goodX) -l_{\thetaTarg}(\goodX) \big) \big] 
\\ & \leq  \e[l_{\thetaTarg}^2(\goodX)]+\epsilon^2 \e[M^2(\goodX)] + 2\sqrt{\epsilon^2 \e[M^2(\goodX)] \e[l_{\thetaTarg}^2(\goodX)]} < \infty .\end{aligned}$$ Above the claim of finiteness follows from Assumptions \ref{assump:SmoothEnoughForAsymptoticLineariaty}(iv) and \ref{assump:SmoothEnoughForAsymptoticLineariaty}(vi). Since $\e[l_{\theta^*}^2(\goodX)] < \infty$ for all $\theta^* \in \mathcal{C}_{\delta}$, by Corollary \ref{cor:PropertiesOfWeightedAveragees}, $\bfs_N l_{\theta^*} \xrightarrow{p} L(\theta^*)$ for all $\theta^* \in \mathcal{C}_{\delta}$. Recalling $T_2 (\delta) \equiv \sup_{\theta^* \in \mathcal{C}_{\delta}} \vert \bfs_N l_{\theta^*} -L(\theta^*) \vert$ and that $\mathcal{C}_{\delta}$ is a finite set, it follows that $T_2(\delta) \xrightarrow{p} 0$. By the continuous mapping theorem and an earlier result $T_1(\delta)+T_2(\delta) \xrightarrow{p} 0$. Hence taking the limit as $N \to \infty$ of each side of an inequality displayed above implies that $$\lim_{N \to \infty} \mathbb{P} \Big( \sup_{\theta \in B(\thetaTarg; \epsilon)} \vert \bfs_N l_{\theta} -L(\theta) \vert > \eta \Big)=0.$$ Since this argument holds for any fixed $\eta>0$, result \eqref{eq:UniformEmiricalLossConvergence} holds.

Now let $\partial B(\thetaTarg;\epsilon) \equiv \{ \theta \in \mathbb{R}^d \ : \ \vert \vert \theta-\thetaTarg \vert \vert_2 =\epsilon \}$. Note that by Assumption \ref{assump:SmoothEnoughForAsymptoticLineariaty}(iv), for any $\theta,\theta' \in \partial B(\thetaTarg ; \epsilon)$, $\vert L(\theta) -L(\theta') \vert \leq \e[ \vert l_{\theta}(\goodX) -l_{\theta'}(\goodX) \vert] \leq \e[M(\goodX)] \vert \vert \theta - \theta' \vert \vert_2$. As a consequence $\theta \mapsto L(\theta)$ is continuous on $\partial B(\thetaTarg;\epsilon)$. Because $\partial B(\thetaTarg;\epsilon)$ is compact, by the Bolzano-Weierstrass theorem and the fact that $\theta \mapsto L(\theta)$ is continuous on $\partial B(\thetaTarg;\epsilon)$, there exists a $\theta_* \in \partial B(\thetaTarg;\epsilon)$ such that $L(\theta_*)= \inf_{\theta \in \partial B(\thetaTarg;\epsilon)} L(\theta)$. Letting $\delta_* \equiv L(\theta_*)-L(\thetaTarg)$, by the uniqueness of $\thetaTarg$ as the minimizer of $\theta \mapsto L(\theta)$ (see Assumption \ref{assump:SmoothEnoughForAsymptoticLineariaty}(ii)), $$\delta_*=  L(\theta_*) -L(\thetaTarg) =  \inf_{\theta \in \partial B(\thetaTarg;\epsilon)} \big( L(\theta) -L(\thetaTarg) \big) > 0.$$ 

Now fix any $\tilde{\theta} \in \Theta \setminus B(\thetaTarg;\epsilon)$ and we will find a lower bound on $\bfs_N l_{\tilde{\theta}}-\bfs_N l_{\thetaTarg}$ that does not depend on the specific choice of $\tilde{\theta}$. To do this define $$\lambda \equiv \frac{\epsilon}{\vert \vert \tilde{\theta}-\thetaTarg \vert \vert_2} \in (0,1] \quad \text{and} \quad \theta' \equiv \lambda \tilde{\theta} +(1-\lambda) \thetaTarg.$$ First observe that $\vert \vert \theta' -\thetaTarg \vert \vert_2= \lambda \vert \vert \tilde{\theta}-\thetaTarg \vert \vert_2  =\epsilon$ and hence $\theta' \in \partial B(\thetaTarg;\epsilon)$. 
Also note that by the definition of convexity and by Assumption \ref{assump:SmoothEnoughForAsymptoticLineariaty}(i), for all $\argx \in \xSpace$ $$l_{\theta'}(\argx) \leq \lambda l_{\tilde{\theta}}(\argx)+(1-\lambda) l_{\thetaTarg}(\argx) \Rightarrow l_{\tilde{\theta}}(\argx)-l_{\thetaTarg}(\argx) \geq \frac{1}{\lambda} \big( l_{\theta'}(\argx)-l_{\thetaTarg}(\argx)\big) \quad \text{for all } \argx \in \xSpace.$$
Since $\bfs_N$ is a linear operator that takes a positive weighted sum of $N$ terms, we can apply the $\bfs_N$ operator to each side of the above inequality to get that if $\goodX_i \in \xSpace$ for each $i \in[N]$ then, 
$$\begin{aligned} \bfs_N l_{\tilde{\theta}} - \bfs_N l_{\thetaTarg} & \geq \frac{1}{\lambda} \big( \bfs_N l_{\theta'} - \bfs_N l_{\thetaTarg} \big)
\\ & \geq \bfs_N l_{\theta'} - \bfs_N l_{\thetaTarg} 
\\ & = \big( \bfs_N l_{\theta'} -L(\theta') \big)+ \big( L(\theta')-L(\thetaTarg) \big) + \big( L(\thetaTarg)- \bfs_N l_{\thetaTarg} \big) 
\\ & \geq \delta_*- 2 \sup_{\theta \in B(\thetaTarg;\epsilon)} \vert \bfs_N l_{\theta} - L(\theta) \vert,
\end{aligned}$$
where the last inequality follows from a previous result and because $\theta' \in \partial B(\thetaTarg;\epsilon) \subset B(\thetaTarg;\epsilon)$. Note that the above lower bound holds for all $\tilde{\theta} \in \Theta \setminus B(\thetaTarg;\epsilon)$, provided that $\goodX_i \in \xSpace$ for each $i \in[N]$ (an almost sure occurrence). Hence taking the infimum of both sides of the above inequality across $\tilde{\theta} \in \Theta \setminus B(\thetaTarg;\epsilon)$, it follows that almost surely $$\inf\limits_{\tilde{\theta} \in \Theta \setminus B(\thetaTarg;\epsilon)} \bfs_N l_{\tilde{\theta}} - \bfs_N l_{\thetaTarg}  \geq \delta_*- 2 \sup_{\theta \in B(\thetaTarg;\epsilon)} \vert \bfs_N l_{\theta} - L(\theta) \vert.$$ Thus if $\vert \vert \htc - \thetaTarg \vert \vert_2 > \epsilon$, then almost surely $$\inf\limits_{\tilde{\theta} \in \Theta \setminus B(\thetaTarg;\epsilon)} \bfs_N l_{\tilde{\theta}} \leq \bfs_N l_{\thetaTarg} \Rightarrow \delta_* \leq 2 \sup_{\theta \in B(\thetaTarg;\epsilon)} \vert \bfs_N l_{\theta} - L(\theta) \vert \Rightarrow \sup_{\theta \in B(\thetaTarg;\epsilon)} \vert \bfs_N l_{\theta} - L(\theta)  \vert > \frac{\delta_*}{3}.$$ By monotonicity of probability measure and recalling that $\epsilon_1 \geq \epsilon$ it follows that $$\mathbb{P}\big( \vert \vert \htc - \thetaTarg \vert \vert_2 > \epsilon_1 \big) \leq \mathbb{P}\big( \vert \vert \htc - \thetaTarg \vert \vert_2 > \epsilon \big) \leq \mathbb{P} \Big( \sup_{\theta \in B(\thetaTarg;\epsilon)} \vert \bfs_N l_{\theta} - L(\theta)  \vert > \frac{\delta_*}{3} \Big).$$
By \eqref{eq:UniformEmiricalLossConvergence}, the right hand side goes to zero as $N \to \infty$, so by nonnegativity of probability measure, $\lim_{N \to \infty} \mathbb{P}\big( \vert \vert \htc - \thetaTarg \vert \vert_2 > \epsilon_1 \big)=0$. Since this argument holds for any fixed $\epsilon_1>0$, $\htc \xrightarrow{p} \thetaTarg$. An analogous argument shows that $\hgc \xrightarrow{p} \gammaTarg$.

\end{proof}

\subsection{Control on local empirical process via symmetrization and chaining}\label{sec:ControlLocalEmpiricalProcessResults}

For each $\theta \in \Theta$ define
 \begin{equation}\label{eq:DeltaDef}\Delta_N(\theta) \equiv \big( \bfs_N l_{\theta}-L(\theta) \big)-\big(\bfs_N l_{\thetaTarg} -L(\thetaTarg) \big) \quad \text{and} \quad \tilde{\Delta}_N(\theta) \equiv \big( \bfts_N l_{\theta}-\tilde{L}(\theta) \big)-\big(\bfts_N l_{\gammaTarg} -\tilde{L}(\gammaTarg) \big).\end{equation} To study the asymptotics of $\htc$ and $\hgc$, we will need to control fluctuations of $\Delta_N(\theta)$ and $\tilde{\Delta}_N(\theta)$ in a neighborhood of $\thetaTarg$ and $\gammaTarg$, respectively. More formally for any $\delta' > 0$ and $\theta' \in \mathbb{R}^d$ define the Euclidean ball $$B(\theta'; \delta') \equiv \{ \theta \in \mathbb{R}^d \ : \ \vert \vert \theta-\theta' \vert \vert_2 \leq \delta' \},$$ and define for each $\delta>0$ and $\delta_0>0$ the following moduli of continuity \begin{equation}\label{eq:ModuliOfContinuityDef}
     \omega_{N,\delta_0}(\delta) \equiv \sup\limits_{\substack{\theta,\theta' \in B(\thetaTarg;\delta_0) \\ \vert \vert \theta-\theta' \vert \vert_2 \leq \delta }} \Big| \Delta_N(\theta)-\Delta_N(\theta') \Big| \quad \text{and} \quad \tilde{\omega}_{N,\delta_0}(\delta) \equiv \sup\limits_{\substack{\theta,\theta' \in B(\gammaTarg;\delta_0) \\ \vert \vert \theta-\theta' \vert \vert_2 \leq \delta }} \Big| \tilde{\Delta}_N(\theta)-\tilde{\Delta}_N(\theta') \Big|.
 \end{equation} In the following lemmas, we use a chaining and symmetrization argument to upper bound $\e[\omega_{N,\delta_0}(\delta)]$ and $\e[\tilde{\omega}_{N,\delta_0}(\delta)]$.  This upperbound is later used to establish $\sqrt{N}$-consistency and asymptotic linearity of $\htc$ and $\hgc$.

To do this it helps to first define for each $\delta_0,\delta$ the function classes
\begin{equation}\label{eq:FunctionClassFDefinition}\begin{split}
     \mathcal{F}_{\delta_0}(\delta) & \equiv \Bigl\{ l_{\theta}-l_{\theta'} : \mathbb{R}^p \to \mathbb{R} \ \text{ such that } \ \theta,\theta' \in B(\thetaTarg; \delta_0), \vert \vert \theta-\theta' \vert \vert_2 \leq \delta \Bigr\}, \text{ and } \\
      \tilde{\mathcal{F}}_{\delta_0}(\delta) & \equiv \Bigl\{ l_{\theta}-l_{\theta'} : \mathbb{R}^p \to \mathbb{R} \ \text{ such that } \ \theta,\theta' \in B(\gammaTarg; \delta_0), \vert \vert \theta-\theta' \vert \vert_2 \leq \delta \Bigr\},
       \end{split}
\end{equation} and to note that the following lemma holds by applying a standard chaining argument.

\begin{lemma}\label{lemma:ApplyChaining}
    Assume Assumption \ref{assump:SmoothEnoughForAsymptoticLineariaty} holds and fix any $\delta >0$, $N \in \mathbb{Z}_+$, and $\delta_0 >0$ such that $B(\thetaTarg; \delta_0) \subseteq \mathcal{L}_{\thetaTarg}$ and $B(\gammaTarg; \delta_0) \subseteq \mathcal{L}_{\gammaTarg}$. If we let $R_1,\dots,R_N \stackrel{\textnormal{i.i.d.}}{\sim} \textnormal{Unif}\{-1,1\}$ be i.i.d. Rademacher variables then for any fixed sequences $a_1,a_2,\dots,a_N \in [0,b^{-K}]$, $\argx_1,\argx_2,\dots,\argx_N \in \xSpace$ and $\argxProxy_1,\argxProxy_2,\dots,\argxProxy_N \in \xProxySpace$, 
    $$\e\Big[ \sup_{f \in \mathcal{F}_{\delta_0}(\delta)} \big|  \sum_{i=1}^N   R_i a_i f(\argx_i)  \big|   \Big] \leq 4 b^{-K} C_u \sqrt{d} \sqrt{\delta \delta_0}   \sqrt{\sum_{i=1}^N M^2(\argx_i)} \quad \text{and}$$ $$\e\Big[ \sup_{f \in \tilde{\mathcal{F}}_{\delta_0}(\delta)} \big|  \sum_{i=1}^N   R_i a_i f(\argxProxy_i)  \big|   \Big] \leq 4 b^{-K} C_u \sqrt{d} \sqrt{\delta \delta_0}   \sqrt{\sum_{i=1}^N \tilde{M}^2(\argxProxy_i)}.$$
    Above $C_u \in (0,\infty)$ is a universal constant that does not depend on $N, \delta, \delta_0,b,K$, $d$, or the sequences $(a_i)_{i=1}^N$, $(\argx_i)_{i=1}^N$, and $(\argxProxy_i)_{i=1}^N$.
\end{lemma}

\begin{proof}
    Fix $N \in \mathbb{Z}_+, \delta >0, \delta_0>0$ such that $B(\thetaTarg; \delta_0) \subseteq \mathcal{L}_{\thetaTarg}$ and $B(\gammaTarg; \delta_0) \subseteq \mathcal{L}_{\gammaTarg}$ ($\mathcal{L}_{\thetaTarg}$ is the neighborhood from Assumption \ref{assump:SmoothEnoughForAsymptoticLineariaty}(iv) in which $\theta \mapsto l_{\theta}(\goodX)$ is locally $M(\goodX)$-Lipschitz). Further fix $a_1,a_2,\dots,a_N \in [0,b^{-K}]$, $\argx_1,\argx_2,\dots,\argx_N \in \xSpace$ and $\argxProxy_1,\argxProxy_2,\dots,\argxProxy_N \in \xProxySpace$, and let $R_1,\dots,R_N \stackrel{\text{i.i.d.}}{\sim} \text{Unif}\{-1,1\}$ be i.i.d. Rademacher variables. 
    
Note that by Assumption \ref{assump:SmoothEnoughForAsymptoticLineariaty}(iv), $$\vert l_{\theta}(\argx_i)- l_{\theta'}(\argx_i) \vert  \leq M(\argx_i) \vert \vert \theta-\theta' \vert \vert_2 \quad \text{for all } \quad \theta,\theta' \in B(\thetaTarg; \delta_0) \quad \text{and} \quad i \in [N].$$ Next define $$Z(\theta) \equiv \sum_{i=1}^N R_i a_i l_{\theta}(\argx_i) \quad \text{for all } \theta \in B(\thetaTarg; \delta_0),$$ and observe that by Definition of $\mathcal{F}_{\delta_0}(\delta)$ at \eqref{eq:FunctionClassFDefinition}, $$\e\Bigg[ \sup_{f \in \mathcal{F}_{\delta_0}(\delta)} \big|  \sum_{i=1}^N   R_i a_i f(\argx_i)  \big|   \Bigg]= \e \Bigg[ \sup_{\substack{\theta,\theta' \in B(\thetaTarg;\delta_0) \\ \vert \vert \theta-\theta' \vert \vert  \leq \delta}} \big| Z(\theta)-Z(\theta')   \big|  \Bigg].$$ Now let $\vert \vert \cdot \vert \vert_{\psi_2}$ denote the sub-Gaussian norm. Applying Lemma 2.2.8 of \cite{VDVAndWellnerTextbook}, for any $\theta,\theta' \in B(\thetaTarg;\delta_0)$,  $$\begin{aligned} \vert \vert Z(\theta)-Z(\theta') \vert \vert_{\psi_2} & = \big| \big| \sum_{i=1}^N R_i a_i \big( l_{\theta}(\argx_i)-l_{\theta'}(\argx_i) \big) \big| \big|_{\psi_2} 
\\ & \leq \sqrt{6} \Big( \sum_{i=1}^N a_i^2 \big( l_{\theta}(\argx_i)-l_{\theta'}(\argx_i) \big)^2 \Big)^{1/2} 
\\ & \leq \sqrt{6}  \Big( \sum_{i=1}^N b^{-2K} M^2(\argx_i) \vert \vert \theta - \theta' \vert \vert_2^2 \Big)^{1/2}
\\ & = \sqrt{6} \cdot b^{-K} \sqrt{\sum_{i=1}^N  M^2(\argx_i)} \cdot \vert \vert \theta-\theta' \vert \vert_2. \end{aligned}$$ Thus letting $s_N \equiv  b^{-K}  \sqrt{\sum_{i=1}^N M^2(\argx_i)}$ and letting $$\rho_N(\theta,\theta') \equiv  s_N \vert \vert \theta - \theta' \vert \vert_2 \quad \text{for all} \quad \theta,\theta' \in B(\thetaTarg;\delta_0),$$ it follows that $\vert \vert Z(\theta)-Z(\theta') \vert \vert_{\psi_2} \leq \sqrt{6} \rho_N(\theta,\theta')$ for all $\theta,\theta' \in B(\thetaTarg;\delta_0)$. Hence the stochastic process $\{ Z(\theta) \ : \ \theta \in B(\thetaTarg;\delta_0) \}$ is sub-Gaussian with respect to the semi-metric $\rho_N$. $\{ Z(\theta) \ : \ \theta \in B(\thetaTarg;\delta_0) \}$ is also a separable stochastic process because there exists a countable dense subset of $B(\thetaTarg;\delta_0)$, and for each fixed realization of $(R_i)_{i=1}^N$, $\theta \mapsto Z(\theta)$ is continuous.

Since $\{ Z(\theta) \ : \ \theta \in B(\thetaTarg;\delta_0) \}$ is a separable and sub-Gaussian process with respect to the semi-metric $\rho_N$, we can apply a variant of Dudley's integral inequality found in Corollary 2.2.9 of \cite{VDVAndWellnerTextbook} to obtain that $$\e \Bigg[ \sup_{\substack{\theta,\theta' \in B(\thetaTarg;\delta_0) \\ \rho_N(\theta, \theta')  \leq \delta s_N}} \vert Z(\theta)-Z(\theta') \vert   \Bigg] \leq C_u \int_0^{\delta s_N} \sqrt{\log \mathcal{D}(\epsilon, \rho_N)} \rd \epsilon ,$$ where $C_u \in (0,\infty)$ is a universal constant and $\mathcal{D}(\epsilon, \rho_N)$ is the maximum number of $\epsilon$-separated points in the semimetric space $\big(B(\thetaTarg;\delta_0), \rho_N \big)$. Next, for any $\epsilon>0$ and semimetric $\rho$, let $\mathcal{N}\big(\epsilon, B(\thetaTarg;\delta_0), \rho \big)$ be the minimal number of $\epsilon$-balls (with respect to the semimetric $\rho$) that cover $B(\thetaTarg;\delta_0)$, and observe that $$\mathcal{D}(\epsilon, \rho_N) \leq \mathcal{N} \big(\epsilon/2, B(\thetaTarg;\delta_0), \rho_N \big) \leq \mathcal{N} \Big( \frac{\epsilon }{2 s_N}, B(\thetaTarg;\delta_0), \vert \vert \cdot \vert \vert_2 \Big) \leq  \Big(1+ \frac{4 s_N \delta_0}{\epsilon } \Big)^d.$$ Hence, we can combine the two previous inequalities to get that $$\begin{aligned} \e \Bigg[\sup_{\substack{\theta,\theta' \in B(\thetaTarg;\delta_0) \\ \rho_N(\theta, \theta')  \leq \delta s_N}} \vert Z(\theta)-Z(\theta') \vert   \Bigg] & \leq C_u \int_0^{\delta s_N} \sqrt{d \log \big( 1+ \frac{4 s_N \delta_0}{\epsilon } \big)} \rd \epsilon 
\\ & \leq C_u \sqrt{d} \int_0^{\delta s_N} \sqrt{\frac{4 s_N \delta_0}{\epsilon}} \rd \epsilon
\\ & = 4 C_u \sqrt{d} \sqrt{\delta \delta_0} s_N.
\end{aligned}$$ Now recalling an earlier formula and noting that for $\theta,\theta' \in B(\thetaTarg;\delta_0)$, $\vert \vert \theta - \theta' \vert \vert \leq \delta$ if and only if $\rho_N(\theta,\theta') \leq \delta s_N$, it follows that  $$\begin{aligned} \e\Bigg[ \sup_{f \in \mathcal{F}_{\delta_0}(\delta)} \big|  \sum_{i=1}^N   R_i a_i f(\argx_i)  \big|   \Bigg] & = \e \Bigg[ \sup_{\substack{\theta,\theta' \in B(\thetaTarg;\delta_0) \\ \vert \vert \theta-\theta' \vert \vert  \leq \delta}} \vert Z(\theta)-Z(\theta') \vert    \Bigg]
\\ & = \e \Bigg[\sup_{\substack{\theta,\theta' \in B(\thetaTarg;\delta_0) \\ \rho_N(\theta, \theta')  \leq \delta s_N}} \vert Z(\theta)-Z(\theta') \vert   \Bigg] 
\\ & \leq 4 C_u \sqrt{d} \sqrt{\delta \delta_0} s_N=4 b^{-K} C_u \sqrt{d} \sqrt{\delta \delta_0}   \sqrt{\sum_{i=1}^N M^2(\argx_i)}.
\end{aligned}.$$ An identical argument shows that
$$\e\Big[ \sup_{f \in \tilde{\mathcal{F}}_{\delta_0}(\delta)} \big|  \sum_{i=1}^N   R_i a_i f(\argxProxy_i)  \big|   \Big] \leq 4 b^{-K} C_u \sqrt{d} \sqrt{\delta \delta_0}   \sqrt{\sum_{i=1}^N \tilde{M}^2(\argxProxy_i)}.$$
\end{proof}

\begin{lemma}\label{lemma:ModContinuityExpectationBound}
    
 Under \samplingSchemeName{} and Assumptions \ref{assump:IIDUnderlyingData}, \ref{assump:LabellingRuleOverlap}, and \ref{assump:SmoothEnoughForAsymptoticLineariaty}, there exists a constant $C \in (0,\infty)$, such that for all $\delta >0$, $N \in \mathbb{Z}_+$, and $\delta_0 >0$ such that $B(\thetaTarg; \delta_0) \subseteq \mathcal{L}_{\thetaTarg}$ and $B(\gammaTarg; \delta_0) \subseteq \mathcal{L}_{\gammaTarg}$, \begin{equation*}
    \e \big[  \omega_{N,\delta_0}(\delta) \big] \leq C \sqrt{\delta \delta_0}  N^{-1/2} \quad \text{and} \quad \e \big[  \tilde{\omega}_{N,\delta_0}(\delta) \big] \leq C \sqrt{\delta \delta_0}  N^{-1/2},
\end{equation*} where $\omega_{N,\delta_0}(\delta)$ and $\tilde{\omega}_{N,\delta_0}(\delta)$ are defined at \eqref{eq:ModuliOfContinuityDef}.
\end{lemma}
\begin{proof}

 Because the proof deriving upper bounds for $\e \big[  \omega_{N,\delta_0}(\delta) \big]$ and $\e \big[  \tilde{\omega}_{N,\delta_0}(\delta) \big]$ is lengthy, we first give an overview of the steps below: 
\begin{enumerate}
    \setlength{\itemsep}{1pt}
    \setlength{\parskip}{1pt}
    \item deriving an upper bound on $\omega_{N,\delta_0}(\delta)$ in which each term in the upper bound has an expectation that can be bounded via a symmetrization argument (see Inequality \eqref{eq:moduliOfcontinuityUpperBoundDecomposition}),
    \item developing a symmetrization argument that is specific to our multiwave sampling setting where the weights are not statistically independent, 
    \item applying the chaining-based result in Lemma \ref{lemma:ApplyChaining} conditionally on the data to upper bound the expectation of the symmetrized processes,
    \item upper bounding the expectation of a remaining term in Inequality \eqref{eq:moduliOfcontinuityUpperBoundDecomposition} using a standard symmetrization and chaining argument for i.i.d. processes, and
    \item combining terms in an upperbound for $\e[\omega_{N,\delta_0}(\delta)]$ and noting that an analogous argument gives an upper bound on $\e \big[  \tilde{\omega}_{N,\delta_0}(\delta) \big]$.
\end{enumerate}  Throughout the proof we will fix $N \in \mathbb{Z}_+, \delta >0, \delta_0>0$ such that $B(\thetaTarg; \delta_0) \subseteq \mathcal{L}_{\thetaTarg}$ and $B(\gammaTarg; \delta_0) \subseteq \mathcal{L}_{\gammaTarg}$ ($\mathcal{L}_{\thetaTarg}$ is the neighborhood from Assumption \ref{assump:SmoothEnoughForAsymptoticLineariaty}(iv) in which $\theta \mapsto l_{\theta}(\goodX)$ is locally $M(\goodX)$-Lipschitz).

\paragraph{Upper bounding $\omega_{N,\delta_0}(\delta)$:} Recalling the definitions in Formulas \eqref{eq:ModuliOfContinuityDef}, \eqref{eq:DeltaDef}, and \eqref{eq:FunctionClassFDefinition} $$\begin{aligned}   \omega_{N,\delta_0}(\delta) & =  \sup\limits_{\substack{\theta,\theta' \in B(\thetaTarg;\delta_0) \\ \vert \vert \theta-\theta' \vert \vert_2 \leq \delta }} \Big| \bfs_N l_{\theta}-L(\theta) -\big( \bfs_N l_{\theta'}-L(\theta') \big)  \Big| 
\\ & = \sup\limits_{\substack{\theta,\theta' \in B(\thetaTarg;\delta_0) \\ \vert \vert \theta-\theta' \vert \vert_2 \leq \delta }} \Big| \frac{1}{N}\sum_{i=1}^N W_i \big(l_{\theta}(\goodX_i)-l_{\theta'}(\goodX_i) \big)-\e[l_{\theta}(\goodX)-l_{\theta'}(\goodX)]  \Big| 
\\ & =  \sup_{f \in \mathcal{F}_{\delta_0}(\delta)} \Big| \frac{1}{N} \sum_{i=1}^N W_i f(\goodX_i) -\e[f(\goodX)]  \Big| 
\\ & =  \sup_{f \in \mathcal{F}_{\delta_0}(\delta)} \Big| \sum_{k=1}^K c_k \Big( \frac{1}{N} \sum_{i=1}^N  W_i^{(k)} f(\goodX_i) - \e[f(\goodX)] \Big)  \Big|
\\ & \leq \sum_{k=1}^K  \sup_{f \in \mathcal{F}_{\delta_0}(\delta)} \Big|  \frac{1}{N} \sum_{i=1}^N  W_i^{(k)} f(\goodX_i) - \e[f(\goodX)]  \Big|,
\end{aligned}$$ where the last two steps hold because $\sum_{k=1}^K c_k=1$ with $c_k \in [0,1]$ for each $k \in [K]$.

To simplify the upper bound above, for each $k \in [K]$ and $i \in [N]$  recall from Definition \eqref{eq:IPWsingleSingle} and Equation \eqref{eq:WkFormulaAlt}, that $$W_i^{(k,j)} = \begin{cases} \frac{I_i^{(j)}}{\LabelRulePi{j}(\proxyX_i)} & \text{if } j=k 
\\  \frac{1-I_i^{(j)}}{1-\LabelRulePi{j}(\proxyX_i) } & \text{if } j < k 
\\ 1 & \text{if } j > k \end{cases} \quad \quad \text{and} \quad  W_i^{(k)} = \prod_{j=1}^k W_i^{(k,j)}.$$ For each $k \in [K]$ and $i \in [N]$ define $$T_i^{(k,0)} \equiv 1, \quad \text{and} \quad T_i^{(k,k')} \equiv \prod_{j=1}^{k'} W_i^{(k,j)} \quad \text{for } k' \in [k],$$ and note that with these definitions, for any $i \in [N]$ and $k \in [K]$, $$W_i^{(k)}=T_i^{(k,k)}-T_i^{(k,0)}+1=1+  \sum_{k'=0}^{k-1} \big( T_i^{(k,k'+1)}-T_i^{(k,k')} \big) = 1+ \sum_{k'=0}^{k-1} \big( W_i^{(k,k'+1)}- 1 \big) T_i^{(k,k')}.$$ Plugging this expression into a previous inequality, $$\begin{aligned} \omega_{N,\delta_0}(\delta) & \leq  \sum_{k=1}^K  \sup_{f \in \mathcal{F}_{\delta_0}(\delta)} \Big|  \frac{1}{N} \sum_{i=1}^N  W_i^{(k)} f(\goodX_i) - \e[f(\goodX)]  \Big|
\\ & = \sum_{k=1}^K  \sup_{f \in \mathcal{F}_{\delta_0}(\delta)} \Big|  \frac{1}{N} \sum_{i=1}^N  \sum_{k'=0}^{k-1} \big( W_i^{(k,k'+1)}- 1 \big) T_i^{(k,k')} f(\goodX_i)  + \frac{1}{N} \sum_{i=1}^N f(\goodX_i)- \e[f(\goodX)] \Big| 
\\ & \leq \sum_{k=1}^K \sum_{k'=0}^{k-1} \sup_{f \in \mathcal{F}_{\delta_0}(\delta)} \Big|  \frac{1}{N} \sum_{i=1}^N   \big( W_i^{(k,k'+1)}- 1 \big) T_i^{(k,k')} f(\goodX_i) \Big| +K \sup_{f \in \mathcal{F}_{\delta_0}(\delta)} \Big|  \frac{1}{N} \sum_{i=1}^N f(\goodX_i)-\e[f(\goodX)] \Big|.
\end{aligned}$$ Thus if we define,
\begin{equation}
\begin{split}
   \omega_{N,\delta_0}^{(k,k')}(\delta) & \equiv  \sup_{f \in \mathcal{F}_{\delta_0}(\delta)} \Big|  \frac{1}{N} \sum_{i=1}^N   \big( W_i^{(k,k'+1)}- 1 \big) T_i^{(k,k')} f(\goodX_i) \Big| \quad \text{ for } k \in [K] \text{ and } k' \in \{0\} \cup [k-1] \text{ and} \\
\omega_{N,\delta_0}^{(\text{IID})}(\delta) & \equiv \sup_{f \in \mathcal{F}_{\delta_0}(\delta)} \Big|  \frac{1}{N} \sum_{i=1}^N f(\goodX_i)-\e[f(\goodX)] \Big|,
\end{split}
\end{equation} then 
\begin{equation}\label{eq:moduliOfcontinuityUpperBoundDecomposition}
    \omega_{N,\delta_0}(\delta) \leq \sum_{k=1}^K \sum_{k'=0}^{k-1} \omega_{N,\delta_0}^{(k,k')}(\delta) +  K \cdot \omega_{N,\delta_0}^{(\text{IID})}(\delta).
\end{equation} Thus to upper bound $\e[\omega_{N,\delta_0}(\delta)]$ it suffices to find an upper bound for $\e[\omega_{N,\delta_0}^{(\text{IID})}(\delta)]$ (which can be done using a standard symmetrization and chaining argument for i.i.d. data) and an upper bound for $\e[\omega_{N,\delta_0}^{(k,k')}(\delta)]$ for each $k \in [K]$ and $k' \in \{0\} \cup [k-1]$, which we derive next.

\paragraph{Bounding $\e[\omega_{N,\delta_0}^{(k,k')}(\delta)]$ with a symmetrization argument for two-phase multiwave sampling:}Fix $k \in [K]$ and $k' \in \{0\} \cup [k-1]$ and we will find an upper bound on $\e[\omega_{N,\delta_0}^{(k,k')}(\delta)]$ using a modification of a symmetrization argument that is appropriate in \samplingSchemeName{} settings. To do this, let $r=\text{sgn}(k'+1-k)$ and recall the functions $\phi_{-1}, \phi_0:  [0,1] \to [0,1]$ from Equation \eqref{eq:ComplementOr1Operator} that are given by $\phi_{-1}(s)=1-s$ and $\phi_0(s)=s$ for all $s \in [0,1]$. By the definitions of $W_i^{(k,k'+1)}$  observe that $$W_i^{(k,k'+1)}= \frac{\phi_r(I_i^{(k'+1)})}{\phi_r \big( \LabelRulePiDoubleArg{k'+1}{k'} (\proxyX_i) \big)}=  \frac{\phi_r \big( \mathbbm{1} \bigl\{ U_i^{(k'+1)} \leq \LabelRulePiDoubleArg{k'+1}{k'} (\proxyX_i) \bigr\} \big)}{\phi_r \big( \LabelRulePiDoubleArg{k'+1}{k'} (\proxyX_i) \big)}  \quad \text{for each } i \in [N],$$ where recall that $U_i^{(k'+1)} \stackrel{\text{i.i.d.}}{\sim} \text{Unif}[0,1]$ were generated independently of the data $\cd_{k'}$.

Next define $\cd_0^+ \equiv (\xmiss_i,\proxyX_i)_{i=1}^N$ and for each $k^* \in [K]$ define $$\cd_{k^*}^+ \equiv \Big( \big(I_i^{(j)},U_i^{(j)},\LabelRulePi{j}(\proxyX_i) \big)_{j=1}^{k^*}, \xmiss_i, \proxyX_i \Big)_{i=1}^N$$ to be an augmented version of $\cd_{k^*}$ (recall $\cd_{k^*}= \big((I_i^{(j)},I_i^{(j)} \xmiss_i)_{j=1}^{k^*},\proxyX_i) \big)_{i=1}^N$ is the observed data and labelling indicators after the $k^*$th wave). The augmentation includes all $\goodX_i=(\xobs_i,\xmiss_i)$ values (not just the fully observed ones) and additional sampling information $\big( U_i^{(j)},\LabelRulePi{j}(\proxyX_i) \big)_{j=1}^{k^*}$ from the first $k^*$ waves, but crucially the augmentation does not contain any information about the $U_i^{(j)}$ for $j>k^*$. Note that for each $i \in [N]$, $W_i^{(k,k'+1)}$, $T_i^{(k,k')}$, and $f(\goodX_i)$ can all be written as measurable functions of $\cd_{k'+1}^+$, so $\e[W_i^{(k,k'+1)} T_i^{(k,k')} f(\goodX_i) \giv \cd_{k'+1}^+]=W_i^{(k,k'+1)} T_i^{(k,k')} f(\goodX_i)$ for each $i \in [N]$. Next note that in \samplingSchemeName{} $U_i^{(k'+1)} \stackrel{\text{i.i.d.}}{\sim} \text{Unif}[0,1]$ were generated independently of the data $\cd_{k'}^+$. 

We introduce additional variables to symmetrize by generating $U_i^* \stackrel{\text{i.i.d.}}{\sim} \text{Unif}[0,1]$ independently of $(U_i^{(k'+1)})_{i=1}^N$, $(I_i^{(k'+1)})_{i=1}^N$ and $\cd_{k'}^+$. Hence $U_i^* \giv \cd_{k'+1}^+ \stackrel{\text{i.i.d.}}{\sim} \text{Unif}[0,1]$. Define $$W_i^{*(k,k'+1)} \equiv \frac{\phi_r \big( \mathbbm{1} \bigl\{ U_i^* \leq \LabelRulePiDoubleArg{k'+1}{k'} (\proxyX_i) \bigr\} \big)}{\phi_r \big( \LabelRulePiDoubleArg{k'+1}{k'} (\proxyX_i) \big)} \quad \text{ for each } i \in [N],$$ and note that regardless of whether $r=0$ or $r=-1,$ $$1 = \e \Big[  \frac{\phi_r \big( \mathbbm{1} \bigl\{ U_i^* \leq \LabelRulePiDoubleArg{k'+1}{k'} (\proxyX_i) \bigr\} \big)}{\phi_r \big( \LabelRulePiDoubleArg{k'+1}{k'} (\proxyX_i) \big)} \Big| \cd_{k'+1}^+ \Big] =\e[ W_i^{*(k,k'+1)} \giv \cd_{k'+1}^+].$$ Combining previous results and definitions, $$\begin{aligned} \omega_{N,\delta_0}^{(k,k')}(\delta) & \equiv  \sup_{f \in \mathcal{F}_{\delta_0}(\delta)} \Big|  \frac{1}{N} \sum_{i=1}^N   \big( W_i^{(k,k'+1)}- 1 \big) T_i^{(k,k')} f(\goodX_i) \Big| 
\\ & =   \sup_{f \in \mathcal{F}_{\delta_0}(\delta)} \Bigg|  \frac{1}{N} \sum_{i=1}^N   \Big( W_i^{(k,k'+1)}- \e \big[ W_i^{*(k,k'+1)} \big| \cd_{k'+1}^+ \big] \Big) T_i^{(k,k')} f(\goodX_i) \Bigg|
\\ & =   \sup_{f \in \mathcal{F}_{\delta_0}(\delta)} \Bigg| \e \Big[ \frac{1}{N} \sum_{i=1}^N   \big( W_i^{(k,k'+1)}-  W_i^{*(k,k'+1)}  \big) T_i^{(k,k')} f(\goodX_i) \Big| \cd_{k'+1}^+ \Big] \Bigg|
\\ & \leq  \e \Bigg[ \sup_{f \in \mathcal{F}_{\delta_0}(\delta)} \Big| \frac{1}{N} \sum_{i=1}^N   \big( W_i^{(k,k'+1)}-  W_i^{*(k,k'+1)}  \big) T_i^{(k,k')} f(\goodX_i)  \Big| \Bigg| \cd_{k'+1}^+ \Bigg].
\end{aligned}$$ Since $\cd_{k'}^+$ contains a subset of the variables in $\cd_{k'+1}^+$, by taking $\e[\cdot \giv \cd_{k'}^+]$ of each side of the above inequality and the tower property, \begin{equation} \label{eq:SymmetrizationStep1UpperBound}
\e \big[ \omega_{N,\delta_0}^{(k,k')}(\delta) \big| \cd_{k'}^+ \big] \leq   \e \Bigg[ \sup_{f \in \mathcal{F}_{\delta_0}(\delta)} \Big| \frac{1}{N} \sum_{i=1}^N   \big( W_i^{(k,k'+1)}-  W_i^{*(k,k'+1)}  \big) T_i^{(k,k')} f(\goodX_i)  \Big| \Bigg| \cd_{k'}^+ \Bigg].\end{equation}

Now let $R_1,R_2,\dots,R_N \stackrel{\text{i.i.d.}}{\sim} \text{Unif}\{-1,1\}$ be $N$ independent Rademacher variables that are independent of all previously described random variables (notably these Rademacher variables are independent of $\cd_k^+$, $\cd_{k'}^+$, $(U_i^{(k'+1)})_{i=1}^N$, $(U_i^*)_{i=1}^N$). Next recall that conditionally on $\cd_{k'}^+$, for each $i \in [N]$, $\LabelRulePiDoubleArg{k'+1}{k'}(\proxyX_i)$ is a constant. In addition, conditionally on  $\cd_{k'}^+$, $(U_i^{(k'+1)})_{i=1}^N \stackrel{\text{i.i.d.}}{\sim} \text{Unif}[0,1]$ and independently $(U_i^*)_{i=1}^N \stackrel{\text{i.i.d.}}{\sim} \text{Unif}[0,1]$. Thus because $$W_i^{(k,k'+1)}-  W_i^{*(k,k'+1)}    =   \frac{\phi_r \big( \mathbbm{1} \bigl\{ U_i^{(k'+1)} \leq \LabelRulePiDoubleArg{k'+1}{k'} (\proxyX_i) \bigr\} \big)}{\phi_r \big( \LabelRulePiDoubleArg{k'+1}{k'} (\proxyX_i) \big)}  -   \frac{\phi_r \big( \mathbbm{1} \bigl\{ U_i^* \leq \LabelRulePiDoubleArg{k'+1}{k'} (\proxyX_i) \bigr\} \big)}{\phi_r \big( \LabelRulePiDoubleArg{k'+1}{k'} (\proxyX_i) \big)},$$ when conditioning on $\cd_{k'}^+$, $\big( W_i^{(k,k'+1)}-  W_i^{*(k,k'+1)} \big)_{i=1}^N$ is a sequence of $N$ independent random variables. Moreover, when conditioning on $\cd_{k'}^+$, for each $i \in [N]$, $$\begin{aligned} W_i^{(k,k'+1)}-  W_i^{*(k,k'+1)}   &  =   \frac{\phi_r \big( \mathbbm{1} \bigl\{ U_i^{(k'+1)} \leq \LabelRulePiDoubleArg{k'+1}{k'} (\proxyX_i) \bigr\} \big)}{\phi_r \big( \LabelRulePiDoubleArg{k'+1}{k'} (\proxyX_i) \big)}  -   \frac{\phi_r \big( \mathbbm{1} \bigl\{ U_i^* \leq \LabelRulePiDoubleArg{k'+1}{k'} (\proxyX_i) \bigr\} \big)}{\phi_r \big( \LabelRulePiDoubleArg{k'+1}{k'} (\proxyX_i) \big)} 
\\ & \stackrel{\text{dist}}{=}   \frac{\phi_r \big( \mathbbm{1} \bigl\{ U_i^* \leq \LabelRulePiDoubleArg{k'+1}{k'} (\proxyX_i) \bigr\} \big)}{\phi_r \big( \LabelRulePiDoubleArg{k'+1}{k'} (\proxyX_i) \big)}  -   \frac{\phi_r \big( \mathbbm{1} \bigl\{ U_i^{(k'+1)} \leq \LabelRulePiDoubleArg{k'+1}{k'} (\proxyX_i) \bigr\} \big)}{\phi_r \big( \LabelRulePiDoubleArg{k'+1}{k'} (\proxyX_i) \big)}
\\ & =  -(W_i^{(k,k'+1)}-W_i^{*(k,k'+1)}),
\end{aligned}$$ and as a consequence for each $i \in [N]$, $W_i^{(k,k'+1)}-  W_i^{*(k,k'+1)}$ and $R_i \big( W_i^{(k,k'+1)}-  W_i^{*(k,k'+1)} \big)$ have the same distribution when conditioning on $\cd_{k'}^+$ (because $R_i \sim \text{Unif}\{-1,1\}$ independently of $\cd_{k'}^+$). Since conditionally on $\cd_{k'}^+$, $(R_i)_{i=1}^N$ and $\big( W_i^{(k,k'+1)}-  W_i^{*(k,k'+1)} \big)_{i=1}^N$ are both sequences of $N$ independent variables (and moreover the sequences are independent of each other), it follows that conditionally on $\cd_{k'}^+$, $\big( R_i ( W_i^{(k,k'+1)}-  W_i^{*(k,k'+1)}) \big)_{i=1}^N$ is a sequence of $N$ independent random variables. Combining these results we have that conditionally on $\cd_{k'}^+$, both $\big( W_i^{(k,k'+1)}-  W_i^{*(k,k'+1)} \big)_{i=1}^N$ and $\big( R_i ( W_i^{(k,k'+1)}-  W_i^{*(k,k'+1)}) \big)_{i=1}^N$ are sequences of $N$ independent random variables and that conditionally on $\cd_{k'}^+$, $W_i^{(k,k'+1)}-  W_i^{*(k,k'+1)} \stackrel{\text{dist}}{=} R_i ( W_i^{(k,k'+1)}-  W_i^{*(k,k'+1)})$ for each $i \in [N]$. Since two random vectors that each have independent components and the same coordinate-wise distribution must have the same joint distribution it follows that conditionally on $\cd_{k'}^+$, $$\big( W_i^{(k,k'+1)}-  W_i^{*(k,k'+1)} \big)_{i=1}^N \stackrel{\text{dist}}{=} \big( R_i ( W_i^{(k,k'+1)}-  W_i^{*(k,k'+1)}) \big)_{i=1}^N.$$ Next recall that conditionally on $\cd_{k'}^+$, for each $i \in [N]$, $T_i^{(k,k')}$ and $f(\goodX_i)$ are constants. Combining this with the previous result, conditionally on $\cd_{k'}^+$, $$\Big( \big(W_i^{(k,k'+1)}-  W_i^{*(k,k'+1)}\big) T_i^{(k,k')} f(\goodX_i) \Big)_{i=1}^N \stackrel{\text{dist}}{=} \Big( R_i \big( W_i^{(k,k'+1)}-  W_i^{*(k,k'+1)} \big) T_i^{(k,k')} f(\goodX_i) \Big)_{i=1}^N.$$ 

We can now combine the above result with the inequality at \eqref{eq:SymmetrizationStep1UpperBound} to get a symmetrization bound: $$\begin{aligned} 
\e \big[ \omega_{N,\delta_0}^{(k,k')}(\delta) \big| \cd_{k'}^+ \big] & 
\leq   \e \Bigg[ \sup_{f \in \mathcal{F}_{\delta_0}(\delta)} \Big| \frac{1}{N} \sum_{i=1}^N   \big( W_i^{(k,k'+1)}-  W_i^{*(k,k'+1)}  \big) T_i^{(k,k')} f(\goodX_i)  \Big| \Bigg| \cd_{k'}^+ \Bigg]
\\ & = \e \Bigg[ \sup_{f \in \mathcal{F}_{\delta_0}(\delta)} \Big| \frac{1}{N} \sum_{i=1}^N   R_i \big( W_i^{(k,k'+1)}-  W_i^{*(k,k'+1)}  \big) T_i^{(k,k')} f(\goodX_i)  \Big| \Bigg| \cd_{k'}^+ \Bigg]
\\ & \leq \e \Big[ \sup_{f \in \mathcal{F}_{\delta_0}(\delta)} \big| \frac{1}{N} \sum_{i=1}^N   R_i   W_i^{(k,k'+1)}  T_i^{(k,k')} f(\goodX_i)  \big| \Big| \cd_{k'}^+ \Big] \\ & \quad + \e \Big[ \sup_{f \in \mathcal{F}_{\delta_0}(\delta)} \big| \frac{1}{N} \sum_{i=1}^N   R_i   W_i^{*(k,k'+1)}  T_i^{(k,k')} f(\goodX_i)  \big| \Big| \cd_{k'}^+ \Big]
\\ & = 2 \cdot \e \Big[ \sup_{f \in \mathcal{F}_{\delta_0}(\delta)} \big| \frac{1}{N} \sum_{i=1}^N   R_i   W_i^{(k,k'+1)}  T_i^{(k,k')} f(\goodX_i)  \big| \Big| \cd_{k'}^+ \Big],
\end{aligned}$$ where the last step holds because conditionally on $\cd_{k'}^+$, $\big(  W_i^{(k,k'+1)} \big)_{i=1}^N$ and $\big(  W_i^{*(k,k'+1)} \big)_{i=1}^N$ have the same joint distribution. Noting that $W_i^{(k,k'+1)}  T_i^{(k,k')}=T_i^{(k,k'+1)}$, by taking the expectation of each side of the above inequality with respect to $\cd_{k'}^+$ and applying the tower property, $$\begin{aligned} \ \e \big[ \omega_{N,\delta_0}^{(k,k')}(\delta) \big] & \leq 2  \e \Big[ \sup_{f \in \mathcal{F}_{\delta_0}(\delta)} \big| \frac{1}{N} \sum_{i=1}^N   R_i T_i^{(k,k'+1)} f(\goodX_i)  \big| \Big] \\ & = \frac{2}{N} \cdot \e \Bigg[ \e\Big[ \sup_{f \in \mathcal{F}_{\delta_0}(\delta)} \big|  \sum_{i=1}^N   R_i T_i^{(k,k'+1)} f(\goodX_i)  \big| \Big| \cd_{k}^+ \Big] \Bigg] \end{aligned}$$ where the outermost expectation is with respect to $\cd_{k}^+$. 

\paragraph{Applying the chaining result from Lemma \ref{lemma:ApplyChaining}:} Since conditionally on $\cd_{k}^+$, $T_i^{(k,k'+1)} f(\goodX_i)$ are constants while $R_i \stackrel{\text{i.i.d.}}{\sim} \text{Unif}\{-1,1\}$, by applying Lemma \ref{lemma:ApplyChaining} with $a_i=T_i^{(k,k'+1)} \in [0,b^{-K}]$ (where almost surely, $a_i \in [0,b^{-K}]$ by Assumption \ref{assump:LabellingRuleOverlap}) and with $\argx_i = \goodX_i \in \xSpace$ for each $i \in [N]$, when conditioning on $\cd_{k}^+$, $$\e\Big[ \sup_{f \in \mathcal{F}_{\delta_0}(\delta)} \big|  \sum_{i=1}^N   R_i T_i^{(k,k'+1)} f(\goodX_i)  \big| \Big| \cd_{k}^+ \Big] \leq 4 b^{-K} C_u \sqrt{d} \sqrt{\delta \delta_0}   \sqrt{\sum_{i=1}^N M^2(\goodX_i)},$$ where $C_u \in (0,\infty)$ is a universal constant.  
Combining this with the previously displayed result, $$\e \big[ \omega_{N,\delta_0}^{(k,k')}(\delta) \big] \leq \frac{2}{N} \cdot \e \Big[ 4 b^{-K} C_u \sqrt{d} \sqrt{\delta \delta_0}   \Big( \sum_{i=1}^N M^2(\goodX_i) \Big)^{1/2} \Big]=\frac{8 C_u\sqrt{d} \sqrt{\delta \delta_0}}{b^K \sqrt{N}} \e \Big[ \Big( \frac{1}{N} \sum_{i=1}^N M^2(\goodX_i) \Big)^{1/2} \Big]$$ Recalling by Assumption \ref{assump:SmoothEnoughForAsymptoticLineariaty}(iv) $\e[M^2(\goodX)]< \infty$ we can let $C_{M,2} \equiv \e[M^2(\goodX)] \in (0,\infty)$ and note that since 2nd moments are always bigger than the square of a first moment, $$\e \Big[ \Big( \frac{1}{N} \sum_{i=1}^N M^2(\goodX_i) \Big)^{1/2} \Big] \leq \Big( \e \Big[  \frac{1}{N} \sum_{i=1}^N M^2(\goodX_i)  \Big] \Big)^{1/2} =\sqrt{\e[M^2(\goodX)]}=\sqrt{C_{M,2}}.$$ Thus combining this with a previous inequality $$ \e \big[ \omega_{N,\delta_0}^{(k,k')}(\delta) \big] \leq C_1 \sqrt{\delta \delta_0} N^{-1/2} \quad \text{ where }  C_1 \equiv 8 b^{-K} C_u\sqrt{d} \sqrt{ C_{M,2}}.$$  Since the above argument holds for any $k \in [K]$ and $k' \in \{0\} \cup [k-1]$ we have thus shown that \begin{equation}\label{eq:UpperBoundAfterSymmetraizatioAndChainingSingle_kkprime}
    \e \big[ \omega_{N,\delta_0}^{(k,k')}(\delta) \big] \leq C_1 \sqrt{\delta \delta_0} N^{-1/2} \quad \text{for each } k \in [K] \text{ and } k' \in \{0\} \cup [k-1].
\end{equation}

\paragraph{Upperbounding $\e[\omega_{N,\delta_0}^{(\text{IID})}(\delta)]$ using techniques for i.i.d. processes:}We next find an upperbound on $\e[\omega_{N,\delta_0}^{(\text{IID})}(\delta)]$ using a standard symmetrization and chaining argument for empirical processes of i.i.d. data. Recall that by Assumption \ref{assump:IIDUnderlyingData}, $\goodX_1,\dots,\goodX_N$ are i.i.d. Next let $R_1^*,\dots,R_N^* \stackrel{\text{i.i.d.}}{\sim} \text{Unif}\{-1,1\}$ be $N$ independent Rademacher variables such that $(R_i^*)_{i=1}^N \indep (\goodX_i)_{i=1}^N$. By a standard symmetrization result (e.g., see Lemma 2.3.1 in \cite{VDVAndWellnerTextbook} or Exercise 8.3.24 in \cite{Vershynin_2018}), $$\e[\omega_{N,\delta_0}^{(\text{IID})}(\delta)]  = \e \Big[ \sup_{f \in \mathcal{F}_{\delta_0}(\delta)} \Big|  \frac{1}{N} \sum_{i=1}^N f(\goodX_i)-\e[f(\goodX)] \Big| \Big] \leq 2 \cdot \e \Big[ \sup_{f \in \mathcal{F}_{\delta_0}(\delta)} \Big|  \frac{1}{N} \sum_{i=1}^N R_i^* f(\goodX_i) \Big| \Big] .$$ Note that by applying Lemma \ref{lemma:ApplyChaining} with $a_i=1 \in [0,b^{-K}]$ and with $\argx_i = \goodX_i \in \xSpace$ for each $i \in [N]$, when conditioning on $(\goodX_i)_{i=1}^N$, $$\e\Big[ \sup_{f \in \mathcal{F}_{\delta_0}(\delta)} \big|  \sum_{i=1}^N   R_i^*  f(\goodX_i)  \big| \Big| (\goodX_i)_{i=1}^N \Big] \leq 4 b^{-K} C_u \sqrt{d} \sqrt{\delta \delta_0}   \Big( \sum_{i=1}^N M^2(\goodX_i)\Big)^{1/2},$$ where $C_u \in (0,\infty)$ is a universal constant. Thus combining previous results and the definitions of $C_{M,2}$ and $C_1$ above, by the tower property, $$\begin{aligned} 
\e[\omega_{N,\delta_0}^{(\text{IID})}(\delta)] & \leq \frac{2}{N} \cdot \e \Big[ \sup_{f \in \mathcal{F}_{\delta_0}(\delta)} \Big|  \sum_{i=1}^N R_i^* f(\goodX_i) \Big| \Big]
\\ & = \frac{2}{N} \cdot \e \Bigg[ \e \Big[ \sup_{f \in \mathcal{F}_{\delta_0}(\delta)} \big|  \sum_{i=1}^N R_i^* f(\goodX_i) \big| \Big| (\goodX_i)_{i=1}^N \Big] \Bigg]
\\ & \leq \frac{2}{N} \cdot \e \Big[ 4 b^{-K} C_u \sqrt{d} \sqrt{\delta \delta_0}   \Big( \sum_{i=1}^N M^2(\goodX_i)\Big)^{1/2} \Big]
\\ & = \frac{8 b^{-K} C_u \sqrt{d} \sqrt{\delta \delta_0}}{\sqrt{N}} \cdot \e\Big[ \Big( \frac{1}{N} \sum_{i=1}^N M^2(\goodX_i)\Big)^{1/2} \Big]
\\ & \leq \frac{8 b^{-K} C_u \sqrt{d} \sqrt{\delta \delta_0} \sqrt{C_{M,2}}}{\sqrt{N}} = C_1 \sqrt{\delta \delta_0} N^{-1/2}.
\end{aligned}$$ 

\paragraph{Combining terms and completing the proof:}Combining the above inequality with the inequalities \eqref{eq:moduliOfcontinuityUpperBoundDecomposition} and \eqref{eq:UpperBoundAfterSymmetraizatioAndChainingSingle_kkprime},  $$\begin{aligned}
    \e[\omega_{N,\delta_0}(\delta)] & \leq \sum_{k=1}^K \sum_{k'=0}^{k-1} \e[\omega_{N,\delta_0}^{(k,k')}(\delta)] +  K \cdot \e[\omega_{N,\delta_0}^{(\text{IID})}(\delta)]
    \\ &  \leq \sum_{k=1}^K \sum_{k'=0}^{k-1} C_1 \sqrt{\delta \delta_0} N^{-1/2} +  K C_1 \sqrt{\delta \delta_0} N^{-1/2}
    \\ & \leq 2K^2 C_1 \sqrt{\delta \delta_0} N^{-1/2}. 
\end{aligned}$$

Thus taking $C_2=2K^2 C_1$, $\e[\omega_{N,\delta_0}(\delta)] \leq C_2 \sqrt{\delta \delta_0} N^{-1/2}$. Since the proof holds for any fixed $N \in \mathbb{Z}_+$, $\delta >0$ and $\delta_0>0$ such that $B(\thetaTarg;\delta_0) \subseteq \mathcal{L}_{\thetaTarg}$, $$\e[\omega_{N,\delta_0}(\delta)] \leq C_2 \sqrt{\delta \delta_0} N^{-1/2} \quad \text{ for all } N \in \mathbb{Z}_+,\delta>0, \delta_0>0 \text{ such that } B(\thetaTarg;\delta_0) \subseteq \mathcal{L}_{\thetaTarg},$$ where $C_2 \in (0,\infty)$ is a constant that does not depend on $N$, $\delta$ or $\delta_0$.

An analogous argument shows that for some constant $\tilde{C}_2 \in (0,\infty)$, $$\e[\tilde{\omega}_{N,\delta_0}(\delta)] \leq \tilde{C}_2 \sqrt{\delta \delta_0} N^{-1/2} \quad \text{ for all } N \in \mathbb{Z}_+,\delta>0, \delta_0>0 \text{ such that } B(\gammaTarg;\delta_0) \subseteq \mathcal{L}_{\gammaTarg}.$$ Taking $C= \max \{C_2,\tilde{C}_2 \} \in (0,\infty)$ completes the proof.

\end{proof}

\subsection{Proof of \texorpdfstring{$\sqrt{N}$}{Root-N}-consistency of point estimators}

 Using the point estimator consistency result (Proposition \ref{prop:ComponentEstimatorConsistency}) and Lemma \ref{lemma:ModContinuityExpectationBound} which controls local fluctuations of the empirical process, we can set up the use of a rate-of-convergence proof technique (e.g., see Theorem 5.52 in \cite{VanderVaartTextbook}) that establishes a stronger, $\sqrt{N}$-consistency result. The $\sqrt{N}$-consistency of $\htc$, $\hgc$, and $\hga$ is formalized in the following theorem. We remark that this $\sqrt{N}$-consistency result should be interpreted with caution as the variance of $\hga$ can be orders of magnitude smaller than the variances of $\htc$ and $\hgc$, especially when the labelling probabilities are close to zero. Nonetheless, establishing $\sqrt{N}$-consistency is a critical step in establishing that an estimator is asymptotically linear. 

\begin{theorem}\label{theorem:RootNConsistency}
    Under \samplingSchemeName{} and Assumptions \ref{assump:IIDUnderlyingData}, \ref{assump:LabellingRuleOverlap}, and \ref{assump:SmoothEnoughForAsymptoticLineariaty}, $$\sqrt{N} (\htc-\thetaTarg) =O_p(1), \quad \sqrt{N}(\hgc - \gammaTarg)=O_p(1), \quad \text{and}  \quad \sqrt{N}(\hga - \gammaTarg)=O_p(1),$$ where $O_p(1)$ denotes a sequence that is bounded in probability as $N \to \infty$.
\end{theorem} 

\begin{proof}

   Recall that $\hga \xrightarrow{p} \gammaTarg$ by Proposition \ref{prop:ComponentEstimatorConsistency}, so by standard M-estimation theory (e.g., see Corollary 5.53 in \cite{VanderVaartTextbook}) $\sqrt{N}(\hga - \gammaTarg)=O_p(1)$. We focus on showing that $\sqrt{N}(\htc - \thetaTarg)=O_p(1)$ and the result that $\sqrt{N}(\hgc - \gammaTarg)=O_p(1)$ will follow from an analogous argument. The proof follows the same general strategy as seen in Theorem 5.52 and Corollary 5.53 in \cite{VanderVaartTextbook}, but it is reproduced and modified for our notation and setting in order to establish that the results still hold in spite of the statistically dependent sample weights.
    
    To show $\sqrt{N}(\htc - \thetaTarg)=O_p(1)$, recall the definitions of $\Delta_N(\cdot)$ and $\omega_{N,\delta_0}(\delta)$ from \eqref{eq:DeltaDef} and \eqref{eq:ModuliOfContinuityDef} and the definition that $B(\thetaTarg;\delta)$ is a ball of radius $\delta$ about $\thetaTarg$, and observe that $\Delta_N(\thetaTarg)=0$. Moreover, $$\sup_{\theta \in B(\thetaTarg;\delta)} \vert \Delta_N (\theta) \vert = \sup_{\theta \in B(\thetaTarg;\delta)} \vert \Delta_N (\theta) -\Delta_N(\thetaTarg) \vert \leq \sup\limits_{\substack{\theta,\theta' \in B(\thetaTarg;\delta) \\ \vert \vert \theta-\theta' \vert \vert_2 \leq \delta }} \Big| \Delta_N(\theta)-\Delta_N(\theta') \Big| = \omega_{N,\delta}(\delta).$$ Let $\delta_* >0$ be small enough such that $B(\thetaTarg;\delta) \subseteq \mathcal{L}_{\thetaTarg}$ for all $\delta \in (0,\delta_*)$. By the above inequality and applying Lemma \ref{lemma:ModContinuityExpectationBound} for the special case where $\delta=\delta_0$, there exists a $C \in (0,\infty)$ such that 
 \begin{equation}\label{eq:SupDeltaExpectationBound}
    \e\Big[ \sup_{\theta \in B(\thetaTarg;\delta)} \vert \Delta_N (\theta) \vert \Big] \leq \e[\omega_{N,\delta}(\delta)] \leq \frac{C \delta}{\sqrt{N}} \quad \text{for all } N \in \mathbb{Z}_+ \text{ and } \delta \in (0,\delta_*).
\end{equation} 

Next we will show that there exists a $C_2 \in (0,\infty)$ and $\delta_2 \in (0,\infty)$ such that \begin{equation}\label{eq:GrowthFactorControl}
    L(\thetaTarg)-L(\theta)  \leq - C_2 \vert \vert \theta-\thetaTarg \vert \vert_2^2 \quad \text{for all} \quad \theta \in B(\thetaTarg;\delta_2).
\end{equation} To verify this recall that by Assumption \ref{assump:SmoothEnoughForAsymptoticLineariaty}(v), $\theta \mapsto L(\theta)$ admits a 2nd order Taylor expansion about $\thetaTarg=\argmin_{\theta \in \Theta} L(\theta)$, and hence $\nabla L(\thetaTarg)=0$. Moreover, by a Taylor expansion, for $\theta$ in a neighborhood of $\thetaTarg$, for some function $h: \mathbb{R}^d \to \mathbb{R}$ such that $\lim_{t \to 0} h(t)=0$, $$L(\theta)-L(\thetaTarg)=\frac{1}{2}(\theta-\thetaTarg)^\tran H_{\thetaTarg} (\theta-\thetaTarg) + h(\theta-\thetaTarg) \vert \vert \theta-\thetaTarg \vert \vert_2^2 \geq \frac{1}{2} \lambda_{\text{min}} ( H_{\thetaTarg}) \vert \vert \theta-\thetaTarg \vert \vert_2^2 +  h(\theta-\thetaTarg) \vert \vert \theta-\thetaTarg \vert \vert_2^2,$$ where $\lambda_{\text{min}} (\cdot)$ is an operator that gives the smallest eigenvalue of a matrix and $H_{\thetaTarg} =\nabla^2 L(\thetaTarg)$ is the Hessian. Now since by Assumption \ref{assump:SmoothEnoughForAsymptoticLineariaty}(i), $\theta \mapsto l_{\theta}(\goodX)$ is convex almost surely, $\theta \mapsto L(\theta)=\e[l_{\theta}(\goodX)]$ is convex and thus $H_{\thetaTarg} =\nabla^2 L(\thetaTarg) \succeq 0$. Since $H_{\thetaTarg}$ is also nonsingular (by Assumption \ref{assump:SmoothEnoughForAsymptoticLineariaty}(v)), $\lambda_{\text{min}} ( H_{\thetaTarg} )>0$. Since $\lim_{t \to 0} h(t)=0$, we can choose $\delta_2 >0$ to be small enough such that for all $\theta \in B(\thetaTarg;\delta_2)$ both $\vert h(\theta-\thetaTarg) \vert \leq \lambda_{\text{min}} ( H_{\thetaTarg} )/4$ and the Taylor expansion displayed above holds. Combining this with a previous result it follows that for all $\theta \in B(\thetaTarg;\delta_2)$, $$L(\theta)-L(\thetaTarg) \geq \frac{1}{4} \lambda_{\text{min}} (H_{\thetaTarg}) \vert \vert \theta-\thetaTarg \vert \vert_2^2  \Rightarrow L(\thetaTarg)- L(\theta) \leq - C_2 \vert \vert \theta-\thetaTarg \vert \vert_2^2,$$ where $C_2 \equiv \lambda_{\text{min}} ( H_{\thetaTarg} )/4 \in (0,\infty)$. This verifies Inequality \eqref{eq:GrowthFactorControl}.

Having established Inequalities \eqref{eq:SupDeltaExpectationBound} and \eqref{eq:GrowthFactorControl}, the rest of the proof follows from a standard rate of convergence argument for M-estimators (e.g., Theorem 5.52 in \cite{VanderVaartTextbook}) using a ``peeling" or ``shelling" technique which we exhibit below. To do this, take $\epsilon = \min \{\delta_*/2,\delta_2/2 \}$, where $\delta_*,\delta_2$ are the small positive constants below which Inequalities \eqref{eq:SupDeltaExpectationBound} and \eqref{eq:GrowthFactorControl} hold. Next fix $N,r \in \mathbb{Z}_+$ and we will find an upper bound for $\mathbb{P}(\sqrt{N} \vert \vert \htc -\thetaTarg \vert \vert_2>2^r)$. Further let $j_* \equiv \min \{j \in \mathbb{N} \ : \ 2^j > \epsilon \sqrt{N} \}$ and we first consider the case where $j_* \geq r+1$. Observe that if $\sqrt{N} \vert \vert \htc -\thetaTarg \vert \vert_2>2^r$, then either $ \sqrt{N} \vert \vert \htc -\thetaTarg \vert \vert_2 \in (2^{j-1},2^j]$ for some $j \in \{r+1,r+2,\dots, j_*\}$ or $\vert \vert \htc -\thetaTarg \vert \vert_2 > \epsilon$. Hence by the union bound, if we define $\Theta_j \equiv \{ \theta \in \Theta \ : \  \sqrt{N} \vert \vert \theta -\thetaTarg \vert \vert_2 \in (2^{j-1},2^j]\}$, $$\begin{aligned} 
\mathbb{P}(\sqrt{N} \vert \vert \htc -\thetaTarg \vert \vert_2>2^r) & \leq \sum_{j=r+1}^{j_*} \mathbb{P} \Big( \sqrt{N} \vert \vert \htc -\thetaTarg \vert \vert_2 \in (2^{j-1},2^j] \Big) + \mathbb{P}(\vert \vert \htc -\thetaTarg \vert \vert_2 > \epsilon)
\\ & \leq \sum_{j=r+1}^{j_*} \mathbb{P} \Big( \inf\limits_{\theta \in \Theta_j} \bfs_N l_{\theta} < \bfs_N l_{\thetaTarg} \Big) + \mathbb{P}(\vert \vert \htc -\thetaTarg \vert \vert_2 > \epsilon)
\\ & \leq \sum_{j=r+1}^{j_*} \mathbb{P} \Big( \inf\limits_{\theta \in \Theta_j} \big( \bfs_N l_{\theta} +L(\thetaTarg)- L(\theta) \big) < \bfs_N l_{\thetaTarg}  + \sup\limits_{\theta \in \Theta_j} \big(L(\thetaTarg)- L(\theta) \big) \Big)
\\ & \quad + \mathbb{P}(\vert \vert \htc -\thetaTarg \vert \vert_2 > \epsilon)
\\ & = \sum_{j=r+1}^{j_*} \mathbb{P} \Big( \inf\limits_{\theta \in \Theta_j}   \Delta_N(\theta) < \sup\limits_{\theta \in \Theta_j} \big(L(\thetaTarg)- L(\theta) \big) \Big)  + \mathbb{P}(\vert \vert \htc -\thetaTarg \vert \vert_2 > \epsilon), \end{aligned}$$ where the above inequalities follow by monotonicity of probability measure and the formulas and definitions of $\htc$, $\bfs_N l_{\theta}$ and $\Delta_N(\theta)$ at \eqref{eq:PTDComponentEstimatorsOperatorNotation}, \eqref{eq:AveragingOperatorDefs}, and \eqref{eq:DeltaDef}. The above inequality was shown for the case where $j_* \geq r+1$ but it also holds for the case where $j_* < r+1$ (in this case $2^r> \epsilon \sqrt{N}$ so $\mathbb{P}(\sqrt{N} \vert \vert \htc -\thetaTarg \vert \vert_2>2^r) \leq \mathbb{P}(\vert \vert \htc -\thetaTarg \vert \vert_2>\epsilon)$).

Next note that for $j \leq j_*$, $2^j/\sqrt{N} \leq \delta_2$ so by Inequality \eqref{eq:GrowthFactorControl}, $\sup\limits_{\theta \in \Theta_j} \big(L(\thetaTarg)- L(\theta) \big) \leq - C_2 \big(2^j/\sqrt{N} \big)^2$ for each $j \leq j_*$. Combining this with the previous inequality,

$$\begin{aligned} \mathbb{P}(\sqrt{N} \vert \vert \htc -\thetaTarg \vert \vert_2>2^r)
& \leq \sum_{j=r+1}^{j_*} \mathbb{P} \Big( \inf\limits_{\theta \in \Theta_j}   \Delta_N(\theta) < -C_2 \cdot 2^{2j}/N \Big) + \mathbb{P}(\vert \vert \htc -\thetaTarg \vert \vert_2 > \epsilon)
\\ & \leq \sum_{j=r+1}^{j_*} \mathbb{P} \Big( \sup\limits_{\theta \in \Theta_j} \vert   \Delta_N(\theta) \vert > C_2 \cdot 2^{2j}/N \big) + \mathbb{P}(\vert \vert \htc -\thetaTarg \vert \vert_2 > \epsilon)
\\ & \leq \sum_{j=r+1}^{j_*} \frac{N \cdot \e \big[ \sup_{\theta \in \Theta_j} \vert   \Delta_N(\theta) \vert \big]}{C_2 \cdot 2^{2j} } + \mathbb{P}(\vert \vert \htc -\thetaTarg \vert \vert_2 > \epsilon)
\\ & \leq \frac{N }{C_2 }\sum_{j=r+1}^{j_*}  2^{-2j} \e \Bigg[ \sup_{\theta \in B(\thetaTarg;\frac{2^j}{\sqrt{N}})} \vert   \Delta_N(\theta) \vert \Bigg] + \mathbb{P}(\vert \vert \htc -\thetaTarg \vert \vert_2 > \epsilon)
\\ & \leq \frac{N }{C_2}\sum_{j=r+1}^{j_*}  2^{-2j} \frac{C \cdot 2^j/\sqrt{N}}{\sqrt{N}} + \mathbb{P}(\vert \vert \htc -\thetaTarg \vert \vert_2 > \epsilon)
\\ & \leq \frac{C }{C_2}\sum_{j=r+1}^{j_*}  2^{-j}  + \mathbb{P}(\vert \vert \htc -\thetaTarg \vert \vert_2 > \epsilon).
 \end{aligned}$$ Above the penultimate inequality follows from an application of Inequality \eqref{eq:SupDeltaExpectationBound} which applies because $2^j/\sqrt{N} \leq \delta_*$ for all $j \leq j_*$. Noting that $\sum_{j=r+1}^{j_*} 2^{-j} < \sum_{j=r+1}^{\infty} 2^{-j} \leq 2^{-r}$ letting $C_3= C/C_2 \in (0,\infty)$ the previous inequality implies that $$\mathbb{P} \big(\sqrt{N} \vert \vert \htc -\thetaTarg \vert \vert_2>2^r \big) \leq C_3 2^{-r} + \mathbb{P}(\vert \vert \htc -\thetaTarg \vert \vert_2> \epsilon).$$ Note that the above argument holds for any fixed $N, r \in \mathbb{Z}_+$ and recall that $\epsilon=\min\{\delta_*/2,\delta_2/2\}$ did not depend on $N$ and $r$. Hence $$\mathbb{P} \big(\sqrt{N} \vert \vert \htc -\thetaTarg \vert \vert_2>2^r \big) \leq C_3 2^{-r} + \mathbb{P}(\vert \vert \htc -\thetaTarg \vert \vert_2 > \epsilon) \quad \text{for all } \quad  N, r \in \mathbb{Z}_+.$$ Now recall by Proposition \ref{prop:ComponentEstimatorConsistency}, $\htc \xrightarrow{p} \thetaTarg$. Thus for any fixed $\eta>0$ we can let $N_*$ be an $N$ such that for $N> N_*$, $\mathbb{P}(\vert \vert \htc -\thetaTarg \vert \vert_2> \epsilon)< \eta/2$ and let $r_*$ be an integer sufficiently large such that $C_3 2^{-r_*} < \eta/2$. Thus $$ \mathbb{P} \big(\sqrt{N} \vert \vert \htc -\thetaTarg \vert \vert_2>2^{r_*} \big) \leq C_3 2^{-r_*} + \mathbb{P}(\vert \vert \htc -\thetaTarg \vert \vert_2> \epsilon) < \eta \quad \text{ for all } N>N_*,$$ and moreover for any $\eta>0$ such a $2^{r_*} \in (0,\infty)$ and an $N_* \in \mathbb{Z}_+$ exist that satisfy the above statement. Thus $\sqrt{N} \vert \vert \htc -\thetaTarg \vert \vert_2=O_p(1)$. $\sqrt{N} \vert \vert \hgc -\gammaTarg \vert \vert_2=O_p(1)$ by an analogous argument.
\end{proof}

 \subsection{Helpful results for proving asymptotic linearity}

To show $\htc$ and $\hgc$ are asymptotically linear, we must first establish some properties of the gradient of the loss function evaluated at $\thetaTarg$ and $\gammaTarg$ under Assumption \ref{assump:SmoothEnoughForAsymptoticLineariaty}. Recall that for all $\theta' \in \Theta$ and $\argx \in \xSpace \cup \xProxySpace$, $\dot{l}_{\theta'}(\argx)$ denotes the vector of upper right-hand Dini partial derivatives of $\theta \mapsto l_{\theta}(\argx)$ evaluated at $\theta=\theta'$.
Thus by Assumption \ref{assump:SmoothEnoughForAsymptoticLineariaty}(iii), almost surely $$\dot{l}_{\thetaTarg}(\goodX)= \nabla_{\theta} l_{\theta}(\goodX) \big|_{\theta=\thetaTarg} \quad \text{and} \quad \dot{l}_{\gammaTarg}(\proxyX)= \nabla_{\theta} l_{\theta}(\proxyX) \big|_{\theta=\gammaTarg}.$$ The following fact about  $\dot{l}_{\thetaTarg}$ and $\dot{l}_{\gammaTarg}$ will be used throughout the remainder of the appendix.

\begin{fact}\label{fact:PropertiesOfGradients}
    Under Assumption \ref{assump:SmoothEnoughForAsymptoticLineariaty}, $\e[\dot{l}_{\thetaTarg}(\goodX)]=0$ and $\e[\dot{l}_{\gammaTarg}(\proxyX)]=0$, and moreover,  for any $\argx \in \xSpace$ and $\argxProxy \in \xProxySpace$, $\vert \vert \dot{l}_{\thetaTarg}(\argx) \vert \vert_{\infty} \leq M(\argx)$ and $\vert \vert \dot{l}_{\gammaTarg}(\argxProxy) \vert \vert_{\infty} \leq \tilde{M}(\argxProxy)$, where $M(\cdot)$ and $\tilde{M}(\cdot)$ are the Lipschitz functions guaranteed by Assumption \ref{assump:SmoothEnoughForAsymptoticLineariaty}(iv).
\end{fact}

\begin{proof}
   Fix $j \in [d]$.
   Note that there exists an $h_0$ such that for all $h \in (0,h_0)$, $\thetaTarg+h e_j \in \mathcal{L}_{\thetaTarg}$. Hence for all $h \in (0,h_0)$, by Assumption \ref{assump:SmoothEnoughForAsymptoticLineariaty}(iv), 
   $$\Big| \frac{l_{\thetaTarg+he_j}(\goodX)-l_{\thetaTarg}(\goodX)}{h} \Big| \leq M(\goodX) \quad \text{where} \quad \e[M(\goodX)] < \e[1+M^2(\goodX)] < \infty.$$ Hence by the definition of a gradient and the dominated convergence theorem, 
   $$\begin{aligned}   \ \big[ \nabla L(\thetaTarg) \big]_j= \lim_{h \to 0} \frac{L(\thetaTarg+he_j)-L(\thetaTarg)}{h} & = \lim_{h \to 0}  \e \Big[ \frac{l_{\thetaTarg+he_j}(\goodX)-l_{\thetaTarg}(\goodX)}{h} \Big]
   \\ & = \e \Big[ \lim_{h \to 0} \frac{l_{\thetaTarg+he_j}(\goodX)-l_{\thetaTarg}(\goodX)}{h} \Big] = \e[ e_j^\tran \dot{l}_{\thetaTarg}(\goodX)], \end{aligned}$$ where the last step holds because $\theta \mapsto l_{\theta}(\goodX)$ is differentiable at $\theta=\thetaTarg$ almost surely by Assumption \ref{assump:SmoothEnoughForAsymptoticLineariaty}(iii). 
   Since this argument holds for each $j \in [d]$, it follows that $\nabla L(\thetaTarg)=\e[ \dot{l}_{\thetaTarg}(\goodX)]$. 

    Now because under Assumptions \ref{assump:SmoothEnoughForAsymptoticLineariaty}(i), \ref{assump:SmoothEnoughForAsymptoticLineariaty}(v) and \ref{assump:SmoothEnoughForAsymptoticLineariaty}(ii), $\theta \mapsto L(\theta)$ is convex, twice differentiable at $\thetaTarg$, and uniquely minimized at $\thetaTarg$, its gradient at $\thetaTarg$ must be zero (i.e., $\nabla L(\thetaTarg)=0$). Combining this with the previous result, $\e[ \dot{l}_{\thetaTarg}(\goodX)]=\nabla L(\thetaTarg)=0$. A similar argument shows $\e[ \dot{l}_{\gammaTarg}(\proxyX)]=\nabla \tilde{L}(\gammaTarg)=0$.

    To establish an upper bound on $\vert \vert \dot{l}_{\thetaTarg}(\argx) \vert \vert_{\infty}$, fix $\argx \in \xSpace$ and $j \in [d]$. Note that there exists an $h_0$ such that for all $h \in (0,h_0)$, $\thetaTarg+h e_j \in \mathcal{L}_{\thetaTarg}$. Hence by the definition of an upper right-hand Dini partial derivative and by Assumption \ref{assump:SmoothEnoughForAsymptoticLineariaty}(iv), $$\vert [\dot{l}_{\thetaTarg}(\argx)]_j \vert=\Big| \limsup_{h \downarrow 0} \frac{l_{\thetaTarg+he_j}(\argx)-l_{\thetaTarg}(\argx)}{h}  \Big| \leq \limsup_{h \downarrow 0}\Big|  \frac{l_{\thetaTarg+he_j}(\argx)-l_{\thetaTarg}(\argx)}{h}  \Big| \leq M(\argx).$$ Since the above argument holds for any fixed $\argx \in \xSpace$ and $j \in [d]$, it follows that $\vert \vert \dot{l}_{\thetaTarg}(\argx) \vert \vert_{\infty} \leq M(\argx)$ for all $\argx \in \xSpace$. An analogous argument shows that $\vert \vert \dot{l}_{\gammaTarg}(\argxProxy) \vert \vert_{\infty} \leq \tilde{M}(\argxProxy)$ for all $\argxProxy \in \xProxySpace$.
\end{proof}

To show $\htc$ and $\hgc$ are asymptotically linear, we must also prove a helpful lemma about certain centered and scaled empirical processes. In particular, define $\mathbb{G}_N$ and $\tilde{\mathbb{G}}_N$ to be operators such that for any $p^* \in \mathbb{Z}_+$ and function $f : \mathbb{R}^p \to \mathbb{R}^{p^*}$, \begin{equation}\label{eq:G_operator}
    \mathbb{G}_N f \equiv \sqrt{N} \big(\bfs_N f -\e[f(\goodX)] \big) \quad \text{and} \quad \tilde{\mathbb{G}}_N f \equiv \sqrt{N} \big(\bfts_N f -\e[f(\proxyX)] \big).
\end{equation} It also helps to define for each $h \in \mathbb{R}^d$ and $N \in \mathbb{Z}_+$, \begin{equation}\label{eq:ZnH_Def}
    Z_N(h) \equiv \mathbb{G}_N \Big( \sqrt{N}(l_{\thetaTarg+h/\sqrt{N}}-l_{\thetaTarg}) \Big) - h^\tran \mathbb{G}_N \dot{l}_{\thetaTarg}  \quad \text{and} \quad \tilde{Z}_N(h) = \tilde{\mathbb{G}}_N \Big( \sqrt{N}(l_{\gammaTarg+h/\sqrt{N}}-l_{\gammaTarg}) \Big) - h^\tran \tilde{\mathbb{G}}_N \dot{l}_{\gammaTarg}.
\end{equation}

With these definitions, we can prove the following lemma which is subsequently used to establish asymptotic linearity.

\begin{lemma}\label{lemma:UniformPConvToZero}
    Under \samplingSchemeName{} and Assumptions \ref{assump:IIDUnderlyingData},\ref{assump:LabellingRuleOverlap}, and \ref{assump:SmoothEnoughForAsymptoticLineariaty}, for any $r \in (0,\infty)$, 
    $$\sup_{\substack{h \in \mathbb{R}^d \\ \vert \vert h \vert \vert_2 \leq r}  } Z_N(h) \xrightarrow{p} 0 \quad \text{and} \quad \sup_{\substack{h \in \mathbb{R}^d \\ \vert \vert h \vert \vert_2 \leq r}  } \tilde{Z}_N(h) \xrightarrow{p} 0,$$ where  $Z_N(h)$ and $\tilde{Z}_N(h)$ are defined at \eqref{eq:ZnH_Def}. 
\end{lemma}

\begin{proof}
    Fix $r \in (0,\infty)$ and let $B_r = \{ h \in \mathbb{R}^d \ : \ \vert \vert h \vert \vert_2 \leq r \}$ denote a Euclidean ball about $0$. 
Next observe that by \eqref{eq:ZnH_Def} for each $N \in \mathbb{Z}_+$ and $h \in B_r$, 
$$Z_N(h) \equiv \mathbb{G}_N \Big( \sqrt{N}(l_{\thetaTarg+h/\sqrt{N}}-l_{\thetaTarg}) \Big) - h^\tran \mathbb{G}_N \dot{l}_{\thetaTarg} = N \cdot \Delta_N \big(\thetaTarg+\frac{h}{\sqrt{N}} \big) -h^\tran \mathbb{G}_N \dot{l}_{\thetaTarg},$$ where $\Delta_N$ and $\mathbb{G}_N$ are defined at \eqref{eq:DeltaDef} and \eqref{eq:G_operator}, respectively. 

We will first show that for each $h \in B_r$, $Z_N(h) \xrightarrow{p} 0$. To do this, fix $h \in B_r$ and observe that $\e[Z_N(h)]=0$ as a consequence of result (I) in Corollary \ref{cor:PropertiesOfWeightedAveragees}. Next note that by result (III) in Corollary \ref{cor:PropertiesOfWeightedAveragees}, $$\begin{aligned} \var\big(Z_N(h) \big) & = N \cdot \var \Big( \bfs_N  \big( \sqrt{N}(l_{\thetaTarg+h/\sqrt{N}}-l_{\thetaTarg}) - h^\tran \dot{l}_{\thetaTarg} \big) \Big) 
\\ & \leq b^{-2K} \cdot \e \Big[ \Big(  \sqrt{N}\big(l_{\thetaTarg+h/\sqrt{N}} (\goodX)-l_{\thetaTarg}(\goodX)  \big) - h^\tran \dot{l}_{\thetaTarg}(\goodX)  \Big)^2 \Big]. \end{aligned}$$ By Assumption \ref{assump:SmoothEnoughForAsymptoticLineariaty}(iii), $\theta \mapsto l_{\theta}(\goodX)$ is differentiable at $\thetaTarg$ almost surely, and hence almost surely, $$\lim_{N \to \infty} \sqrt{N}\big(l_{\thetaTarg+h/\sqrt{N}} (\goodX)-l_{\thetaTarg}(\goodX)  \big) - h^\tran \dot{l}_{\thetaTarg}(\goodX)=0.$$ Also observe that for almost every $\goodX$, by Assumption \ref{assump:SmoothEnoughForAsymptoticLineariaty}(iv) for $N$ sufficiently large such that $\thetaTarg+h/\sqrt{N} \in \mathcal{L}_{\thetaTarg}$, $$\begin{aligned}
\Big(  \sqrt{N}\big(l_{\thetaTarg+h/\sqrt{N}} (\goodX)-l_{\thetaTarg}(\goodX)  \big) - h^\tran \dot{l}_{\thetaTarg}(\goodX)  \Big)^2 & \leq \Big(  \sqrt{N} \cdot \big| l_{\thetaTarg+h/\sqrt{N}} (\goodX)-l_{\thetaTarg}(\goodX)  \big| + \vert h^\tran \dot{l}_{\thetaTarg}(\goodX)  \vert \Big)^2 
\\ & \leq \Big(  \vert \vert h \vert \vert_2 M(\goodX) + \vert \vert h \vert \vert_1 \vert \vert \dot{l}_{\thetaTarg}(\goodX)  \vert \vert_{\infty} \Big)^2 
\\ & \leq \Big(  \vert \vert h \vert \vert_2 M(\goodX) + \vert \vert h \vert \vert_1  M(\goodX) \Big)^2
\\ & = (\vert \vert h \vert \vert_2  +\vert \vert h \vert \vert_1 )^2 \big( M(\goodX) \big)^2,
\end{aligned}$$ where above, the penultimate step follows from Fact \ref{fact:PropertiesOfGradients}. Since the expectation of the right hand side of the above inequality is finite by Assumption \ref{assump:SmoothEnoughForAsymptoticLineariaty}(iv), we can apply the dominated convergence theorem and a previous pointwise convergence result to get that $$  \lim\limits_{N \to \infty} \e \Big[ \Big(  \sqrt{N}\big(l_{\thetaTarg+h/\sqrt{N}} (\goodX)-l_{\thetaTarg}(\goodX)  \big) - h^\tran \dot{l}_{\thetaTarg}(\goodX)  \Big)^2 \Big] =0.$$ Combining this with an earlier inequality, it follows that $\lim_{N \to \infty} \var \big( Z_N(h) \big)=0$. Recalling that $\e[Z_N(h)]=0$, by Chebyshev's inequality, for any $\eta>0$, $$\limsup_{N \to \infty} \mathbb{P} \big( \vert Z_N(h) \vert > \eta \big) \leq \limsup_{N \to \infty} \var \big( Z_N(h) \big) \cdot \eta^{-2}=0,$$ and hence $Z_N(h) \xrightarrow{p} 0$. Moreover, this argument holds for any fixed $h \in B_r$.

Next let $N_0 \in \mathbb{Z}_+$ be large enough such that for all $N>N_0$, $B(\thetaTarg; rN^{-1/2}) \subseteq \mathcal{L}_{\thetaTarg}$, where $B(\thetaTarg; r N^{-1/2})$ denotes a ball of radius $rN^{-1/2}$ about $\thetaTarg$. Moreover note that for all $h \in B_r$ and $N>N_0$,  $\thetaTarg +hN^{-1/2} \in B(\thetaTarg; r N^{-1/2}) \subseteq \mathcal{L}_{\thetaTarg}$. Observe that for any $\epsilon >0$ and $N>N_0$, by rearranging terms and applying the Cauchy-Schwarz inequality and Lemma \ref{lemma:ModContinuityExpectationBound}, $$\begin{aligned} \e \Bigg[ \sup\limits_{\substack{h_1,h_2 \in B_r \\ \vert \vert h_1 -h_2 \vert \vert_2 \leq \epsilon }} \Big| Z_N(h_1) - Z_N(h_2) \Big| \Bigg] & \leq  N \cdot \e \Bigg[ \sup\limits_{\substack{h_1,h_2 \in B_r \\ \vert \vert h_1 -h_2 \vert \vert_2 \leq \epsilon }} \Big|  \Delta_N (\thetaTarg+h_1N^{-1/2}) -\Delta_N (\thetaTarg+h_2 N^{-1/2}) \Big| \Bigg]
\\ & \quad + \e \Bigg[ \sup\limits_{\substack{h_1,h_2 \in B_r \\ \vert \vert h_1 -h_2 \vert \vert_2 \leq \epsilon }} \Big|   (h_2-h_1)^\tran \mathbb{G}_N \dot{l}_{\thetaTarg} \Big| \Bigg]
\\ & \leq  N \cdot \e \Bigg[ \sup\limits_{\substack{\theta,\theta' \in B(\thetaTarg; r N^{-1/2} ) \\ \vert \vert \theta -\theta' \vert \vert_2 \leq \epsilon N^{-1/2} }} \Big|  \Delta_N (\theta) -\Delta_N (\theta') \Big| \Bigg] + \epsilon \cdot \e \big[ \vert \vert \mathbb{G}_N \dot{l}_{\thetaTarg} \vert \vert_2 \big] 
\\ & =  N \cdot \e \big[ \omega_{N,rN^{-1/2}}( \epsilon N^{-1/2} )\big] + \epsilon \cdot \e \big[ \vert \vert \mathbb{G}_N \dot{l}_{\thetaTarg} \vert \vert_2 \big] 
\\ & \leq N \cdot C \sqrt{\epsilon N^{-1/2} \cdot r N^{-1/2}} \cdot N^{-1/2} + \epsilon \cdot \e \big[ \vert \vert \mathbb{G}_N \dot{l}_{\thetaTarg} \vert \vert_2 \big] 
\\ & = C \sqrt{\epsilon r}  + \epsilon \cdot \e \big[ \vert \vert \mathbb{G}_N \dot{l}_{\thetaTarg} \vert \vert_2 \big].
\end{aligned}$$ Also observe that since by result (I) in Corollary \ref{cor:PropertiesOfWeightedAveragees}, $\e[\mathbb{G}_N f]=0$ for all $f : \mathbb{R}^p \to \mathbb{R}$, 
$$ \e \big[ \vert \vert \mathbb{G}_N \dot{l}_{\thetaTarg} \vert \vert_2 \big] = \e \Big[ \sqrt{ \sum_{j=1}^d  \big( \mathbb{G}_N (e_j^\tran \dot{l}_{\thetaTarg}) \big)^2 } \Big]
 \leq \sqrt{\e \Big[  \sum_{j=1}^d  \big( \mathbb{G}_N (e_j^\tran \dot{l}_{\thetaTarg}) \big)^2 \Big] }
 = \sqrt{   \sum_{j=1}^d  \var \big( \mathbb{G}_N (e_j^\tran \dot{l}_{\thetaTarg}) \big)  },$$
and thus recalling the definition of $\mathbb{G}_N$ at \eqref{eq:G_operator} and result (III) of Corollary \ref{cor:PropertiesOfWeightedAveragees},
$$  \e \big[ \vert \vert \mathbb{G}_N \dot{l}_{\thetaTarg} \vert \vert_2 \big]  = \sqrt{   \sum_{j=1}^d N \var \big( \bfs_N (e_j^\tran \dot{l}_{\thetaTarg}) \big)  } 
 \leq \sqrt{   \sum_{j=1}^d \frac{1}{b^{2K}} \cdot \e\big[ \big( e_j^\tran \dot{l}_{\thetaTarg}(\goodX) \big)^2 \big]}
 \leq  \frac{\sqrt { d \cdot \e\big[ \vert \vert  \dot{l}_{\thetaTarg}(\goodX) \vert \vert_{\infty}^2 \big]}}{b^K}.$$ Letting $C_2=b^{-K} \sqrt{d \cdot \e[M^2(\goodX)]} \in (0,\infty)$, $\e \big[ \vert \vert \mathbb{G}_N \dot{l}_{\thetaTarg} \vert \vert_2 \big] \leq C_2$ by Fact \ref{fact:PropertiesOfGradients} and the above inequality. Combining this with an earlier inequality (which held for any $\epsilon>0$ and $N>N_0$), $$ \e \Bigg[ \sup\limits_{\substack{h_1,h_2 \in B_r \\ \vert \vert h_1 -h_2 \vert \vert_2 \leq \epsilon }} \Big| Z_N(h_1) - Z_N(h_2) \Big| \Bigg] \leq C \sqrt{\epsilon r} + C_2 \epsilon \quad \text{for any } \quad \epsilon >0, N>N_0.$$

Now fix $\eta>0$. Next fix $\epsilon>0$ and define $\mathcal{C}_{\epsilon} \subseteq B_r $ to be a finite $\epsilon$-covering of $B_r$ (i.e., $\mathcal{C}_{\epsilon}$ is a finite set such that for any $h \in B_r$, there is an $h' \in \mathcal{C}_{\epsilon}$ for which $\vert \vert h-h' \vert \vert \leq \epsilon$). Observe that for any  $N > N_0$, by Markov's inequality and the above result, $$\begin{aligned}  \mathbb{P} \big( \big|
    \sup_{h \in B_r}  Z_N(h) \big| > \eta \big)   & = \mathbb{P} \Big( \sup\limits_{\substack{h' \in \mathcal{C}_{\epsilon},h \in B_r \\ \vert \vert h -h' \vert \vert_2 \leq \epsilon }} \Big| Z_N(h') + Z_N(h) - Z_N(h') \Big| > \eta \Big)
    \\ & \leq \mathbb{P} \big(
    \sup_{h' \in \mathcal{C}_{\epsilon}} \vert Z_N(h') \vert > \frac{\eta}{2} \big) + \mathbb{P} \Big( \sup\limits_{\substack{h,h' \in B_r \\ \vert \vert h -h' \vert \vert_2 \leq \epsilon }} \Big| Z_N(h) - Z_N(h') \Big| > \frac{\eta}{2} \Big)
      \\ & \leq \mathbb{P} \big(
    \max_{h \in \mathcal{C}_{\epsilon}} \vert Z_N(h) \vert > \frac{\eta}{2} \big) + \frac{2C \sqrt{\epsilon r} + 2C_2 \epsilon}{\eta}.
\end{aligned}$$ Since $\mathcal{C}_{\epsilon}$ is a finite subset of $B_r$ and we showed that $Z_N(h) \xrightarrow{p} 0$ for any fixed $h \in B_r$, it follows that $\max_{h \in \mathcal{C}_{\epsilon}} \vert Z_N(h) \vert \xrightarrow{p} 0$. Thus taking the limsup as $N \to \infty$ of each side of the above inequality, we get that $$\limsup_{N \to \infty} \mathbb{P} \big( \big|
    \sup_{h \in B_r}  Z_N(h) \big| > \eta \big) \leq 0 + \frac{2C \sqrt{\epsilon r} + 2C_2 \epsilon}{\eta}.$$
Since the above argument holds for any $\epsilon>0$, we can consider $\epsilon \downarrow 0$ and it follows that $$\limsup_{N \to \infty} \mathbb{P} \big( \big|
    \sup_{h \in B_r}  Z_N(h) \big| > \eta \big) \leq 0 \Rightarrow   \lim_{N \to \infty} \mathbb{P} \big( \big|
    \sup_{h \in B_r}  Z_N(h) \big| > \eta \big) =0  .$$
Since this argument holds for any $\eta>0$, $\sup_{h \in B_r} \big(   Z_N(h) \big) \xrightarrow{p} 0$, so by definition of $B_r$, the first claim in the lemma holds. By an analogous argument, $$\sup_{\substack{h \in \mathbb{R}^d \\ \vert \vert h \vert \vert_2 \leq r}  } \tilde{Z}_N(h) \xrightarrow{p} 0 .$$
\end{proof}

We are now ready to prove Theorem \ref{theorem:AsymptoticLInearMestsStacked}, which establishes asymptotic linearity of our M-estimators of interest. The proof mirrors that in Theorem 5.23 of \cite{VanderVaartTextbook} which assumes an i.i.d. setting, and cites our previous findings from Theorem \ref{theorem:RootNConsistency} and Lemma \ref{lemma:UniformPConvToZero} which apply to our \samplingSchemeName{} setting where statistical dependency is induced by the sampling process.

\subsection{Proof of asymptotic linearity (Theorem \ref{theorem:AsymptoticLInearMestsStacked})}

    We will show that $\sqrt{N}(\htc-\thetaTarg) = -\sqrt{N} H_{\thetaTarg}^{-1} \bfs_N \dot{l}_{\thetaTarg}+o_p(1)$. $\sqrt{N}(\hgc-\gammaTarg) = -\sqrt{N} H_{\gammaTarg}^{-1} \bfts_N \dot{l}_{\gammaTarg}+o_p(1)$ will hold by an analogous argument while $\sqrt{N}(\hga-\gammaTarg) = -\sqrt{N} H_{\gammaTarg}^{-1} \bftp_N \dot{l}_{\gammaTarg}+o_p(1)$ follows from a standard result for M-estimators (e.g., Theorem 5.23 in \cite{VanderVaartTextbook}). 
    
    To show $\sqrt{N}(\htc-\thetaTarg) = -\sqrt{N} H_{\thetaTarg}^{-1} \bfs_N \dot{l}_{\thetaTarg}+o_p(1)$, fix any sequence $\{ h_N \}_{N=1}^{\infty}$ of elements of $\mathbb{R}^d$ that are bounded in probability (denoted by $h_N =O_p(1)$). 
     We will start by showing $Z_N(h_N) \xrightarrow{p} 0$, where $Z_N(\cdot)$ is defined at $\eqref{eq:ZnH_Def}$. To do this fix $\epsilon>0$. Next fix $\eta>0$. Since $h_N =O_p(1)$, there exists an $N_{\eta} \in \mathbb{Z}_+$ and an $r_{\eta} \in (0,\infty)$ such that $\mathbb{P}( \vert \vert h_N \vert \vert_2 > r_{\eta} ) \leq \eta$ for all $N>N_{\eta}$. Thus choosing such an $N_{\eta}$ and $r_{\eta}$, for any $N>N_{\eta}$ $$\begin{aligned}
    \mathbb{P} \big( \vert Z_N(h_N) \vert > \epsilon \big) \leq \mathbb{P}( \vert \vert h_N \vert \vert_2 > r_{\eta} ) + \mathbb{P} \Big( \Big| \sup_{\substack{h \in \mathbb{R}^d \\ \vert \vert h \vert \vert_2 \leq r_{\eta}}  } Z_N(h) \Big| > \epsilon \Big) \leq \eta + \mathbb{P} \Big( \Big| \sup_{\substack{h \in \mathbb{R}^d \\ \vert \vert h \vert \vert_2 \leq r_{\eta}}  } Z_N(h) \Big| > \epsilon \Big).
    \end{aligned}$$ Taking the $\limsup$ as $N \to \infty$ of each side of the above inequality and noting that by Lemma \ref{lemma:UniformPConvToZero} 
    the second term converges to $0$ as $N \to \infty$, it follows that $\limsup_{N \to \infty} \mathbb{P} \big( \vert Z_N(h_N) \vert > \epsilon \big) \leq \eta$. Since this argument holds for any $\eta >0$ no matter how small, $\limsup_{N \to \infty} \mathbb{P} \big( \vert Z_N(h_N) \vert > \epsilon \big) \leq 0$ implying that $\lim_{N \to \infty} \mathbb{P} \big( \vert Z_N(h_N) \vert > \epsilon \big)=0$. Because this argument holds for any $\epsilon>0$, $Z_N(h_N) \xrightarrow{p} 0$. Because $Z_N(h_N) \xrightarrow{p} 0$, the equivalent statement (see definitions for $\mathbb{G}_N$ and $Z_N(\cdot)$ at \eqref{eq:G_operator} and \eqref{eq:ZnH_Def}) that 
    $$N \bfs_N \big( l_{\thetaTarg+h_N/\sqrt{N}} -l_{\thetaTarg} \big) = N \big( L(\thetaTarg+h_N/\sqrt{N}) - L(\thetaTarg) \big) +   h_N^\tran \mathbb{G}_N \dot{l}_{\thetaTarg}  +o_p(1),$$ must hold. Now since $ L(\cdot)$ has a 2nd-order Taylor expansion about $\thetaTarg$ (by Assumption \ref{assump:SmoothEnoughForAsymptoticLineariaty}(v)), and because $\nabla L(\thetaTarg)=0$ (by Assumptions \ref{assump:SmoothEnoughForAsymptoticLineariaty}(i), \ref{assump:SmoothEnoughForAsymptoticLineariaty}(ii), and \ref{assump:SmoothEnoughForAsymptoticLineariaty}(v)), there is a function $\zeta : \mathbb{R}^d \to \mathbb{R}$ that is continuous at $\thetaTarg$ that satisfies $\lim_{z \to \thetaTarg} \zeta(z)=0$ such that $$N \big( L(\thetaTarg+h_N/\sqrt{N}) - L(\thetaTarg) \big) = N \Big( \frac{1}{2 N} h_N^\tran H_{\thetaTarg} h_N + \zeta \big(\thetaTarg+\frac{h_N}{\sqrt{N}} \big) \big| \big| \frac{h_N}{\sqrt{N}} \big| \big|_2^2 \Big) = \frac{1}{2} h_N^\tran H_{\thetaTarg} h_N + o_p(1).$$ Above the last step holds because $h_N =O_p(1)$ so $\vert \vert h_N \vert \vert_2^2=O_p(1)$ and $\thetaTarg+h_N/\sqrt{N} \xrightarrow{p} \thetaTarg$, and hence because $\lim_{z \to \thetaTarg} \zeta(z)=0$, $\zeta(\thetaTarg+h_N/\sqrt{N}) \xrightarrow{p} 0$ by the continuous mapping theorem. Combining the two expressions displayed above, it follows that \begin{equation}\label{eq:ToolForProvingAsymptoticLinearity}
        N \bfs_N \big( l_{\thetaTarg+h_N/\sqrt{N}} -l_{\thetaTarg} \big) = \frac{1}{2} h_N^\tran H_{\thetaTarg} h_N +  h_N^\tran \mathbb{G}_N   \dot{l}_{\thetaTarg} +o_p(1).
    \end{equation} Notably, \eqref{eq:ToolForProvingAsymptoticLinearity} holds for any sequence $\{ h_N \}_{N=1}^{\infty}$ of elements of $\mathbb{R}^d$ that satisfy $h_N=O_p(1)$.

    Now let $$\hat{h}_N \equiv \sqrt{N} (\htc-\thetaTarg) \quad \text{and} \quad \hat{h}_N^* \equiv -H_{\thetaTarg}^{-1} \mathbb{G}_N \dot{l}_{\thetaTarg} \quad \text{for each } N \in \mathbb{Z}_+.$$ By Theorem \ref{theorem:RootNConsistency}, $\hat{h}_N=O_p(1)$, and we will next show that $\hat{h}_N^*=O_p(1)$. To do this observe that for any $j \in [d]$, $\e[e_j^\tran \hat{h}_N^*]=0$ by result (I) in Corollary \ref{cor:PropertiesOfWeightedAveragees}. Next note that for any $\kappa>0$ and $j \in [d]$, by Chebyshev's inequality, the definition for $\mathbb{G}_N$ at \eqref{eq:G_operator}, and by applying result (III) of Corollary \ref{cor:PropertiesOfWeightedAveragees}, $$\mathbb{P}( \vert e_j^\tran \hat{h}_N^* \vert > \kappa ) \leq \frac{\var( e_j^\tran \hat{h}_N^*) }{\kappa^2} =   \frac{ N \var \big(\mathbb{S}_N (e_j^\tran H_{\thetaTarg}^{-1} \dot{l}_{\thetaTarg}) \big) }{\kappa^2} \leq \frac{ \e \big[ \big( e_j^\tran H_{\thetaTarg}^{-1} \dot{l}_{\thetaTarg}(\goodX) \big)^2 \big]}{b^{2K} \kappa^2}.$$ This further implies that for each $\kappa>0$ and $j \in [d]$,
    $$\mathbb{P}( \vert e_j^\tran \hat{h}_N^* \vert > \kappa ) \leq \frac{1}{\kappa^2 b^{2K}} \cdot \e \big[ \big( d \vert \vert e_j^\tran H_{\thetaTarg}^{-1} \vert \vert_{\infty} \vert \vert \dot{l}_{\thetaTarg}(\goodX) \vert \vert_{\infty} \big)^2 \big] \leq \frac{ d^2}{\kappa^2 \cdot b^{2K}} \cdot \vert \vert e_j^\tran H_{\thetaTarg}^{-1} \vert \vert_{\infty}^2 \e[M^2(\goodX)],$$ where the last step follows from Fact \ref{fact:PropertiesOfGradients}. Since $\e[M^2(\goodX)]< \infty$ (by Assumption \ref{assump:SmoothEnoughForAsymptoticLineariaty}(iv)), and $\vert \vert e_j^\tran H_{\thetaTarg}^{-1} \vert \vert_{\infty} < \infty$  (by Assumption \ref{assump:SmoothEnoughForAsymptoticLineariaty}(v)), all terms except $\kappa$ on the right hand side of the above inequality do not depend on $\kappa$ or $N$ and are finite. Hence for any $\epsilon>0$, we can pick $\kappa$ to be sufficiently large such that the above inequality implies $\mathbb{P}( \vert e_j^\tran \hat{h}_N^* \vert > \kappa ) \leq \epsilon$. This implies $e_j^\tran \hat{h}_N^*=O_p(1)$, and because this argument holds for each $j \in [d]$, $\hat{h}_N^*=O_p(1)$.

    Since $\hat{h}_N=O_p(1)$ and $\hat{h}_N^*=O_p(1)$, and moreover $\thetaTarg+\hat{h}_N/\sqrt{N}=\htc$ and $\mathbb{G}_N \dot{l}_{\thetaTarg}=-H_{\thetaTarg} \hat{h}_N^*$, by plugging in $\hat{h}_N$ and then $\hat{h}_N^*$ into Equation \eqref{eq:ToolForProvingAsymptoticLinearity}, $$N \bfs_N (l_{\htc}-l_{\thetaTarg}) = \frac{1}{2} \hat{h}_N^\tran H_{\thetaTarg} \hat{h}_N +   \hat{h}_N^\tran  \mathbb{G}_N \dot{l}_{\thetaTarg}  + o_p(1), \quad \text{and}$$  $$N \bfs_N \big( l_{\thetaTarg+\hat{h}_N^*/\sqrt{N}} -l_{\thetaTarg} \big) = \frac{1}{2} (\hat{h}_N^*)^\tran H_{\thetaTarg} \hat{h}_N^* +   (\hat{h}_N^*)^\tran  \mathbb{G}_N \dot{l}_{\thetaTarg}  +o_p(1)=-\frac{1}{2} (\hat{h}_N^*)^\tran H_{\thetaTarg} \hat{h}_N^* +o_p(1).$$ Observing that $\htc$ minimizes $\theta \mapsto \bfs_N l_{\theta}$ (see \eqref{eq:PTDComponentEstimatorsOperatorNotation}), $\bfs_N (l_{\htc}-l_{\thetaTarg}) \leq \bfs_N (l_{\thetaTarg+\hat{h}_N^*/\sqrt{N}}-l_{\thetaTarg})$, and hence by the previous two results $$\begin{aligned}
        0 & \geq N \bfs_N (l_{\htc}-l_{\thetaTarg})- N \bfs_N (l_{\thetaTarg+\hat{h}_N^*/\sqrt{N}}-l_{\thetaTarg})
        \\ & = \frac{1}{2} \hat{h}_N^\tran H_{\thetaTarg} \hat{h}_N + \hat{h}_N^\tran \mathbb{G}_N \dot{l}_{\thetaTarg} +\frac{1}{2} (\hat{h}_N^*)^\tran H_{\thetaTarg} \hat{h}_N^*  +o_p(1)
        \\ & = \frac{1}{2} \big(\hat{h}_N - \hat{h}_N^* \big)^\tran H_{\thetaTarg} \big(\hat{h}_N - \hat{h}_N^* \big) + \frac{1}{2} (\hat{h}_N^*)^\tran H_{\thetaTarg}  \hat{h}_N + \frac{1}{2} \hat{h}_N^\tran H_{\thetaTarg}  \hat{h}_N^* + \hat{h}_N^\tran \mathbb{G}_N \dot{l}_{\thetaTarg}+o_p(1)
        \\ & = \frac{1}{2} \big(\hat{h}_N - \hat{h}_N^* \big)^\tran H_{\thetaTarg} \big(\hat{h}_N - \hat{h}_N^* \big)+o_p(1),
         \\ & \geq \frac{1}{2} \lambda_{\text{min}} ( H_{\thetaTarg} ) \vert \vert \hat{h}_N - \hat{h}_N^* \vert \vert_2^2 +o_p(1),
    \end{aligned}$$ where $\lambda_{\text{min}}(\cdot)$ is an operator giving the minimum eigenvalue of a matrix. Above the penultimate step follows because $H_{\thetaTarg}$ is a symmetric matrix and because $H_{\thetaTarg}  \hat{h}_N^*=-\mathbb{G}_N \dot{l}_{\thetaTarg}$. To complete the proof, note that by Assumptions \ref{assump:SmoothEnoughForAsymptoticLineariaty}(i) and \ref{assump:SmoothEnoughForAsymptoticLineariaty}(v), $H_{\thetaTarg} \succ 0$, and hence $\lambda_{\text{min}} ( H_{\thetaTarg} ) >0$. Thus dividing each side of the above inequality by $\lambda_{\text{min}} ( H_{\thetaTarg} )/2$, a positive constant, it follows that $$0 \geq \vert \vert \hat{h}_N - \hat{h}_N^* \vert \vert_2^2 +o_p(1).$$ Since norms are nonnegative, it must be the case that $\vert \vert \hat{h}_N - \hat{h}_N^* \vert \vert_2^2 \xrightarrow{p} 0$ and hence $\hat{h}_N=  \hat{h}_N^* +o_p(1)$. Recalling our definitions of $\hat{h}_N$ and $\hat{h}_N^*$, $$\sqrt{N} ( \htc - \thetaTarg) = - H_{\thetaTarg}^{-1} \mathbb{G}_N \dot{l}_{\thetaTarg}+o_p(1)= -\sqrt{N} H_{\thetaTarg}^{-1} \bfs_N \dot{l}_{\thetaTarg} +o_p(1),$$ where the last step above holds by Definition \eqref{eq:G_operator} for $\mathbb{G}_N$ and because $\e[\dot{l}_{\thetaTarg}(\goodX)]=0$ by Fact \ref{fact:PropertiesOfGradients}. An analogous argument shows that $\sqrt{N}(\hgc-\gammaTarg) = -\sqrt{N} H_{\gammaTarg}^{-1} \bfts_N \dot{l}_{\gammaTarg}+o_p(1)$.

    Combining these results with the standard asymptotic linear expansion of $\sqrt{N}(\hga-\gammaTarg)$ for M-estimation in i.i.d. settings, $$\sqrt{N} \Bigg( \begin{bmatrix}
        \htc \\ \hgc \\ \hga
    \end{bmatrix} - \begin{bmatrix}
        \thetaTarg \\ \gammaTarg \\ \gammaTarg
    \end{bmatrix}  \Bigg) = - \sqrt{N} \begin{bmatrix}
        H_{\thetaTarg}^{-1} \bfs_N \dot{l}_{\thetaTarg}
        \\ H_{\gammaTarg}^{-1} \bfts_N \dot{l}_{\gammaTarg}
        \\ H_{\gammaTarg}^{-1} \bftp_N \dot{l}_{\gammaTarg}
    \end{bmatrix} + o_p(1).$$ Plugging in the definitions of the averaging operators $\bfs_N$, $\bfts_N$, and $\bftp_N$ at \eqref{eq:AveragingOperatorDefs}, $$\sqrt{N} \Bigg( \begin{bmatrix}
        \htc \\ \hgc \\ \hga
    \end{bmatrix} - \begin{bmatrix}
        \thetaTarg \\ \gammaTarg \\ \gammaTarg
    \end{bmatrix}  \Bigg) = - \frac{1}{\sqrt{N}} \sum_{i=1}^N \begin{bmatrix}
       H_{\thetaTarg}^{-1} & 0 & 0 
       \\ 0 & H_{\gammaTarg}^{-1} & 0   
       \\ 0 & 0 & H_{\gammaTarg}^{-1}
    \end{bmatrix} \begin{bmatrix}
       W_i \dot{l}_{\thetaTarg}(\goodX_i)
        \\ W_i \dot{l}_{\gammaTarg}(\proxyX_i)
        \\ \dot{l}_{\gammaTarg}(\proxyX_i)
    \end{bmatrix} + o_p(1).$$ Moreover, the left hand side of the above equation is $O_p(1)$ as a direct consequence of Theorem \ref{theorem:RootNConsistency}.

\subsection{Proof of asymptotic linearity of the Multiwave PTD estimator (Corollary \ref{cor:PTDEstAsymptoticallyLinear})}

In the setting of Corollary \ref{cor:PTDEstAsymptoticallyLinear}, $\hat{\Omega}=\Omega+o_p(1)$ and Assumptions \ref{assump:IIDUnderlyingData}, \ref{assump:LabellingRuleOverlap}, and \ref{assump:SmoothEnoughForAsymptoticLineariaty} hold. Recalling the definition of $\htPTD$ at \eqref{eq:PTD_estimator}, by applying Proposition \ref{prop:ComponentEstimatorConsistency}, $$\htPTD= \hat{\Omega} \hga + (\htc- \hat{\Omega} \hgc) = \big(\Omega+o_p(1) \big) \big(\gammaTarg+o_p(1) \big) + \Big( \thetaTarg+o_p(1)- \big(\Omega+o_p(1) \big) \big(\gammaTarg+o_p(1) \big) \Big)=\thetaTarg+o_p(1).$$ Hence $\htPTD \xrightarrow{p} \thetaTarg$ as $N \to \infty$.

Since $\sqrt{N}(\hgc-\gammaTarg)=O_p(1)$ and $\sqrt{N}(\hga-\gammaTarg)=O_p(1)$ by Theorem \ref{theorem:RootNConsistency},

$$\begin{aligned}
        \sqrt{N} \big( \htPTD -\thetaTarg \big) & = \sqrt{N} \big( \hat{\Omega} \hga + (\htc- \hat{\Omega} \hgc) - \thetaTarg \big)
        \\ & = \sqrt{N}(\htc -\thetaTarg) + \hat{\Omega} \Big( -\sqrt{N}(\hgc -\gammaTarg) +\sqrt{N}(\hga -\gammaTarg) \Big) 
        \\ & = \sqrt{N}(\htc -\thetaTarg) + \big( \Omega+o_p(1) \big) \Big( -\sqrt{N}(\hgc -\gammaTarg) +\sqrt{N}(\hga -\gammaTarg) \Big) 
          \\ & = \sqrt{N}(\htc -\thetaTarg) + \Omega \Big( -\sqrt{N}(\hgc -\gammaTarg) +\sqrt{N}(\hga -\gammaTarg) \Big) + o_p(1)O_p(1).
    \end{aligned}$$

Letting  $A_{\Omega} \equiv \begin{bmatrix}
    I_{d \times d} & -\Omega & \Omega
 \end{bmatrix} \in \mathbb{R}^{d \times 3d}$ and by applying Theorem \ref{theorem:AsymptoticLInearMestsStacked} and rearranging terms, the previous expression simplifies as follows: $$\begin{aligned} \sqrt{N} (\htPTD-\thetaTarg) & = A_{\Omega} \sqrt{N} \Bigg( \begin{bmatrix}
        \htc \\ \hgc \\ \hga
    \end{bmatrix} - \begin{bmatrix}
        \thetaTarg \\ \gammaTarg \\ \gammaTarg
    \end{bmatrix}  \Bigg) +o_p(1) O_p(1)
    \\ & = A_{\Omega} \Bigg( - \frac{1}{\sqrt{N}}  \sum_{i=1}^N \begin{bmatrix}
       H_{\thetaTarg}^{-1} & 0 & 0 
       \\ 0 & H_{\gammaTarg}^{-1} & 0   
       \\ 0 & 0 & H_{\gammaTarg}^{-1}
    \end{bmatrix} \begin{bmatrix}
       W_i \dot{l}_{\thetaTarg}(\goodX_i)
        \\ W_i \dot{l}_{\gammaTarg}(\proxyX_i)
        \\ \dot{l}_{\gammaTarg}(\proxyX_i)
    \end{bmatrix} + o_p(1) \Bigg) +o_p(1)
    \\ & = -\frac{1}{\sqrt{N}} \sum_{i=1}^N \Big( W_i H_{\thetaTarg}^{-1} \dot{l}_{\thetaTarg}(\goodX_i) +(1-W_i) \Omega H_{\gammaTarg}^{-1} \dot{l}_{\gammaTarg}(\proxyX_i) \Big)  + o_p(1). \end{aligned}$$

\section{Establishing asymptotic normality}\label{sec:CLTProof}

To establish asymptotic normality of $\htPTD$ from its asymptotic linear expansion at \eqref{eq:AsympLinearExpansion} we introduce the following weights 

\begin{equation}\label{eq:IndepAggregatedWeightsMainText}\bar{W}_i \equiv \sum_{k=1}^K c_k \prod_{j=1}^{k-1} \Big( \frac{1-\mathbbm{1} \{ U_i^{(j)} \leq \LimitingLabelRulePi{j}(\proxyX_i) \} }{1- \LimitingLabelRulePi{j}(\proxyX_i) } \Big) \frac{\mathbbm{1} \{ U_i^{(k)} \leq \LimitingLabelRulePi{k}(\proxyX_i) \} }{\LimitingLabelRulePi{k}(\proxyX_i) } \quad \text{for} \quad i \in [N].\end{equation} Above $\LimitingLabelRulePi{k} : \xProxySpace \to [b,1-b]$ is a fixed function introduced in Assumption \ref{assump:CLTRgularityConditions}(ii) for each $k \in [K]$, while $U_i^{(k)} \stackrel{\text{i.i.d.}}{\sim} \text{Unif}[0,1]$ for $i \in [N]$, $k \in [K]$ were generated independently of the data during \samplingSchemeName{.} As a consequence $(\bar{W}_i)_{i=1}^N$ are i.i.d. Although these weights cannot be calculated from the data, they still exist as useful theoretical tools.

In the next subsection we prove some results about the i.i.d. weights $\bar{W}_i$, which help us establish asymptotic normality of $\htPTD$. Notably, in Proposition \ref{prop:CanReduceToIIDWeights}, we prove that $N^{-1/2} \sum_{i=1}^N (W_i-\bar{W}_i) f(\datvecraw_i) \xrightarrow{p} 0$ as $N \to \infty$ under our assumptions for any $f:\mathbb{R}^q \to \mathbb{R}$ satisfying a certain bounded moment condition. Proposition \ref{prop:CanReduceToIIDWeights} enables us to prove a CLT for $\htPTD$ by removing the statistical dependency from the $W_i$ in the asymptotic linear expansion for $\htPTD$. In particular, using Proposition \ref{prop:CanReduceToIIDWeights}, the asymptotic linear expansion of $\htPTD$ at \eqref{eq:AsympLinearExpansion} can be restated in terms of the i.i.d. weights $(\bar{W}_i)_{i=1}^N$. With i.i.d. samples, the multivariate central limit theorem applies, enabling us to derive a CLT for $\htPTD$ under \samplingSchemeName{.}

\subsection{Construction and properties of i.i.d. approximations to the weights}

To study properties of the weights $\bar{W}_i$ we start by introducing and motivating an alternative formula for $\bar{W}_i$. Recall from \eqref{eq:ComplementOr1Operator} that $\phi_{-1},\phi_0,\phi_1: [0,1] \to [0,1]$ are functions given by $$\phi_{-1}(t) = 1-t, \quad \phi_0(t)=t,  \quad  \text{and} \quad \phi_1(t)=1 \quad  \text{ for all } t \in [0,1],$$ from \eqref{eq:IPWsingleSingle} and \eqref{eq:WkFormulaAlt} that for each $i \in [N]$ and $k \in [K]$ $$W_i^{(k)}= \prod_{j=1}^k  W_i^{(k,j)} =\prod_{j=1}^K  W_i^{(k,j)} \quad \text{where} \quad  W_i^{(k,j)} =  \frac{\phi_{\text{sgn}(j-k)} \big( I_i^{(j)} \big)}{\phi_{\text{sgn}(j-k)} \big( \LabelRulePi{j}(\proxyX_i) \big)} \quad \text{ for each } j \in [K],$$ and from \eqref{eq:aggregated_Wk} that $W_i =\sum_{k=1}^K c_k W_i^{(k)}$ for each $i \in [N]$. Also recall that for each $j \in [K]$ and $i \in [N]$, $I_i^{(j)} = \mathbbm{1} \{ U_i^{(j)} \leq \LabelRulePi{j}(\proxyX_i) \}$ where the  $U_i^{(j)}$ are i.i.d. Unif$[0,1]$ random variables. For each $j \in [K]$ let $\LimitingLabelRulePi{j} : \mathbb{R}^p \to [b,1-b]$ denote the limiting labelling rule specified in Assumption \ref{assump:CLTRgularityConditions}. We can thus construct i.i.d. weights that are approximately equal to $W_i^{(k,j)}$ and $W_i^{(k)}$ by defining for each $i \in [N]$ and $k \in [K]$ \begin{equation}\label{eq:WkIndepVersionDecomp}
    \bar{W}_i^{(k)} \equiv \prod_{j=1}^k  \bar{W}_i^{(k,j)} =\prod_{j=1}^K  \bar{W}_i^{(k,j)} \quad \text{where} \quad  \bar{W}_i^{(k,j)} \equiv  \frac{\phi_{\text{sgn}(j-k)} \big( \mathbbm{1} \{ U_i^{(j)} \leq \LimitingLabelRulePi{j}(\proxyX_i) \} \big)}{\phi_{\text{sgn}(j-k)} \big( \LimitingLabelRulePi{j}(\proxyX_i) \big)} \quad \text{ for each } j \in [K].
\end{equation} Using these definitions and recalling the definition of $\bar{W}_i$ from Formula \eqref{eq:IndepAggregatedWeightsMainText}, observe that \begin{equation}\label{eq:WkAggregatedIndepVersion}
\bar{W}_i = \sum_{k=1}^K c_k \bar{W}_i^{(k)}  \quad \text{for each } \quad i \in [N].\end{equation} The following lemmas state some helpful properties of the weights $\bar{W}_i$ and the terms $\bar{W}_i^{(k)}$ and $\bar{W}_i^{(k,j)}$ in its decomposition.

\begin{lemma}\label{lemma:PropertiesOfIndependentVersionsOfWeights}
    Under \samplingSchemeName{} and Assumptions \ref{assump:IIDUnderlyingData} and \ref{assump:CLTRgularityConditions}, the following properties hold:
    \begin{enumerate}[(I)]
        \item $(\bar{W}_i,\datvecraw_i)_{i=1}^N$ is an i.i.d. sequence of $N$ random vectors in $\mathbb{R}^{q+1}$. Moreover, for each $i \in [N]$, $(W_i,\bar{W}_i,\datvecraw_i)$ has the same joint distribution as $(W_1,\bar{W}_1,\datvecraw_1)$.
        \item For any $i \in [N]$ and measurable $f: \mathbb{R}^q \to \mathbb{R}$, $\e[ \bar{W}_i^{(k)} f(\datvecraw_i)]=\e[f(\datvecraw)]$ for each $k \in [K]$ and moreover $\e[ \bar{W}_i f(\datvecraw_i)]=\e[f(\datvecraw)]$. 
        \item For any $i \in [N]$ and measurable $f,g: \mathbb{R}^q \to \mathbb{R}$, $\cov \big(\bar{W}_i f(\datvecraw_i), g(\datvecraw_i) \big) = \cov \big(f(\datvecraw), g(\datvecraw) \big)$.
        \item For each $i \in [N]$ and measurable $f,g: \mathbb{R}^q \to \mathbb{R}$, $$\cov \big( \bar{W}_i f(\datvecraw_i), \bar{W}_i g(\datvecraw_i) \big)= \sum_{k=1}^K c_k^2 \e \Big[ \frac{ f(\datvecraw) g(\datvecraw)}{\LimitingLabelRulePi{1:k}(\proxyX)} \Big] - \e[f(\datvecraw)] \e[g(\datvecraw)],$$ where $\LimitingLabelRulePi{1:k}(\cdot) \text{ is defined at \eqref{eq:PiBarProd1tok}}.$
    \end{enumerate}
    
\end{lemma}

\begin{proof} 

To prove (I), recall from Assumption \ref{assump:IIDUnderlyingData} that $(\datvecraw_i)_{i=1}^N$ are i.i.d. Also recall that under \samplingSchemeName{,} $\big( (U_i^{(j)})_{j=1}^K \big)_{i=1}^N$ are each i.i.d. uniform random variables generated independently of the underlying data $(\datvecraw_i)_{i=1}^N$. Thus $\big(\datvecraw_i, (U_i^{(j)})_{j=1}^K \big)_{i=1}^N$ is a sequence of $N$ i.i.d. random vectors in $\mathbb{R}^{q+K}$. 
Noting that for $j \in [K]$ the function $\bar{\pi}^{(j)} : \mathbb{R}^p \to [b,1-b]$ is a fixed measurable function given by Assumption \ref{assump:CLTRgularityConditions}(ii), and recalling that by \eqref{eq:WkIndepVersionDecomp} and \eqref{eq:WkAggregatedIndepVersion} $$\bar{W}_i= \sum_{k=1}^K c_k \prod_{j=1}^k \frac{\phi_{\text{sgn}(j-k)} \big( \mathbbm{1} \{ U_i^{(j)} \leq \LimitingLabelRulePi{j}(\proxyX_i) \} \big)}{\phi_{\text{sgn}(j-k)} \big( \LimitingLabelRulePi{j}(\proxyX_i) \big)} \quad \text{and} \quad \datvecraw_i=(\proxyX_i,\xmiss_i) \quad  \text{for each} \quad i \in [N],$$ it is clear that there exists a fixed, measurable function $h: \mathbb{R}^{q+K} \to \mathbb{R}$ such that almost surely $\bar{W}_i =h \big(\datvecraw_i, (U_i^{(j)})_{j=1}^K \big)$ for each $i \in [N]$. Let $h$ be such a fixed measurable function. Since $\big(\datvecraw_i, (U_i^{(j)})_{j=1}^K \big)_{i=1}^N$ is a sequence of $N$ i.i.d. random vectors and $\bar{W}_i$ is a fixed, measurable function of the $i$th vector in the sequence, $\big(\bar{W}_i,\datvecraw_i, (U_i^{(j)})_{j=1}^K \big)_{i=1}^N$ are i.i.d. This further implies $\big(\bar{W}_i,\datvecraw_i \big)_{i=1}^N$ is an i.i.d. sequence of $N$ random vectors in $\mathbb{R}^{q+1}$. 

Now fix $i \in [N]$, and note that as a consequence of Lemma \ref{lemma:StrongerExchangeability}, $\big(W_i,\datvecraw_i,(U_i^{(j)})_{j=1}^K \big)$ has the same joint distribution as $\big(W_1,\datvecraw_1,(U_1^{(j)})_{j=1}^K \big)$. It follows that 
$$(\bar{W}_i,W_i,\datvecraw_i) = \Big(h \big(\datvecraw_i, (U_i^{(j)})_{j=1}^K \big),W_i,\datvecraw_i \Big) \stackrel{\text{dist}}{=} \Big(h \big(\datvecraw_1, (U_1^{(j)})_{j=1}^K \big),W_1,\datvecraw_1 \Big) =(\bar{W}_1,W_1,\datvecraw_1).$$ We have thus shown $(\bar{W}_i,W_i,\datvecraw_i)$ has the same joint distribution as $(\bar{W}_1,W_1,\datvecraw_1)$, and this argument holds for any $i \in [N]$, completing the proof of (I).

To prove (II), fix a measurable $f: \mathbb{R}^q \to \mathbb{R}$ and $i \in [N]$. Since $U_i^{(1)},\dots,U_i^{(K)} \stackrel{\text{i.i.d.}}{\sim} \text{Unif}[0,1]$ and are independent of $\datvecraw_i =(\proxyX_i,\xmiss_i)$ it follows that conditionally on $\datvecraw_i$, $U_i^{(1)},\dots,U_i^{(K)} \stackrel{\text{i.i.d.}}{\sim} \text{Unif}[0,1]$. Thus, by the definition of $\bar{W}_i^{(k)}$ at \eqref{eq:WkIndepVersionDecomp}, linearity of expectation and the tower property, for each $k \in [K]$ $$\begin{aligned} 
\e [\bar{W}_i^{(k)} f(\datvecraw_i)] & =  \e \Bigg[ \e \Big[ \prod_{j=1}^k \frac{\phi_{\text{sgn}(j-k)} \big( \mathbbm{1} \{ U_i^{(j)} \leq \LimitingLabelRulePi{j}(\proxyX_i) \} \big)}{\phi_{\text{sgn}(j-k)} \big( \LimitingLabelRulePi{j}(\proxyX_i) \big)} \cdot f(\datvecraw_i) \Big| \datvecraw_i \Big]\Bigg]
\\ & =  \e \Bigg[   \frac{ \e \big[ \prod_{j=1}^k  \phi_{\text{sgn}(j-k)} \big( \mathbbm{1} \{ U_i^{(j)} \leq \LimitingLabelRulePi{j}(\proxyX_i) \} \big) \giv \datvecraw_i\big]}{\prod_{j=1}^k \phi_{\text{sgn}(j-k)} \big( \LimitingLabelRulePi{j}(\proxyX_i) \big)} \cdot f(\datvecraw_i) \Bigg]
\\ & =  \e \Bigg[   \frac{\prod_{j=1}^k \e \big[ \phi_{\text{sgn}(j-k)} \big( \mathbbm{1} \{ U_i^{(j)} \leq \LimitingLabelRulePi{j}(\proxyX_i) \} \big) \giv \datvecraw_i\big]}{\prod_{j=1}^k \phi_{\text{sgn}(j-k)} \big( \LimitingLabelRulePi{j}(\proxyX_i) \big)} \cdot f(\datvecraw_i) \Bigg]
\\ & = \e \Bigg[   \frac{ \e \big[ \mathbbm{1} \{ U_i^{(k)} \leq \LimitingLabelRulePi{k}(\proxyX_i) \}  \giv \datvecraw_i\big] \cdot \prod_{j=1}^{k-1}  \big( 1-\e \big[\mathbbm{1} \{ U_i^{(j)} \leq \LimitingLabelRulePi{j}(\proxyX_i) \} \giv \datvecraw_i\big] \big) }{\LimitingLabelRulePi{k}(\proxyX_i) \prod_{j=1}^{k-1} \big( 1- \LimitingLabelRulePi{j}(\proxyX_i) \big)} \cdot f(\datvecraw_i) \Bigg]
\\ & =  \e \Bigg[  \frac{\LimitingLabelRulePi{k}(\proxyX_i) \prod_{j=1}^{k-1} \big( 1- \LimitingLabelRulePi{j}(\proxyX_i) \big)} {\LimitingLabelRulePi{k}(\proxyX_i) \prod_{j=1}^{k-1} \big( 1- \LimitingLabelRulePi{j}(\proxyX_i) \big)} \cdot f(\datvecraw_i) \Bigg]
\\ & = \e [  f(\datvecraw_i) ] = \e[f(\datvecraw)].
\end{aligned}$$ Above the third step follows from the independence of $U_i^{(1)},\dots,U_i^{(K)}$ conditionally on $\datvecraw_i$, the fourth step uses the definition of $\phi_{-1}$ and $\phi_0$ at \eqref{eq:ComplementOr1Operator} and the final step uses Assumption \ref{assump:IIDUnderlyingData}. Thus we have shown that $\e [\bar{W}_i^{(k)} f(\datvecraw_i)]= \e[f(\datvecraw)]$ for each $k \in [K]$. Combining this with Equation \eqref{eq:WkAggregatedIndepVersion} that $\bar{W}_i =\sum_{k=1}^K c_k \bar{W}_i^{(k)}$ and the fact that $\sum_{k=1}^K c_k=1$, $\e[\bar{W}_i f(\datvecraw_i)]=\e[f(\datvecraw)]$. Since this argument holds for any fixed, measurable $f: \mathbb{R}^q \to \mathbb{R}$ and $i \in [N]$, (II) holds.

To prove (III), we directly apply (II). In particular, fix $i \in [N]$ and measurable $f,g: \mathbb{R}^q \to \mathbb{R}$. Observe that the product of $f$ and $g$ denoted by $fg: \mathbb{R}^q \to \mathbb{R}$ is measurable, so applying property (II) to both $f$ and $fg$, $$\begin{aligned} \cov \big( \bar{W}_i f(\datvecraw_i), g(\datvecraw_i) \big)= \e[\bar{W}_i f(\datvecraw_i) g(\datvecraw_i)]-\e[\bar{W}_i f(\datvecraw_i)] \cdot \e[g(\datvecraw_i)] & =\e[f(\datvecraw) g(\datvecraw)]-\e[f(\datvecraw)] \cdot  \e[g(\datvecraw)]. \end{aligned}$$ Hence, simplifying the right hand side, $\cov \big( \bar{W}_i f(\datvecraw_i), g(\datvecraw_i) \big)=\cov \big( f(\datvecraw), g(\datvecraw) \big)$, and this argument holds for any fixed $i \in [N]$ and measurable $f,g: \mathbb{R}^q \to \mathbb{R}$. 

To prove (IV), fix $i \in [N]$ and measurable functions $f,g : \mathbb{R}^q \to \mathbb{R}$. Observe from \eqref{eq:WkIndepVersionDecomp} that for $k,k' \in [K]$ such that $k \neq k'$, $\bar{W}_i^{(k)} \bar{W}_i^{(k')}=0$. Combining this with formula \eqref{eq:WkAggregatedIndepVersion} that $\bar{W}_i=\sum_{k=1}^K c_k \bar{W}_i^{(k)}$, it follows that $$\bar{W}_i^2 = \sum_{k=1}^K c_k^2 \big( \bar{W}_i^{(k)} \big)^2=\sum_{k=1}^K c_k^2 \Bigg( \prod_{j=1}^{k-1} \Big( \frac{  1- \mathbbm{1} \{ U_i^{(j)} \leq \LimitingLabelRulePi{j}(\proxyX_i) \} }{  1- \LimitingLabelRulePi{j}(\proxyX_i) } \Big) \frac{  \mathbbm{1} \{ U_i^{(k)} \leq \LimitingLabelRulePi{k}(\proxyX_i) \} }{ \LimitingLabelRulePi{k}(\proxyX_i) } \Bigg)^2.$$ By recalling definition $\eqref{eq:PiBarProd1tok}$ and noting that for $t \in \{0,1\}$, $t^2=t$, this expression simplifies to $$\bar{W}_i^2= \sum_{k=1}^K  \frac{c_k^2}{\big( \LimitingLabelRulePi{1:k}(\proxyX_i) \big)^2}  \prod_{j=1}^{k-1} \big( 1- \mathbbm{1} \{ U_i^{(j)} \leq \LimitingLabelRulePi{j}(\proxyX_i) \}  \big) \cdot  \mathbbm{1} \{ U_i^{(k)} \leq \LimitingLabelRulePi{k}(\proxyX_i) \}.$$ Recalling that $\datvecraw_i =(\proxyX_i,\xmiss_i)$ and that conditionally on $\datvecraw_i$, $U_i^{(1)},\dots,U_i^{(K)} \stackrel{\text{i.i.d.}}{\sim} \text{Unif}[0,1]$, $$\begin{aligned} 
\e [\bar{W}_i^2 f(\datvecraw_i) g(\datvecraw_i) ] & = \e \big[  f(\datvecraw_i) g(\datvecraw_i) \cdot \e[\bar{W}_i^2 \giv \datvecraw_i ] \big]
\\ & = \e \Bigg[  \sum_{k=1}^K  \frac{f(\datvecraw_i) g(\datvecraw_i) c_k^2}{\big( \LimitingLabelRulePi{1:k}(\proxyX_i) \big)^2} \cdot \e \Big[\prod_{j=1}^{k-1} \big( 1- \mathbbm{1} \{ U_i^{(j)} \leq \LimitingLabelRulePi{j}(\proxyX_i) \}  \big) \cdot  \mathbbm{1} \{ U_i^{(k)} \leq \LimitingLabelRulePi{k}(\proxyX_i) \} \Big| \datvecraw_i  \Big] \Bigg]
\\ & = \e \Bigg[ \sum_{k=1}^K  \frac{f(\datvecraw_i) g(\datvecraw_i) c_k^2}{\big( \LimitingLabelRulePi{1:k}(\proxyX_i) \big)^2} \cdot \prod_{j=1}^{k-1} \e \big[ 1-  \mathbbm{1} \{ U_i^{(j)} \leq \LimitingLabelRulePi{j}(\proxyX_i) \} \big| \datvecraw_i \big]  \cdot  \e \big[ \mathbbm{1} \{ U_i^{(k)} \leq \LimitingLabelRulePi{k}(\proxyX_i) \} \big| \datvecraw_i  \big] \Bigg]
\\ & = \e \Bigg[ \sum_{k=1}^K  \frac{f(\datvecraw_i) g(\datvecraw_i) c_k^2}{\big( \LimitingLabelRulePi{1:k}(\proxyX_i) \big)^2} \cdot \prod_{j=1}^{k-1} \big(1-\LimitingLabelRulePi{j}(\proxyX_i) \big)  \cdot  \LimitingLabelRulePi{k}(\proxyX_i) \Bigg]
\\ & = \e \Bigg[ \sum_{k=1}^K  \frac{f(\datvecraw_i) g(\datvecraw_i) c_k^2}{ \LimitingLabelRulePi{1:k}(\proxyX_i) } \Bigg]
\\ & = \sum_{k=1}^K c_k^2 \e \Big[ \frac{f(\datvecraw) g(\datvecraw)}{ \LimitingLabelRulePi{1:k}(\proxyX) } \Big].
\end{aligned}$$ Above the penultimate step follows from \eqref{eq:PiBarProd1tok} and the last step by Assumption \ref{assump:IIDUnderlyingData}. Combining this result with Property (II) proved earlier, $$\begin{aligned} \cov \Big(\bar{W}_i f(\datvecraw_i) , \bar{W}_i g(\datvecraw_i) \Big) & = \e [\bar{W}_i^2 f(\datvecraw_i) g(\datvecraw_i) ]-\e [\bar{W}_i f(\datvecraw_i) ] \cdot \e[\bar{W}_i g(\datvecraw_i) ] 
\\ & =\e [\bar{W}_i^2 f(\datvecraw_i) g(\datvecraw_i) ]-\e[f(\datvecraw)]\e[g(\datvecraw)]
\\ & = \sum_{k=1}^K c_k^2 \e \Big[ \frac{f(\datvecraw) g(\datvecraw)}{ \LimitingLabelRulePi{1:k}(\proxyX) } \Big]-\e[f(\datvecraw)]\e[g(\datvecraw)]. \end{aligned}$$ Noting that this result holds for any fixed $i \in [N]$ and measurable $f,g : \mathbb{R}^q \to \mathbb{R}$ completes the proof of (IV).

\end{proof}

The next lemma can be shown using a similar proof strategy as that used to prove Proposition \ref{prop:NoCovarianceMultiwaveIPW}.

\begin{lemma}\label{lemma:CovWeightInepAndStandardIsZero}

   Under \samplingSchemeName{} and Assumptions \ref{assump:IIDUnderlyingData} and \ref{assump:CLTRgularityConditions}, for any measurable $f,g: \mathbb{R}^q \to \mathbb{R}$, $\cov \big ( W_i f(\datvecraw_i), \bar{W}_{i'} g(\datvecraw_{i'}) \big)=0$ for each $i,i' \in [N]$ such that $i \neq i'$.
\end{lemma}

\begin{proof}
Fix measurable functions $f,g : \mathbb{R}^q \to \mathbb{R}$. Next fix $i, i' \in [N]$ such that $i \neq i'$. Next recall that for $k \in [K]$, $\cd_{k} =\big( (I_i^{(j)}, I_i^{(j)} \xmiss_i)_{j=1}^k, \proxyX_i \big)_{i=1}^N$ is the observed data after wave $k$. Next define $\cd_{0}^{\text{aug}} \equiv (\proxyX_i,\datvecraw_i)_{i=1}^N$ and for $k \in [K]$ let $$\cd_{k}^{\text{aug}} \equiv \Big( \big(I_i^{(j)},U_i^{(j)},\LabelRulePi{j}(\proxyX_i), I_i^{(j)} \xmiss_i \big)_{j=1}^k, \proxyX_i,\datvecraw_i \Big)_{i=1}^N$$ be an augmented version of the data after wave $k$ (which contains the unobserved $\datvecraw_i$ values, the labelling probabilities $\LabelRulePi{j}(\proxyX_i)$, and the uniform variables $U_i^{(j)}$ for $j=1,\dots,k$, but crucially does not carry information about $U_i^{(j)}$ in later waves).  We will first show that \begin{equation}\label{eq:Step1CovNullWithIndepWeights_CondExpSimplification}
    \e \big[  W_i^{(k,j)}   \bar{W}_{i'}^{(k',j)}  f(\datvecraw_i) g(\datvecraw_{i'}) \giv \cd_{j-1}^{\text{aug}} \big]= f(\datvecraw_i) g(\datvecraw_{i'})  \quad \text{for each } k,k',j \in [K],
\end{equation} where $W_i^{(k,j)}$ and $\bar{W}_i^{(k,j)}$ are defined at \eqref{eq:IPWsingleSingle} and \eqref{eq:WkIndepVersionDecomp}. 

To show \eqref{eq:Step1CovNullWithIndepWeights_CondExpSimplification} further fix $k,k',j \in [K]$. For convenience, let $r = \text{sgn}(j-k)$ and $r' = \text{sgn}(j-k')$. Recall that by \eqref{eq:ComplementOr1Operator} $\phi_{-1},\phi_0,\phi_1: [0,1] \to [0,1]$ are functions given by $\phi_{-1}(t)=1-t$, $\phi_{0}(t)=t$, and $\phi_1(t)=1$. Since conditionally on $\cd_{j-1}^{\text{aug}}$, $U_{i}^{(j)},U_{i'}^{(j)} \stackrel{\text{i.i.d.}}{\sim} \text{Unif}[0,1]$ and since $\LabelRulePi{j}(\proxyX_i)$, $\proxyX_i$ and $\proxyX_{i'}$ are each components of $\cd_{j-1}^{\text{aug}}$, by considering all 3 possible values of $r$ and $r'$, $$\e \big[\phi_{r} \big( \mathbbm{1} \{ U_i^{(j)} \leq \LabelRulePi{j}(\proxyX_i) \} \big) \big| \cd_{j-1}^{\text{aug}} \big] = \phi_{r} \big( \LabelRulePi{j}(\proxyX_i) \big) \quad  \text{and}$$  $$\e \big[\phi_{r'} \big( \mathbbm{1} \{ U_{i'}^{(j)} \leq \LimitingLabelRulePi{j}(\proxyX_{i'}) \} \big) \big| \cd_{j-1}^{\text{aug}} \big] = \phi_{r'} \big( \LimitingLabelRulePi{j}(\proxyX_{i'}) \big).$$  Under \samplingSchemeName{} $U_{i}^{(j)} \indep U_{i'}^{(j)} \giv \cd_{j-1}^{\text{aug}}$, so combining this with the above results $$\e \big[\phi_{r} \big( \mathbbm{1} \{ U_i^{(j)} \leq \LabelRulePi{j}(\proxyX_i) \} \big) \phi_{r'} \big( \mathbbm{1} \{ U_{i'}^{(j)} \leq \LimitingLabelRulePi{j}(\proxyX_{i'}) \} \big) \big| \cd_{j-1}^{\text{aug}} \big]= \phi_{r} \big( \LabelRulePi{j}(\proxyX_i) \big) \cdot  \phi_{r'} \big( \LimitingLabelRulePi{j}(\proxyX_{i'}) \big).$$

Thus by definition \eqref{eq:IPWsingleSingle} and \eqref{eq:WkIndepVersionDecomp} and the above result,  $$\begin{aligned} \e [  W_i^{(k,j)}  \bar{W}_{i'}^{(k',j)}  f(\datvecraw_i) g(\datvecraw_{i'}) \giv \cd_{j-1}^{\text{aug}} ] & =  \e \Bigg[ \frac{ \phi_{r} \big( \mathbbm{1} \{ U_i^{(j)} \leq \LabelRulePi{j}(\proxyX_i) \} \big) \cdot \phi_{r'} \big( \mathbbm{1} \{ U_{i'}^{(j)} \leq \LimitingLabelRulePi{j}(\proxyX_{i'}) \} \big) f(\datvecraw_i) g(\datvecraw_{i'}) }{ \phi_r \big( \LabelRulePi{j}(\proxyX_i) \big) \cdot \phi_{r'} \big( \LimitingLabelRulePi{j}(\proxyX_{i'}) \big) }  \Bigg| \cd_{j-1}^{\text{aug}}\Bigg]
    \\ & = \frac{ f(\datvecraw_i) g(\datvecraw_{i'}) \cdot \e \big[  \phi_{r} \big( \mathbbm{1} \{ U_i^{(j)} \leq \LabelRulePi{j}(\proxyX_i) \} \big) \cdot \phi_{r'} \big( \mathbbm{1} \{ U_{i'}^{(j)} \leq \LimitingLabelRulePi{j}(\proxyX_{i'}) \} \big)   \big| \cd_{j-1}^{\text{aug}} \big] }{ \phi_r \big( \LabelRulePi{j}(\proxyX_i) \big) \cdot \phi_{r'} \big( \LimitingLabelRulePi{j}(\proxyX_{i'}) \big) } 
    \\ & = f(\datvecraw_i) g(\datvecraw_{i'}). \end{aligned}$$  
    Noting that the above result holds for each fixed $k,k',j \in [K]$ proves \eqref{eq:Step1CovNullWithIndepWeights_CondExpSimplification}.

    Again fix $k,k' \in [K]$ and we will show that $\cov \big(  W_i^{(k)}  f(\datvecraw_i), \bar{W}_{i'}^{(k')} g(\datvecraw_{i'})  \big)=0$. To do this we will first prove by induction that for each $j^* \in [K]$, \begin{equation}\label{eq:InductiveStep1stOrderExpectationsIndepWeightVersion}
   \e \Big[ \Big( \prod_{j=1}^{j^*} W_i^{(k,j)} \bar{W}_{i'}^{(k',j)} \Big) f(\datvecraw_i) g(\datvecraw_{i'})  \Big] = \e [ f(\datvecraw_i) g(\datvecraw_{i'})]. \end{equation} In the case where $j^*=1$ \eqref{eq:InductiveStep1stOrderExpectationsIndepWeightVersion} holds because by the tower property and \eqref{eq:Step1CovNullWithIndepWeights_CondExpSimplification}, $$\e \big[  W_i^{(k,1)} \bar{W}_{i'}^{(k',1)} f(\datvecraw_i) g(\datvecraw_{i'})  \big] = \e \big[  \e[ W_i^{(k,1)} \bar{W}_{i'}^{(k',1)} f(\datvecraw_i) g(\datvecraw_{i'}) \giv \cd_0^{\text{aug}}]  \big] =\e [ f(\datvecraw_i) g(\datvecraw_{i'})].$$ Next fix some $j^* \in [K-1]$ and assume \eqref{eq:InductiveStep1stOrderExpectationsIndepWeightVersion} holds for $j^*$. Since for $j \leq j^*$, $ W_i^{(k,j)}$ and $\bar{W}_{i'}^{(k',j)}$ can be written as fixed, measurable functions of $\cd_{j^*}^{\text{aug}}$, by the tower property and \eqref{eq:Step1CovNullWithIndepWeights_CondExpSimplification}, $$\begin{aligned} \e \Big[ \prod_{j=1}^{j^* +1} W_i^{(k,j)} \bar{W}_{i'}^{(k',j)} f(\datvecraw_i) g(\datvecraw_{i'})  \Big] & = \e \Big[ \Big( \prod_{j=1}^{j^* } W_i^{(k,j)} \bar{W}_{i'}^{(k',j)}  \Big) \e[ W_i^{(k,j^*+1)} \bar{W}_{i'}^{(k',j^*+1)} f(\datvecraw_i) g(\datvecraw_{i'}) \giv \cd_{j^*}^{\text{aug}}]  \Big] 
   \\ & = \e \Big[ \Big( \prod_{j=1}^{j^* } W_i^{(k,j)} \bar{W}_{i'}^{(k',j)} \Big)  f(\datvecraw_i) g(\datvecraw_{i'})   \Big]
   \\ & = \e [ f(\datvecraw_i) g(\datvecraw_{i'})], \end{aligned}$$ where the last step holds provided that \eqref{eq:InductiveStep1stOrderExpectationsIndepWeightVersion} holds for $j^*$. Hence we have shown that for any $j^* \in [K-1]$, if \eqref{eq:InductiveStep1stOrderExpectationsIndepWeightVersion} holds for $j^*$, then \eqref{eq:InductiveStep1stOrderExpectationsIndepWeightVersion} also holds for $j^*+1$, and additionally, \eqref{eq:InductiveStep1stOrderExpectationsIndepWeightVersion} holds in the case where $j^*=1$. Thus by induction \eqref{eq:InductiveStep1stOrderExpectationsIndepWeightVersion} holds for all $j^* \in [K]$. 
   
   By recalling the formulas for $W_i^{(k)}$ and $\bar{W}_{i'}^{(k')}$ at \eqref{eq:WkFormulaAlt} and \eqref{eq:WkIndepVersionDecomp}, and applying Equation \eqref{eq:InductiveStep1stOrderExpectationsIndepWeightVersion} in the case where $j^*=K$, $$\e \big[ W_i^{(k)} \bar{W}_{i'}^{(k')} f(\datvecraw_i) g(\datvecraw_{i'})  \big]= \e \Big[ \prod_{j=1}^K W_i^{(k,j)} \bar{W}_{i'}^{(k',j)} f(\datvecraw_i) g(\datvecraw_{i'})  \Big]=\e[f(\datvecraw_i) g(\datvecraw_{i'})]=\e[f(\datvecraw_i)]\e[g(\datvecraw_{i'})],$$ where the last step above holds by Assumption \ref{assump:IIDUnderlyingData} since $i \neq i'$. By the definition of covariance, Property (II) in Lemma \ref{lemma:PropertiesOfIndependentVersionsOfWeights}, Corollary \ref{cor:SimplifyWeightedExpectations1Term}, and the above result, $$\begin{aligned}\cov \big(  W_i^{(k)}  f(\datvecraw_i), \bar{W}_{i'}^{(k')} g(\datvecraw_{i'})  \big) & = \e \big[ W_i^{(k)} \bar{W}_{i'}^{(k')} f(\datvecraw_i) g(\datvecraw_{i'})  \big]- \e[W_i^{(k)}  f(\datvecraw_i)] \e[\bar{W}_{i'}^{(k')} g(\datvecraw_{i'}) ] 
   \\ & = \e[f(\datvecraw_i)]\e[g(\datvecraw_{i'})]-\e[f(\datvecraw_i)]\e[g(\datvecraw_{i'})]=0. \end{aligned}$$ Thus we have shown that $\cov \big(  W_i^{(k)}  f(\datvecraw_i), \bar{W}_{i'}^{(k')} g(\datvecraw_{i'})  \big)=0$ and this argument holds for any fixed $k,k' \in [K]$.

   To complete the proof recalling \eqref{eq:aggregated_Wk} and \eqref{eq:WkAggregatedIndepVersion}, by the previous result, $$\cov \big( W_i f(\datvecraw_i), \bar{W}_{i'} g(\datvecraw_{i'}) \big) = \sum_{k=1}^K \sum_{k'=1}^K c_k c_{k'} \cov \big(  W_i^{(k)}  f(\datvecraw_i), \bar{W}_{i'}^{(k')} g(\datvecraw_{i'})  \big)=0.$$ We have thus shown the desired result for any fixed measurable functions $f,g : \mathbb{R}^q \to \mathbb{R}$ and $i, i' \in [N]$ such that $i \neq i'$.
\end{proof}

\begin{lemma}\label{lemma:GapWithIndepIPWWeightsDisappearsInExpecation} Under \samplingSchemeName{} and Assumptions \ref{assump:IIDUnderlyingData}, \ref{assump:LabellingRuleOverlap}, and \ref{assump:CLTRgularityConditions}, $\lim_{N \to \infty} \e \big[ \vert W_i^{(k,j)}-\bar{W}_i^{(k,j)} \vert \big]=0$ for each $i \in [N]$ and $k,j \in [K]$ such that $j \leq k$.
\end{lemma}

\begin{proof}
    Fix any $i \in [N]$ and $k, j \in [K]$ such that $j \leq k$. Let $r=\text{sgn}(j-k)$ and recall from \eqref{eq:ComplementOr1Operator} that $\phi_r: [0,1] \to [0,1]$ is given by $\phi_r(t)=t$ if $r=0$ and $\phi_r(t)=1-t$ if $r=-1$. Next let $$T \equiv \frac{\phi_{r} \big( \mathbbm{1} \{ U_i^{(j)} \leq \LabelRulePi{j}(\proxyX_i) \} \big)}{\phi_{r} \big( \LimitingLabelRulePi{j}(\proxyX_i) \big)},$$ and we will show $\lim_{N \to \infty} \e[\vert W_i^{(k,j)} - T \vert]=0$ and $\lim_{N \to \infty} \e[\vert T- \bar{W}_i^{(k,j)} \vert]=0$. Note that by the definition of $W_i^{(k,j)}$ at \eqref{eq:IPWsingleSingle},
    
    $$\begin{aligned} \e[\vert W_i^{(k,j)} - T \vert ] & = \e \Bigg[ \Bigg| \phi_{r} \big( \mathbbm{1} \{ U_i^{(j)} \leq \LabelRulePi{j}(\proxyX_i) \} \big) \Big( \frac{1}{\phi_{r} \big( \LabelRulePi{j}(\proxyX_i) \big)}  - \frac{1}{\phi_{r} \big( \LimitingLabelRulePi{j}(\proxyX_i) \big)} \Big) \Bigg| \Bigg]  
    \\ & \leq \e\Bigg[ \Bigg|  \frac{1}{\phi_{r} \big( \LabelRulePi{j}(\proxyX_i) \big)}  - \frac{1}{\phi_{r} \big( \LimitingLabelRulePi{j}(\proxyX_i) \big)}  \Bigg| \Bigg]. \end{aligned}$$ Taking the $\limsup_{N \to \infty}$ of each side of the above inequality and applying Lemma \ref{lemma:InverseLpConvergence} it follows that $\lim_{N \to \infty} \e[\vert W_i^{(k,j)} - T \vert]=0$. Next letting $$L_{\text{lower}}^{(j)}(\proxyX_i,\cd_{j-1})  \equiv \frac{\LabelRulePi{j}(\proxyX_i)+\LimitingLabelRulePi{j}(\proxyX_i)}{2}- \frac{ \vert \LabelRulePi{j}(\proxyX_i)-\LimitingLabelRulePi{j}(\proxyX_i) \vert}{2} \quad \text{and}$$  $$L_{\text{upper}}^{(j)}(\proxyX_i,\cd_{j-1}) \equiv \frac{\LabelRulePi{j}(\proxyX_i)+\LimitingLabelRulePi{j}(\proxyX_i)}{2}+ \frac{ \vert \LabelRulePi{j}(\proxyX_i)-\LimitingLabelRulePi{j}(\proxyX_i) \vert}{2},$$
    by the definition for $\bar{W}_i^{(k,j)}$ at \eqref{eq:WkIndepVersionDecomp}, regardless of whether $r=0$ or $r=-1$
    $$\begin{aligned}
   \e[ \vert T -\bar{W}_i^{(k,j)} \vert ] & = \e \Bigg[ \Bigg| \frac{1}{\phi_{r} \big( \LimitingLabelRulePi{j}(\proxyX_i) \big)} \Big( \phi_{r} \big( \mathbbm{1} \{ U_i^{(j)} \leq \LabelRulePi{j}(\proxyX_i) \} \big) - \phi_{r} \big( \mathbbm{1} \{ U_i^{(j)} \leq \LimitingLabelRulePi{j}(\proxyX_i) \} \big) \Big) \Bigg| \Bigg]
    \\ & \leq b^{-1} \e \Big[ \Big| \phi_{r} \big( \mathbbm{1} \{ U_i^{(j)} \leq \LabelRulePi{j}(\proxyX_i) \} \big) - \phi_{r} \big( \mathbbm{1} \{ U_i^{(j)} \leq \LimitingLabelRulePi{j}(\proxyX_i) \} \big) \Big| \Big]
    \\ & = b^{-1} \e \big[ \big|  \mathbbm{1} \{ U_i^{(j)} \leq \LabelRulePi{j}(\proxyX_i) \}  -  \mathbbm{1} \{ U_i^{(j)} \leq \LimitingLabelRulePi{j}(\proxyX_i) \}  \big| \big]
    \\ & \leq b^{-1} \e \Big[   \mathbbm{1} \bigr\{ U_i^{(j)} \in  [L_{\text{lower}}^{(j)}(\proxyX_i,\cd_{j-1}),L_{\text{upper}}^{(j)}(\proxyX_i,\cd_{j-1})]  \bigl\}   \Big]
    \\ & = b^{-1} \e \Big[ \e \big[  \mathbbm{1} \bigr\{ U_i^{(j)} \in  [L_{\text{lower}}^{(j)}(\proxyX_i,\cd_{j-1}),L_{\text{upper}}^{(j)}(\proxyX_i,\cd_{j-1})]  \bigl\}  \big| \cd_{j-1} \big]   \Big]
    \\ & = b^{-1} \e \Big[  \big| L_{\text{upper}}^{(j)}(\proxyX_i,\cd_{j-1})-L_{\text{lower}}^{(j)}(\proxyX_i,\cd_{j-1})    \big|    \Big]
    \\ & = b^{-1} \e \big[ \vert \LabelRulePi{j}(\proxyX_i)-\LimitingLabelRulePi{j}(\proxyX_i) \vert \big].
    \end{aligned}$$ Above the second step follows from Assumption \ref{assump:LabellingRuleOverlap} and the penultimate step follows because in \samplingSchemeName{,} $U_i^{(j)} \sim \text{Unif}[0,1]$ and is generated independently of $\cd_{j-1}$ while $L_{\text{lower}}^{(j)}(\proxyX_i,\cd_{j-1})$ and $L_{\text{upper}}^{(j)}(\proxyX_i,\cd_{j-1})$ can be expressed as measurable functions of $\cd_{j-1}$. 
    
    Taking the limsup as $N \to \infty$ of each side of the above inequality and by applying the exchangeability result in Lemma \ref{lemma:StrongerExchangeability} and Assumption \ref{assump:CLTRgularityConditions}, $$ \limsup_{N \to \infty} \e[ \vert T -\bar{W}_i^{(k,j)} \vert ] \leq \limsup_{N \to \infty} b^{-1} \e \big[ \vert \LabelRulePi{j}(\proxyX_i)-\LimitingLabelRulePi{j}(\proxyX_i) \vert \big]=b^{-1} \cdot \limsup_{N \to \infty}\e \big[ \vert \LabelRulePi{j}(\proxyX_1)-\LimitingLabelRulePi{j}(\proxyX_1) \vert \big]=0.$$ 
    
    Combining this with an earlier result, $$0 \leq \limsup_{N \to \infty} \e[ \vert W_i^{(k,j)} - \bar{W}_i^{(k,j)} \vert ] \leq \limsup_{N \to \infty} \Big( \e[ \vert W_i^{(k,j)} -T \vert] +  \e[ \vert T-\bar{W}_i^{(k,j)} \vert] \Big) \leq 0.$$ Thus we have shown that $\lim_{N \to \infty} \e[ \vert W_i^{(k,j)} - \bar{W}_i^{(k,j)} \vert ]=0$ and this argument held for any fixed $i \in [N]$ and $k, j \in [K]$ such that $j \leq k$.
    
\end{proof}

As a corollary to the previous lemma, we can also use boundedness of the $W^{(k,j)}$ and $\bar{W}^{(k,j)}$ terms to establish that $\lim_{N \to \infty} \e[ \vert W_i -\bar{W}_i \vert ] =0$ for each $i \in [N]$. 

\begin{corollary}\label{cor:GapWithIndepWeightsDisappearsInExpecation}
    Under \samplingSchemeName{} and Assumptions \ref{assump:IIDUnderlyingData}, \ref{assump:LabellingRuleOverlap}, and \ref{assump:CLTRgularityConditions}, $\lim_{N \to \infty} \e \big[ \vert W_i-\bar{W}_i \vert \big]=0$ for each $i \in [N]$.
\end{corollary}

\begin{proof}
    Fix $i \in [N]$. Define $$T_i^{(k,j^*)} \equiv \Big( \prod_{j=1}^{j^*} W_i^{(k,j)} \Big) \Big( \prod_{j=j^*+1}^{k} \bar{W}_i^{(k,j)} \Big) \quad \text{ for each } k \in [K] \text { and } j^* \in \{0 \} \cup [k],$$ where above we use the convention that a product equals $1$ if the lower limit is greater than the upper limit, and where $W_i^{(k,j)}$ and  $\bar{W}_i^{(k,j)}$ are defined at \eqref{eq:IPWsingleSingle} and \eqref{eq:WkIndepVersionDecomp}. Next note for each $k \in [K]$, $W_i^{(k)}= T_i^{(k,k)}$ by \eqref{eq:WkFormulaAlt} while $\bar{W}_i^{(k)}= T_i^{(k,0)}$ by \eqref{eq:WkIndepVersionDecomp}, and hence $$W_i^{(k)}-\bar{W}_i^{(k)} = T_i^{(k,k)} - T_i^{(k,0)} =\sum_{j'=1}^k \big( T_i^{(k,j')} - T_i^{(k,j'-1)} \big),$$
    for each $k \in [K]$. 

    Note that for each $k,j \in [K]$ such that $j \leq k$, $\vert W_i^{(k,j)} \vert \leq b^{-1}$ almost surely by Assumption \ref{assump:LabellingRuleOverlap} and $\vert \bar{W}_i^{(k,j)} \vert \leq b^{-1}$ almost surely by the fact that $\LimitingLabelRulePi{k}(\argxProxy) \in [b,1-b]$ for each $\argxProxy \in \xProxySpace$. Combining this with the previous expressions, for each $k \in [K]$, $$\begin{aligned} \e [ \vert W_i^{(k)} -\bar{W}_i^{(k)} \vert ] & \leq \sum_{j'=1}^k \e \big[ \vert T_i^{(k,j')} - T_i^{(k,j'-1)} \vert \big] 
    \\ & = \sum_{j'=1}^k \e \Bigg[ \Bigg| \Big(W_i^{(k,j')}-\bar{W}_i^{(k,j')} \Big) \Big( \prod_{j=1}^{j'-1} W_i^{(k,j)} \Big) \Big( \prod_{j=j'+1}^{k} \bar{W}_i^{(k,j)} \Big) \Bigg| \Bigg]  
    \\ & \leq \sum_{j'=1}^k b^{-k} \e[ \vert W_i^{(k,j')}-\bar{W}_i^{(k,j')} \vert]. \end{aligned}$$  Next recall by \eqref{eq:aggregated_Wk} and \eqref{eq:WkAggregatedIndepVersion} that $W_i=\sum_{k=1}^K c_k W_i^{(k)}$ and $\bar{W}_i=\sum_{k=1}^K c_k \bar{W}_i^{(k)}$ where $c_k \in [0,1]$ for each $k \in [K]$, and hence the previous result implies that $$\begin{aligned} \e[ \vert W_i-\bar{W}_i \vert]  = \e \big[ \big| \sum_{k=1}^K c_k (W_i^{(k)}-\bar{W}_i^{(k)}) \big| \big] 
    \leq \sum_{k=1}^K \e[ \vert W_i^{(k)}-\bar{W}_i^{(k)} \vert] \leq \sum_{k=1}^K \sum_{j'=1}^k b^{-k} \e[ \vert W_i^{(k,j')}-\bar{W}_i^{(k,j')} \vert].
    \end{aligned}$$  Taking the limsup as $N \to \infty$ of each side of the above inequality and applying Lemma \ref{lemma:GapWithIndepIPWWeightsDisappearsInExpecation}, $\lim_{N \to \infty} \e [ \vert W_i -\bar{W}_i \vert ]=0$.
\end{proof}

\begin{proposition}\label{prop:CanReduceToIIDWeights}
    Under \samplingSchemeName{} and Assumptions \ref{assump:IIDUnderlyingData}, \ref{assump:LabellingRuleOverlap}, and \ref{assump:CLTRgularityConditions}, for any measurable $f: \mathbb{R}^q \to \mathbb{R}$ such that $\e[\vert f(\datvecraw) \vert^{2+\eta} ] < \infty$ for some $\eta >0$, $$\frac{1}{\sqrt{N}} \sum_{i=1}^N (W_i-\bar{W}_i) f(\datvecraw_i) \xrightarrow{p} 0 \quad \text{as } N \to \infty.$$ 
\end{proposition}

 \begin{proof}
    First fix a measurable $f: \mathbb{R}^q \to \mathbb{R}$ such that $\e[\vert f(\datvecraw) \vert^{2+\eta} ] < \infty$ for some $\eta>0$ and let $\eta \in (0,\infty)$ be one of the values  such that $\e[\vert f(\datvecraw) \vert^{2+\eta} ] < \infty$. Next define $$C_f \equiv \Big( \e[\vert f(\datvecraw) \vert^{2+\eta} ] \Big)^{\frac{2}{2+\eta}} \in (0,\infty).$$
    
    Now fix $i,i' \in [N]$ such that $i \neq i'$. Note that $\cov\big( W_i f(\datvecraw_i), W_{i'} f(\datvecraw_{i'}) \big)=0$ by Proposition \ref{prop:NoCovarianceMultiwaveIPW} and $\cov\big( \bar{W}_i f(\datvecraw_i), \bar{W}_{i'} f(\datvecraw_{i'}) \big)=0$ since $(\bar{W}_i,\datvecraw_i)_{i=1}^N$ are i.i.d. by Property (I) in Lemma \ref{lemma:PropertiesOfIndependentVersionsOfWeights}. Meanwhile by Lemma \ref{lemma:CovWeightInepAndStandardIsZero}, $\cov\big( W_i f(\datvecraw_i), \bar{W}_{i'} f(\datvecraw_{i'}) \big)=0$ and $ \cov\big( \bar{W}_i f(\datvecraw_i), W_{i'} f(\datvecraw_{i'}) \big)=0$. 
    Combining these results, which hold for any fixed $i,i' \in [N]$ such that $i \neq i'$,
    $$\cov\Big( (W_i-\bar{W}_i) f(\datvecraw_i), (W_{i'}-\bar{W}_{i'}) f(\datvecraw_{i'}) \Big) =0 \quad \text{for all } i,i' \in [N] \text{ such that } i \neq i'.$$ Hence $$\begin{aligned}\var \Big( \frac{1}{\sqrt{N}} \sum_{i=1}^N (W_i-\bar{W}_i) f(\datvecraw_i) \Big) & = \frac{1}{N} \sum_{i=1}^N \sum_{i'=1}^N \cov\Big( (W_i-\bar{W}_i) f(\datvecraw_i), (W_{i'}-\bar{W}_{i'}) f(\datvecraw_{i'}) \Big) 
    \\ & = \frac{1}{N} \sum_{i=1}^N \var \big( (W_i-\bar{W}_i) f(\datvecraw_i) \big) 
    \\ & = \var \big( (W_1-\bar{W}_1) f(\datvecraw_1) \big) 
    \\ & \leq \e[ (W_1-\bar{W}_1)^2 f^2(\datvecraw_1)]
    \\ & \leq \Big( \e\Big[ \big| W_1-\bar{W}_1\big|^{ \frac{2 (2+\eta)}{\eta}} \Big]  \Big)^{\frac{\eta}{2+\eta}} \Big( \e\Big[ \big| f(\datvecraw_1) \big|^{\frac{2(2+\eta)}{2}} \Big] \Big)^{\frac{2}{2+\eta}}
    \\ & = C_f \cdot \Big( \e\Big[ \big| W_1-\bar{W}_1\big|^{2+4/\eta} \Big]  \Big)^{\frac{\eta}{2+\eta}}. 
    \end{aligned}$$ Above we use Property (I) in Lemma \ref{lemma:PropertiesOfIndependentVersionsOfWeights} in the third step, Hölder's inequality for the $(2+\eta)/\eta$ and $(2+\eta)/2$  pair in the penultimate step, and we use Assumption \ref{assump:IIDUnderlyingData} and the definition of $C_f$ in the final step. Next note that by the Cauchy-Schwarz inequality, $$\begin{aligned} \e\big[ \big| W_1-\bar{W}_1\big|^{2+4/\eta} \big] = \e\big[ \big|  W_1-\bar{W}_1\big|^{3/2+4/\eta} \cdot \big| W_1-\bar{W}_1\big|^{1/2}   \big] & \leq \sqrt{  \e\big[ \vert W_1-\bar{W}_1\vert^{3+8/\eta} \big] \cdot \e[ \vert W_1-\bar{W}_1\vert]  } 
    \\ & \leq \sqrt{b^{-K(3+8/\eta)} \cdot \e[ \vert W_1-\bar{W}_1\vert]  },
    \end{aligned}$$ where the second step holds because $W_1 \in [0,b^{-K}]$ almost surely by Fact \ref{fact:WeightBound} and $\bar{W}_1 \in [0,b^{-K}]$ almost surely because $\LimitingLabelRulePi{k}(\argxProxy) \in [b,1-b]$ for each $\argxProxy \in \xProxySpace$. Combining two previous inequalities, $$\var \Big( \frac{1}{\sqrt{N}} \sum_{i=1}^N (W_i-\bar{W}_i) f(\datvecraw_i) \Big) \leq C_f \Big( b^{-K(3+8/\eta)} \Big)^{\frac{\eta}{2(2+\eta)}} \cdot \Big( \e[ \vert W_1 -\bar{W}_1 \vert ] \Big)^{\frac{\eta}{2(2+\eta)}}.$$ Since $\lim_{N \to \infty} \e[ \vert W_1 -\bar{W}_1 \vert ]=0$ by Corollary \ref{cor:GapWithIndepWeightsDisappearsInExpecation}, taking the limsup as $N \to \infty$ of each side of the above inequality implies that $$\lim_{N \to \infty} \var \Big( \frac{1}{\sqrt{N}} \sum_{i=1}^N (W_i-\bar{W}_i) f(\datvecraw_i) \Big) =0.$$ Next note by Proposition \ref{prop:WeightsEnableUnbiasedMeanEst} and Property (II) in Lemma \ref{lemma:PropertiesOfIndependentVersionsOfWeights}, $$\e\Big[ \frac{1}{\sqrt{N}} \sum_{i=1}^N (W_i-\bar{W}_i) f(\datvecraw_i) \Big] =\sqrt{N} \e \big[  \frac{1}{N} \sum_{i=1}^N W_i f(\datvecraw_i) \big] - \frac{1}{\sqrt{N}} \sum_{i=1}^N \e[\bar{W}_i f(\datvecraw_i)]=\sqrt{N} \e[f(\datvecraw)]-\sqrt{N}\e[f(\datvecraw)]=0.$$ Thus for any $\epsilon >0$, by Chebyshev's inequality and the previous results, $$\limsup_{N \to \infty} \mathbb{P} \Bigg( \Bigg| \frac{1}{\sqrt{N}} \sum_{i=1}^N (W_i-\bar{W}_i) f(\datvecraw_i)  \Bigg| > \epsilon \Bigg) \leq \limsup_{N \to \infty} \Bigg( \epsilon^{-2} \cdot \var \Big( \frac{1}{\sqrt{N}} \sum_{i=1}^N (W_i-\bar{W}_i) f(\datvecraw_i) \Big) \Bigg) =0.$$ Hence we have shown that $$\frac{1}{\sqrt{N}} \sum_{i=1}^N (W_i-\bar{W}_i) f(\datvecraw_i) \xrightarrow{p} 0 \quad \text{as } N \to \infty,$$ and this argument holds for any measurable $f: \mathbb{R}^q \to \mathbb{R}$ such that $\e[\vert f(\datvecraw) \vert^{2+\eta} ] < \infty$ for some $\eta >0$.

\end{proof}

\subsection{Proof of Theorem \ref{theorem:CLT_PTDEstimator}}
The proof involves three main parts. In the first part, we use Proposition \ref{prop:CanReduceToIIDWeights} above to express the asymptotic linear expansion of $\htPTD$ from Corollary \ref{cor:PTDEstAsymptoticallyLinear} as an asymptotic linear expansion with the i.i.d. weights $(\bar{W}_i)_{i=1}^N$. This alternate asymptotic linear expansion is stated in Equation \eqref{eq:asympLinearExpansionPTDMidProofIndepWights}. The second part of the proof involves deploying the standard multivariate CLT. The third part involves simplifying the asymptotic covariance matrix expression derived in the CLT.

\paragraph{Expressing the asymptotic linear expansion with i.i.d. weights.} By Corollary \ref{cor:PTDEstAsymptoticallyLinear}, $$ \sqrt{N} \big( \htPTD -\thetaTarg \big) = -\frac{1}{\sqrt{N}} \sum_{i=1}^N \Big( W_i  H_{\thetaTarg}^{-1} \dot{l}_{\thetaTarg}(\goodX_i) +(1-W_i) \Omega H_{\gammaTarg}^{-1} \dot{l}_{\gammaTarg}(\proxyX_i) \Big) + o_p(1).$$ Next define $g_1,g_2: \mathbb{R}^p \to \mathbb{R}^d$ to be given by $$g_1(\argx)=H_{\thetaTarg}^{-1} \dot{l}_{\thetaTarg}(\argx) \quad \text{and} \quad  g_2(\argxProxy) = \Omega H_{\gammaTarg}^{-1} \dot{l}_{\gammaTarg}(\argxProxy) \quad \text{ for } \argx \in \xSpace,\argxProxy \in \xProxySpace.$$ Also let $\bar{W}_i$ be defined as in \eqref{eq:IndepAggregatedWeightsMainText} for each $i \in [N]$, and note that we can rewrite the asymptotic linear expansion for $\htPTD$ as $$\begin{aligned}
    \sqrt{N} \big( \htPTD -\thetaTarg \big) & = -\frac{1}{\sqrt{N}} \sum_{i=1}^N \Big( W_i  g_1 (\goodX_i) +(1-W_i) g_2(\proxyX_i) \Big) + o_p(1)
    \\ & = -\frac{1}{\sqrt{N}} \sum_{i=1}^N \Big( \bar{W}_i  g_1 (\goodX_i) +(1-\bar{W}_i) g_2(\proxyX_i) \Big) + o_p(1)
    \\ & \quad - \frac{1}{\sqrt{N}} \sum_{i=1}^N (W_i-\bar{W}_i) g_1(\goodX_i) +\frac{1}{\sqrt{N}} \sum_{i=1}^N (W_i-\bar{W}_i) g_2(\proxyX_i).
\end{aligned}$$ We will next show that the last two terms in the above expression are $o_p(1)$. To do this fix $j \in [d]$. Note that since for each $i \in [N]$, $\goodX_i$ and $\proxyX_i$ are each subsets of the components of $\datvecraw_i$, we can let $f_j^{(1)},f_j^{(2)} : \mathbb{R}^q \to \mathbb{R}$ be the measurable functions such that $f_j^{(1)}(\datvecraw)=[g_1(\goodX)]_j$ and $f_j^{(2)}(\datvecraw)=[g_2(\proxyX)]_j$ for any realization of $\datvecraw=(\xobs,\txmiss,\xmiss)$. Now letting $\eta_*>0$ be the constant from Assumption \ref{assump:CLTRgularityConditions}(iii), and since $t \mapsto \vert t \vert^{2+\eta_*}$ is convex, by Jensen's inequality and Assumption \ref{assump:CLTRgularityConditions}(iii), $$\e \big[\vert \vert \dot{l}_{\thetaTarg}(\goodX) \vert \vert_{1}^{2+\eta_*}  \big] = d^{2+\eta_*} \cdot \e \Big[ \Big(  \sum_{j'=1}^d d^{-1} \cdot \big| [\dot{l}_{\thetaTarg}(\goodX)]_{j'} \big| \Big)^{2+\eta_*} \Big] \leq d^{2+\eta_*} \cdot \e \Big[ \sum_{j'=1}^d d^{-1} \big| [\dot{l}_{\thetaTarg}(\goodX)]_{j'} \big|^{2+\eta_*} \Big] < \infty.$$ Since $H_{\thetaTarg}$ is nonsingular by Assumption \ref{assump:SmoothEnoughForAsymptoticLineariaty}(v), by the definition of $f_j^{(1)}$ and the previous inequality,

$$\e[ \vert f_j^{(1)} (\datvecraw) \vert^{2+\eta_*}]= \e[ \vert e_j^{\tran} H_{\thetaTarg}^{-1} \dot{l}_{\thetaTarg}(\goodX) \vert^{2+\eta_*}] \leq \vert \vert e_j^\tran H_{\thetaTarg}^{-1} \vert \vert_{\infty}^{2+\eta_*} \cdot \e \big[\vert \vert \dot{l}_{\thetaTarg}(\goodX) \vert \vert_{1}^{2+\eta_*}  \big]< \infty.$$ Thus we can apply Proposition \ref{prop:CanReduceToIIDWeights} to get that $$\frac{1}{\sqrt{N}} \sum_{i=1}^N (W_i-\bar{W}_i) [g_1(\goodX_i)]_j = \frac{1}{\sqrt{N}} \sum_{i=1}^N (W_i-\bar{W}_i) f_j^{(1)} (\datvecraw_i)=o_p(1).$$ Since the above result holds for each $j \in [d]$, $$\frac{1}{\sqrt{N}} \sum_{i=1}^N (W_i-\bar{W}_i) g_1(\goodX_i)=o_p(1) \quad \text{and} \quad \frac{1}{\sqrt{N}} \sum_{i=1}^N (W_i-\bar{W}_i) g_2(\proxyX_i)=o_p(1),$$ where the latter claim follows from an analogous argument. Combining this result with an earlier expression for $\sqrt{N} \big( \htPTD -\thetaTarg \big)$ we have thus shown that \begin{equation}\label{eq:asympLinearExpansionPTDMidProofIndepWights} \sqrt{N} \big( \htPTD -\thetaTarg \big)=
     -\frac{1}{\sqrt{N}} \sum_{i=1}^N \Big( \bar{W}_i  g_1 (\goodX_i) +(1-\bar{W}_i) g_2(\proxyX_i) \Big) + o_p(1).
\end{equation} 

\paragraph{Applying the multivariate CLT.} Note that by Property (II) in Lemma \ref{lemma:PropertiesOfIndependentVersionsOfWeights} and by Fact \ref{fact:PropertiesOfGradients}, for each $i \in [N]$, $$\e[\bar{W}_i g_1(\goodX_i)+(1-\bar{W}_i) g_2(\proxyX_i)]=\e[g_1(\goodX)]+\e[g_2(\proxyX)]-\e[g_2(\proxyX)]= H_{\thetaTarg}^{-1} \e[  \dot{l}_{\thetaTarg} (\goodX)]=0.$$ In addition, $$\begin{aligned} \e \Big[ \big| \big| \bar{W}_i g_1(\goodX_i)+(1-\bar{W}_i) g_2(\proxyX_i) \big| \big|_2^2 \Big] < \infty \end{aligned}$$ because as $ \vert \bar{W}_i \vert \leq b^{-K} < \infty$ almost surely, and because both $$\e[\vert \vert g_1(\goodX_i) \vert \vert_2^2 ] =\e[\vert \vert H_{\thetaTarg}^{-1} \dot{l}_{\thetaTarg}(\goodX) \vert \vert_2^2 ] < \infty \quad \text{and} \quad \e[\vert \vert g_2(\proxyX_i) \vert \vert_2^2 ] = \e[\vert \vert \Omega H_{\gammaTarg}^{-1} \dot{l}_{\gammaTarg}(\proxyX) \vert \vert_2^2 ] < \infty$$ as consequences of Assumptions \ref{assump:IIDUnderlyingData}, \ref{assump:SmoothEnoughForAsymptoticLineariaty}(v), and  \ref{assump:CLTRgularityConditions}(iii). Moreover, by Property (I) in Lemma \ref{lemma:PropertiesOfIndependentVersionsOfWeights} $(\bar{W}_i, \goodX_i,\proxyX_i)_{i=1}^N$ is a sample of $N$ i.i.d. random vectors. Thus in \eqref{eq:asympLinearExpansionPTDMidProofIndepWights} the summation is over $N$ i.i.d. random vectors with mean $0$ and bounded second moment, so by the multivariate central limit theorem and Slutsky's lemma, $$ \sqrt{N} \big( \htPTD -\thetaTarg \big) =
     -\frac{1}{\sqrt{N}} \sum_{i=1}^N \Big( \bar{W}_i  g_1 (\goodX_i) +(1-\bar{W}_i) g_2(\proxyX_i) \Big) + o_p(1) \xrightarrow{d} \mathcal{N} \big( 0, \Sigma_{\Omega} \big),$$ where
   $$\Sigma_{\Omega} \equiv \var \Big( \bar{W}_1  g_1 (\goodX_1) +(1-\bar{W}_1) g_2(\proxyX_1) \Big).
     $$

\paragraph{Simplifying the asymptotic variance expression.} To complete the proof it remains to check that $\Sigma_{\Omega}$ defined above equals $\Sigma^{\PTDSuperScriptAcr}(\Omega)$ defined in \eqref{eq:AsympVarPTDFormula}. Observe that by Property (III) from Lemma \ref{lemma:PropertiesOfIndependentVersionsOfWeights}, $$\cov \big(\bar{W}_1 \big[g_1(\goodX_1)-g_2(\proxyX_1)]_j,[g_2(\proxyX_1)]_{j'}\big)=\cov \big( \big[g_1(\goodX)-g_2(\proxyX)]_j,[g_2(\proxyX)]_{j'}\big) \quad \text {for each } j,j' \in [d].$$ By this result and Assumption \ref{assump:IIDUnderlyingData}, the formula for $\Sigma_{\Omega}$ can thus be simplified as follows

$$\begin{aligned} \Sigma_{\Omega} & = \var \Big( \bar{W}_1  \big( g_1(\goodX_1) - g_2(\proxyX_1) \big) \Big) + \var \big( g_2(\proxyX_1) \big)
\\ & \quad + \cov \Big( \bar{W}_1  \big( g_1(\goodX_1) - g_2(\proxyX_1) \big) , g_2(\proxyX_1) \Big) + \cov \Big( g_2(\proxyX_1), \bar{W}_1  \big( g_1(\goodX_1) - g_2(\proxyX_1) \big) \Big)
\\ & = \var \Big( \bar{W}_1  \big( g_1(\goodX_1) - g_2(\proxyX_1) \big) \Big) + \var \big( g_2(\proxyX) \big) 
\\ & \quad + \cov \big(  g_1(\goodX) - g_2(\proxyX) , g_2(\proxyX) \big)+ \cov \big( g_2(\proxyX),  g_1(\goodX) - g_2(\proxyX)  \big).
\end{aligned}$$ This simplifies to \begin{equation}\label{eq:SigmaPTDIntermediateFormulaInProof}
   \Sigma_{\Omega}=  \var \Big( \bar{W}_1  \big( g_1(\goodX_1) - g_2(\proxyX_1) \big) \Big) + \cov \big( g_1(\goodX),  g_2(\proxyX) \big) + \cov \big( g_2(\proxyX),  g_1(\goodX)  \big) -  \var \big( g_2(\proxyX) \big).
\end{equation}

We next simplify each term in \eqref{eq:SigmaPTDIntermediateFormulaInProof}. To simplify the first term note as a consequence of Fact \ref{fact:PropertiesOfGradients}, $\e[g_1(\goodX)-g_2(\proxyX)]=0$. Hence for each $j,j' \in [d]$, by Property (IV) in Lemma \ref{lemma:PropertiesOfIndependentVersionsOfWeights}

$$\cov \Big( \bar{W}_1 \big[ g_1(\goodX_1) - g_2(\proxyX_1) \big]_j,\bar{W}_1 \big[ g_1(\goodX_1) - g_2(\proxyX_1) \big]_{j'} \Big)= \sum_{k=1}^K c_k^2 \e \Big[ \frac{ \big[ g_1(\goodX) - g_2(\proxyX) \big]_j  \big[ g_1(\goodX) - g_2(\proxyX) \big]_{j'}}{\LimitingLabelRulePi{1:k}(\proxyX)} \Big].$$ Since this holds for each $j,j' \in [d]$, $$\var \Big( \bar{W}_1  \big( g_1(\goodX_1) - g_2(\proxyX_1) \big) \Big)= \sum_{k=1}^K c_k^2 \e \Bigg[ \frac{ [ g_1(\goodX) - g_2(\proxyX) ]  [ g_1(\goodX) - g_2(\proxyX) ]^\tran}{\LimitingLabelRulePi{1:k}(\proxyX)} \Bigg].$$ Recalling that $g_1(\goodX)=H_{\thetaTarg}^{-1} \dot{l}_{\thetaTarg}(\goodX)$, $g_2(\proxyX) = \Omega H_{\gammaTarg}^{-1} \dot{l}_{\gammaTarg}(\proxyX)$, and recalling from definition \eqref{eq:GradientAsympCovFormulae}, 

$$    \Sigma_{11} \equiv  \sum_{k=1}^K c_k^2 \e \Big[ \frac{ \dot{l}_{\thetaTarg}(\goodX) [  \dot{l}_{\thetaTarg}(\goodX) ]^\tran}{\LimitingLabelRulePi{1:k}(\proxyX)} \Big], \Sigma_{12} \equiv  \sum_{k=1}^K c_k^2 \e \Big[ \frac{ \dot{l}_{\thetaTarg}(\goodX) [  \dot{l}_{\gammaTarg}(\proxyX) ]^\tran}{\LimitingLabelRulePi{1:k}(\proxyX)} \Big],  
    \Sigma_{22} \equiv \sum_{k=1}^K c_k^2 \e \Big[ \frac{ \dot{l}_{\gammaTarg}(\proxyX) [  \dot{l}_{\gammaTarg}(\proxyX) ]^\tran}{\LimitingLabelRulePi{1:k}(\proxyX)} \Big],$$ we get the following simplification $$\var \Big( \bar{W}_1  \big( g_1(\goodX_1) - g_2(\proxyX_1) \big) \Big) = H_{\thetaTarg}^{-1} \Sigma_{11} H_{\thetaTarg}^{-1} -H_{\thetaTarg}^{-1} \Sigma_{12} H_{\gammaTarg}^{-1} \Omega^\tran - \Omega H_{\gammaTarg}^{-1} \Sigma_{12}^\tran H_{\thetaTarg}^{-1} + \Omega H_{\gammaTarg}^{-1} \Sigma_{22} H_{\gammaTarg}^{-1} \Omega^\tran.$$

    Also since by Fact \ref{fact:PropertiesOfGradients}, $\e[g_1(\goodX)]=0$ and $\e[g_2(\proxyX)]=0$, and recalling from definition \eqref{eq:GradientAsympCovFormulae} that $\Sigma_{13}=\e\big[\dot{l}_{\thetaTarg}(\goodX) [\dot{l}_{\gammaTarg}(\proxyX)]^{\tran} \big]$ and $\Sigma_{33}=\e\big[\dot{l}_{\gammaTarg}(\proxyX) [\dot{l}_{\gammaTarg}(\proxyX)]^{\tran} \big]$, $$\cov \big( g_1(\goodX),g_2(\proxyX) \big)= \e \big[ g_1(\goodX) [g_2(\proxyX)]^\tran \big] = H_{\thetaTarg}^{-1} \Sigma_{13} H_{\gammaTarg}^{-1} \Omega^\tran,$$
    $$\cov \big( g_2(\proxyX), g_1(\goodX) \big)= \e \big[ g_2(\proxyX)  [g_1(\goodX)]^\tran \big] = \Omega H_{\gammaTarg}^{-1} \Sigma_{13}^\tran H_{\thetaTarg}^{-1}, \quad \text{and} $$ 
    $$\var \big( g_2 (\proxyX) \big) = \e \big[ g_2(\proxyX)  [g_2(\proxyX)]^\tran \big]= \Omega H_{\gammaTarg}^{-1} \Sigma_{33} H_{\gammaTarg}^{-1} \Omega^\tran.$$

Plugging the above expressions into \eqref{eq:SigmaPTDIntermediateFormulaInProof} and rearranging terms, 

$$\begin{aligned}
    \Sigma_{\Omega} 
     & = H_{\thetaTarg}^{-1} \Sigma_{11} H_{\thetaTarg}^{-1} -H_{\thetaTarg}^{-1} \Sigma_{12} H_{\gammaTarg}^{-1} \Omega^\tran - \Omega H_{\gammaTarg}^{-1} \Sigma_{12}^\tran H_{\thetaTarg}^{-1} + \Omega H_{\gammaTarg}^{-1} \Sigma_{22} H_{\gammaTarg}^{-1} \Omega^\tran
    \\ & \quad + H_{\thetaTarg}^{-1} \Sigma_{13} H_{\gammaTarg}^{-1} \Omega^\tran + \Omega H_{\gammaTarg}^{-1} \Sigma_{13}^\tran H_{\thetaTarg}^{-1} - \Omega H_{\gammaTarg}^{-1} \Sigma_{33} H_{\gammaTarg}^{-1} \Omega^\tran
    \\ & = H_{\thetaTarg}^{-1} \Sigma_{11} H_{\thetaTarg}^{-1} + \Omega H_{\gammaTarg}^{-1} (\Sigma_{22}-\Sigma_{33}) H_{\gammaTarg}^{-1} \Omega^\tran
    \\ & \quad + H_{\thetaTarg}^{-1} (\Sigma_{13}-\Sigma_{12}) H_{\gammaTarg}^{-1} \Omega^\tran + \big( H_{\thetaTarg}^{-1} (\Sigma_{13}-\Sigma_{12}) H_{\gammaTarg}^{-1} \Omega^\tran \big)^{\tran}
    \\ & = \Sigma^{\PTDSuperScriptAcr}(\Omega),
\end{aligned}$$ where
$\Sigma^{\PTDSuperScriptAcr}(\Omega)$ is defined at \eqref{eq:AsympVarPTDFormula}. Since $\Sigma_{\Omega}=\Sigma^{\PTDSuperScriptAcr}(\Omega)$, and we showed earlier that $\sqrt{N}(\htPTD-\thetaTarg) \xrightarrow{d} \mathcal{N} (0, \Sigma_{\Omega})$ it follows that $\sqrt{N}(\htPTD-\thetaTarg) \xrightarrow{d} \mathcal{N} \big(0, \Sigma^{\PTDSuperScriptAcr}(\Omega) \big)$.

\section{Consistent covariance estimation and valid confidence intervals}\label{sec:AppendixProovingConsistentCovarianceAndValidCIs}

In this appendix, we prove that under Assumptions \ref{assump:IIDUnderlyingData}--\ref{assump:SmoothEnoughForConsitentVarEst} the covariance matrix component estimators defined at \eqref{eq:HessianAndSigmaEstimators} satisfy $\hat{\Sigma}_{11} \xrightarrow{p} \Sigma_{11}$, $\hat{\Sigma}_{12} \xrightarrow{p} \Sigma_{12}$, $\hat{\Sigma}_{22} \xrightarrow{p} \Sigma_{22}$, $\hat{\Sigma}_{13} \xrightarrow{p} \Sigma_{13}$, $\hat{\Sigma}_{33} \xrightarrow{p} \Sigma_{33}$, $\hat{H}_{\thetaTarg} \xrightarrow{p} H_{\thetaTarg}$, and $\hat{H}_{\gammaTarg} \xrightarrow{p} H_{\gammaTarg}$ as $N \to \infty$. These consistency results are then used alongside the asymptotic normality result for $\htPTD$ (Theorem \ref{theorem:CLT_PTDEstimator}) to prove that the corresponding confidence intervals are asymptotically valid (Proposition \ref{prop:AsymptoticallyValidCIs}). We start by proving two lemmas that help streamline the proofs of consistency for the seven matrices $\hat{\Sigma}_{11}, \hat{\Sigma}_{12}, \hat{\Sigma}_{22}, \hat{\Sigma}_{13}, \hat{\Sigma}_{33}, \hat{H}_{\thetaTarg}$, and $ \hat{H}_{\gammaTarg}$.

\subsection{Lemmas for proving consistency of covariance estimators}

\begin{lemma}\label{lemma:DCT_with_SubsequenceTrick}
    Under \samplingSchemeName{} and Assumptions \ref{assump:IIDUnderlyingData}, \ref{assump:LabellingRuleOverlap}, \ref{assump:SmoothEnoughForAsymptoticLineariaty}, \ref{assump:CLTRgularityConditions}, and \ref{assump:SmoothEnoughForConsitentVarEst}, for each $j,j' \in [d]$,
    $$ \lim_{N \to \infty} \e \big[ \big| [\dot{l}_{\htc}(\goodX_1)-\dot{l}_{\thetaTarg}(\goodX_1)]_j  \big|^2 \big] = 0 \quad \text{and} \quad \lim_{N \to \infty} \e \big[ \big| [\ddot{l}_{\htc}(\goodX_1)-\ddot{l}_{\thetaTarg}(\goodX_1)]_{jj'}  \big| \big]=0,$$ and similarly,   $$ \lim_{N \to \infty} \e \big[ \big| [\dot{l}_{\hga}(\proxyX_1)-\dot{l}_{\gammaTarg}(\proxyX_1)]_j  \big|^2 \big] = 0 \quad \text{and} \quad \lim_{N \to \infty} \e \big[ \big| [\ddot{l}_{\hga}(\proxyX_1)-\ddot{l}_{\gammaTarg}(\proxyX_1)]_{jj'}  \big| \big]=0.$$
\end{lemma}

\begin{proof}
    By Proposition \ref{prop:ComponentEstimatorConsistency}, $\htc \xrightarrow{p} \thetaTarg$ as $N \to \infty$. Throughout this proof we will use the notation $\htc_N=\htc$ to emphasize the dependence of the estimator $\htc$ on $N$. We will also fix $j,j' \in [d]$ throughout the proof.
    
    Fix $(N_m)_{m=1}^{\infty}$ to be any increasing subsequence of the natural numbers. Since $\htc_N \xrightarrow{p} \thetaTarg$ as $N \to \infty$, there must exist a further subsequence $(N_{m_r})_{r=1}^{\infty}$ satisfying $\{N_{m_r}\}_{r=1}^{\infty} \subset \{N_m\}_{m=1}^{\infty}$ and $m_1 < m_2 < \dots$ such that $\htc_{N_{m_r}} \xrightarrow{a.s.} \thetaTarg$ as $r \to \infty$. Fixing such a subsequence $(N_{m_r})_{r=1}^{\infty}$ and noting that with probability 1, $\theta \mapsto \dot{l}_{\theta}(\goodX_1)$ and $\theta \mapsto \ddot{l}_{\theta}(\goodX_1)$ are both continuous at $\theta=\thetaTarg$ (by Assumption \ref{assump:SmoothEnoughForConsitentVarEst}(i)), by the composite limit theorem,  $$\lim_{r \to \infty} \big| [\dot{l}_{\htc_{N_{m_r}}} (\goodX_1) - \dot{l}_{\thetaTarg}(\goodX_1) ]_j \big|^2=0 \quad \text{and} \quad \lim_{r \to \infty} \big| [\ddot{l}_{\htc_{N_{m_r}}} (\goodX_1) - \ddot{l}_{\thetaTarg}(\goodX_1) ]_{jj'} \big|=0 \quad  \text{with probability 1}.$$  It follows that $ \big| [\dot{l}_{\htc_{N_{m_r}}} (\goodX_1) - \dot{l}_{\thetaTarg}(\goodX_1) ]_j \big|^2 \xrightarrow{a.s.} 0$ and $\big| [\ddot{l}_{\htc_{N_{m_r}}} (\goodX_1) - \ddot{l}_{\thetaTarg}(\goodX_1) ]_{jj'} \big| \xrightarrow{a.s.} 0$ as $r \to \infty$.

    Since $\htc_{N_{m_r}} \xrightarrow{a.s.} \thetaTarg$, there exists an $r_0$ such that for all $r>r_0$, $\htc_{N_{m_r}} \in \mathcal{L}_{\thetaTarg} \cap \mathcal{B}_{jj'}$ almost surely, where $\mathcal{L}_{\thetaTarg}$ is the neighborhood of $\thetaTarg$ for the local Lipschitz condition in Assumption \ref{assump:SmoothEnoughForAsymptoticLineariaty}(iv) and $\mathcal{B}_{jj'}$ is the neighborhood of $\thetaTarg$ from Assumption \ref{assump:SmoothEnoughForConsitentVarEst}(ii) in which $[\ddot{l}_{\theta}(\goodX)]_{jj'}$ is bounded by an integrable function. For all $r>r_0$ since $\htc_{N_{m_r}},\thetaTarg \in  \mathcal{B}_{jj'}$ almost surely, by Assumption \ref{assump:SmoothEnoughForConsitentVarEst}(ii), almost surely $$ \big| [\ddot{l}_{\htc_{N_{m_r}}} (\goodX_1) - \ddot{l}_{\thetaTarg}(\goodX_1) ]_{jj'} \big| \leq 2 L_{jj'} (\goodX_1) \quad \text{for }  r> r_0 \quad \text{ where} \quad \e[L_{jj'}(\goodX_1)] < \infty.$$ Combining this with the earlier result that $\big| [\ddot{l}_{\htc_{N_{m_r}}} (\goodX_1) - \ddot{l}_{\thetaTarg}(\goodX_1) ]_{jj'} \big| \xrightarrow{a.s.} 0$ as $r \to \infty$, by the dominated convergence theorem, $$\lim_{r \to \infty} \e \big[ \big| [\ddot{l}_{\htc_{N_{m_r}}} (\goodX_1) - \ddot{l}_{\thetaTarg}(\goodX_1) ]_{jj'} \big| \big] =0.$$ Next note that for $\theta \in \mathcal{L}_{\thetaTarg}$ and $\argx \in \xSpace$, $\vert [\dot{l}_{\theta}(\argx)]_j \vert \leq M(\argx)$ as a consequence of Assumption \ref{assump:SmoothEnoughForAsymptoticLineariaty}(iv) (for a proof in the case where $\theta=\thetaTarg$, see Fact \ref{fact:PropertiesOfGradients}). Thus for all $r>r_0$ since $\htc_{N_{m_r}},\thetaTarg \in  \mathcal{L}_{\thetaTarg}$ almost surely, almost surely, $$ \big| [\dot{l}_{\htc_{N_{m_r}}} (\goodX_1) - \dot{l}_{\thetaTarg}(\goodX_1) ]_{j} \big|^2 \leq 4 M^2(\goodX_1) \quad \text{for }  r> r_0, \quad \text{ where } \quad \e[M^2(\goodX_1)] < \infty$$ by Assumption \ref{assump:SmoothEnoughForAsymptoticLineariaty}(iv). Combining this with the earlier result that $\big| [\dot{l}_{\htc_{N_{m_r}}} (\goodX_1) - \dot{l}_{\thetaTarg}(\goodX_1) ]_j \big|^2 \xrightarrow{a.s.} 0$ as $r \to \infty$, by the dominated convergence theorem $$\lim_{r \to \infty} \e \big[ \big| [\dot{l}_{\htc_{N_{m_r}}} (\goodX_1) - \dot{l}_{\thetaTarg}(\goodX_1) ]_j \big|^2 \big]=0.$$

   Thus we have shown that for any subsequence $(N_m)_{m=1}^{\infty}$, there exists a further subsequence  $(N_{m_r})_{r=1}^{\infty}$ such that $$\lim_{r \to \infty} \e \big[ \big| [\dot{l}_{\htc_{N_{m_r}}} (\goodX_1) - \dot{l}_{\thetaTarg}(\goodX_1) ]_j \big|^2 \big]=0 \quad \text{and} \quad \lim_{r \to \infty} \e \big[ \big| [\ddot{l}_{\htc_{N_{m_r}}} (\goodX_1) - \ddot{l}_{\thetaTarg}(\goodX_1) ]_{jj'} \big| \big] =0.$$ It follows by way of contradiction that $$\lim_{N \to \infty} \e \big[ \big| [\dot{l}_{\htc_{N}} (\goodX_1) - \dot{l}_{\thetaTarg}(\goodX_1) ]_j \big|^2 \big]=0 \quad \text{and} \quad \lim_{N \to \infty} \e \big[ \big| [\ddot{l}_{\htc_{N}} (\goodX_1) - \ddot{l}_{\thetaTarg}(\goodX_1) ]_{jj'} \big| \big] =0.$$ The proofs that $\lim_{N \to \infty} \e \big[ \big| [\dot{l}_{\hga}(\proxyX_1)-\dot{l}_{\gammaTarg}(\proxyX_1)]_j  \big|^2 \big] = 0$ and $\lim_{N \to \infty} \e \big[ \big| [\ddot{l}_{\hga}(\proxyX_1)-\ddot{l}_{\gammaTarg}(\proxyX_1)]_{jj'}  \big| \big]=0$ follow from analogous arguments.
\end{proof}

As a consequence of the above lemma, we can derive the following result which will be useful in proving consistency of $\hat{\Sigma}_{11}$, $\hat{\Sigma}_{12}$, $\hat{\Sigma}_{22}$, $\hat{\Sigma}_{13}$, and $\hat{\Sigma}_{33}$. Heuristically the following lemma states that asymptotically the errors in $\htc$ and $\hga$ can be ignored when taking averages.

\begin{lemma}\label{lemma:AsymptoticallyIgnoreEstimationErrorInAverages}
     Let $\hat{\zeta}_1 \equiv \htc$ and $\hat{\zeta}_2 \equiv \hga$, and let $\zeta_1 \equiv \thetaTarg$ and $\zeta_2 \equiv \gammaTarg$. Next, as shorthand notation, for each $i \in \mathbb{Z}_+$, define $Y_i^{(1)} \equiv \goodX_i$ and $Y_i^{(2)} \equiv \proxyX_i$.  Under \samplingSchemeName{} and Assumptions \ref{assump:IIDUnderlyingData}, \ref{assump:LabellingRuleOverlap}, \ref{assump:SmoothEnoughForAsymptoticLineariaty}, \ref{assump:CLTRgularityConditions}, and \ref{assump:SmoothEnoughForConsitentVarEst}, for any $s,s' \in \{1,2\}$, $r \in \{0,1,2\}$, and any fixed measurable function $g : \xProxySpace \to [0,1]$, $$\frac{1}{N} \sum_{i=1}^N  W_i^r g(\proxyX_i) \dot{l}_{\hat{\zeta}_s} ( Y_i^{(s)} ) \big[ \dot{l}_{\hat{\zeta}_{s'}} (Y_i^{(s')}) \big]^\tran = \frac{1}{N} \sum_{i=1}^N W_i^r g(\proxyX_i) \dot{l}_{\zeta_s}  (Y_i^{(s)})  \big[ \dot{l}_{\zeta_{s'}} ( Y_i^{(s')} ) \big]^\tran +o_p(1).$$
\end{lemma}

\begin{proof}

Fix $s,s' \in \{1,2\}$ and $r \in \{0,1,2\}$. Next define $$\bar{T}_{\hat{\zeta}} \equiv \frac{1}{N} \sum_{i=1}^N W_i^r g(\proxyX_i)\dot{l}_{\hat{\zeta}_s} ( Y_i^{(s)} )  \big[ \dot{l}_{\hat{\zeta}_{s'}} (Y_i^{(s')}) \big]^\tran \quad \text{and} \quad \bar{T}_{\zeta} \equiv  \frac{1}{N} \sum_{i=1}^N W_i^r g(\proxyX_i) \dot{l}_{\zeta_s}  (Y_i^{(s)})  \big[ \dot{l}_{\zeta_{s'}} ( Y_i^{(s')} ) \big]^\tran,$$ and we wish to show that $\bar{T}_{\hat{\zeta}}=\bar{T}_{\zeta}+o_p(1)$. To do this observe that $$\bar{T}_{\hat{\zeta}}=\bar{T}_{\zeta}+\bar{T}^{(1)}+\bar{T}^{(2)}+\bar{T}^{(3)}, \quad \text{where} $$  

\begin{equation*}
\begin{split}
   \bar{T}^{(1)} & \equiv \frac{1}{N} \sum_{i=1}^N W_i^r g(\proxyX_i) \big[ \dot{l}_{\hat{\zeta}_s} (Y_i^{(s)} )-\dot{l}_{\zeta_s} (Y_i^{(s)} ) \big] \big[ \dot{l}_{\zeta_{s'}} ( Y_i^{(s')}) \big]^\tran,  \\  
   \bar{T}^{(2)} & \equiv \frac{1}{N} \sum_{i=1}^N W_i^r  g(\proxyX_i) \big[ \dot{l}_{\zeta_s} (Y_i^{(s)} ) \big] \big[ \dot{l}_{\hat{\zeta}_{s'}} ( Y_i^{(s')}) -\dot{l}_{\zeta_{s'}} ( Y_i^{(s')}) \big]^\tran, \quad \text{and} \\ 
   \bar{T}^{(3)} & \equiv \frac{1}{N} \sum_{i=1}^N W_i^r g(\proxyX_i) \big[ \dot{l}_{\hat{\zeta}_s} (Y_i^{(s)}) -\dot{l}_{\zeta_s} (Y_i^{(s)}) \big] \big[ \dot{l}_{\hat{\zeta}_{s'}} ( Y_i^{(s')}) -\dot{l}_{\zeta_{s'}} ( Y_i^{(s')}) \big]^\tran.
\end{split}
\end{equation*}

We will next show that $\bar{T}^{(1)} \xrightarrow{p} 0$. To do this fix $j, j' \in [d]$. Next fix $\epsilon>0$. By Fact \ref{fact:WeightBound}, for $r \in \{0,1,2\}$, $W_i^r \in [0, b^{-2K}]$ almost surely. Also by Corollary \ref{cor:SameDistWithEstimators}, for each fixed $N$ and $i \in [N]$, $(\htc,\hga,\goodX_i,\proxyX_i)$ has the same joint distribution as $(\htc,\hga,\goodX_1,\proxyX_1)$. Thus regardless of the value of $s,s' \in \{1,2\}$ and $r \in \{0,1,2\}$, and the fixed measurable function $g: \xProxySpace \to [0,1]$ choice, by Markov's inequality and the previous statements,

$$\begin{aligned} 
\mathbb{P} \big( \big| [\bar{T}^{(1)}]_{jj'} \big| > \epsilon \big) & \leq \frac{1}{\epsilon} \e \Big[ \Big| \frac{1}{N} \sum_{i=1}^N W_i^r g(\proxyX_i) \big[ \dot{l}_{\hat{\zeta}_s} (Y_i^{(s)} )-\dot{l}_{\zeta_s} (Y_i^{(s)} ) \big]_j \big[ \dot{l}_{\zeta_{s'}} ( Y_i^{(s')}) \big]_{j'} \Big|  \Big]
\\ & \leq \frac{1}{N b^{2K} \cdot \epsilon}  \sum_{i=1}^N \e \Big[ \Big|  \big[ \dot{l}_{\hat{\zeta}_s} (Y_i^{(s)} )-\dot{l}_{\zeta_s} (Y_i^{(s)} ) \big]_j \big[ \dot{l}_{\zeta_{s'}} ( Y_i^{(s')}) \big]_{j'} \Big|  \Big]
\\ & \leq \frac{1}{ b^{2K} \cdot \epsilon} \cdot  \e \Big[ \Big|  \big[ \dot{l}_{\hat{\zeta}_s} (Y_1^{(s)} )-\dot{l}_{\zeta_s} (Y_1^{(s)} ) \big]_j \big[ \dot{l}_{\zeta_{s'}} ( Y_1^{(s')}) \big]_{j'} \Big|  \Big]
\\ & \leq \frac{1}{ b^{2K} \cdot \epsilon} \cdot \sqrt{\e \big[ \big| [ \dot{l}_{\hat{\zeta}_s} (Y_1^{(s)} )-\dot{l}_{\zeta_s} (Y_1^{(s)} ) ]_j \big|_2^2 \big] \cdot \e \big[ \big( [  \dot{l}_{\zeta_{s'}} ( Y_1^{(s')})]_{j'} \big)^2  \big] },
\end{aligned}$$ where the last step follows from the Cauchy-Schwarz inequality. By Fact \ref{fact:PropertiesOfGradients}, $\vert \vert \dot{l}_{\thetaTarg} (\goodX_1) \vert \vert_{\infty} \leq M(\goodX_1)$ and $\vert \vert \dot{l}_{\gammaTarg} (\proxyX_1) \vert \vert_{\infty} \leq \tilde{M}(\proxyX_1)$, so $$\e \big[ \big( [  \dot{l}_{\zeta_{s'}} ( Y_1^{(s')})]_{j'} \big)^2  \big] \leq \e\big[ M^2(\goodX_1)+\tilde{M}^2(\proxyX_1) \big]=\e[M^2(\goodX)]+\e[\tilde{M}^2(\proxyX)]< \infty,$$ where the last step follows from Assumption \ref{assump:SmoothEnoughForAsymptoticLineariaty}(iv). Moreover, by Lemma \ref{lemma:DCT_with_SubsequenceTrick}, regardless of whether $s$ equals 1 or 2, $\lim_{N \to \infty} \e \big[ \big| [ \dot{l}_{\hat{\zeta}_s} (Y_1^{(s)} )-\dot{l}_{\zeta_s} (Y_1^{(s)} ) ]_j \big|_2^2 \big]=0$. 
Thus, by taking limit as $N \to \infty$ of each side of the above inequality, $\lim_{N \to \infty} \mathbb{P} \big( \big| [\bar{T}^{(1)}]_{jj'} \big| > \epsilon \big) =0$. Since this argument holds for any fixed $\epsilon>0$, $[\bar{T}^{(1)}]_{jj'} \xrightarrow{p} 0$. Moreover, this argument holds for any fixed $j,j' \in [d]$, so $\bar{T}^{(1)} \xrightarrow{p} 0$. Similar arguments show that $\bar{T}^{(2)} \xrightarrow{p} 0$ and $\bar{T}^{(3)} \xrightarrow{p} 0$.

Combining this with an earlier decomposition, $$\bar{T}_{\hat{\zeta}}=\bar{T}_{\zeta}+\bar{T}^{(1)}+\bar{T}^{(2)}+\bar{T}^{(3)}=\bar{T}_{\zeta}+o_p(1).$$ Recalling the definitions of $\bar{T}_{\hat{\zeta}}$ and $\bar{T}_{\zeta}$, and noting the above argument holds for any $r \in \{0,1,2\}$ and $s,s' \in \{1,2\}$ completes the proof.
\end{proof}

\subsection{Proof of consistency of covariance matrix components}\label{sec:ConsistencyOfCovarianceMatrixComponents}

 In this subsection, we prove that under \samplingSchemeName{} and Assumptions \ref{assump:IIDUnderlyingData}--\ref{assump:SmoothEnoughForConsitentVarEst}, $\hat{\Sigma}_{11} \xrightarrow{p} \Sigma_{11}$, $\hat{\Sigma}_{12} \xrightarrow{p} \Sigma_{12}$, $\hat{\Sigma}_{22} \xrightarrow{p} \Sigma_{22}$, $\hat{\Sigma}_{13} \xrightarrow{p} \Sigma_{13}$, $\hat{\Sigma}_{33} \xrightarrow{p} \Sigma_{33}$, $\hat{H}_{\thetaTarg} \xrightarrow{p} H_{\thetaTarg}$, and $\hat{H}_{\gammaTarg} \xrightarrow{p} H_{\gammaTarg}$ as $N \to \infty$. We will prove this claim by breaking it up into 3 propositions (Propositions \ref{prop:ConsistencyOfHessians}, \ref{prop:ConsistencyOfSig13andSig33}, and \ref{prop:ConsistencyOfSigmaUpperLeft}), each of which establishes some of the above consistency results. Taken together, Propositions \ref{prop:ConsistencyOfHessians}, \ref{prop:ConsistencyOfSig13andSig33}, and \ref{prop:ConsistencyOfSigmaUpperLeft} establish $(\hat{\Sigma}_{11}, \hat{\Sigma}_{12}, \hat{\Sigma}_{22}, \hat{\Sigma}_{13}, \hat{\Sigma}_{33}, \hat{H}_{\thetaTarg},\hat{H}_{\gammaTarg}) \xrightarrow{p} (\Sigma_{11}, \Sigma_{12}, \Sigma_{22}, \Sigma_{13}, \Sigma_{33}, H_{\thetaTarg},H_{\gammaTarg})$ as $N \to \infty$. We start by proving the consistency of the Hessian estimators where, recall from \eqref{eq:HessianAndSigmaEstimators}, that $$\hat{H}_{\thetaTarg} =\frac{1}{N} \sum_{i=1}^N W_i \ddot{l}_{\htc} (\goodX_i) \quad \text{and} \quad \hat{H}_{\gammaTarg} =\frac{1}{N} \sum_{i=1}^N  \ddot{l}_{\hga} (\proxyX_i).$$

\begin{proposition} \label{prop:ConsistencyOfHessians}  Under \samplingSchemeName{} and Assumptions \ref{assump:IIDUnderlyingData}--\ref{assump:SmoothEnoughForConsitentVarEst}, $\hat{H}_{\thetaTarg} \xrightarrow{p} H_{\thetaTarg}$ and $\hat{H}_{\gammaTarg} \xrightarrow{p} H_{\gammaTarg}$ as $N \to \infty$.
\end{proposition}
\begin{proof}

  Fix $j,j' \in [d]$. First observe that $$[\hat{H}_{\thetaTarg}]_{jj'}= \frac{1}{N} \sum_{i=1}^N W_i [\ddot{l}_{\thetaTarg}(\goodX_i)]_{jj'} + \frac{1}{N} \sum_{i=1}^N W_i [\ddot{l}_{\htc}(\goodX_i)-\ddot{l}_{\thetaTarg}(\goodX_i)]_{jj'}.$$ Now observe that for any $\epsilon>0$, by Markov's inequality, Fact \ref{fact:WeightBound}, and the fact that for a fixed $N$ and $i \in [N]$, $(\htc,\goodX_i)$ and $(\htc,\goodX_1)$ have the same joint distribution (see Corollary \ref{cor:SameDistWithEstimators}), $$\begin{aligned}  \mathbb{P} \Big( \Big| \frac{1}{N} \sum_{i=1}^N  W_i [\ddot{l}_{\htc}(\goodX_i)-\ddot{l}_{\thetaTarg}(\goodX_i)]_{jj'} \Big| > \epsilon \Big) & \leq \frac{1}{N \epsilon} \cdot  \sum_{i=1}^N  \e \big[ \big|  W_i [\ddot{l}_{\htc}(\goodX_i)-\ddot{l}_{\thetaTarg}(\goodX_i)]_{jj'} \big| \big] 
 \\ & \leq \frac{1}{N \epsilon b^K} \cdot  \sum_{i=1}^N  \e \big[ \big|  [\ddot{l}_{\htc}(\goodX_i)-\ddot{l}_{\thetaTarg}(\goodX_i)]_{jj'} \big| \big] 
 \\ & = (\epsilon b^K)^{-1}   \e \big[ \big|  [\ddot{l}_{\htc}(\goodX_1)-\ddot{l}_{\thetaTarg}(\goodX_1)]_{jj'} \big| \big]. \end{aligned}$$ Taking $N \to \infty$ of each side of the above inequality and applying Lemma \ref{lemma:DCT_with_SubsequenceTrick}, $$\limsup_{N \to \infty} \mathbb{P} \Big( \Big| \frac{1}{N} \sum_{i=1}^N  W_i [\ddot{l}_{\htc}(\goodX_i)-\ddot{l}_{\thetaTarg}(\goodX_i)]_{jj'} \Big| > \epsilon \Big) \leq (\epsilon b^K)^{-1}  \limsup_{N \to \infty} \e \big[ \big|  [\ddot{l}_{\htc}(\goodX_1)-\ddot{l}_{\thetaTarg}(\goodX_1)]_{jj'} \big| \big]=0.$$ Since the above holds for any $\epsilon>0,$ $\frac{1}{N} \sum_{i=1}^N W_i [\ddot{l}_{\htc}(\goodX_i)-\ddot{l}_{\thetaTarg}(\goodX_i)]_{jj'} \xrightarrow{p} 0$, so a previous expression simplifies to $$[\hat{H}_{\thetaTarg}]_{jj'}= \frac{1}{N} \sum_{i=1}^N W_i [\ddot{l}_{\thetaTarg}(\goodX_i)]_{jj'} + o_p(1).$$ A similar argument (that does not involve bounding the weights $W_i$) establishes that $$[\hat{H}_{\gammaTarg}]_{jj'}= \frac{1}{N} \sum_{i=1}^N  [\ddot{l}_{\gammaTarg}(\proxyX_i)]_{jj'} + \frac{1}{N} \sum_{i=1}^N [\ddot{l}_{\hga}(\proxyX_i)-\ddot{l}_{\gammaTarg}(\proxyX_i)]_{jj'} =\frac{1}{N} \sum_{i=1}^N  [\ddot{l}_{\gammaTarg}(\proxyX_i)]_{jj'} +o_p(1).$$

Now note that by Assumption \ref{assump:SmoothEnoughForConsitentVarEst}(iii), $\e\big[\big([\ddot{l}_{\thetaTarg}(\goodX)]_{jj'} \big)^2 \big]<\infty$, so by directly applying Lemma \ref{lemma:VarianceOfWeightedAveragesGeneric}, $$[\hat{H}_{\thetaTarg}]_{jj'}= \frac{1}{N} \sum_{i=1}^N W_i [\ddot{l}_{\thetaTarg}(\goodX_i)]_{jj'} + o_p(1) \xrightarrow{p} \e \big[ [\ddot{l}_{\thetaTarg}(\goodX)]_{jj'} \big] + 0=[H_{\thetaTarg}]_{jj'},$$ where the last step holds by Assumption \ref{assump:SmoothEnoughForConsitentVarEst}(i). Similarly, note that by Assumption \ref{assump:SmoothEnoughForConsitentVarEst}(ii), $\e\big[\big| [\ddot{l}_{\gammaTarg}(\proxyX)]_{jj'} \big| \big] \leq \e[ \tilde{L}_{jj'}(\proxyX)]<\infty$, and hence by the weak law of large numbers, $$[\hat{H}_{\gammaTarg}]_{jj'}= \frac{1}{N} \sum_{i=1}^N  [\ddot{l}_{\gammaTarg}(\proxyX_i)]_{jj'} + o_p(1) \xrightarrow{p} \e\big[ [\ddot{l}_{\gammaTarg}(\proxyX)]_{jj'} \big]+0=[H_{\gammaTarg}]_{jj'},$$ where the last step holds by Assumption \ref{assump:SmoothEnoughForConsitentVarEst}(i). 
Thus $[\hat{H}_{\thetaTarg}]_{jj'} \xrightarrow{p} [H_{\thetaTarg}]_{jj'}$ and $[\hat{H}_{\gammaTarg}]_{jj'} \xrightarrow{p} [H_{\gammaTarg}]_{jj'}$. Since this argument holds for any fixed $j,j' \in [d]$, $\hat{H}_{\thetaTarg} \xrightarrow{p} H_{\thetaTarg}$ and $\hat{H}_{\gammaTarg} \xrightarrow{p} H_{\gammaTarg}$.

\end{proof}

We next show that $\hat{\Sigma}_{13}$ and $\hat{\Sigma}_{33}$ are consistent, where recall from \eqref{eq:GradientAsympCovFormulae} that $$\Sigma_{13} =\e \big[\dot{l}_{\thetaTarg}(\goodX) [\dot{l}_{\gammaTarg}(\proxyX)]^\tran \big] \quad \text{and} \quad \Sigma_{33} =\e \big[\dot{l}_{\gammaTarg}(\proxyX) [\dot{l}_{\gammaTarg}(\proxyX)]^\tran \big], \quad \text{ and recall from \eqref{eq:HessianAndSigmaEstimators} that }$$ $$\hat{\Sigma}_{13} = \frac{1}{N} \sum_{i=1}^N W_i \dot{l}_{\htc} (\goodX_i) \big[ \dot{l}_{\hga} (\proxyX_i) \big]^\tran \quad \text{and} \quad \hat{\Sigma}_{33} \equiv \frac{1}{N} \sum_{i=1}^N  \dot{l}_{\hga} (\proxyX_i) \big[ \dot{l}_{\hga} (\proxyX_i) \big]^\tran.$$

\begin{proposition}\label{prop:ConsistencyOfSig13andSig33}
 Under \samplingSchemeName{} and Assumptions \ref{assump:IIDUnderlyingData}--\ref{assump:SmoothEnoughForConsitentVarEst}, $\hat{\Sigma}_{13} \xrightarrow{p} \Sigma_{13}$ and $\hat{\Sigma}_{33} \xrightarrow{p} \Sigma_{33}$ as $N \to \infty$.
\end{proposition}
\begin{proof}
   First note that by direct application of Lemma \ref{lemma:AsymptoticallyIgnoreEstimationErrorInAverages}, $$\hat{\Sigma}_{13} = \frac{1}{N} \sum_{i=1}^N W_i \dot{l}_{\thetaTarg}(\goodX_i) \big[ \dot{l}_{\gammaTarg} (\proxyX_i) \big]^\tran +o_p(1) \quad \text{and} \quad \hat{\Sigma}_{33} = \frac{1}{N} \sum_{i=1}^N  \dot{l}_{\gammaTarg}(\proxyX_i) \big[ \dot{l}_{\gammaTarg} (\proxyX_i) \big]^\tran +o_p(1).$$ Note that by Fact \ref{fact:PropertiesOfGradients} and Assumption \ref{assump:SmoothEnoughForAsymptoticLineariaty}(iv), for any $j,j' \in [d]$, $$\e \Big[ \Big| [\dot{l}_{\gammaTarg}(\proxyX) ]_j [\dot{l}_{\gammaTarg} (\proxyX) ]_{j'} \Big| \Big] \leq \e [ \tilde{M}^2(\proxyX) ]  < \infty.$$ Hence since $(\proxyX_i)_{i=1}^N$ are i.i.d. (Assumption \ref{assump:IIDUnderlyingData}), by the weak law of large numbers for any $j,j' \in [d]$, $$[\hat{\Sigma}_{33}]_{jj'} = \frac{1}{N} \sum_{i=1}^N  [\dot{l}_{\gammaTarg}(\proxyX_i)]_j [ \dot{l}_{\gammaTarg} (\proxyX_i)]_{j'} +o_p(1) \xrightarrow{p} \e\big[ [\dot{l}_{\gammaTarg}(\proxyX)]_j [ \dot{l}_{\gammaTarg} (\proxyX)]_{j'} \big] +0=[\Sigma_{33}]_{jj'},$$ and thus $\hat{\Sigma}_{33} \xrightarrow{p} \Sigma_{33}$. 
   
   Now fix $j,j' \in [d]$. 
   Note that by Assumption \ref{assump:SmoothEnoughForConsitentVarEst}(iii) and the Cauchy-Schwarz inequality, $$\e \Big[ \Big( [\dot{l}_{\thetaTarg}(\goodX)]_j [\dot{l}_{\gammaTarg} (\proxyX) ]_{j'} \Big)^2 \Big] \leq \sqrt{ \e \big[ \big( [\dot{l}_{\thetaTarg}(\goodX)]_j \big)^4 \big] \cdot \e \big[ \big( [\dot{l}_{\gammaTarg} (\proxyX) ]_{j'} \big)^4 \big]} < \infty.$$ Thus by noting that for each $i \in \mathbb{Z}_+$, $\proxyX_i$ and $\goodX_i$ are each random vectors that can be written as a subset of the components of the random vector $\datvecraw_i$, we can apply Lemma \ref{lemma:VarianceOfWeightedAveragesGeneric} to show $$\frac{1}{N} \sum_{i=1}^N W_i [\dot{l}_{\thetaTarg}(\goodX_i)]_j [\dot{l}_{\gammaTarg} (\proxyX_i) ]_{j'} \xrightarrow{p} \e\Big[ [\dot{l}_{\thetaTarg}(\goodX)]_j [\dot{l}_{\gammaTarg} (\proxyX) ]_{j'} \Big] = [\Sigma_{13}]_{jj'}.$$ Above the last step follows from the definition of $\Sigma_{13}$ at \eqref{eq:GradientAsympCovFormulae}. Combining this with an earlier result, $$[\hat{\Sigma}_{13}]_{jj'} = \frac{1}{N} \sum_{i=1}^N W_i [\dot{l}_{\thetaTarg}(\goodX_i)]_j [\dot{l}_{\gammaTarg} (\proxyX_i) ]_{j'} +o_p(1) \xrightarrow{p}  [\Sigma_{13}]_{jj'}.$$  Hence we have shown that $[\hat{\Sigma}_{13}]_{jj'} \xrightarrow{p} [\Sigma_{13}]_{jj'}$. Because this argument holds for any fixed $j,j' \in [d]$, $\hat{\Sigma}_{13} \xrightarrow{p} \Sigma_{13}$. 

\end{proof}

We next show that $\hat{\Sigma}_{11}$, $\hat{\Sigma}_{12}$, and $\hat{\Sigma}_{22}$ are consistent, where recall from \eqref{eq:GradientAsympCovFormulae} and \eqref{eq:HessianAndSigmaEstimators} that 

\begin{equation*}
\begin{split}
   \Sigma_{11} & \equiv  \sum_{k=1}^K c_k^2 \e \Big[ \frac{ \dot{l}_{\thetaTarg}(\goodX) [  \dot{l}_{\thetaTarg}(\goodX) ]^\tran}{\LimitingLabelRulePi{1:k}(\proxyX)} \Big], \quad 
   \quad \hat{\Sigma}_{11}  \equiv \frac{1}{N} \sum_{i=1}^N W_i^2 \dot{l}_{\htc} (\goodX_i) [ \dot{l}_{\htc} (\goodX_i)]^\tran, 
   \\ 
   \Sigma_{12} & \equiv  \sum_{k=1}^K c_k^2 \e \Big[ \frac{ \dot{l}_{\thetaTarg}(\goodX) [  \dot{l}_{\gammaTarg}(\proxyX) ]^\tran}{\LimitingLabelRulePi{1:k}(\proxyX)} \Big], \quad \quad \hat{\Sigma}_{12} \equiv \frac{1}{N} \sum_{i=1}^N W_i^2 \dot{l}_{\htc} (\goodX_i) [ \dot{l}_{\hga} (\proxyX_i)]^\tran, 
   \\   
    \Sigma_{22} & \equiv \sum_{k=1}^K c_k^2 \e \Big[ \frac{ \dot{l}_{\gammaTarg}(\proxyX) [  \dot{l}_{\gammaTarg}(\proxyX) ]^\tran}{\LimitingLabelRulePi{1:k}(\proxyX)} \Big], \quad \quad \hat{\Sigma}_{22}  \equiv \frac{1}{N} \sum_{i=1}^N W_i^2  \dot{l}_{\hga} (\proxyX_i) [ \dot{l}_{\hga} (\proxyX_i) ]^\tran. 
\end{split}
\end{equation*}

\begin{proposition}\label{prop:ConsistencyOfSigmaUpperLeft}
 Under \samplingSchemeName{} and Assumptions \ref{assump:IIDUnderlyingData}--\ref{assump:SmoothEnoughForConsitentVarEst}, $\hat{\Sigma}_{11} \xrightarrow{p} \Sigma_{11}$, $\hat{\Sigma}_{12} \xrightarrow{p} \Sigma_{12}$, and $\hat{\Sigma}_{22} \xrightarrow{p} \Sigma_{22}$ as $N \to \infty$.
\end{proposition}

\begin{proof}
    Fix $j,j' \in [d]$. We start by showing that $N^{-1} \sum_{i=1}^N W_i^2 [\dot{l}_{\thetaTarg} (\goodX_i)]_j [ \dot{l}_{\gammaTarg} (\proxyX_i)]_{j'} \xrightarrow{p} [\Sigma_{12}]_{jj'}$. To do this first note that for a fixed $N$ and any $i,i' \in [N]$ such that $i \neq i'$, by Proposition \ref{prop:ExchangeableData},   $(W_i,\goodX_i,\proxyX_i,W_{i'},\goodX_{i'},\proxyX_{i'})$ has the same joint distribution as $(W_1,\goodX_1,\proxyX_1,W_2,\goodX_2,\proxyX_2)$. Hence
    $$\begin{aligned} \var \Big( \frac{1}{N} \sum_{i=1}^N W_i^2 [\dot{l}_{\thetaTarg} (\goodX_i)]_j [ \dot{l}_{\gammaTarg} (\proxyX_i)]_{j'} \Big) & = \frac{1}{N^2} \sum_{i=1}^N \sum_{i'=1}^N \cov \Big( W_i^2 [\dot{l}_{\thetaTarg} (\goodX_{i})]_j [ \dot{l}_{\gammaTarg} (\proxyX_i)]_{j'}, W_{i'}^2 [\dot{l}_{\thetaTarg} (\goodX_{i'})]_j [ \dot{l}_{\gammaTarg} (\proxyX_{i'})]_{j'}\Big)
    \\ & = \frac{N}{N^2} \cdot \var \Big( W_1^2 [\dot{l}_{\thetaTarg} (\goodX_{1})]_j [ \dot{l}_{\gammaTarg} (\proxyX_1)]_{j'} \Big)
    \\ & + \frac{N (N-1)}{N^2} \cdot \cov \Big( W_1^2 [\dot{l}_{\thetaTarg} (\goodX_{1})]_j [ \dot{l}_{\gammaTarg} (\proxyX_1)]_{j'}, W_{2}^2 [\dot{l}_{\thetaTarg} (\goodX_{2})]_j [ \dot{l}_{\gammaTarg} (\proxyX_{2})]_{j'}\Big).
    \end{aligned}$$
    
Now observe that by Cauchy-Schwarz and Assumption \ref{assump:SmoothEnoughForConsitentVarEst}(iii) $$\e \Big[ \Big( [\dot{l}_{\thetaTarg} (\goodX)]_j [ \dot{l}_{\gammaTarg} (\proxyX)]_{j'} \Big)^2 \Big] \leq \sqrt{\e \big[ \big([\dot{l}_{\thetaTarg}(\goodX)]_j \big)^4 \big] \cdot \e \big[ \big([\dot{l}_{\gammaTarg}(\proxyX)]_{j'} \big)^4 \big] } < \infty,$$ so by applying Proposition \ref{prop:Vanishing2ndOrderCov} (in the case where $\eta=1$ and $f=g$), $$\lim_{N \to \infty} \cov \Big( W_1^2 [\dot{l}_{\thetaTarg} (\goodX_{1})]_j [ \dot{l}_{\gammaTarg} (\proxyX_1)]_{j'}, W_{2}^2 [\dot{l}_{\thetaTarg} (\goodX_{2})]_j [ \dot{l}_{\gammaTarg} (\proxyX_{2})]_{j'}\Big) =0.$$ Moreover by Fact \ref{fact:WeightBound}, $$\var \Big( W_1^2 [\dot{l}_{\thetaTarg} (\goodX_{1})]_j [ \dot{l}_{\gammaTarg} (\proxyX_1)]_{j'} \Big) \leq \e \Big[  W_1^4 \Big( [\dot{l}_{\thetaTarg} (\goodX_{1})]_j [ \dot{l}_{\gammaTarg} (\proxyX_1)]_{j'} \Big)^2  \Big] \leq b^{-4K} \e \Big[   \Big( [\dot{l}_{\thetaTarg} (\goodX_{1})]_j [ \dot{l}_{\gammaTarg} (\proxyX_1)]_{j'} \Big)^2  \Big] < \infty.$$ Thus we can take $N \to \infty$ of each side of a previous equation to get that, $$\lim_{N \to \infty} \var \Big( \frac{1}{N} \sum_{i=1}^N W_i^2 [\dot{l}_{\thetaTarg} (\goodX_i)]_j [ \dot{l}_{\gammaTarg} (\proxyX_i)]_{j'} \Big) =0.$$

Next let $$\mu_{N} \equiv \e \Big[ \frac{1}{N} \sum_{i=1}^N W_i^2 [\dot{l}_{\thetaTarg} (\goodX_i)]_j [ \dot{l}_{\gammaTarg} (\proxyX_i)]_{j'} \Big]= \e \big[ W_1^2 [\dot{l}_{\thetaTarg} (\goodX_1)]_j [ \dot{l}_{\gammaTarg} (\proxyX_1)]_{j'} \big],$$ and note that by Chebyshev's inequality and the previous result, for any $\epsilon>0$, $$\limsup_{N \to \infty} \mathbb{P} \Big( \Big| \frac{1}{N} \sum_{i=1}^N W_i^2 [\dot{l}_{\thetaTarg} (\goodX_i)]_j [ \dot{l}_{\gammaTarg} (\proxyX_i)]_{j'} - \mu_{N} \Big| > \epsilon \Big) \leq 0,$$ and hence $N^{-1}\sum_{i=1}^N W_i^2 [\dot{l}_{\thetaTarg} (\goodX_i)]_j [ \dot{l}_{\gammaTarg} (\proxyX_i)]_{j'} -\mu_N \xrightarrow{p} 0.$ Next note that by Fact \ref{fact:WsqFormula}, $W_1^2 = \sum_{k=1}^K c_k^2 (W_1^{(k)})^2$, so by a previous expression, $$\mu_{N} =  \sum_{k=1}^K c_k^2 \e \big[ (W_1^{(k)})^2 [\dot{l}_{\thetaTarg} (\goodX_1)]_j [ \dot{l}_{\gammaTarg} (\proxyX_1)]_{j'} \big].$$ Taking the limit as $N \to \infty$ of each side of the above equation, by Corollary \ref{cor:SimplifyHigherOrderExpectationsOneTerm},  $$\lim_{N \to \infty} \mu_{N}  =  \sum_{k=1}^K c_k^2 \lim_{N \to \infty} \e \big[ (W_1^{(k)})^2 [\dot{l}_{\thetaTarg} (\goodX_1)]_j [ \dot{l}_{\gammaTarg} (\proxyX_1)]_{j'} \big] = \sum_{k=1}^K c_k^2 \e \Bigg[ \frac{[\dot{l}_{\thetaTarg} (\goodX)]_j [ \dot{l}_{\gammaTarg} (\proxyX)]_{j'}}{\LimitingLabelRulePi{1:k}(\proxyX)} \Bigg]=[\Sigma_{12}]_{jj'},$$ where the last step follows from \eqref{eq:GradientAsympCovFormulae}. 

Clearly, as a consequence $ \mu_{N}=[\Sigma_{12}]_{jj'}+o_p(1)$, so combining this with a previous result $$\begin{aligned} \frac{1}{N} \sum_{i=1}^N W_i^2 [\dot{l}_{\thetaTarg} (\goodX_i)]_j [ \dot{l}_{\gammaTarg} (\proxyX_i)]_{j'} & = \Big( \frac{1}{N} \sum_{i=1}^N W_i^2 [\dot{l}_{\thetaTarg} (\goodX_i)]_j [ \dot{l}_{\gammaTarg} (\proxyX_i)]_{j'} - \mu_N \Big) +(\mu_{N}-[\Sigma_{12}]_{jj'}) + [\Sigma_{12}]_{jj'}
\\ & = o_p(1)+o_p(1)+[\Sigma_{12}]_{jj'}=[\Sigma_{12}]_{jj'}+o_p(1). \end{aligned}$$ 
Since the above convergence in probability holds for any fixed $j,j' \in [d]$, $$\frac{1}{N} \sum_{i=1}^N W_i^2 \dot{l}_{\thetaTarg} (\goodX_i) [ \dot{l}_{\gammaTarg} (\proxyX_i)]^\tran \xrightarrow{p} \Sigma_{12}.$$ Thus recalling the definition of $\hat{\Sigma}_{12}$ at \eqref{eq:HessianAndSigmaEstimators} and applying Lemma \ref{lemma:AsymptoticallyIgnoreEstimationErrorInAverages} with $r=2$, $s=1$, and $s'=2$, and $g(\argxProxy)=1$ for all $\argxProxy \in \xProxySpace$, $$\hat{\Sigma}_{12}= \frac{1}{N} \sum_{i=1}^N W_i^2 \dot{l}_{\htc} (\goodX_i) [ \dot{l}_{\hga} (\proxyX_i)]^\tran = \frac{1}{N} \sum_{i=1}^N W_i^2 \dot{l}_{\thetaTarg} (\goodX_i) [ \dot{l}_{\gammaTarg} (\proxyX_i)]^\tran +o_p(1) \xrightarrow{p} \Sigma_{12}.$$ Hence we have shown that $\hat{\Sigma}_{12} \xrightarrow{p} \Sigma_{12}$. Analogous arguments show that $\hat{\Sigma}_{11} \xrightarrow{p} \Sigma_{11}$ and $\hat{\Sigma}_{22} \xrightarrow{p} \Sigma_{22}$.

\end{proof}

\subsection{Proof of Proposition \ref{prop:AsymptoticallyValidCIs}}

Suppose $\hat{\Omega}$ is given by \eqref{eq:HatOmegaContinuousFunctionChoice} for some fixed $f: (\mathbb{R}^{d \times d})^7 \to \mathbb{R}^{d \times d}$ that does not depend on $N$ and that is continuous at $(\Sigma_{11}, \Sigma_{12}, \Sigma_{22}, \Sigma_{13}, \Sigma_{33}, H_{\thetaTarg},H_{\gammaTarg} )$. Define $$\Omega_f \equiv f \big(\Sigma_{11}, \Sigma_{12}, \Sigma_{22}, \Sigma_{13}, \Sigma_{33}, H_{\thetaTarg},H_{\gammaTarg} \big).$$ Note that by the continuous mapping theorem and Propositions \ref{prop:ConsistencyOfHessians}, \ref{prop:ConsistencyOfSig13andSig33}, and \ref{prop:ConsistencyOfSigmaUpperLeft}, as $N \to \infty$, $$\hat{\Omega}= f \big(\hat{\Sigma}_{11}, \hat{\Sigma}_{12}, \hat{\Sigma}_{22}, \hat{\Sigma}_{13}, \hat{\Sigma}_{33}, \hat{H}_{\thetaTarg},\hat{H}_{\gammaTarg} \big) \xrightarrow{p} f \big(\Sigma_{11}, \Sigma_{12}, \Sigma_{22}, \Sigma_{13}, \Sigma_{33}, H_{\thetaTarg},H_{\gammaTarg} \big)=\Omega_f.$$ Thus by the above result, by Propositions \ref{prop:ConsistencyOfHessians}, \ref{prop:ConsistencyOfSig13andSig33}, and \ref{prop:ConsistencyOfSigmaUpperLeft}, and by the definitions for $\hat{\Sigma}^{\PTDSuperScriptAcr}$ and the matrix-valued function $\Sigma^{\PTDSuperScriptAcr}(\cdot)$ at \eqref{eq:PTDAsympVarEstimator} and \eqref{eq:AsympVarPTDFormula},  $\hat{\Sigma}^{\PTDSuperScriptAcr} \xrightarrow{p} \Sigma^{\PTDSuperScriptAcr}(\Omega_f)$ as $N \to \infty$. Hence, defining $\Sigma_f^{\PTDSuperScriptAcr}\equiv \Sigma^{\PTDSuperScriptAcr}(\Omega_f)$, where $\Sigma^{\PTDSuperScriptAcr}(\cdot)$ is the function defined at \eqref{eq:AsympVarPTDFormula}, establishes the first claim that $\hat{\Sigma}^{\PTDSuperScriptAcr} \xrightarrow{p} \Sigma_f^{\PTDSuperScriptAcr}$ for this choice of $\Sigma_f^{\PTDSuperScriptAcr} \in \mathbb{R}^{d \times d}$.

Next note that since $\hat{\Omega} \xrightarrow{p} \Omega_f$ and $\Sigma_f^{\PTDSuperScriptAcr} = \Sigma^{\PTDSuperScriptAcr}(\Omega_f)$, we can apply Theorem \ref{theorem:CLT_PTDEstimator} to show that $\sqrt{N} \big( \htPTD -\thetaTarg \big) \xrightarrow{d} \mathcal{N} \big(0,\Sigma_f^{\PTDSuperScriptAcr} \big)$. Now fix $j \in [d]$ such that $[\Sigma_f^{\PTDSuperScriptAcr}]_{jj} \neq 0$. Also fix $\alpha \in (0,1)$ and let $z_{1-\alpha/2}$ be the $(1-\alpha/2)$-quantile of a standard Gaussian distribution. By Slutsky's lemma, since $\sqrt{N} \big( \htPTD_j -[\thetaTarg]_j \big) \xrightarrow{d} \mathcal{N} \big(0,[\Sigma_f^{\PTDSuperScriptAcr}]_{jj} \big)$ and $\hat{\Sigma}_{jj}^{\PTDSuperScriptAcr} \xrightarrow{p} [\Sigma_f^{\PTDSuperScriptAcr}]_{jj} \neq 0$, $$(\hat{\Sigma}_{jj}^{\PTDSuperScriptAcr} )^{-1/2}  \cdot \sqrt{N} \big( \htPTD_j -[\thetaTarg]_j \big) \xrightarrow{d} \mathcal{N} \big(0,1 \big).$$ Thus, recalling from \eqref{eq:CIDef} that $\mathcal{C}_j^{(1-\alpha)} \equiv \Big[ \htPTD_j - z_{1-\alpha/2} \sqrt{\hat{\Sigma}_{jj}^{\PTDSuperScriptAcr}/N},\htPTD_j + z_{1-\alpha/2} \sqrt{\hat{\Sigma}_{jj}^{\PTDSuperScriptAcr} /N} \Big],$ observe

$$ \lim_{N \to \infty} \mathbb{P} \big( [\thetaTarg]_j \in \mathcal{C}_j^{(1-\alpha)} \big)  =  \lim_{N \to \infty} \mathbb{P} \Big( (\hat{\Sigma}_{jj}^{\PTDSuperScriptAcr} )^{-1/2}  \cdot \sqrt{N} \big( \htPTD_j -[\thetaTarg]_j \big) \in \big[-z_{1-\alpha/2},z_{1-\alpha/2} \big]  \Big)=1-\alpha,$$ where the second step follows by definition of convergence in distribution, the symmetry of the standard Gaussian $\mathcal{N}(0,1)$, and the continuity of the CDF of $\mathcal{N}(0,1)$.

\section{General two-phase multiwave sampling}\label{sec:HowToGeneralizeToTwoPhaseMultiwave}

While the exposition of the paper and our theoretical results focus on \samplingSchemeName{}, in this appendix we demonstrate how the theoretical results still apply to more general (not necessarily proxy-assisted) two-phase multiwave sampling settings. Notably, the theoretical results in this paper can be used to establish theoretical guarantees for M-estimators, including generalized raking-type estimators, in more general two-phase multiwave sampling settings provided that each wave of Phase II involves (conditionally) independent Bernoulli sampling.

In two-phase multiwave sampling designs, an investigator collects cheap-to-measure variables $X^{\textnormal{(I)}}$ from $N$ samples (or subjects) in Phase I, and then in Phase II measures more expensive variables $X^{\textnormal{(II)}}$ on a subset of the $N$ samples. Phase II is broken up into $K$ waves and the sampling scheme for each wave can depend on the $X^{\textnormal{(II)}}$ data from previous waves as well as the $X^{\textnormal{(I)}}$ data from Phase I. Throughout this appendix, we suppose that in each wave of Phase II, $X^{\textnormal{(II)}}$ measurements are collected according to independent Bernoulli sampling, conditionally on the data from Phase I and previous waves. We will refer to these two-phase multiwave sampling designs studied in this appendix as two-phase multi-Bernoulli-wave sampling designs.

The only difference between such two-phase multiwave sampling designs and the proxy-assisted ones we study in the main text is that in the former, the vector $X^{\textnormal{(I)}}$ need not contain components that are predictions (or estimates) of $X^{\textnormal{(II)}}$. At first glance, two-phase multiwave sampling designs may appear to be a strictly more general setting; however, they can always be reformulated as a \samplingSchemeName{} design by introducing a noninformative proxy. In particular, under any two-phase multiwave sampling design, one can define $\tilde{X}^{\textnormal{(II)}}=0 \in \mathbb{R}^{p_2}$ always, where $p_2$ is the dimension of the random vector $X^{\textnormal{(II)}}$. Then letting $\xobs=X^{\textnormal{(I)}}$, $\txmiss=\tilde{X}^{\textnormal{(II)}}=0$, $\xmiss= X^{\textnormal{(II)}}$, $\datvecraw= (X^{\textnormal{(I)}},\tilde{X}^{\textnormal{(II)}},X^{\textnormal{(II)}})$, $\goodX = (X^{\textnormal{(I)}},X^{\textnormal{(II)}})$, and $\proxyX = (X^{\textnormal{(I)}},\tilde{X}^{\textnormal{(II)}})$, the two-phase multiwave sampling design can be written as a special case of a \samplingSchemeName{} design. In this reformulation, the noninformative $\txmiss$ are not used at all to determine the labelling probabilities $\LabelRulePi{k}(\proxyX_i)$ for $i \in [N]$ and $k \in [K]$.

Using this reformulation, our theoretical results directly provide guarantees under two-phase multi-Bernoulli-wave sampling for estimators that use only the Phase II data (Appendix \ref{sec:General2PhaseIPWType}). Moreover, under additional regularity conditions, we use them to prove asymptotic linearity and asymptotic normality of generalized raking-type estimators (Appendix \ref{sec:GeneralizedRakingTheory}).

\subsection{Theoretical guarantees for Phase II-only estimator}\label{sec:General2PhaseIPWType}

When reformulating a two-phase multi-Bernoulli-wave sampling design as a \samplingSchemeName{} design with a noninformative proxy $\txmiss=0$ and when setting $\xobs=X^{\textnormal{(I)}}$, $\xmiss= X^{\textnormal{(II)}}$, and $\goodX = (X^{\textnormal{(I)}},  X^{\textnormal{(II)}})$, one can conduct M-estimation tasks using the estimator $\htc$ given in Equation \eqref{eq:PTDComponentEstimatorsDEF} of the main text. This corresponds to the Multiwave Predict-Then-Debias estimator $\htPTD$ for the case where the tuning matrix is set as $\hat{\Omega}=0$. In this case, our theoretical results establishing asymptotic linearity (Theorem \ref{theorem:AsymptoticLInearMestsStacked}) and asymptotic normality (Theorem \ref{theorem:CLT_PTDEstimator}) of $\htc=\htPTD$ still hold, and the corresponding confidence intervals are still asymptotically valid under additional regularity conditions (Proposition \ref{prop:AsymptoticallyValidCIs}). Moreover, when setting $\hat{\Omega}=0$ (i.e., when using the estimator $\htc$), regularity conditions involving $\gammaTarg$, the proxy loss $l_{\theta}(\proxyX)$, and the proxy population loss $\tilde{L}(\theta)$ in Assumptions \ref{assump:SmoothEnoughForAsymptoticLineariaty}--\ref{assump:SmoothEnoughForConsitentVarEst} need not be verified. We note that while this estimator has theoretical guarantees, it is not the most efficient estimator because it does not directly leverage the Phase I data.

\subsection{Theoretical guarantees for generalized raking-type estimators}\label{sec:GeneralizedRakingTheory}

Generalized raking is an approach for leveraging both Phase I and Phase II data that was originally developed to estimate means (or population totals) \citep{DelvilleAndSarndalGeneralizedRaking,DelvilleSarndalSautory93GeneralizedRaking} and can more generally be used to estimate regression coefficients \citep{BreslowEtAlRaking,LumleyGeneralizedRaking2011}. In this appendix, we present a slight modification of the generalized raking procedure for M-estimation tasks and develop theoretical guarantees for this procedure. Our procedure differs from the typical generalized raking procedure in the sense that it uses multiwave inverse probability weights (defined in Equation \eqref{eq:aggregated_Wk}) rather than inverse probability weights as the design weights. We then show that under two-phase multi-Bernoulli-wave sampling designs and under suitable regularity conditions, the estimator from our generalized raking-type procedure is asymptotically linear and asymptotically normal.

\subsubsection{Generalized raking with multiwave inverse probability weights}\label{sec:GeneralizedRakingWithMultiwaveIPWWeights}

Suppose $N$ i.i.d. samples with measurements of $X^{\textnormal{(I)}} \in \mathbb{R}^{p^{\textnormal{(I)}}}$, denoted by $(X_i^{\textnormal{(I)}})_{i=1}^N$, are collected in Phase I and that the corresponding Phase II variables $X^{\textnormal{(II)}} \in \mathbb{R}^{p^{\textnormal{(II)}}}$ from the $N$ samples, denoted by $(X_i^{\textnormal{(II)}})_{i=1}^N$, are initially unobserved but are measured across multiple measurement waves. For each $i \in [N]$, let $\goodX_i=(X_i^{\textnormal{(I)}},X_i^{\textnormal{(II)}})$, and $\proxyX_i=\big(X_i^{\textnormal{(I)}},h(X_i^{\textnormal{(I)}}) \big)$ for some fixed, prespecified function $h : \mathbb{R}^{p^{\textnormal{(I)}}} \to \mathbb{R}^{p^{\textnormal{(II)}}}$. (In our motivating \samplingSchemeName{} applications, estimates of $X_i^{\textnormal{(II)}}$ are measured in Phase I, and $h$ selects the corresponding indices of these estimates so that $h(X_i^{\textnormal{(I)}})$ is an estimate of $X_i^{\textnormal{(II)}}$ for each $i \in [N]$). For convenience we will let $(\goodX,\proxyX)$ denote an independent copy of $(\goodX_i,\proxyX_i)$, let $p=p^{\textnormal{(I)}}+p^{\textnormal{(II)}}$ and let $\mathcal{I}^{(K)} = \{i \in [N] \ : \ I_i=1\}$ be the set of indices of samples for which $X_i^{\textnormal{(II)}}$ is fully observed.

In the M-estimation tasks we consider, the goal is to estimate $\thetaTarg= \argmin_{\theta \in \Theta} \e[ l_{\theta}(\goodX)]$ for some prespecified loss function $l_{\theta} : \mathbb{R}^{p} \to \mathbb{R}$ parameterized by $\theta \in \Theta$. 
One way to estimate $\thetaTarg$ is to let $\hat{\theta}$ be the solution to $N^{-1} \sum_{i=1}^N W_i \dot{l}_{\theta}(\goodX_i)=0$ (which equals $\htc$ from Equation \eqref{eq:PTDComponentEstimatorsDEF} under suitable regularity conditions). 

The idea in generalized raking is to use auxiliary variables to improve efficiency beyond that of $\htc$. This is done by prespecifying auxiliary variables that can be derived from Phase I data and modifying the weights for the Phase II sample such that the weighted average of the auxiliary variable across the Phase II sample equals its average across all $N$ Phase I samples. Subject to this constraint, the modified weights are set to be as close to the original weights as possible, in the sense of a prespecified distance function.

For the auxiliary variables, we consider a fixed function $r: \mathbb{R}^{p^{\textnormal{(I)}}} \to \mathbb{R}^{p'}$, and let $A_i \equiv r(X_i^{\textnormal{(I)}})$ for each $i \in [N]$  and $A \equiv r(X^{\textnormal{(I)}})$ be an independent copy of $A_i$. (One can always add nuisance variables to $X_i^{\textnormal{(I)}}$ to incorporate auxiliary variables that are not strictly necessary for estimating $\thetaTarg$). 
Letting $G : \mathbb{R} \times \mathbb{R} \to \mathbb{R}$ denote a prespecified distance function, we study the following two-step estimation procedure: \begin{enumerate}
    \item Set $(W_1^*,\dots,W_N^*)$ to be the $(w_1',\dots,w_N')$ that solve the optimization problem $$\text{minimize} \quad  \sum_{i \in \ci^{(K)}} G(w_i',W_i) \quad \text{subject to} \quad  \sum_{i \in \ci^{(K)}} w_i' A_i =  \sum_{i=1}^N A_i.$$
    \item Set $\hat{\theta}^{\text{GR}}$ to be the solution to the weighted estimating equation $$\sum_{i \in \mathcal{I}^{(K)}} W_i^* \dot{l}_{\theta}(\goodX_i)=0.$$
\end{enumerate} $\hat{\theta}^{\text{GR}}$ is a generalized raking-type estimator \citep{DelvilleAndSarndalGeneralizedRaking,DelvilleSarndalSautory93GeneralizedRaking,LumleyGeneralizedRaking2011}. We remark that the construction of $\hat{\theta}^{\text{GR}}$ differs slightly from that typically seen in the generalized raking literature because in Step 1 above, the multiwave inverse probability weights $W_i$ (defined at \eqref{eq:aggregated_Wk}) rather than the more commonly seen inverse probability weights, are used as design weights. Nonetheless, the Lagrange multiplier arguments in \cite{DelvilleAndSarndalGeneralizedRaking} that are used to rewrite Step 1 hold regardless of whether $W_i$ is an inverse probability weight or a multiwave inverse probability weight. 

We consider the class of distance functions printed in Table 1 of \cite{DelvilleAndSarndalGeneralizedRaking} (reproduced in Table \ref{tab:distance-functions}). For the class of distance functions considered in \cite{DelvilleAndSarndalGeneralizedRaking}, there exists a function $F$ (also presented in Table 1 of \cite{DelvilleAndSarndalGeneralizedRaking}) such that $W_i^*$ has a simplified expression that can be derived using Lagrange multipliers. We will refer to $F$ as the ``calibration function''. For all $i \in [N]$, $$W_i^* = W_i F(\hat{\lambda}^\tran A_i) \quad \text{where } \hat{\lambda} \text { solves } \sum_{i \in \mathcal{I}^{(K)}} W_i F(\hat{\lambda}^\tran A_i) A_i =\sum_{i=1}^N A_i.$$ All 5 calibration functions $F$ in Table 1 of \cite{DelvilleAndSarndalGeneralizedRaking} satisfy $F(0)=1$ and $F'(0)=1$. Observing that for all $i \notin \ci^{(K)}$, $W_i=0$, $\hat{\lambda}$ and $\hat{\theta}^{\text{GR}}$ must satisfy \begin{equation}\label{eq:EstimatingEquationsSplit} \sum_{i =1}^N  W_i F(\hat{\lambda}^\tran A_i)  \dot{l}_{\hat{\theta}^{\text{GR}}}(\goodX_i)=0 \quad \text{and} \quad
\sum_{i=1}^N \big(W_i F(\hat{\lambda}^\tran A_i)-1 \big) A_i =0  .\end{equation}

\begin{table}\caption{  \small Examples of distance functions $G: \mathbb{R} \times \mathbb{R} \to \mathbb{R}$ with the associated calibration functions $F: \mathbb{R} \to \mathbb{R}$. This table is reproduced from \cite{DelvilleAndSarndalGeneralizedRaking} using notation convenient to us.}
\centering
\footnotesize

\label{tab:distance-functions}
\begin{tabular}{ccc}
\toprule
Case & $G(w',w)$ & $F(u)$ \\
\midrule
1 & $(w'-w)^2/(2w)$          & $1 + u$            \\[1.4ex]
2 & $w' \log(w'/w) - w' + w$           & $\exp(u)$          \\[1.4ex]
3 & $2\left(\sqrt{w'}-\sqrt{w}\right)^2$ & $(1 - u/2)^{-2}$ \\[1.4ex]
4 & $-w\log(w'/w) + w' - w$          & $(1 - u)^{-1}$     \\[1.4ex]
5 & $(w'-w)^2/(2w')$          & $(1 - 2u)^{-1/2}$  \\[1.4ex]
\bottomrule
\end{tabular}
\end{table}


We can re-express the $(\hat{\theta}^{\text{GR}}, \hat{\lambda})$ satisfying \eqref{eq:EstimatingEquationsSplit} as follows. Define for each $\eta=(\theta,\lambda)$, $\psi_{\eta} : \mathbb{R}^{p'} \times \mathbb{R}^{p} \to \mathbb{R}^{d+p'}$ to be the function given by $$\psi_{\eta}(a,\argx) \equiv F(\lambda^\tran a)\begin{bmatrix}
    \dot{l}_{\theta}(\argx) \\ a
\end{bmatrix}.$$  Also define for each $\eta=(\theta,\lambda)$, $$\Psi_N(\theta,\lambda) \equiv \frac{1}{N} \sum_{i=1}^N  \begin{bmatrix} W_i F(\lambda^\tran A_i) \dot{l}_{\theta}(\goodX_i) \\  W_i F(\lambda^\tran A_i) A_i -A_i
\end{bmatrix}=\frac{1}{N} \sum_{i=1}^N W_i \psi_{\eta}(A_i,\goodX_i)-\frac{1}{N} \sum_{i=1}^N \begin{bmatrix}
    0 \\ A_i
\end{bmatrix}.$$ Observe that $(\hat{\theta}^{\text{GR}}, \hat{\lambda})$ is the solution to $\Psi_N(\theta,\lambda) =0.$ 

\subsubsection{Derivation of asymptotic linearity of $\hat{\theta}^{\text{GR}}$}

We next show that under suitable regularity conditions, $\hat{\theta}^{\text{GR}}$ is asymptotically linear. For brevity, we assume consistency of $\hat{\theta}^{\text{GR}}$ and $\hat{\lambda}$: \begin{assumption}[Consistency of generalized raking]\label{assump:ConsistencyOfRakingEstimators} \
$\hat{\theta}^{\text{GR}} \xrightarrow{p} \thetaTarg$ and $\hat{\lambda} \xrightarrow{p} 0$ as $N \to \infty$.
\end{assumption}

We also introduce an assumption on the smoothness of $\theta \mapsto l_{\theta}(x)$ and on $F$ that ensures that the remainder term in a 2nd order Taylor expansion of $\Psi_N(\hat{\theta},\hat{\lambda})$ about $(\thetaTarg,0)$ is negligible asymptotically. To specify this smoothness assumption, for each $j \in [d+p']$ define $\psi_{\eta}^{(j)} : \mathbb{R}^{p'} \times \mathbb{R}^{p} \to \mathbb{R}$ to be the $j$th component of $\psi_{\eta}$ given by $\psi_{\eta}^{(j)}(a,\argx)=e_j^\tran \psi_{\eta}(a,\argx)$. For all $(a,\argx)$ we use $\ddot{\psi}_{\eta}^{(j)}(a,\argx) \in \mathbb{R}^{(d+p') \times (d+p')}$ to denote the Hessian of the map $\eta \mapsto \psi_{\eta}^{(j)}(a,\argx)$. 

\begin{assumption}[Regularity conditions for generalized raking]\label{assump:SmoothEnoughForGeneralizedRaking} \

\begin{enumerate}[(i)] 
\item The second order partial derivatives of each component of $\eta \mapsto \psi_{\eta}(A,\goodX)$ are locally bounded in a neighborhood of $\eta_0=(\thetaTarg,0)$ by integrable functions of $(A,\goodX)$. In particular, for each $j,j',j'' \in [d+p']$, there exists a $\delta_{j,j',j''}>0$ and a function $M_{j,j',j''}: \mathbb{R}^{p'} \times \mathbb{R}^{p} \to \mathbb{R}$ satisfying $\e[M_{j,j',j''}(A,\goodX)]<\infty$, such that if $\vert \vert \eta -\eta_0 \vert \vert_2 < \delta_{j,j',j''}$, $\vert [\ddot{\psi}_{\eta}^{(j)}(A,\goodX)]_{j'j''} \vert \leq M_{j,j',j''}(A,\goodX)$ almost surely.
\item \textit{Finite 2nd moments:} $\e[\vert \vert \ddot{l}_{\thetaTarg}(X) \vert \vert_F^2] < \infty$,  $\e[\vert \vert \dot{l}_{\thetaTarg}(X) A^\tran \vert \vert_F^2] < \infty$, $\e[\vert \vert A A^\tran \vert \vert_F^2] < \infty$, $\e[\vert \vert \dot{l}_{\thetaTarg}(X) \vert \vert_2^2] < \infty$, and $\e[\vert \vert A \vert \vert_2^2] < \infty$.
\item The matrices $\e[\ddot{l}_{\thetaTarg}(\goodX)]$ and $\e[A A^\tran]$ are invertible. Moreover, $\e[\ddot{l}_{\thetaTarg}(\goodX)]=H_{\thetaTarg}$, where $H_{\thetaTarg}$ is the Hessian of $\e[l_{\thetaTarg}(\goodX)]$.
\item $\e[\dot{l}_{\thetaTarg}(\goodX)]=0$.
\end{enumerate}
\end{assumption}

Using theoretical results from Appendices \ref{sec:PropertiesOfWeights} and \ref{sec:AsymptoticsForMestimators}, we next show that $\hat{\theta}^{\text{GR}}$ is asymptotically linear under these assumptions. To state the asymptotic linear expansion it is convenient to define \begin{equation}\label{eq:RakingHessianComponents} H_{11} \equiv \e[\ddot{l}_{\thetaTarg}(\goodX)], \quad H_{12} \equiv \e[\dot{l}_{\thetaTarg}(\goodX) A^\tran]  \quad \text{and} \quad H_{22} \equiv \e[AA^\tran].\end{equation}

\begin{theorem}\label{theorem:AsympLineartyOfGeneralizedRakingEstimators} Under two-phase multi-Bernoulli-wave sampling and Assumptions \ref{assump:IIDUnderlyingData}, \ref{assump:LabellingRuleOverlap}, \ref{assump:ConsistencyOfRakingEstimators}, and \ref{assump:SmoothEnoughForGeneralizedRaking}, the generalized raking estimator (based on any distance function considered in Table \ref{tab:distance-functions} and on reference weights $W_i$ defined at \eqref{eq:aggregated_Wk}) is asymptotically linear in the sense that $$\sqrt{N} (\hat{\theta}^{\textnormal{GR}}-\thetaTarg)= - \frac{1}{\sqrt{N}} \sum_{i=1}^N \Big( W_i H_{\thetaTarg}^{-1} \dot{l}_{\thetaTarg}(\goodX_i) + (1-W_i) H_{\thetaTarg}^{-1} H_{12} H_{22}^{-1} A_i  \Big) +o_p(1).$$
    
\end{theorem}
\begin{proof}

Define $\dot{\Psi}_N(\theta,\lambda)$ to be the Jacobian of $\Psi_N(\theta,\lambda)$ with respect to $(\theta,\lambda)$ which is given by $$\dot{\Psi}_N(\theta,\lambda) = \frac{1}{N} \sum_{i=1}^N W_i \begin{bmatrix}  F(\lambda^\tran A_i) \ddot{l}_{\theta}(\goodX_i) & F'(\lambda^\tran A_i) \dot{l}_{\theta}(\goodX_i) A_i^\tran \\  
    0 & F'(\lambda^\tran A_i) A_i A_i^\tran
\end{bmatrix}.$$ Also define $\lambda_0 \equiv 0$ and $\eta_0 \equiv (\thetaTarg,\lambda_0)$. By applying the mean value version of Taylor's theorem to each component of $\Psi_N(\theta,\lambda)$ in an expansion about $(\thetaTarg,\lambda_0)$, \begin{equation}\label{eq:TaylorFirstStep} \Psi_N(\hat{\theta}^{\text{GR}},\hat{\lambda})=\Psi_N(\thetaTarg,\lambda_0) +  \dot{\Psi}_N(\theta_0,\lambda_0) \begin{bmatrix}
     \hat{\theta}^{\text{GR}} -\thetaTarg \\ \hat{\lambda}-\lambda_0
 \end{bmatrix} + R_N \begin{bmatrix}
     \hat{\theta}^{\text{GR}} -\thetaTarg \\ \hat{\lambda}-\lambda_0
 \end{bmatrix},\end{equation} where $R_N$ is a remainder matrix such that for each $j \in [d+p']$, the $j$th row of $R_N$ is given by $$[R_N]_{j \cdot} = \frac{1}{2} \begin{bmatrix}
     \hat{\theta}^{\text{GR}} -\thetaTarg \\ \hat{\lambda}-\lambda_0
 \end{bmatrix}^\tran \Big(\frac{1}{N}\sum_{i=1}^N W_i \ddot{\psi}_{\tilde{\eta}^{(j)}}^{(j)} (A_i,\goodX_i) \Big),$$ for some $\tilde{\eta}^{(j)}$ on the line segment between $\eta_0=(\thetaTarg,\lambda_0)$ and $\hat{\eta}=(\hat{\theta}^{\text{GR}},\hat{\lambda})$. Let $\tilde{\eta}^{(j)}$ be such quantities and let $\vert \vert  \cdot \vert \vert_2$ and $\vert \vert  \cdot \vert \vert_F$ denote the spectral norm and Frobenius norms of a matrix, respectively. Note that for each $j \in [d+p']$ by the Cauchy-Schwarz inequality, $$\begin{aligned} \vert \vert  [R_N]_{j\cdot} \vert \vert_2 & \leq \frac{1}{2} \vert \vert \hat{\eta}-\eta_0 \vert \vert_2 \cdot \Big| \Big| \frac{1}{N}\sum_{i=1}^N W_i \ddot{\psi}_{\tilde{\eta}^{(j)}}^{(j)} (A_i,\goodX_i)  \Big| \Big|_2 
\\ & \leq   \frac{1}{2} \vert \vert \hat{\eta}-\eta_0 \vert \vert_2 \cdot \Big| \Big| \frac{1}{N}\sum_{i=1}^N W_i \ddot{\psi}_{\tilde{\eta}^{(j)}}^{(j)} (A_i,\goodX_i) \Big| \Big|_F
\\ & =   \frac{1}{2} \vert \vert \hat{\eta}-\eta_0 \vert \vert_2 \cdot \sqrt{\sum_{j'=1}^{d+p'} \sum_{j''=1}^{d+p'} \Big(\frac{1}{N}\sum_{i=1}^N W_i \big[ \ddot{\psi}_{\tilde{\eta}^{(j)}}^{(j)} (A_i,\goodX_i) \big]_{j'j''}\Big)^2} 
 \\ & \leq   \frac{b^{-K}}{2} \vert \vert \hat{\eta}-\eta_0 \vert \vert_2 \cdot \sqrt{\sum_{j'=1}^{d+p'} \sum_{j''=1}^{d+p'} \Big(\frac{1}{N}\sum_{i=1}^N \big| \big[ \ddot{\psi}_{\tilde{\eta}^{(j)}}^{(j)} (A_i,\goodX_i) \big]_{j'j''}\big| \Big)^2} 
\end{aligned}$$ where the last step above follows from Fact \ref{fact:WeightBound} that $W_i \in [0,b^{-K}]$. Next let $E_N$ be the event that $\vert \vert \hat{\eta}-\eta_0 \vert \vert_2 \leq \delta_{j,j',j''}$ for each $j,j',j'' \in [d+p']$. Under the event $E_N$, $\vert \vert \tilde{\eta}^{(j)}-\eta_0 \vert \vert_2 \leq \delta_{j,j',j''}$ for all $j,j',j'' \in [d+p']$. Hence by Assumption \ref{assump:SmoothEnoughForGeneralizedRaking}(i), under the event $E_N$,

$$\begin{aligned}
\vert \vert  [R_N]_{j\cdot} \vert \vert_2 & \leq   \frac{b^{-K}}{2} \vert \vert \hat{\eta}-\eta_0 \vert \vert_2 \cdot \sqrt{\sum_{j'=1}^{d+p'} \sum_{j''=1}^{d+p'} \Big(\frac{1}{N}\sum_{i=1}^N  M_{j,j',j''}(A_i,X_i) \Big)^2} 
\\ & = \frac{b^{-K}}{2} \vert \vert \hat{\eta}-\eta_0 \vert \vert_2 \cdot \sqrt{\sum_{j'=1}^{d+p'} \sum_{j''=1}^{d+p'} \Big( \e[ M_{j,j',j''}(A,X)] +o_p(1) \Big)^2},
\end{aligned}$$  for each $j \in [d+p']$.
Above the final step holds by the law of large numbers and the fact that by Assumption \ref{assump:SmoothEnoughForGeneralizedRaking}(i), $\e[ \vert M_{j,j',j''}(A,X) \vert] < \infty$. Since $\hat{\eta} \xrightarrow{p} \eta_0$ (by Assumption \ref{assump:ConsistencyOfRakingEstimators}), the above result implies that under the event $E_N$, $\vert \vert [R_N]_{j\cdot} \vert \vert \leq o_p(1)$ for each $j \in [d+p']$. Recalling that $E_N$ is the event that $\vert \vert \hat{\eta}-\eta_0 \vert \vert_2 \leq \delta_{j,j',j''}$ for each $j,j',j'' \in [d+p']$, $\lim_{N \to \infty} \mathbb{P}(E_N)=1$ because $\hat{\eta} \xrightarrow{p} \eta_0$ (again by Assumption \ref{assump:ConsistencyOfRakingEstimators}). Since $\lim_{N \to \infty} \mathbb{P}(E_N)=1$ and since under the event $E_N$, $\vert \vert [R_N]_{j\cdot} \vert \vert \leq o_p(1)$ for each $j \in [d+p']$, $R_N=o_p(1)$.

Next recalling the formulas for $\Psi_N(\theta,\lambda)$ and $\dot{\Psi}_N(\theta,\lambda)$ and noting that $\lambda_0=0$ so $F(\lambda_0^\tran A)=F(0)=1$ and $F'(\lambda_0^\tran A)=F'(0)=1$, $$\dot{\Psi}_N(\thetaTarg,\lambda_0) = \frac{1}{N} \sum_{i=1}^N W_i \begin{bmatrix}   \ddot{l}_{\thetaTarg}(\goodX_i) &  \dot{l}_{\thetaTarg}(\goodX_i) A_i^\tran \\  
    0 &  A_i A_i^\tran
\end{bmatrix} = \begin{bmatrix}   \e[\ddot{l}_{\thetaTarg}(\goodX)] &  \e[\dot{l}_{\thetaTarg}(\goodX) A^\tran] \\  
    0 &  \e[A A^\tran]
\end{bmatrix} +o_p(1),$$ by Assumption \ref{assump:SmoothEnoughForGeneralizedRaking}(ii) and Corollary \ref{cor:PropertiesOfWeightedAveragees} (II). It is convenient to define $$H_{\eta_0} \equiv \begin{bmatrix}   \e[\ddot{l}_{\thetaTarg}(\goodX)] &  \e[\dot{l}_{\thetaTarg}(\goodX) A^\tran] \\  
    0 &  \e[A A^\tran]
\end{bmatrix} = 
\begin{bmatrix}  H_{\thetaTarg} & H_{12} \\  
    0 &  H_{22}
\end{bmatrix}
\quad \text{and note that } H_{\eta_0}^{-1} = \begin{bmatrix}   H_{\thetaTarg}^{-1} &   -H_{\thetaTarg}^{-1} H_{12} H_{22}^{-1} \\  
    0 &  H_{22}^{-1}
\end{bmatrix},$$ where the inverses exist by Assumption \ref{assump:SmoothEnoughForGeneralizedRaking}(iii). The above formula simplifies to $\dot{\Psi}_N(\thetaTarg,\lambda_0) = H_{\eta_0}+o_p(1)$.

Recall that $(\hat{\theta}^{\text{GR}}, \hat{\lambda})$ is the solution to $\Psi_N(\theta,\lambda) =0$. Rearranging terms in Equation \eqref{eq:TaylorFirstStep}, $$-\Psi_N(\thetaTarg,\lambda_0)=\big(\dot{\Psi}_N(\thetaTarg,\lambda_0) +R_N \big) (\hat{\eta}-\eta_0).$$ Left multiplying each side by $\sqrt{N}\big(\dot{\Psi}_N(\thetaTarg,\lambda_0) + R_N \big)^{-1}$ and noting $R_N=o_p(1)$ and $\dot{\Psi}_N(\thetaTarg,\lambda_0) = H_{\eta_0}+o_p(1)$, $$\sqrt{N} (\hat{\eta}-\eta_0)= - \big(\dot{\Psi}_N(\thetaTarg,\lambda_0) +R_N \big)^{-1} \sqrt{N} \Psi_N(\thetaTarg,\lambda_0) = - \big(H_{\eta_0} +o_p(1) \big)^{-1} \sqrt{N} \Psi_N(\thetaTarg,\lambda_0)$$

Plugging in the formula for $\Psi_N(\theta,\lambda)$ and recalling that $\lambda_0=0$ so $F(\lambda_0^\tran A)=F(0)=1$ and $F'(\lambda_0^\tran A)=F'(0)=1$, the above expressions can be written as \begin{equation}\label{eq:RakingExpansionStillNuisanceTerms} \sqrt{N} \begin{bmatrix}
    \hat{\theta}^{\text{GR}}-\thetaTarg \\ \hat{\lambda}-0
\end{bmatrix} = - \big( H_{\eta_0} +o_p(1) \big)^{-1} \underbrace{\Big(\frac{1}{\sqrt{N}} \sum_{i=1}^N W_i \begin{bmatrix} \dot{l}_{\thetaTarg}(\goodX_i) \\   A_i 
\end{bmatrix} -  \frac{1}{\sqrt{N}} \sum_{i=1}^N \begin{bmatrix} 0 \\   A_i
\end{bmatrix} \Big) }_{=O_p(1)}.\end{equation}

To show that the rightmost term above is indeed $O_p(1)$, note that since $A_i=r(X_i^{\textnormal{(I)}})$ and $\proxyX_i=\big(X_i^{\textnormal{(I)}},h(X_i^{\textnormal{(I)}}) \big)$ always for fixed deterministic functions $r: \mathbb{R}^{p^{\textnormal{(I)}}} \to \mathbb{R}^{p'}$ and $h: \mathbb{R}^{p^{\textnormal{(I)}}} \to \mathbb{R}^{p^{\textnormal{(II)}}}$, there exists a deterministic function $\tilde{r} : \mathbb{R}^p \to \mathbb{R}^{p'}$ such that $\tilde{r}(\proxyX_i)=A_i$ always and $\tilde{r}(\proxyX)=A$ always. Letting $\tilde{r}(\cdot)$ be such a function, by Corollary \ref{cor:PropertiesOfWeightedAveragees} (III) and Assumption \ref{assump:SmoothEnoughForGeneralizedRaking}(ii), for each $j \in [d]$ and $j' \in [p']$, $$\var \big(\sqrt{N} \cdot N^{-1} \sum_{i=1}^N W_i e_j^\tran \dot{l}_{\thetaTarg}(\goodX_i) \big) \leq b^{-2K} \e\big[ \big(e_j^\tran \dot{l}_{\thetaTarg}(\goodX) \big)^2 \big] \leq b^{-2K} \e[ \vert \vert \dot{l}_{\thetaTarg}(\goodX) \vert \vert_2^2] < \infty \quad  \text{and}$$ $$\var \big( N^{-1/2} \sum_{i=1}^N W_i e_{j'}^\tran A_i \big)= \var \big(\sqrt{N} \cdot N^{-1} \sum_{i=1}^N W_i e_{j'}^\tran \tilde{r}(\proxyX_i) \big) \leq b^{-2K} \e\big[ \big(e_{j'}^\tran \tilde{r}(\proxyX) \big)^2 \big] \leq b^{-2K} \e[ \vert \vert A \vert \vert_2^2] < \infty.$$ Since each component of $N^{-1/2} \sum_{i=1}^N W_i \dot{l}_{\thetaTarg}(\goodX_i)$ and $N^{-1/2} \sum_{i=1}^N W_i A_i$ has variance that is bounded above by a finite number, and so does $N^{-1/2} \sum_{i=1}^N A_i$, the rightmost term in \eqref{eq:RakingExpansionStillNuisanceTerms} (marked by a bracket) has variance upper bounded by some finite number, regardless of the size of $N$. Moreover, for each $N$ the rightmost term in \eqref{eq:RakingExpansionStillNuisanceTerms} (marked by a bracket) has mean $0$ (by Proposition \ref{prop:WeightsEnableUnbiasedMeanEst} and Assumption \ref{assump:SmoothEnoughForGeneralizedRaking}(iv)). Combining these results, the quantity on the right hand side of Equation \eqref{eq:RakingExpansionStillNuisanceTerms} is indeed $O_p(1)$.

Recalling the formula for $H_{\eta_0}^{-1}$, we can thus simplify Equation \eqref{eq:RakingExpansionStillNuisanceTerms} as follows, $$\begin{aligned} \sqrt{N} \begin{bmatrix}
    \hat{\theta}^{\text{GR}}-\thetaTarg \\ \hat{\lambda}-0
\end{bmatrix} & = - \big( H_{\eta_0}^{-1} +o_p(1) \big) \Big(\frac{1}{\sqrt{N}} \sum_{i=1}^N W_i \begin{bmatrix} \dot{l}_{\thetaTarg}(\goodX_i) \\   A_i 
\end{bmatrix} -  \frac{1}{\sqrt{N}} \sum_{i=1}^N \begin{bmatrix} 0 \\   A_i
\end{bmatrix} \Big)
\\ & = - \begin{bmatrix}   H_{\thetaTarg}^{-1} &   -H_{\thetaTarg}^{-1} H_{12} H_{22}^{-1} \\  
    0 &  H_{22}^{-1}
\end{bmatrix} \Big(\frac{1}{\sqrt{N}} \sum_{i=1}^N W_i \begin{bmatrix} \dot{l}_{\thetaTarg}(\goodX_i) \\   A_i 
\end{bmatrix} -  \frac{1}{\sqrt{N}} \sum_{i=1}^N \begin{bmatrix} 0 \\   A_i
\end{bmatrix} \Big) +o_p(1) O_p(1).
\end{aligned}
$$
Considering the first $d$ components in the above expansion and simplify we proved the desired result that $$\sqrt{N} (\hat{\theta}^{\text{GR}}-\thetaTarg)= - \frac{1}{\sqrt{N}} \sum_{i=1}^N \Big( W_i H_{\thetaTarg}^{-1} \dot{l}_{\thetaTarg}(\goodX_i) + (1-W_i) H_{\thetaTarg}^{-1} H_{12} H_{22}^{-1} A_i  \Big) +o_p(1).$$
\end{proof}

\subsubsection{Derivation of asymptotic normality of $\hat{\theta}^{\text{GR}}$}\label{sec:RakingAsymptoticNormality}


To characterize the asymptotic variance of $\hat{\theta}^{\text{GR}}$ it is helpful to define $$\Sigma_{\Delta,A} \equiv  \sum_{k=1}^K c_k^2 \e\Bigg[ \frac{\big[ \dot{l}_{\thetaTarg}(\goodX) - H_{12} H_{22}^{-1} A \big] \big[ \dot{l}_{\thetaTarg}(\goodX) - H_{12} H_{22}^{-1} A \big]^\tran }{\LimitingLabelRulePi{1:k}(\proxyX)}  \Bigg]$$ and $$\Sigma_{\dot{l},A} \equiv \cov \big(\dot{l}_{\thetaTarg}(\goodX), H_{12} H_{22}^{-1} A \big),$$ and to recall from \eqref{eq:RakingHessianComponents} that $H_{12} \equiv \e[\dot{l}_{\thetaTarg}(\goodX) A^\tran]$ and $H_{22} \equiv \e[AA^\tran]$.

\begin{theorem}\label{theorem:asymptoticNormalityRaking}
    In the setting of Theorem \ref{theorem:AsympLineartyOfGeneralizedRakingEstimators}, if Assumption \ref{assump:CLTRgularityConditions} also holds and if for some $\delta>0$, $\e[\vert e_j^\tran A \vert^{2+\delta}] < \infty$ for all $j \in [p']$, then $$ \sqrt{N} (\hat{\theta}^{\textnormal{GR}}-\thetaTarg) \xrightarrow{d} \mathcal{N} \Big(0,  H_{\thetaTarg}^{-1} \big( \Sigma_{\Delta,A} + \Sigma_{\dot{l},A}+ \Sigma_{\dot{l},A}^\tran  -  H_{12} H_{22}^{-1} H_{12}^\tran \big) H_{\thetaTarg}^{-1} \Big).$$
\end{theorem}

\begin{proof}
    
This proof follows an analogous argument to the proof of Theorem \ref{theorem:CLT_PTDEstimator}. In particular, By Theorem \ref{theorem:AsympLineartyOfGeneralizedRakingEstimators}, $$\begin{aligned} \sqrt{N} (\hat{\theta}^{\textnormal{GR}}-\thetaTarg) & = - \frac{1}{\sqrt{N}} \sum_{i=1}^N \Big( W_i H_{\thetaTarg}^{-1} \dot{l}_{\thetaTarg}(\goodX_i) + (1-W_i) H_{\thetaTarg}^{-1} H_{12} H_{22}^{-1} A_i  \Big) +o_p(1) 
\\ & = - \frac{1}{\sqrt{N}} \sum_{i=1}^N \Big( \bar{W}_i H_{\thetaTarg}^{-1} \dot{l}_{\thetaTarg}(\goodX_i) + (1-\bar{W}_i) H_{\thetaTarg}^{-1} H_{12} H_{22}^{-1} A_i  \Big) +o_p(1).
\end{aligned}$$ Above the second step follows from replacing each $W_i$ term with $\bar{W}_i+(W_i-\bar{W}_i)$, rearranging terms, and by applying Proposition \ref{prop:CanReduceToIIDWeights} to each coordinate to show that the terms $N^{-1/2} \sum_{i=1}^N (W_i-\bar{W}_i) \dot{l}_{\thetaTarg}(X_i)=o_p(1)$ and $N^{-1/2} \sum_{i=1}^N (W_i-\bar{W}_i) A_i=o_p(1)$. Since $(\bar{W}_i,X_i,A_i)_{i=1}^N$ are i.i.d., by the standard multivariate central limit theorem and Slutsky's lemma, $$\sqrt{N} (\hat{\theta}^{\textnormal{GR}}-\thetaTarg) \xrightarrow{d} \mathcal{N}(0, \Sigma^{\text{GR}}), \quad \text{where} \quad \Sigma^{\text{GR}} \equiv \var \Big( \bar{W}_1 H_{\thetaTarg}^{-1} \dot{l}_{\thetaTarg}(\goodX_1) + (1-\bar{W}_1) H_{\thetaTarg}^{-1} H_{12} H_{22}^{-1} A_1  \Big).$$ To complete the proof note that by Lemma \ref{lemma:PropertiesOfIndependentVersionsOfWeights} (III), $\Sigma^{\text{GR}}$ simplifies as follows $$\begin{aligned} \Sigma^{\text{GR}}  & = \var \Big( \bar{W}_1 (H_{\thetaTarg}^{-1} \dot{l}_{\thetaTarg}(\goodX_1) - H_{\thetaTarg}^{-1} H_{12} H_{22}^{-1} A_1 ) \Big) + \var \big( H_{\thetaTarg}^{-1} H_{12} H_{22}^{-1} A_1 \big)
\\ & \quad + \cov \Big(H_{\thetaTarg}^{-1} \dot{l}_{\thetaTarg}(\goodX_1)-H_{\thetaTarg}^{-1} H_{12} H_{22}^{-1} A_1, H_{\thetaTarg}^{-1} H_{12} H_{22}^{-1} A_1 \Big) 
\\ & \quad +\cov \Big(H_{\thetaTarg}^{-1} H_{12} H_{22}^{-1} A_1, H_{\thetaTarg}^{-1} \dot{l}_{\thetaTarg}(\goodX_1)-H_{\thetaTarg}^{-1} H_{12} H_{22}^{-1} A_1 \Big) 
\\ & = \var \big( \bar{W}_1 (H_{\thetaTarg}^{-1} \dot{l}_{\thetaTarg}(\goodX_1) - H_{\thetaTarg}^{-1} H_{12} H_{22}^{-1} A_1 ) \big) + H_{\thetaTarg}^{-1} \big( \Sigma_{\dot{l},A}+ \Sigma_{\dot{l},A}^\tran -  H_{12} H_{22}^{-1} \var(A) H_{22}^{-1} H_{12}^\tran \big) H_{\thetaTarg}^{-1} 
\\ & = H_{\thetaTarg}^{-1} \Bigg( \sum_{k=1}^K c_k^2 \e\Bigg[ \frac{\big[ \dot{l}_{\thetaTarg}(\goodX) - H_{12} H_{22}^{-1} A \big] \big[ \dot{l}_{\thetaTarg}(\goodX) - H_{12} H_{22}^{-1} A \big]^\tran }{\LimitingLabelRulePi{1:k}(\proxyX)}  \Bigg] - H_{12} H_{22}^{-1} \e[A] \e[A^\tran]  H_{22}^{-1} H_{12}^\tran \Bigg) H_{\thetaTarg}^{-1} 
\\ & \quad + H_{\thetaTarg}^{-1} \Big( \Sigma_{\dot{l},A}+ \Sigma_{\dot{l},A}^\tran -  H_{12} H_{22}^{-1} \big(\e[AA^\tran]-\e[A] \e[A]^\tran \big) H_{22}^{-1} H_{12}^\tran \Big) H_{\thetaTarg}^{-1}  
\\ & = H_{\thetaTarg}^{-1} \big( \Sigma_{\Delta,A} + \Sigma_{\dot{l},A}+ \Sigma_{\dot{l},A}^\tran  -  H_{12} H_{22}^{-1} H_{12}^\tran \big) H_{\thetaTarg}^{-1}.
\end{aligned}$$ Above the third equality follows from applying Lemma \ref{lemma:PropertiesOfIndependentVersionsOfWeights} (IV) and the final step follows from the definition of $H_{22} \equiv \e[AA^\tran]$ and the definition of $\Sigma_{\Delta,A}$.
\end{proof}

\section{Finding efficient labelling strategies}\label{sec:EfficientLabellingStrategyAppendix}

In this appendix we first frame the task of minimizing asymptotic variance of a single component of the Multiwave Predict-Then-Debias estimator (using only the first $k^*$ waves) as an optimization problem. 
We then present a solution to a functional optimization problem to derive the greedy optimal sampling rule. We next discuss how we estimate a parameter that enforces the budget constraint that appears in the formula for the greedy optimal sampling rule.
We also give further details on the adaptive procedure that uses prespecified strata to determine the labelling rule, and show that this strata-based procedure satisfies Assumption \ref{assump:CLTRgularityConditions}(ii) under suitable conditions.


Throughout this appendix we fix $j \in [d]$ and focus on sampling strategies designed to reduce the asymptotic variance $\htPTD_j$. It will be convenient to define $\psi_j : \xSpace \times \xProxySpace \to \mathbb{R}$ to be a function given by \begin{equation}\label{eq:psi_j_DefSimplifyAsympCovjj}
  \psi_j(\argx,\argxProxy) \equiv \Big(e_j^\tran  H_{\thetaTarg}^{-1} \dot{l}_{\thetaTarg}(\argx) -e_j^\tran H_{\gammaTarg}^{-1} \dot{l}_{\gammaTarg}(\argxProxy) \Big)^2 
\quad \text{for each} \quad  \argx \in \xSpace ,\argxProxy \in \xProxySpace.
\end{equation}

\subsection{Formula for the greedy optimal objective}\label{sec:FormulaForGreedyOptimalObjective}

Fix $k^* \in [K]$ and we will derive a formula for the asymptotic variance of a variant of $\htPTD_j$ that uses Phase I data and only the first $k^*$ waves of Phase II data. We consider the case where $\hat{\Omega}=I_{d \times d}$ is the tuning matrix.
To state the asymptotic variance of $\htPTD_j$ when $\hat{\Omega}=I_{d \times d}$, it helps to define $$\Sigma_c \equiv  \begin{bmatrix}
    H_{\thetaTarg}^{-1} \\ -H_{\gammaTarg}^{-1} 
\end{bmatrix}^\tran
\begin{bmatrix}
   0 & \Sigma_{13} \\
    \Sigma_{13}^\tran & \Sigma_{33}
\end{bmatrix} 
\begin{bmatrix}
    H_{\thetaTarg}^{-1} \\ -H_{\gammaTarg}^{-1} 
\end{bmatrix}.$$ 

In the setting of Theorem \ref{theorem:CLT_PTDEstimator}, for any fixed $c_1,\dots,c_K \in [0,1]$, the asymptotic variance of $\htPTD_j$ is given by  
$$\big[\Sigma^{\PTDSuperScriptAcr} (I_{d \times d}) \big]_{jj} = \sum_{k=1}^{K} c_{k}^2 \e \Bigg[ \frac{\psi_j(\goodX,\proxyX)}{\LimitingLabelRulePi{1:k}(\proxyX)} \Bigg] -[\Sigma_c]_{jj},$$ as a consequence of Theorem \ref{theorem:CLT_PTDEstimator} and definition \eqref{eq:psi_j_DefSimplifyAsympCovjj}. Note that by construction of $\htPTD$ (see Equations \eqref{eq:aggregated_Wk}, \eqref{eq:PTDComponentEstimatorsDEF} and \eqref{eq:PTD_estimator}), the estimator $\htPTD$ will disregard Phase II data from waves after the $k^*$th wave if and only if $c_k=0$ for each $k>k^*$. Hence a version of the estimator $\htPTD$ that uses only the first $k^*$ waves of Phase II data uses fixed, prespecified wave-specific weights that satisfy $c^{(k^*)}_1,\dots,c^{(k^*)}_K \in [0,1]$, $\sum_{k=1}^K c^{(k^*)}_k=1$, and $c^{(k^*)}_k=0$ for $k>k^*$. Thus, by the formula displayed above, a version of the estimator $\htPTD$ that uses the Phase I data and only the first $k^*$ waves of Phase II data (or that assumes that the $k^*$th wave is the last wave) has asymptotic variance of the $j$th component that satisfies $$\big[\Sigma^{\PTDSuperScriptAcr,(1:k^*)} (I_{d \times d}) \big]_{jj} = \sum_{k=1}^{k^*-1} (c^{(k^*)}_{k})^2 \cdot \e \Bigg[ \frac{\psi_j(\goodX,\proxyX)}{\LimitingLabelRulePi{1:k}(\proxyX)} \Bigg] + (c^{(k^*)}_{k^*})^2 \cdot \e \Bigg[ \frac{\psi_j(\goodX,\proxyX)}{\LimitingLabelRulePi{1:k^*}(\proxyX)} \Bigg] +0 -[\Sigma_c]_{jj}.$$ Thus if we define $$\kappa_1 \equiv c^{(k^*)}_{k^*} \quad \text{and} \quad \kappa_2 \equiv \sum_{k=1}^{k^*-1} (c^{(k^*)}_{k})^2 \cdot \e \Bigg[ \frac{\big( \big[ H_{\thetaTarg}^{-1} \dot{l}_{\thetaTarg}(\goodX) -H_{\gammaTarg}^{-1} \dot{l}_{\gammaTarg}(\proxyX) \big]_j \big)^2 }{\LimitingLabelRulePi{k}(\proxyX) \prod_{k'=1}^{k-1} \big(1-\LimitingLabelRulePi{k'}(\proxyX) \big)}  \Bigg]-[\Sigma_c]_{jj},$$ by definitions \eqref{eq:PiBarProd1tok} and \eqref{eq:psi_j_DefSimplifyAsympCovjj} and the previous result \begin{equation*} \big[\Sigma^{\PTDSuperScriptAcr,(1:k^*)} (I_{d \times d}) \big]_{jj} = \kappa_1^2 \cdot \e \Bigg[ \frac{\big( \big[  H_{\thetaTarg}^{-1} \dot{l}_{\thetaTarg}(\goodX) - H_{\gammaTarg}^{-1} \dot{l}_{\gammaTarg}(\proxyX) \big]_j \big)^2 }{ \LimitingLabelRulePi{k^*}(\proxyX) \prod_{k=1}^{k^*-1} \big(1-\LimitingLabelRulePi{k}(\proxyX) \big)
} \Bigg] + \kappa_2.\end{equation*} This confirms Equation \eqref{eq:AsympVarUntunedFirstKstarWaves} from the main text. Note that by \eqref{eq:GradientAsympCovFormulae} the matrix $\Sigma_c$ does not depend on $\LimitingLabelRulePi{k^*}$, and moreover, $\{c_k^{(k^*)}\}_{k=1}^K$ are prespecified constants that do not depend on $\LimitingLabelRulePi{k^*}$ either. Hence $\kappa_1$ and $\kappa_2$ are quantities that do not depend on the function $\LimitingLabelRulePi{k^*} \in \mathcal{P}$.

Prior to the start of wave $k^*$, the investigator uses a labelling strategy $\LabelStrategy{k^*}$ to choose a labelling rule $\LabelRulePi{k^*}=\LabelStrategy{k^*}(\cd_{k^*-1}) \in \mathcal{P}$. The choice of labelling strategy $\LabelStrategy{k^*}$ will influence $\LimitingLabelRulePi{k^*}$ (since $\LimitingLabelRulePi{k^*}$ is the asymptotic limit of $\LabelRulePi{k^*}$), but crucially the choice of labelling strategy for the $k^*$th wave does not impact the functions $\LimitingLabelRulePi{k}$ for $k< k^*$. Thus to find a labelling strategy for the $k^*$th wave that minimizes $\big[\Sigma^{\PTDSuperScriptAcr,(1:k^*)} (I_{d \times d}) \big]_{jj}$ we must find the function $\LimitingLabelRulePi{k^*}$ that minimizes $\big[\Sigma^{\PTDSuperScriptAcr,(1:k^*)} (I_{d \times d}) \big]_{jj}$ under fixed values for the functions $\{\LimitingLabelRulePi{k}\}_{k=1}^{k^*-1}$. By the previous result and since $\kappa_1$ and $\kappa_2$ do not depend on $\LimitingLabelRulePi{k^*}$, we suppose that the investigator seeking to minimize $\big[\Sigma^{\PTDSuperScriptAcr,(1:k^*)} (I_{d \times d}) \big]_{jj}$ will choose the labelling strategy $\LabelStrategy{k^*}$ with the goal of achieving a near optimal solution to the optimization problem 
\begin{equation}\label{eq:OptimizationProblemForGreedyApproach}
\begin{split}
    \text{find a} & \quad \quad \bar{\pi}^{(k^*)}: \mathbb{R}^p \to (0,\infty) \\
    \text{minimizing} & \quad \quad \e \Bigg[ \frac{\big( \big[  H_{\thetaTarg}^{-1} \dot{l}_{\thetaTarg}(\goodX) - H_{\gammaTarg}^{-1} \dot{l}_{\gammaTarg}(\proxyX) \big]_j \big)^2}{\LimitingLabelRulePi{k^*}(\proxyX)  \prod_{k=1}^{k^*-1} \big(1-\LimitingLabelRulePi{k}(\proxyX) \big)} \Bigg]  \\  
    \text{subject to} & \quad \quad \e[\bar{\pi}^{(k^*)}(\proxyX)] \leq B^{(k^*)} \quad \text{and} \quad \bar{\pi}^{(k^*)}(\argxProxy) \in [b_l,b_u] \quad \text{for } \argxProxy \in \xProxySpace.
\end{split}
\end{equation} Above $B^{(k^*)}$ imposes a budget constraint on the expected number of labels that can be collected in the $k^*$th wave and $b_l,b_u \in (0,1)$ are user-specified overlap bounds that constrain the search to labelling rules that will satisfy Assumption \ref{assump:LabellingRuleOverlap}.

\subsection{A solution to the optimization problem}\label{sec:OptimizationProblem}

In this appendix we establish that as stated in Equation \eqref{eq:OptimalGreedySolutionWithBounds}, if $0< b_l < B^{(k^*)} < b_u <1$, the solution to the optimization problem given in \eqref{eq:OptimizationProblemForGreedyApproach} is given by \begin{equation*}   
    \bar{\pi}_{\text{opt}}^{(k^*)} (\proxyX)  = \min \Biggl\{ b_u, \max \Bigl\{ b_l, \sqrt{\frac{\varrho_j(\proxyX)}{ \lambda^{(k^*)} \prod_{k=1}^{k^*-1} \big(1-\LimitingLabelRulePi{k}(\proxyX) \big)} } \Bigr\} \Biggr\},
\end{equation*}
where $\lambda^{(k^*)}>0$ is an existing constant such that $\e[\bar{\pi}_{\text{opt}}^{(k^*)} (\proxyX)]=B^{(k^*)}$. Recalling from \eqref{eq:varrhoDef} that $\varrho_j(\proxyX) \equiv \e \big[  \big( \big[  H_{\thetaTarg}^{-1} \dot{l}_{\thetaTarg}(\goodX) - H_{\gammaTarg}^{-1} \dot{l}_{\gammaTarg}(\proxyX) \big]_j \big)^2  \big| \proxyX \big]$ this is the direct result of applying the following lemma in the case where $D=B^{(k^*)}$ and $$g(\goodX,\proxyX) =  \frac{\big( \big[  H_{\thetaTarg}^{-1} \dot{l}_{\thetaTarg}(\goodX) - H_{\gammaTarg}^{-1} \dot{l}_{\gammaTarg}(\proxyX) \big]_j \big)^2}{ \prod_{k=1}^{k^*-1} \big(1-\LimitingLabelRulePi{k}(\proxyX) \big)} .$$ 

\begin{lemma}\label{lemma:ImprovementOnCSOptimalityDeriv}
    Fix any $b_l,b_u,D$ such that $0<b_l < D < b_u < 1$ and a nonnegative function $g: \xProxySpace \times \xProxySpace \to [0,\infty)$ such that $\e[g(X,\tilde{X})]< \infty$ and $g(X,\tilde{X})>0$ almost surely, and define $h(\tilde{X}) \equiv \e[g(X,\tilde{X})\giv \tilde{X}]$. A function $\pi : \xProxySpace \to \mathbb{R}$ that minimizes $\e[g(X,\tilde{X})/\pi(\tilde{X})]$ subject to the constraints that (i) $\e[\pi(\tilde{X})] \leq D$ and (ii) $\pi(\tilde{X}) \in [b_l,b_u]$ always, is given by $$\pi^*(\tilde{X})= \min \big\{ b_u, \max \{ b_l, \sqrt{h(\tilde{X})/\lambda^* } \} \big\},$$ where $\lambda^*$ is a some existing $\lambda>0$ such that $\e \big[\min \big\{ b_u, \max \{ b_l, \sqrt{h(\tilde{X})/\lambda } \} \big\} \big]=D$.
\end{lemma}

\begin{proof}
    For each $\lambda>0$, let $$\pi_{\lambda}(\tilde{X}) = \min \big\{ b_u, \max \{ b_l, \sqrt{h(\tilde{X})/\lambda } \} \big\}.$$ Observe that $\lambda \mapsto \e[\pi_{\lambda}(\tilde{X})]$ is monotone nonincreasing and since $\pi_{\lambda}(\tilde{X})$ is dominated by $b_u$, one can check that $\lambda \mapsto \e[\pi_{\lambda}(\tilde{X})]$ is continuous by the dominated convergence theorem. Moreover, $\lim_{\lambda \downarrow 0} \e[\pi_{\lambda}(\tilde{X})]=b_u$ while $\lim_{\lambda \uparrow \infty} \e[\pi_{\lambda}(\tilde{X})]=b_l$. Since $b_l<D<b_u$, by the intermediate value theorem and the previous results, there exists a $\lambda>0$ such that $\e[\pi_{\lambda}(\tilde{X})]=D$. Letting $\lambda^*$ be such a $\lambda>0$, we set $$\pi^*(\tilde{X})= \min \big\{ b_u, \max \{ b_l, \sqrt{h(\tilde{X})/\lambda^* } \} \big\}.$$ By the choice of $\lambda^*$, $\e[\pi^*(\tilde{X})]=D$ and moreover $b_l \leq \pi^*(\tilde{X}) \leq b_u$ always by the above construction.

    Next we will show that, \begin{equation}\label{eq:IntermidiateStepFindOptimal}
        \frac{h(\argxProxy)}{p} + \lambda^* p \geq  \frac{h(\argxProxy)}{\pi^*(\argxProxy)} + \lambda^*  \pi^*(\argxProxy) \quad \text{for any} 
    \quad p \in [b_l,b_u] \text{ and any }\argxProxy \in \xProxySpace. \end{equation} To do this fix $\tilde{x} \in \xProxySpace$. Note that since $h(\tilde{x}) \geq 0$ and $\lambda^*>0$, $r \mapsto  h(\argxProxy)/r + \lambda^* r$ is a convex function on $(0,\infty)$ that is minimized when $r= \sqrt{h(\argxProxy)/\lambda^*}$. Now consider three cases 
    
    \begin{enumerate}
        \item \textbf{If $\sqrt{h(\argxProxy)/\lambda^*} \leq b_l$:} $r \mapsto  h(\argxProxy)/r + \lambda^* r$ is nondecreasing on $[b_l,b_u]$. This implies that $\inf_{r \in [b_l,b_u]} \{ h(\argxProxy)/r + \lambda^* r \} = h(\argxProxy)/b_l + \lambda^* b_l$. Moreover in this case, $\pi^*(\tilde{x})=b_l$, so for any  $p \in [b_l,b_u]$, $$ \frac{h(\argxProxy)}{p} + \lambda^* p \geq \inf_{r \in [b_l,b_u]} \Bigl\{ \frac{h(\argxProxy)}{r} + \lambda^* r \Bigr\} = \frac{h(\argxProxy)}{\pi^*(\argxProxy)} + \lambda^*  \pi^*(\argxProxy).$$

     \item \textbf{If $\sqrt{h(\argxProxy)/\lambda^*} \geq b_u$:} $r \mapsto  h(\argxProxy)/r + \lambda^* r$ is nonincreasing on $[b_l,b_u]$. This implies that $\inf_{r \in [b_l,b_u]} \{ h(\argxProxy)/r + \lambda^* r \} = h(\argxProxy)/b_u + \lambda^* b_u$. Moreover in this case, $\pi^*(\tilde{x})=b_u$, so for any $p \in [b_l,b_u]$, $$ \frac{h(\argxProxy)}{p} + \lambda^* p \geq \inf_{r \in [b_l,b_u]} \Bigl\{ \frac{h(\argxProxy)}{r} + \lambda^* r \Bigr\} = \frac{h(\argxProxy)}{\pi^*(\argxProxy)} + \lambda^*  \pi^*(\argxProxy).$$

      \item \textbf{If $b_l<\sqrt{h(\argxProxy)/\lambda^*} < b_u$:} $r \mapsto  h(\argxProxy)/r + \lambda^* r$ on $[b_l,b_u]$ is minimized when $r= \sqrt{h(\argxProxy)/\lambda^*}$ so $\inf_{r \in [b_l,b_u]} \{ h(\argxProxy)/r + \lambda^* r \} = 2 \sqrt{h(\tilde{x}) \lambda^*}$. Moreover in this case, $\pi^*(\tilde{x})=\sqrt{h(\argxProxy)/\lambda^*}$, so for any $p \in [b_l,b_u]$, $$ \frac{h(\argxProxy)}{p} + \lambda^* p \geq \inf_{r \in [b_l,b_u]} \Bigl\{ \frac{h(\argxProxy)}{r} + \lambda^* r \Bigr\} = 2 \sqrt{h(\tilde{x}) \lambda^*} = \frac{h(\argxProxy)}{\pi^*(\argxProxy)} + \lambda^*  \pi^*(\argxProxy).$$
    \end{enumerate} Thus, regardless of the value of $\sqrt{h(\argxProxy)/\lambda^*}$ (whether it is in $[0,b_l]$, $(b_l,b_u)$, or $[b_u, \infty)$) we have shown that $$\frac{h(\argxProxy)}{p} + \lambda^* p \geq \frac{h(\argxProxy)}{\pi^*(\argxProxy)} + \lambda^*  \pi^*(\argxProxy) \quad \text{for all } p \in [b_l,b_u]. $$ Since this argument holds for any fixed $\argxProxy \in \xProxySpace$, we have proved inequality \eqref{eq:IntermidiateStepFindOptimal}.

    Now we show that $\pi^*(\cdot)$ is indeed optimal among the stated class of functions. To do this, fix any $\pi : \xProxySpace \to \mathbb{R}$ satisfying $\e[\pi(\tilde{X})] \leq D$ and $\pi(\tilde{X}) \in [b_l,b_u]$ always. Observe that by the tower property $$\begin{aligned} \e \Big[ \frac{g(X,\tilde{X})}{\pi(\tilde{X})}\Big] & = \e \Big[ \frac{\e[g(X,\tilde{X}) \giv \tilde{X}]}{ \pi(\tilde{X})}\Big] 
    \\ & = \e \Big[ \frac{h(\tilde{X})}{ \pi(\tilde{X})}\Big] +\lambda^* D -\lambda^* D  
    \\ & \geq \e \Big[ \frac{h(\tilde{X})}{ \pi(\tilde{X})}\Big] +\lambda^* \cdot \e[\pi(\tilde{X})]-  \lambda^* D
    \\ & = \e \Big[ \frac{h(\tilde{X})}{ \pi(\tilde{X})} +\lambda^* \pi(\tilde{X}) \Big]-  \lambda^* \cdot \e[ \pi^*(\tilde{X})]
    \\ & \geq \e \Big[ \frac{h(\tilde{X})}{ \pi^*(\tilde{X})} +\lambda^* \pi^*(\tilde{X}) \Big]-  \lambda^* \cdot \e[ \pi^*(\tilde{X})] 
    \\ & =  \e \Big[ \frac{h(\tilde{X})}{ \pi^*(\tilde{X})} \Big]
    \\ & = \e \Big[ \frac{\e[g(X,\tilde{X}) \giv \tilde{X}]}{ \pi^*(\tilde{X})} \Big]= \e \Big[ \frac{g(X,\tilde{X})}{\pi^*(\tilde{X})}\Big].\end{aligned}$$
 Above, the fifth step holds by inequality \eqref{eq:IntermidiateStepFindOptimal}, because $\pi(\tilde{X}) \in [b_l,b_u]$ always (by assumption). We have thus shown that $\pi^*(\cdot)$ is the function $\pi : \xProxySpace \to \mathbb{R}$ that minimizes $\e [ g(X,\tilde{X})/\pi(\tilde{X}) ]$ subject to the constraints that $\e[\pi(\tilde{X})] \leq D$ and $\pi(\tilde{X}) \in [b_l,b_u]$ always. 
\end{proof}

\subsection{Calculating parameter to enforce budget constraints}\label{sec:EnforceBudgetConstraints}

In this appendix, we describe how the parameter $\hat{\lambda}^{(k^*)}$ in the estimated greedy optimal labelling rule in Equation \eqref{eq:EstimatedGreedyOptLabellingRule} is calculated in our experiments. In summary, $\hat{\lambda}^{(k^*)}$ is calculated numerically to enforce the budget constraint. 

Let $b_l,b_u \in (0,1)$ and $n_{\text{targ}}^{(k^*)}$ denote the desired, user-specified overlap thresholds (such that $b_l<b_u$) and the number of labels to be collected in wave $k^*$ (in expectation).  Also let \begin{equation}\label{eq:YetToBeLabelledPoints}\mathcal{U}^{(k')} \equiv \bigl\{ i \in [N] \ : \ I_i^{(k)}=0 \ \ \text{for all } k \in [k'-1] \bigr\} \quad \text{for each } k' \in [K] \setminus \{1\} \end{equation} be the set of samples for which the label $\xmiss_i$ has not yet been obtained prior to the start of a wave $k'$. Define for each $\lambda \in (0,\infty)$ the function $\hat{\bar{\pi}}_{\text{opt},\lambda} : \xProxySpace \to \mathbb{R}$ to be given by $$\hat{\bar{\pi}}_{\text{opt},\lambda}^{(k^*)} (\argxProxy) =  \min \Biggl\{ b_u, \max \Bigl\{ b_l, \sqrt{\frac{\hat{\varrho}^{(k^*)}_j(\argxProxy)}{ \lambda \prod_{k=1}^{k^*-1} \big(1-\LabelRulePi{k}(\argxProxy) \big)} } \Bigr\} \Biggr\} \quad \text{ for each } \argxProxy \in \xProxySpace,$$ where $\hat{\varrho}^{(k^*)}_j(\cdot)$ is some estimate for the function $\varrho_j(\cdot)$ based on $\cd_{k^*-1}$. We let $\hat{\lambda}^{(k^*)}$ be the solution to $$ \sum_{i \in \mathcal{U}^{(k^*)}} \hat{\bar{\pi}}_{\text{opt},\lambda}(\proxyX_i) = n_{\text{targ}}^{(k^*)},$$ which can be computed numerically using a univariate root solver because the left hand side is monotone nonincreasing in $\lambda$ and the right hand side is a constant. In our numerical experiments, we use the labelling rule $\hat{\bar{\pi}}_{\text{opt},\hat{\lambda}^{(k^*)}}^{(k^*)}$ (see Equation \eqref{eq:EstimatedGreedyOptLabellingRule}), where $\hat{\lambda}^{(k^*)}$ is the $\lambda \in (0,\infty)$ solving the above equation and is calculated using the \texttt{uniroot} function in \texttt{R}.

\subsection{Formula for estimated optimal strata-specific labelling probabilities}\label{sec:FormulaForEstimatedStrataSpecificLabellingRule}

In this appendix, we derive the stratified estimate of $\varrho_j(\cdot)$ presented in Section \ref{sec:EstimateOptimalLabellingProbabilityInEachStrata} that takes the same value within each stratum. This estimate of $\varrho_j(\cdot)$ uses data from the first $k^*-1$ waves and is used to determine an efficient labelling rule for the $k^*$th wave. As in Section \ref{sec:EstimateOptimalLabellingProbabilityInEachStrata}, we consider a 
partitioning of $\xProxySpace$ into $L$ prespecified strata $\mathcal{S}_1,\dots,\mathcal{S}_L$ such that $\cup_{r=1}^L \mathcal{S}_r = \xProxySpace$ and $\mathcal{S}_r \cap \mathcal{S}_{r'} =\emptyset$ for $r \neq r'$. We also demonstrate that using a greedy optimal sampling rule of Section \ref{sec:DescribeApproximateGreedyOptimalStrategy} with the stratified estimate of $\varrho_j(\cdot)$ presented in Section \ref{sec:EstimateOptimalLabellingProbabilityInEachStrata} yields a labelling rule that converges under suitable conditions.

Recall that by the definitions of $\varrho_j(\cdot)$ and $\psi_j(\cdot,\cdot)$ at \eqref{eq:varrhoDef} and \eqref{eq:psi_j_DefSimplifyAsympCovjj}, $$\varrho_j(\proxyX)=\e[\psi_j(\goodX,\proxyX) \giv \proxyX].$$ We then let $\varrho_j^{(\text{strat})}(\cdot)$ be a coarsening of $\varrho_j(\cdot)$ given by $$\varrho_j^{(\text{strat})}(\argxProxy) = \sum_{r=1}^L \mathbbm{1} \{ \argxProxy \in \mathcal{S}_r \} \cdot  \e [ \psi_j(\goodX,\proxyX) \giv \proxyX \in \mathcal{S}_r ] \quad \text{for } \argxProxy \in \xProxySpace.$$ To estimate $\varrho_j^{(\text{strat})}(\argxProxy)$, we must estimate $\e [ \psi_j(\goodX,\proxyX) \giv \proxyX \in \mathcal{S}_r]$ for each $r \in [L]$.

To do this, recall that $$ \psi_j(\goodX,\proxyX) = \Big(e_j^\tran  H_{\thetaTarg}^{-1} \dot{l}_{\thetaTarg}(\goodX) -e_j^\tran H_{\gammaTarg}^{-1} \dot{l}_{\gammaTarg}(\proxyX) \Big)^2,$$ by \eqref{eq:psi_j_DefSimplifyAsympCovjj}. Further, note that in the setting of Section \ref{sec:EstimateOptimalLabellingProbabilityInEachStrata} for each $i \in [N]$ that has an available label prior to the start of wave $k^*$, we also have an empirical estimate of $\psi_j(\goodX_i,\proxyX_i)$ given by $$\cy_i^{(k^*-1)} = \Big(e_j^\tran  [\hat{H}_{\thetaTarg}^{(k^*-1)}]^{-1} \dot{l}_{\hat{\theta}^{\completeSampNotation,(k^*-1)}}(\goodX_i) -e_j^\tran \hat{H}_{\gammaTarg}^{-1} \dot{l}_{\hga}(\proxyX_i) \Big)^2.$$

Note that for each $r \in [L]$, by Corollary \ref{cor:SimplifyWeightedExpectations1Term},

$$\begin{aligned}
    \e [ \psi_j(\goodX,\proxyX) \giv \proxyX \in \mathcal{S}_r ] & = \big( \mathbb{P} ( \proxyX \in \mathcal{S}_r )\big)^{-1} \e \big[ \mathbbm{1} \{ \proxyX \in \mathcal{S}_r \} \cdot \psi_j(\goodX,\proxyX) \big]
    \\ & = \frac{1} {\mathbb{P} ( \proxyX \in \mathcal{S}_r) }  \cdot   \e \Big[ \sum_{k=1}^{k^*-1} \Big( \frac{c_k}{\sum_{k'=1}^{k^*-1} c_{k'} } \Big) \cdot W_i^{(k)} \mathbbm{1} \{ \proxyX_i \in \mathcal{S}_r \} \cdot \psi_j(\goodX_i,\proxyX_i) \Big]
    \\ & = \frac{1}{ \mathbb{P} ( \proxyX \in \mathcal{S}_r) } \cdot \e \Big[ \cw_i^{(k^*-1)}  \cdot \mathbbm{1} \{ \proxyX_i \in \mathcal{S}_r \} \cdot \big(e_j^\tran  H_{\thetaTarg}^{-1} \dot{l}_{\thetaTarg}(\goodX_i) -e_j^\tran H_{\gammaTarg}^{-1} \dot{l}_{\gammaTarg}(\proxyX_i) \big)^2    \Big],
\end{aligned}$$ where the last step holds by the definition of $\cw_i^{(k^*-1)}$ at \eqref{eq:WeightsAggFirstKstar}. Thus, in this setting, for each $r \in [L]$ we estimate $\e [ \psi_j(\goodX,\proxyX) \giv \proxyX \in \mathcal{S}_r ]$ by $$\frac{ N^{-1} \sum_{i=1}^N \cw_i^{(k^*-1)} \cdot \mathbbm{1} \{ \proxyX_i \in \mathcal{S}_r \} \cdot  \cy_i^{(k^*-1)}   }{ N^{-1} \sum_{i=1}^N \mathbbm{1} \{ \proxyX_i \in \mathcal{S}_r \} }.$$

Plugging in estimates for $\e [ \psi_j(\goodX,\proxyX) \giv \proxyX \in \mathcal{S}_r]$ to the formula for $\varrho_j^{(\text{strat})}(\cdot)$, we obtain the coarsened estimator for $\varrho_j(\cdot)$ given by the piecewise function \begin{equation}\label{eq:varrhoHatStrat}\hat{\varrho}_j^{(\text{strat},k^*)}(\argxProxy) = \sum_{r=1}^L \mathbbm{1} \{ \argxProxy \in \mathcal{S}_r \} \cdot \Bigg(\frac{ N^{-1} \sum_{i=1}^N \cw_i^{(k^*-1)} \cdot \mathbbm{1} \{ \proxyX_i \in \mathcal{S}_r \} \cdot  \cy_i^{(k^*-1)}   }{ N^{-1} \sum_{i=1}^N \mathbbm{1} \{ \proxyX_i \in \mathcal{S}_r \} } \Bigg) \quad \text{for } \argxProxy \in \xProxySpace,\end{equation} that is presented in Section \ref{sec:EstimateOptimalLabellingProbabilityInEachStrata} of the main text.

\subsubsection{Convergence of the stratified variant of the greedy optimal sampling rule}\label{sec:ProofOfStrataLabRuleConvergence}

In this subsection we verify that under suitable conditions, when using the greedy optimal sampling approach described in Section \ref{sec:DescribeApproximateGreedyOptimalStrategy}
(with labelling rules recursively defined by Equation \eqref{eq:EstimatedGreedyOptLabellingRule}), and when using the strategy for estimating $\varrho_j(\cdot)$ described in Section \ref{sec:EstimateOptimalLabellingProbabilityInEachStrata}, Assumption \ref{assump:CLTRgularityConditions}(ii) that the labelling rules converge holds. We start by defining  \begin{equation}\label{eq:varrhoHatStratNumerator}\hat{\varrho}_{j,r}^{(\text{num},k^*)}= \frac{1}{N} \sum_{i=1}^N \cw_i^{(k^*-1)} \cdot \mathbbm{1} \{ \proxyX_i \in \mathcal{S}_r \} \cdot  \cy_i^{(k^*-1)} \quad \text{for each } r \in [L]. \end{equation} The next lemma shows that for each $r \in [L]$, $\hat{\varrho}_{j,r}^{(\text{num},k^*)} \xrightarrow{p} \e \big[ \mathbbm{1} \{ \proxyX \in \mathcal{S}_r \} \cdot \psi_j(\goodX,\proxyX) \big]$ under suitable conditions.

\begin{lemma}\label{lemma:varrhoHatStratConvergence}
    Under \samplingSchemeName{} suppose Assumptions \ref{assump:IIDUnderlyingData}, \ref{assump:LabellingRuleOverlap}, \ref{assump:SmoothEnoughForAsymptoticLineariaty},  \ref{assump:CLTRgularityConditions}(i),  \ref{assump:CLTRgularityConditions}(iii),  and \ref{assump:SmoothEnoughForConsitentVarEst} hold. If the labelling rules for the first $k^*-1$ waves converge in the sense of  Assumption \ref{assump:CLTRgularityConditions}(ii), then $$\hat{\varrho}_{j,r}^{(\textnormal{num},k^*)} \xrightarrow{p} \e \big[ \mathbbm{1} \{ \proxyX \in \mathcal{S}_r \} \cdot \psi_j(\goodX,\proxyX) \big] \quad \text{for each } r \in [L].$$ 
\end{lemma}

\begin{proof}
    
For each $i \in [N]$ consider the following decomposition of  $\cy_i^{(k^*-1)}$: $$\begin{aligned}
    \cy_i^{(k^*-1)} &  = \underbrace{e_j^\tran [\hat{H}_{\thetaTarg}^{(k^*-1)}]^{-1} \dot{l}_{\hat{\theta}^{\completeSampNotation,(k^*-1)}}(\goodX_i) \big[ \dot{l}_{\hat{\theta}^{\completeSampNotation,(k^*-1)}}(\goodX_i) \big]^\tran [\hat{H}_{\thetaTarg}^{(k^*-1)}]^{-1} e_j}_{\equiv T_i^{(1,1)}} 
    \\ & - 2 \underbrace{e_j^\tran [\hat{H}_{\thetaTarg}^{(k^*-1)}]^{-1} \dot{l}_{\hat{\theta}^{\completeSampNotation,(k^*-1)}}(\goodX_i) \big[ \dot{l}_{\hga}(\proxyX_i) \big]^\tran [\hat{H}_{\gammaTarg}]^{-1} e_j}_{\equiv T_i^{(1,2)}} 
    \\ & + \underbrace{e_j^\tran [\hat{H}_{\gammaTarg}]^{-1} \dot{l}_{\hga}(\proxyX_i) \big[ \dot{l}_{\hga}(\proxyX_i) \big]^\tran [\hat{H}_{\gammaTarg}]^{-1} e_j}_{\equiv T_i^{(2,2)}}.
\end{aligned}$$

 Since all Assumptions \ref{assump:IIDUnderlyingData}--\ref{assump:SmoothEnoughForConsitentVarEst} are assumed to hold for the first $k^*-1$ waves used to compute $\hat{H}_{\thetaTarg}^{(k^*-1)}$, by Proposition \ref{prop:ConsistencyOfHessians}, both $\hat{H}_{\thetaTarg}^{(k^*-1)}=H_{\thetaTarg}+o_p(1)$ and $\hat{H}_{\gammaTarg}=H_{\gammaTarg}+o_p(1)$. Next define $$\bar{T}_{r}^{(s,s')} =\frac{1}{N} \sum_{i=1}^N \cw_i^{(k^*-1)} \cdot \mathbbm{1} \{ \proxyX_i \in \mathcal{S}_r \} \cdot T_i^{(s,s')}  \quad \text{for each } s,s' \in \{1,2\}, r \in [L],$$ and we will show that $\bar{T}_{r}^{(s,s')}$ converges in probability to some limit for each $s,s' \in \{1,2\}, r \in [L]$. To do this fix $r \in [L]$, and note that by applying Lemma \ref{lemma:AsymptoticallyIgnoreEstimationErrorInAverages} and then Corollary \ref{cor:PropertiesOfWeightedAveragees}(II) for ($k^*$-1)-wave \samplingSchemeName{}, $$\begin{aligned} \bar{T}_{r}^{(1,1)} & = e_j^\tran [\hat{H}_{\thetaTarg}^{(k^*-1)}]^{-1}  \Bigg(  \frac{1}{N} \sum_{i=1}^N \cw_i^{(k^*-1)} \cdot \mathbbm{1} \{ \proxyX_i \in \mathcal{S}_r \}\dot{l}_{\hat{\theta}^{\completeSampNotation,(k^*-1)}}(\goodX_i) \big[ \dot{l}_{\hat{\theta}^{\completeSampNotation,(k^*-1)}}(\goodX_i) \big]^\tran \Bigg)  [\hat{H}_{\thetaTarg}^{(k^*-1)}]^{-1} e_j  \\ & =  e_j^\tran [\hat{H}_{\thetaTarg}^{(k^*-1)}]^{-1}  \Bigg(  \frac{1}{N} \sum_{i=1}^N \cw_i^{(k^*-1)} \cdot \mathbbm{1} \{ \proxyX_i \in \mathcal{S}_r \}\dot{l}_{\thetaTarg}(\goodX_i) \big[ \dot{l}_{\thetaTarg}(\goodX_i) \big]^\tran +o_p(1) \Bigg)  [\hat{H}_{\thetaTarg}^{(k^*-1)}]^{-1} e_j 
 \\ & =  e_j^\tran [\hat{H}_{\thetaTarg}^{(k^*-1)}]^{-1}  \Bigg( \e\Big[ \mathbbm{1} \{ \proxyX_i \in \mathcal{S}_r \} \cdot \dot{l}_{\thetaTarg}(\goodX) \big[ \dot{l}_{\thetaTarg}(\goodX) \big]^\tran \Big]+o_p(1) \Bigg)  [\hat{H}_{\thetaTarg}^{(k^*-1)}]^{-1} e_j 
 \\ & =  e_j^\tran [H_{\thetaTarg}]^{-1}  \Bigg( \e\Big[ \mathbbm{1} \{ \proxyX_i \in \mathcal{S}_r \} \cdot \dot{l}_{\thetaTarg}(\goodX) \big[ \dot{l}_{\thetaTarg}(\goodX) \big]^\tran \Big] \Bigg)  [H_{\thetaTarg}]^{-1} e_j +o_p(1)
  \\ & =   \e\Big[ \mathbbm{1} \{ \proxyX \in \mathcal{S}_r \} \cdot e_j^\tran H_{\thetaTarg}^{-1}  \dot{l}_{\thetaTarg}(\goodX) \big[ \dot{l}_{\thetaTarg}(\goodX) \big]^\tran H_{\thetaTarg}^{-1} e_j \Big] +o_p(1).
 \end{aligned}$$
Analogous arguments show that $$\bar{T}_{r}^{(1,2)}= \e\Big[ \mathbbm{1} \{ \proxyX \in \mathcal{S}_r \} \cdot e_j^\tran H_{\thetaTarg}^{-1}  \dot{l}_{\thetaTarg}(\goodX) \big[ \dot{l}_{\gammaTarg}(\proxyX) \big]^\tran H_{\gammaTarg}^{-1} e_j \Big] +o_p(1)$$
and $$\bar{T}_{r}^{(2,2)}= \e\Big[ \mathbbm{1} \{ \proxyX \in \mathcal{S}_r \} \cdot e_j^\tran H_{\gammaTarg}^{-1}  \dot{l}_{\gammaTarg}(\proxyX) \big[ \dot{l}_{\gammaTarg}(\proxyX) \big]^\tran H_{\gammaTarg}^{-1} e_j \Big] +o_p(1).$$ Thus combining previous results and definitions, $$\begin{aligned} \hat{\varrho}_{j,r}^{(\text{num},k^*)} & = \frac{1}{N} \sum_{i=1}^N \cw_i^{(k^*-1)} \cdot \mathbbm{1} \{ \proxyX_i \in \mathcal{S}_r \} \cdot  \cy_i^{(k^*-1)}
\\ & =  \frac{1}{N} \sum_{i=1}^N \cw_i^{(k^*-1)} \cdot \mathbbm{1} \{ \proxyX_i \in \mathcal{S}_r \} \cdot (T_i^{(1,1)}-2T_i^{(1,2)}+T_i^{(2,2)})
\\ & = \bar{T}_r^{(1,1)}-2\bar{T}_r^{(1,2)}+\bar{T}_r^{(2,2)}
\\ & = \e \big[ \mathbbm{1} \{ \proxyX \in \mathcal{S}_r \} \cdot \psi_j(\goodX,\proxyX) \big]+o_p(1).
\end{aligned}$$ Noting that this argument holds for each $r \in [L]$ completes the proof.
    
\end{proof}


In the following remark, we show that Assumption 4(ii) that the labelling rules converge in $L^1$ holds when using the greedy optimal sampling approach described in Section \ref{sec:DescribeApproximateGreedyOptimalStrategy} and a coarsened estimate of $\varrho_j$ defined at \eqref{eq:varrhoHatStrat}.

\begin{remark}\label{remark:StratLabellingRuleConverges}
       Suppose Assumptions \ref{assump:IIDUnderlyingData}, \ref{assump:SmoothEnoughForAsymptoticLineariaty},   \ref{assump:CLTRgularityConditions}(iii), and \ref{assump:SmoothEnoughForConsitentVarEst} hold, and that the investigator conducts \samplingSchemeName{} using labelling rules recursively defined by $\LabelRulePi{1}(\argxProxy)=n_{\text{targ}}^{(1)}/N$ for all $\argxProxy \in \xProxySpace$ and by Equation \eqref{eq:EstimatedGreedyOptLabellingRule} with $\hat{\varrho}_j^{(k)}$ defined at \eqref{eq:varrhoHatStrat} and $\hat{\lambda}^{(k)}$ defined by the procedure in Section \ref{sec:EnforceBudgetConstraints} for each $k \in [K] \setminus \{1\}$. Assume an asymptotic regime where for some $\theta_1,\dots,\theta_K \in (0,1)$, $\lim_{N \to \infty} n_{\text{targ}}^{(1)}/N = \theta_1$, and $$\lim_{N \to \infty} n_{\text{targ}}^{(k)}/N = \theta_k \quad \text{and} \quad \lim_{N \to \infty} \mathbb{P} \Big( b_l \vert \mathcal{U}^{(k)} \vert <n_{\text{targ}}^{(k)} < b_u \vert \mathcal{U}^{(k)} \vert \Big)=1 \quad \text{for each } k \in [K] \setminus \{1\},$$ where $\mathcal{U}^{(k)}$ is defined at \eqref{eq:YetToBeLabelledPoints} to be the set of unlabelled observations prior to wave $k$. Then, Assumptions \ref{assump:LabellingRuleOverlap}, \ref{assump:CLTRgularityConditions}(i), and \ref{assump:CLTRgularityConditions}(ii) hold.
\end{remark}

\begin{proof}

By construction for each $k \in [K]$, $\LabelRulePi{k}(\argxProxy) \in [b_l,b_u]$ for each $\argxProxy \in \xProxySpace$ where $0<b_l <b_u <1$, so Assumption \ref{assump:LabellingRuleOverlap} that the labelling probabilities are bounded away from $0$ and $1$ holds. Also note that for each $k^*$ the function $\hat{\varrho}_j^{(\text{strat},k^*)}$ in Equation \eqref{eq:varrhoHatStrat} and the quantity $\hat{\lambda}^{(k^*)}$ defined by the procedure in Section \ref{sec:EnforceBudgetConstraints} would remain unchanged under any permutation of the $N$ Phase I samples, and hence $\LabelRulePi{k^*}$ given by Equation \eqref{eq:EstimatedGreedyOptLabellingRule} would remain unchanged under any permutation of the $N$ Phase I samples (formally this can be argued by strong induction). Hence Assumption \ref{assump:CLTRgularityConditions}(i) that the labelling strategy is symmetric holds. Also observe that by Equations \eqref{eq:EstimatedGreedyOptLabellingRule}, \eqref{eq:varrhoHatStrat}, and the fact that for each $r \in [L]$, $\LabelRulePi{1}(\argxProxy)$ takes the same value for all $\argxProxy \in \mathcal{S}_r$,  for each $r \in [L]$ and $k \in [K]$, $\LabelRulePi{k}(\argxProxy)$ takes the same value for all $\argxProxy \in \mathcal{S}_r$ (this can also formally be argued by strong induction).

It remains to check that Assumption \ref{assump:CLTRgularityConditions}(ii) holds, which we will prove by strong induction. First consider the case where $k^*=1$. Since $\LabelRulePi{1}(\argxProxy)=n_{\text{targ}}^{(1)}/N$ for all $\argxProxy \in \xProxySpace$, letting $\LimitingLabelRulePi{1}(\argxProxy)=\theta_1$ for all $\argxProxy \in \xProxySpace$, $\lim_{N \to \infty} \e \big[ \vert \LabelRulePi{1}(\proxyX_1)-\LimitingLabelRulePi{1}(\proxyX_1) \vert \big]=0$. 

Next fix $k^* \in [K] \setminus \{1\}$, and assume the inductive hypothesis that \begin{equation}\label{eq:InductiveHypothesisForLabellingL1Convergence}
    \lim_{N \to \infty} \e \big[ \vert \LabelRulePi{k}(\proxyX_1)-\LimitingLabelRulePi{k}(\proxyX_1) \vert \big]=0 \quad \text{for each } k \in [k^*-1].
\end{equation} We will show this implies a similar convergence for $\LabelRulePi{k^*}$. 

To do this, recall that for each $k \in [K]$, $\LabelRulePi{k}(\argxProxy)$ is constant within each stratum of $\xProxySpace$, and note that by \eqref{eq:InductiveHypothesisForLabellingL1Convergence}, this further implies that $\LimitingLabelRulePi{k}(\argxProxy)$ is constant within each stratum of $\xProxySpace$ when $k \in [k^*-1]$. (If $\LimitingLabelRulePi{k}(\argxProxy)$ is nonconstant on a set of measure zero it can be redefined to still satisfy \eqref{eq:InductiveHypothesisForLabellingL1Convergence} and be constant within each stratum of $\xProxySpace$). Next for each $r \in [L]$ and $k \in [K]$, define $\check{\pi}_r^{(k)} \in [b_l,b_u]$ such that $$\check{\pi}_r^{(k)} = \LabelRulePi{k}(\argxProxy) \quad \text{for all } \argxProxy \in \mathcal{S}_r,$$ and for each $r \in [L]$ and $k \in [k^*-1]$ define $\bar{\pi}_r^{(k)} \in [b_l,b_u]$ such that $$\bar{\pi}_r^{(k)} = \LimitingLabelRulePi{k}(\argxProxy) \quad \text{for all } \argxProxy \in \mathcal{S}_r.$$ By the inductive hypothesis stated at \eqref{eq:InductiveHypothesisForLabellingL1Convergence}, it follows that for each $r \in [L]$ and $k \in [k^*-1]$, $\check{\pi}_r^{(k)} \xrightarrow{p} \bar{\pi}_r^{(k)}$. Next define for each $r \in [L]$, $$\check{u}_r^{(k^*)} \equiv \prod_{k=1}^{k^*-1} \big(1-\check{\pi}_r^{(k)} \big), \quad \bar{u}_r^{(k^*)} \equiv \prod_{k=1}^{k^*-1} \big(1-\bar{\pi}_r^{(k)} \big), \quad \bar{\varrho}_{j,r}^{(\text{num},k^*)}=\e \big[ \mathbbm{1} \{ \proxyX \in \mathcal{S}_r \} \cdot \psi_j(\goodX,\proxyX) \big],$$ $$\hat{\varrho}_{j,r}^{(\text{den},k^*)}=N^{-1} \sum_{i=1}^N \mathbbm{1} \{ \proxyX_i \in \mathcal{S}_r \}, \quad \text{and} \quad \bar{\varrho}_{j,r}^{(\text{den},k^*)} = \mathbb{P}(\proxyX \in \mathcal{S}_r).$$ Observe that as $N \to \infty$, $$\check{u}_r^{(k^*)} \xrightarrow{p} \bar{u}_r^{(k^*)}, \quad \hat{\varrho}_{j,r}^{(\text{den},k^*)} \xrightarrow{p} \bar{\varrho}_{j,r}^{(\text{den},k^*)}, \quad \text{and} \quad \hat{\varrho}_{j,r}^{(\text{num},k^*)} \xrightarrow{p} \bar{\varrho}_{j,r}^{(\text{num},k^*)}, \quad  \text{ for each } r \in [L].$$ Above the first claim follows from an earlier result and the continuous mapping theorem. The second claim follows from Assumption \ref{assump:IIDUnderlyingData} combined with the law of large numbers. The third claim is a direct application of Lemma \ref{lemma:varrhoHatStratConvergence}.

By Equation \eqref{eq:EstimatedGreedyOptLabellingRule}, the greedy optimal sampling rule discussed in Sections \ref{sec:DescribeApproximateGreedyOptimalStrategy} and  \ref{sec:EstimateOptimalLabellingProbabilityInEachStrata} is given by \begin{equation*} \LabelRulePi{k^*}(\argxProxy) =  \min \Biggl\{ b_u, \max \Bigl\{ b_l, \sqrt{\frac{\hat{\varrho}_j^{(\text{strat},k^*)}(\argxProxy)}{ \hat{\lambda}^{(k^*)} \prod_{k=1}^{k^*-1} \big(1-\LabelRulePi{k}(\argxProxy) \big)} } \Bigr\} \Biggr\} \quad \text{ for each } \argxProxy \in \xProxySpace,\end{equation*} where $\hat{\varrho}_j^{(\text{strat},k^*)}$ is defined at \eqref{eq:varrhoHatStrat}. It is convenient to define $\mathcal{C}_{[b_l,b_u]} : [0,\infty) \to [b_l,b_u]$ to be the function given by $$\mathcal{C}_{[b_l,b_u]}(z) = \min \big\{b_u, \max \{b_l,\sqrt{z} \} \big\} \quad \text{for each } z \in [0,\infty),$$ and to define $$ \check{\pi}_r^{(k^*)} \equiv \mathcal{C}_{[b_l,b_u]} \Bigg( \frac{\hat{\varrho}_{j,r}^{(\text{num},k^*)}}{ \hat{\lambda}^{(k^*)} \check{u}_r^{(k^*)}  \hat{\varrho}_{j,r}^{(\text{den},k^*)}} \Bigg) \quad \text{for each } r \in [L].$$ Since $\LabelRulePi{k^*}(\argxProxy)$ is constant within each stratum, by combining formulas, including the definition for $\hat{\varrho}_j^{(\text{strat},k^*)}$ at \eqref{eq:varrhoHatStrat}, observe that \begin{equation}\label{eq:PiStratDecompForProofStep} \LabelRulePi{k^*}(\argxProxy) = \sum_{r=1}^L \mathbbm{1} \{ \argxProxy \in \mathcal{S}_r \} \cdot \check{\pi}_r^{(k^*)} \quad \text{for each } \argxProxy \in \xProxySpace.\end{equation} We wish to show that $\check{\pi}_r^{(k^*)}$ converges in probability for each $r \in [L]$, but doing so requires characterizing the behavior of $\hat{\lambda}^{(k^*)}$.

Now recall from formula \eqref{eq:YetToBeLabelledPoints} that $$\mathcal{U}^{(k^*)} \equiv \bigl\{ i \in [N] \ : \ I_i^{(k)}=0 \ \ \text{for all } k \in [k^*-1] \bigr\}$$ is the set of samples for which the label $\xmiss_i$ has not yet been obtained prior to the start of wave $k^*$. Also recall from Appendix \ref{sec:EnforceBudgetConstraints} that $\hat{\lambda}^{(k^*)}$ satisfies \begin{equation}\label{eq:ParameterForBudgetConstraintInProof}
n_{\text{targ}}^{(k^*)} = \sum_{i \in \mathcal{U}^{(k^*)}} \mathcal{C}_{[b_l,b_u]} \Bigg( \frac{\hat{\varrho}_j^{(\text{strat},k^*)}(\proxyX_i)}{ \hat{\lambda}^{(k^*)} \prod_{k=1}^{k^*-1} \big(1-\LabelRulePi{k}(\proxyX_i) \big)}\Bigg).\end{equation} Hence dividing each side by $N$, $$\begin{aligned} \theta_{k^*}+o_p(1) = \frac{n_{\text{targ}}^{(k^*)}}{N} & = \frac{1}{N} \sum_{i \in \mathcal{U}^{(k^*)}}  \mathcal{C}_{[b_l,b_u]} \Bigg( \frac{\hat{\varrho}_j^{(\text{strat},k^*)}(\proxyX_i)}{ \hat{\lambda}^{(k^*)} \prod_{k=1}^{k^*-1} \big(1-\LabelRulePi{k}(\proxyX_i) \big)}\Bigg)
\\ & =  \sum_{r=1}^L  \mathcal{C}_{[b_l,b_u]} \Bigg( \frac{\hat{\varrho}_{j,r}^{(\text{num},k^*)}}{\hat{\lambda}^{(k^*)} \check{u}_r^{(k^*)} \hat{\varrho}_{j,r}^{(\text{den},k^*)}}  \Bigg) \frac{1}{N} \sum_{i=1}^N \mathbbm{1}\{ \proxyX_i \in \mathcal{S}_r\}  \prod_{k=1}^{k^*-1}(1-I_i^{(k)})
\\ & =  \sum_{r=1}^L  \mathcal{C}_{[b_l,b_u]} \Big( \hat{\varrho}_{j,r}^{(\text{num},k^*)} /\big( \hat{\lambda}^{(k^*)} \check{u}_r^{(k^*)} \hat{\varrho}_{j,r}^{(\text{den},k^*)} \big) \Big) \big(\bar{\varrho}_{j,r}^{(\text{den},k^*)} \check{u}_r^{(k^*)} +o_p(1) \big)
\\ & =  \sum_{r=1}^L  \mathcal{C}_{[b_l,b_u]} \Big( \hat{\varrho}_{j,r}^{(\text{num},k^*)} /\big( \hat{\lambda}^{(k^*)} \check{u}_r^{(k^*)} \hat{\varrho}_{j,r}^{(\text{den},k^*)} \big) \Big) \bar{\varrho}_{j,r}^{(\text{den},k^*)} \bar{u}_r^{(k^*)} +o_p(1),
\end{aligned}$$ where the last step follows from boundedness of $\mathcal{C}_{[b_l,b_u]} (\cdot)$ and since $\check{u}_r^{(k^*)} \xrightarrow{p} \bar{u}_r^{(k^*)}$. 

Next, we define $$M \equiv \frac{1}{2b_u^2} \min_{r \in [L]} \Bigl\{ \frac{\bar{\varrho}_{j,r}^{(\text{num},k^*)}}{ \bar{u}_r^{(k^*)} \bar{\varrho}_{j,r}^{(\text{den},k^*)}} \Bigr\},$$ and show that $ \lim_{N \to \infty} \mathbb{P}(\hat{\lambda}^{(k^*)} > M) = 1$. To do this note that by the remark's assumption, $n_{\text{targ}}^{(k^*)} < b_u \vert \mathcal{U}^{(k^*)} \vert$ with probability converging to $1$. By Equation \eqref{eq:ParameterForBudgetConstraintInProof}, whenever the event that $n_{\text{targ}}^{(k^*)} < b_u \vert \mathcal{U}^{(k^*)} \vert$ occurs, $$ \min_{r \in [L]} \Bigl\{ \frac{\hat{\varrho}_{j,r}^{(\text{num},k^*)} }{ \hat{\lambda}^{(k^*)} \check{u}_r^{(k^*)} \hat{\varrho}_{j,r}^{(\text{den},k^*)} } \Bigr\} < b_u^2 \Rightarrow \hat{\lambda}^{(k^*)} > \min_{r \in [L]} \Bigl\{ \frac{\hat{\varrho}_{j,r}^{(\text{num},k^*)} }{ b_u^2 \check{u}_r^{(k^*)} \hat{\varrho}_{j,r}^{(\text{den},k^*)} } \Bigr\} = 2M+o_p(1).$$ Since $M>0$ is a constant, and since on an event with probability converging to $1$, $\hat{\lambda}^{(k^*)}>2M +o_p(1)$, $ \lim_{N \to \infty} \mathbb{P}(\hat{\lambda}^{(k^*)} > M) = 1$. This further implies that $(1/\hat{\lambda}^{(k^*)})=O_p(1)$ because $M>0$.

Next observe that $\mathcal{C}_{[b_l,b_u]} :[0,\infty) \to [b_l,b_u]$ is $1/(2 b_l)$-Lipschitz as it has $0$ slope for $z \leq b_l^2$ and $z \geq b_u^2$, and in between the maximum derivative is $\sup_{z \in (b_l^2,b_u^2)} \{ \frac{1}{2} z^{-1/2} \}= 1/(2 b_l)$. 

Thus, for each $r \in [L]$, $$\begin{aligned} \Bigg| \mathcal{C}_{[b_l,b_u]} \Big( \frac{\bar{\varrho}_{j,r}^{(\text{num},k^*)} }{ \hat{\lambda}^{(k^*)} \bar{u}_r^{(k^*)} \bar{\varrho}_{j,r}^{(\text{den},k^*)}  } \Big) - \mathcal{C}_{[b_l,b_u]} \Big( \frac{ \hat{\varrho}_{j,r}^{(\text{num},k^*)}} {  \hat{\lambda}^{(k^*)} \check{u}_r^{(k^*)} \hat{\varrho}_{j,r}^{(\text{den},k^*)} } \Big) \Bigg| & \leq \frac{1}{2b_l \hat{\lambda}^{(k^*)}} \Big| \frac{\bar{\varrho}_{j,r}^{(\text{num},k^*)} }{ \bar{u}_r^{(k^*)} \bar{\varrho}_{j,r}^{(\text{den},k^*)}  }- \frac{ \hat{\varrho}_{j,r}^{(\text{num},k^*)}} {   \check{u}_r^{(k^*)} \hat{\varrho}_{j,r}^{(\text{den},k^*)} }\Big| 
\\ & = \frac{1}{2b_l} O_p(1) o_p(1) =o_p(1). \end{aligned}$$ Hence, $$\mathcal{C}_{[b_l,b_u]} \Big( \frac{\bar{\varrho}_{j,r}^{(\text{num},k^*)} }{ \hat{\lambda}^{(k^*)} \bar{u}_r^{(k^*)} \bar{\varrho}_{j,r}^{(\text{den},k^*)}  } \Big) = \mathcal{C}_{[b_l,b_u]} \Big( \frac{ \hat{\varrho}_{j,r}^{(\text{num},k^*)}} {  \hat{\lambda}^{(k^*)} \check{u}_r^{(k^*)} \hat{\varrho}_{j,r}^{(\text{den},k^*)} } \Big)+o_p(1), \quad \text{for each } r \in [L],$$ so plugging this into an earlier result, 
$$\begin{aligned}
    \theta_{k^*} & = \sum_{r=1}^L  \Bigg( \mathcal{C}_{[b_l,b_u]} \Big( \frac{\bar{\varrho}_{j,r}^{(\text{num},k^*)} }{ \hat{\lambda}^{(k^*)} \bar{u}_r^{(k^*)} \bar{\varrho}_{j,r}^{(\text{den},k^*)}} \Big)  +o_p(1) \Bigg)\bar{\varrho}_{j,r}^{(\text{den},k^*)} \bar{u}_r^{(k^*)} +o_p(1)
    \\ & = \sum_{r=1}^L   \mathcal{C}_{[b_l,b_u]} \Big( \frac{\bar{\varrho}_{j,r}^{(\text{num},k^*)} }{ \hat{\lambda}^{(k^*)} \bar{u}_r^{(k^*)} \bar{\varrho}_{j,r}^{(\text{den},k^*)}} \Big)  \bar{\varrho}_{j,r}^{(\text{den},k^*)} \bar{u}_r^{(k^*)} +o_p(1).
     \\ & =  f^{(k^*)}(\hat{\lambda}^{(k^*)}) +o_p(1),
\end{aligned}$$
where $f^{(k^*)}:(0,\infty) \to (0,\infty)$ is defined to be the function given by \begin{equation}\label{eq:MonotoneFuncOfLambda}
f^{(k^*)}(\lambda)= \sum_{r=1}^L   \mathcal{C}_{[b_l,b_u]} \Big( \frac{\bar{\varrho}_{j,r}^{(\text{num},k^*)} }{ \lambda \bar{u}_r^{(k^*)} \bar{\varrho}_{j,r}^{(\text{den},k^*)}} \Big)  \bar{\varrho}_{j,r}^{(\text{den},k^*)} \bar{u}_r^{(k^*)} \quad \text{for } \lambda \in (0,\infty). \end{equation}
 
 Now recalling by assumption, $$\lim_{N \to \infty} \mathbb{P} \Big( b_l \frac{\vert \mathcal{U}^{(k^*)} \vert}{N}  < \frac{n_{\text{targ}}^{(k^*)}}{N} < b_u \frac{\vert \mathcal{U}^{(k^*)} \vert}{N} \Big)=1,$$ and $\lim_{N \to \infty} n_{\text{targ}}^{(k^*)}/N=\theta_{k^*}$. Also note that $$\frac{\vert \mathcal{U}^{(k^*)} \vert}{N} = \frac{1}{N} \sum_{i=1}^N \sum_{r=1}^L  \mathbbm{1}\{ \proxyX_i \in \mathcal{S}_r\}  \prod_{k=1}^{k^*-1}(1-I_i^{(k)}) \xrightarrow{p} \sum_{r=1}^L \bar{\varrho}_{j,r}^{(\text{den},k^*)} \bar{u}_r^{(k^*)}.$$ Combining these results implies that $$b_l \sum_{r=1}^L \bar{\varrho}_{j,r}^{(\text{den},k^*)} \bar{u}_r^{(k^*)} < \theta_{k^*} < b_u \sum_{r=1}^L \bar{\varrho}_{j,r}^{(\text{den},k^*)} \bar{u}_r^{(k^*)}.$$ Next note that $f^{(k^*)}:(0,\infty) \to (0,\infty)$ is monotone nonincreasing and continuous. Moreover, by the previous result $f^{(k^*)}$ must cross $\theta_{k^*}$ and the function $f^{(k^*)}$ is strictly decreasing at the point where it crosses $\theta_{k^*}$ (because it is only flat for values of $\lambda$ where clipping occurs in all strata, but we know from the inequality displayed above that clipping cannot occur in all strata). Thus there exists a unique $\bar{\lambda}^{(k^*)}$ such that $\theta_{k^*} = f^{(k^*)}(\bar{\lambda}^{(k^*)})$. Using such a choice of $\bar{\lambda}^{(k^*)}$, by an earlier result, $$f^{(k^*)}(\bar{\lambda}^{(k^*)}) = \theta_{k^*} =  f^{(k^*)} (\hat{\lambda}^{(k^*)})+o_p(1).$$ 
 
 Now assume by way of contradiction that $\hat{\lambda}^{(k^*)}$ does not converge in probability to $\bar{\lambda}^{(k^*)}$. Hence for some $\epsilon >0$, $\liminf_{N \to \infty} \mathbb{P} (\vert \hat{\lambda}^{(k^*)}- \bar{\lambda}^{(k^*)} \vert > \epsilon) > 0$. Thus, with non-vanishing probability either $\hat{\lambda}^{(k^*)} > \bar{\lambda}^{(k^*)} + \epsilon$ or $\hat{\lambda}^{(k^*)} < \bar{\lambda}^{(k^*)} - \epsilon$. Since the function defined at \eqref{eq:MonotoneFuncOfLambda} is strictly decreasing in a neighborhood of $\bar{\lambda}^{(k^*)}$, this implies that with non-vanishing probability $$\big| f^{(k^*)}(\bar{\lambda}^{(k^*)}) - f^{(k^*)} (\hat{\lambda}^{(k^*)}) \big| \geq \min \bigl\{ f^{(k^*)}(\bar{\lambda}^{(k^*)}) - f^{(k^*)} (\bar{\lambda}^{(k^*)} + \epsilon) , f^{(k^*)}(\bar{\lambda}^{(k^*)} - \epsilon) - f^{(k^*)} (\bar{\lambda}^{(k^*)} ) \bigr\} > 0,$$ contradicting the result that $f^{(k^*)}(\bar{\lambda}^{(k^*)}) = f^{(k^*)} (\hat{\lambda}^{(k^*)})+o_p(1)$. Hence $\hat{\lambda}^{(k^*)} \xrightarrow{p} \bar{\lambda}^{(k^*)}$. 

Now recall that for each $r \in [L]$, $\check{\pi}_r^{(k^*)}$ is the value $\LabelRulePi{k^*}$ takes in stratum $r$, so by combining earlier formulas and results $$\check{\pi}_r^{(k^*)}=  \mathcal{C}_{[b_l,b_u]} \Big( \frac{\hat{\varrho}_{j,r}^{(\text{num},k^*)} }{ \hat{\lambda}^{(k^*)} \check{u}_r^{(k^*)} \hat{\varrho}_{j,r}^{(\text{den},k^*)}}\Big) \xrightarrow{p} \mathcal{C}_{[b_l,b_u]} \Big( \frac{\bar{\varrho}_{j,r}^{(\text{num},k^*)} }{ \bar{\lambda}^{(k^*)} \bar{u}_r^{(k^*)} \bar{\varrho}_{j,r}^{(\text{den},k^*)}}\Big) \quad \text{for each } r \in [L].$$ Above the convergence in probability follows from the continuous mapping theorem and previously shown convergence in probability results. Defining $$\bar{\pi}_r^{(k^*)} \equiv \mathcal{C}_{[b_l,b_u]} \Big( \frac{\bar{\varrho}_{j,r}^{(\text{num},k^*)} }{ \bar{\lambda}^{(k^*)} \bar{u}_r^{(k^*)} \bar{\varrho}_{j,r}^{(\text{den},k^*)}}\Big) \quad \text{for each } r \in [L],$$ $\check{\pi}_r^{(k^*)} \xrightarrow{p} \bar{\pi}_r^{(k^*)}$ for each $r \in [L]$. 
Hence defining $\LimitingLabelRulePi{k^*}: \xProxySpace \to [b_l,b_u]$ to be the function given by $$\LimitingLabelRulePi{k^*}(\argxProxy)=\sum_{r=1}^L \mathbbm{1} \{\argxProxy \in \mathcal{S}_r\} \cdot \bar{\pi}_r^{(k^*)} \quad \text{for each } \argxProxy \in \xProxySpace.$$ Observe that for any $\epsilon>0$, $$\begin{aligned} \limsup_{N \to \infty} \e \Big[ \big| \LabelRulePi{k^*}(\proxyX_1)- \LimitingLabelRulePi{k^*}(\proxyX_1) \big| \Big] & =  \limsup_{N \to \infty} \e \Big[ \big| \sum_{r=1}^L \mathbbm{1} \{\proxyX_1 \in \mathcal{S}_r\} \cdot \big( \check{\pi}_r^{(k^*)} -\bar{\pi}_r^{(k^*)} \big) \big| \Big]
\\ & \leq   \sum_{r=1}^L  \limsup_{N \to \infty} \e \big[ \vert \check{\pi}_r^{(k^*)} -\bar{\pi}_r^{(k^*)} \vert \big]
\\ & \leq   \sum_{r=1}^L  \limsup_{N \to \infty} \e \Big[ b_u \mathbbm{1} \{ \vert \check{\pi}_r^{(k^*)} -\bar{\pi}_r^{(k^*)}  \vert > \epsilon \}  + \epsilon \Big]
\\ & \leq   \sum_{r=1}^L    b_u \limsup_{N \to \infty} \mathbb{P} (\vert \check{\pi}_r^{(k^*)} -\bar{\pi}_r^{(k^*)}  \vert > \epsilon ) + \epsilon \Big] =L \epsilon . \end{aligned}$$

Since $L$ is finite and the above holds for any $\epsilon>0$, we can take $\epsilon \downarrow 0$ to establish that $$\limsup_{N \to \infty} \e \Big[ \big| \LabelRulePi{k^*}(\proxyX_1)- \LimitingLabelRulePi{k^*}(\proxyX_1) \big| \Big]=0.$$ Since the above convergence result holds provided the inductive hypothesis that \eqref{eq:InductiveHypothesisForLabellingL1Convergence} holds, by strong induction, $$\lim_{N \to \infty} \e \Big[ \big| \LabelRulePi{k}(\proxyX_1)- \LimitingLabelRulePi{k}(\proxyX_1) \big| \Big]= 0, \quad \text{for each } k \in [K].$$ The above establishes Assumption \ref{assump:CLTRgularityConditions}(ii) in this setting.

\end{proof}

\section{Choices to increase efficiency after data collection}\label{sec:OtherWaysToIncreaseEfficiencyAppendix}

After Phase II is complete and all data have been collected, it is still possible to increase efficiency by choosing the tuning matrix $\hat{\Omega}$ and the wave-specific weights $c_1,\dots,c_K \in [0,1]$ strategically. In this appendix, we first derive the asymptotically optimal choice of the tuning matrix $\hat{\Omega}$ among the class of diagonal tuning matrices for fixed values of $\{c_k\}_{k=1}^K$. This asymptotically optimal diagonal tuning matrix is used in all of our numerical experiments. We then briefly discuss the choice of $\{c_k\}_{k=1}^K$ and describe the default that is used in our numerical experiments.

\subsection{Asymptotically optimal diagonal tuning matrix}\label{sec:AsymptoticallyOptimalDiagTuningMatrix}

In this subsection, we derive the asymptotically optimal tuning matrix among the class of diagonal tuning matrices for $\htPTD$ when the constants $\{c_k\}_{k=1}^K$ are fixed. Recall that in the setting of Theorem \ref{theorem:CLT_PTDEstimator}, if $\hat{\Omega} \xrightarrow{p} \Omega$, the asymptotic variance of $\htPTD$ is given by 
\begin{equation*}
    \begin{split}
    \Sigma^{\PTDSuperScriptAcr}(\Omega) & = H_{\thetaTarg}^{-1} \Sigma_{11} H_{\thetaTarg}^{-1} + \Omega H_{\gammaTarg}^{-1} \big( \Sigma_{22}-\Sigma_{33} \big) H_{\gammaTarg}^{-1} \Omega^\tran 
\\ & \quad + H_{\thetaTarg}^{-1} (\Sigma_{13}-\Sigma_{12}) H_{\gammaTarg}^{-1} \Omega^\tran +  \big(  H_{\thetaTarg}^{-1} (\Sigma_{13}-\Sigma_{12}) H_{\gammaTarg}^{-1} \Omega^\tran \big)^\tran.
\end{split}
\end{equation*} We next find the diagonal matrix $\Omega \in \mathbb{R}^{d \times d}$ that minimizes the diagonal components of the asymptotic variance $\Sigma^{\PTDSuperScriptAcr}(\Omega)$.

To do this, fix $j \in [d]$, suppose that $\Omega$ is a diagonal matrix and that $\omega_j \equiv [\Omega]_{jj} \in \mathbb{R}$, in which case $$[\Sigma^{\PTDSuperScriptAcr}(\Omega)]_{jj}  =  [H_{\thetaTarg}^{-1} \Sigma_{11} H_{\thetaTarg}^{-1}]_{jj} + \omega_j^2 [ H_{\gammaTarg}^{-1} ( \Sigma_{22}-\Sigma_{33} ) H_{\gammaTarg}^{-1} ]_{jj} + 2 \omega_j [ H_{\thetaTarg}^{-1} (\Sigma_{13}-\Sigma_{12}) H_{\gammaTarg}^{-1} ]_{jj}.$$ The above quantity is quadratic in $\omega_j$ and convex in $\omega_j$ ($\Sigma_{22}-\Sigma_{33} \succeq 0$, while generally being strictly positive definite, because it can be shown using Lemma \ref{lemma:PropertiesOfIndependentVersionsOfWeights} and Fact \ref{fact:PropertiesOfGradients} that $\Sigma_{22}-\Sigma_{33}= \var \big( (\bar{W}_1-1) \dot{l}_{\gammaTarg} (\proxyX_1) \big)$). Convex quadratic functions can be minimized by setting the derivative to $0$, implying that the choice $$\omega_j^{\text{opt,diag}}= \frac{[ H_{\thetaTarg}^{-1} (\Sigma_{12}-\Sigma_{13}) H_{\gammaTarg}^{-1} ]_{jj}}{[ H_{\gammaTarg}^{-1} ( \Sigma_{22}-\Sigma_{33} ) H_{\gammaTarg}^{-1} ]_{jj}},$$ minimizes $[\Sigma^{\PTDSuperScriptAcr}(\Omega)]_{jj}$.  Since for diagonal $\Omega$, $[\Sigma^{\PTDSuperScriptAcr}(\Omega)]_{jj}$ does not depend on other diagonal entries of $\Omega$ (besides $[\Omega]_{jj}$) and since the above formula for the optimal choice of the $j$th diagonal entry of $\Omega$ holds for each $j \in [d]$, the diagonal matrix $\Omega^{\text{opt,diag}}$ minimizing each diagonal component of $\Sigma^{\PTDSuperScriptAcr}(\Omega)$ is given by $$\Omega^{\text{opt,diag}} = \text{diag} \Bigg( \frac{[ H_{\thetaTarg}^{-1} (\Sigma_{12}-\Sigma_{13}) H_{\gammaTarg}^{-1} ]_{11}}{[ H_{\gammaTarg}^{-1} ( \Sigma_{22}-\Sigma_{33} ) H_{\gammaTarg}^{-1} ]_{11}},\dots,\frac{[ H_{\thetaTarg}^{-1} (\Sigma_{12}-\Sigma_{13}) H_{\gammaTarg}^{-1} ]_{dd}}{[ H_{\gammaTarg}^{-1} ( \Sigma_{22}-\Sigma_{33} ) H_{\gammaTarg}^{-1} ]_{dd}} \Bigg) \in \mathbb{R}^{d \times d}.$$ Thus we propose the use of the following tuning matrix \begin{equation}\label{eq:EstimatedOptimalDiagTuningMatrix}
    \hat{\Omega}^{\text{opt,diag}} = \text{diag} \Bigg( \frac{[ \hat{H}_{\thetaTarg}^{-1} (\hat{\Sigma}_{12}-\hat{\Sigma}_{13}) \hat{H}_{\gammaTarg}^{-1} ]_{11}}{[ \hat{H}_{\gammaTarg}^{-1} ( \hat{\Sigma}_{22}-\hat{\Sigma}_{33} ) \hat{H}_{\gammaTarg}^{-1} ]_{11}},\dots, \frac{[ \hat{H}_{\thetaTarg}^{-1} (\hat{\Sigma}_{12}-\hat{\Sigma}_{13}) \hat{H}_{\gammaTarg}^{-1} ]_{dd}}{[ \hat{H}_{\gammaTarg}^{-1} ( \hat{\Sigma}_{22}-\hat{\Sigma}_{33} ) \hat{H}_{\gammaTarg}^{-1} ]_{dd}} \Bigg) \in \mathbb{R}^{d \times d},
\end{equation} which is consistent for $\Omega^{\text{opt,diag}}$ in the settings of Propositions \ref{prop:ConsistencyOfHessians}, \ref{prop:ConsistencyOfSig13andSig33}, and \ref{prop:ConsistencyOfSigmaUpperLeft} (provided that $\Sigma_{22}-\Sigma_{33} \succ 0$). Crucially, using $\hat{\Omega}^{\text{opt,diag}}$ as a tuning matrix will minimize the asymptotic variance of $\htPTD_j$ (simultaneously for all $j \in [d]$) among all choices of diagonal tuning matrices, and in the setting of Proposition \ref{prop:AsymptoticallyValidCIs}, will result in asymptotically valid confidence intervals.

\begin{remark}
     We use the optimal diagonal tuning matrix ($d$ tuning parameters) rather than the optimal tuning matrix among all possible $d \times d$ tuning matrices ($d^2$ tuning parameters) that originates in \cite{ChenAndChen2000} in order to increase stability. In particular, for similar Predict-Then-Debias approaches, \cite{PTDBootstrapPaper} empirically found that reducing the number of tuning parameters from $d^2$ to $d$ in this manner led to limited losses in efficiency, while increasing stability and avoiding undercoverage. Nonetheless, in some of our supplementary experiments, we also consider the use of the tuning matrix \begin{equation}\label{eq:EstimatedFullOptimalTuningMatrix}
    \hat{\Omega}^{\text{opt}} = \hat{H}_{\thetaTarg}^{-1} (\hat{\Sigma}_{12}-\hat{\Sigma}_{13}) ( \hat{\Sigma}_{22}-\hat{\Sigma}_{33})^{-1} \hat{H}_{\gammaTarg},
 \end{equation} which estimates the optimal tuning matrix among all $d \times d$ tuning matrices.
\end{remark}

\subsection{Choice of \texorpdfstring{$c_k$}{wave-specific weights}}\label{sec:ChoiceOfCk}

After data collection the $\{c_k\}_{k=1}^K$ could also be tuned. In particular, for a fixed $\Omega$, one can find an analytic formula for the asymptotically optimal choice of $\{c_k\}_{k=1}^K$.
However, our current theory does not guarantee the validity of confidence intervals with post-hoc tuning of the $\{c_k\}_{k=1}^K$, so we do not pursue this further in our manuscript.  In our simulations we set $c_k = n_{\text{targ}}^{(k)}/ \big(\sum_{k'=1}^K n_{\text{targ}}^{(k')} \big)$ to be the fraction of the budget that is allocated to wave $k$ and put this forward as a reasonable default.

\newpage

\section{Quantile estimation-specific approaches}\label{sec:QuantileEstimationImplementationApproaches}

In this appendix we present implementation strategies that are specific to quantile estimation tasks. Throughout this appendix, our focus is on settings where $\goodX$ and $\proxyX$ are univariate random variables with probability density (and cumulative distribution) functions given by $f_{\goodX}(\cdot)$ (and $F_{\goodX}(\cdot)$) and $f_{\proxyX}$ (and $F_{\proxyX}(\cdot)$), respectively, and the goal is to estimate the $q$-quantile of $X$, given by $\thetaTarg =F_{\goodX}^{-1}(q)$. Throughout this appendix, we define the loss function $l_{\theta}: \mathbb{R} \to \mathbb{R}$ parameterized by $\theta \in \mathbb{R}$ to be given by \begin{equation}\label{eq:QuantileLoss}
l_{\theta}(\argx) = q (\argx-\theta) \mathbbm{1} \{ \argx > \theta \} + (1-q) (\theta -\argx) \mathbbm{1} \{ \argx \leq \theta \} \quad \text{for } \argx \in \mathbb{R}.\end{equation} The population quantile of $\goodX$ is given by \begin{equation*}
    \thetaTarg = \argmin_{\theta \in \mathbb{R}} \e[l_{\theta} (\goodX)] =F_{\goodX}^{-1}(q).
\end{equation*} Similarly define $\gammaTarg = \argmin_{\theta \in \mathbb{R}} \e[l_{\theta} (\proxyX)] =F_{\proxyX}^{-1}(q)$. In these settings it can be shown that \begin{equation}\label{eq:HessianQuantileFormula}   
H_{\thetaTarg} \equiv \nabla^2 \e[l_{\theta}(\goodX)] \big|_{\theta=\thetaTarg} =f_{\goodX}(\thetaTarg) \quad \text{and} \quad H_{\gammaTarg} \equiv \nabla^2 \e[l_{\theta}(\proxyX)] \big|_{\theta=\gammaTarg} =f_{\proxyX}(\gammaTarg).\end{equation}

We first present simplifications of the formula for the point estimator $\htPTD$ and its asymptotic variance. We then present Phase II labelling strategies that exploit these simplifications. Finally, we discuss Hessian estimation for confidence interval constructions.

\subsection{Estimator and asymptotic variance formula}

Recall from Equations \eqref{eq:aggregated_Wk} and \eqref{eq:PTDComponentEstimatorsDEF} that \begin{equation*}
    \htc \equiv \argmin_{\theta \in \mathbb{R}} \frac{1}{N} \sum_{i=1}^N W_i l_{\theta}(\goodX_i), \ \  \hgc \equiv \argmin_{\theta \in \mathbb{R}} \frac{1}{N} \sum_{i=1}^N W_i l_{\theta}(\proxyX_i), \ \  \text{and }  \hga \equiv \argmin_{\theta \in \mathbb{R}} \frac{1}{N} \sum_{i=1}^N  l_{\theta}(\proxyX_i).
\end{equation*} Since the estimators above and the estimand are all univariate in the quantile estimation setting, the Multiwave Predict-Then-Debias estimator is given by 
\begin{equation*} \htPTD = \hat{\omega} \hga +  (\htc -  \hat{\omega} \hgc),\end{equation*}
where $\hat{\omega} \in \mathbb{R}$ is a tuning parameter.

To present a formula for the asymptotic variance of $\htPTD$, note that by Equations \eqref{eq:PiBarProd1tok} and \eqref{eq:GradientAsympCovFormulae}, \begin{equation*}
    \LimitingLabelRulePi{1:k}(\argxProxy) \equiv \LimitingLabelRulePi{k}(\argxProxy) \prod_{j=1}^{k-1} \big( 1-  \LimitingLabelRulePi{j}(\argxProxy) \big) \quad \text{for } k \in [K], \argxProxy \in \xProxySpace, \quad \text{and}
\end{equation*} 
$$\begin{aligned}
    \Sigma_{11} & \equiv  \sum_{k=1}^K c_k^2 \e \Big[ \frac{ \big(\dot{l}_{\thetaTarg}(\goodX) \big)^2}{\LimitingLabelRulePi{1:k}(\proxyX)} \Big], \quad  \Sigma_{12} \equiv  \sum_{k=1}^K c_k^2 \e \Big[ \frac{ \dot{l}_{\thetaTarg}(\goodX)   \dot{l}_{\gammaTarg}(\proxyX) }{\LimitingLabelRulePi{1:k}(\proxyX)} \Big], \quad \Sigma_{13} \equiv  \e \big[  \dot{l}_{\thetaTarg}(\goodX)   \dot{l}_{\gammaTarg}(\proxyX) \big],  \\  
    \Sigma_{22} & \equiv \sum_{k=1}^K c_k^2 \e \Big[ \frac{ \big( \dot{l}_{\gammaTarg}(\proxyX) \big)^2}{\LimitingLabelRulePi{1:k}(\proxyX)} \Big], \quad  
      \text{and} \quad \Sigma_{33} \equiv \e \big[ \big( \dot{l}_{\gammaTarg}(\proxyX) \big)^2 \big]. 
\end{aligned}$$ By Theorem \ref{theorem:CLT_PTDEstimator},  under \samplingSchemeName{} and Assumptions 1--4, 

$$\text{if} \quad \hat{\omega} \xrightarrow{p} \omega \quad \text{for some } \omega \in \mathbb{R} \ \text{then} \quad \sqrt{N} (\htPTD-\thetaTarg) \xrightarrow{d} \mathcal{N} \big(0,\Sigma^{\PTDSuperScriptAcr}(\omega) \big) \quad \text{as } N \to \infty,$$ where by \eqref{eq:AsympVarPTDFormula}, \begin{equation*}
    \Sigma^{\PTDSuperScriptAcr}(\omega)  \equiv H_{\thetaTarg}^{-2} \Sigma_{11}  + \omega^2 H_{\gammaTarg}^{-2} \big( \Sigma_{22}-\Sigma_{33} \big)   + 2 \omega H_{\thetaTarg}^{-1} H_{\gammaTarg}^{-1} (\Sigma_{13}-\Sigma_{12})  \quad \text{for each } \omega \in \mathbb{R}
\end{equation*}

We now simplify the asymptotic variance. Note that $$\dot{l}_{\thetaTarg}(\goodX)= \mathbbm{1} \{ \goodX \leq \thetaTarg \} - q, \quad \dot{l}_{\gammaTarg}(\proxyX)= \mathbbm{1} \{ \proxyX \leq \gammaTarg \} - q, \quad H_{\thetaTarg} = f_{\goodX}(\thetaTarg) \quad \text{and} \quad  H_{\gammaTarg} = f_{\proxyX}(\gammaTarg),$$ where $f_{\goodX}(\cdot)$ and $f_{\proxyX}(\cdot)$ are the probability density functions of $\goodX$ and $\proxyX$ respectively. Next for any $\omega \in \mathbb{R}$ define $\varpi= \omega f_{\goodX}(\thetaTarg)/f_{\proxyX}(\gammaTarg)$ and $\Sigma_{c,\varpi}= (2 \varpi \Sigma_{13} -\varpi^2 \Sigma_{33})/[f_{\goodX}(\thetaTarg)]^2$, then the asymptotic variance can be re-xpressed as follows

$$\begin{aligned} \Sigma^{\PTDSuperScriptAcr}(\omega) &  = \frac{1}{[f_{\goodX}(\thetaTarg)]^2} \Big( \Sigma_{11}  + \varpi^2  \big( \Sigma_{22}-\Sigma_{33} \big)   + 2 \varpi (\Sigma_{13}-\Sigma_{12}) \Big) 
\\ & = \frac{1}{[f_{\goodX}(\thetaTarg)]^2} \Big( \Sigma_{11}  - 2 \varpi \Sigma_{12} + \varpi^2   \Sigma_{22}  \Big) +\Sigma_{c,\varpi} 
\\ & = \frac{1}{[f_{\goodX}(\thetaTarg)]^2} \sum_{k=1}^K c_k^2 \e \Bigg[ \frac{ \big(\dot{l}_{\thetaTarg}(\goodX) \big)^2- 2 \varpi \dot{l}_{\thetaTarg}(\goodX)  \dot{l}_{\gammaTarg}(\proxyX) + \varpi^2 \big( \dot{l}_{\gammaTarg}(\proxyX) \big)^2}{\LimitingLabelRulePi{1:k}(\proxyX)} \Bigg] +\Sigma_{c,\varpi} 
\\ & = \frac{1}{[f_{\goodX}(\thetaTarg)]^2} \sum_{k=1}^K c_k^2 \e \Bigg[ \frac{ \big(\dot{l}_{\thetaTarg}(\goodX) -\varpi  \dot{l}_{\gammaTarg}(\proxyX) \big)^2}{\LimitingLabelRulePi{1:k}(\proxyX)} \Bigg] +\Sigma_{c,\varpi} 
\\ & = \frac{1}{[f_{\goodX}(\thetaTarg)]^2} \sum_{k=1}^K c_k^2 \e \Bigg[ \frac{ \big( \mathbbm{1} \{ \goodX \leq \thetaTarg \} - q - \varpi ( \mathbbm{1} \{ \proxyX \leq \gammaTarg \} - q) \big)^2}{\bar{\pi}^{(k)}(\proxyX) \prod_{j=1}^{k-1} (1-\bar{\pi}^{(j)}(\proxyX)) } \Bigg] + \Sigma_{c,\varpi}.
\end{aligned}$$

Note that by solving the quadratic equation in the first expression displayed above, the choice of $\varpi$ that minimizes the asymptotic variance is $$\varpi_{\text{opt}}=\frac{\Sigma_{12}-\Sigma_{13}}{\Sigma_{22}-\Sigma_{33}} \quad \text{so} \quad  \omega_{\text{opt}}=\frac{\Sigma_{12}-\Sigma_{13}}{\Sigma_{22}-\Sigma_{33}} \cdot \frac{f_{\proxyX}(\gammaTarg)}{ f_{\goodX}(\thetaTarg)},$$ and we note that our numerical experiments leverage the fact that $\varpi_{\text{opt}}$ can be estimated in intermediate waves of Phase II without requiring density estimation. We next present a greedy optimal labelling strategy based on these formulas which exploits the fact that $\Sigma_{c,\varpi}$ does not depend on the labelling strategy. In contrast to the greedy optimal labelling strategy presented in Section \ref{sec:DescribeApproximateGreedyOptimalStrategy}, the tuning parameter $\hat{\omega}$ is not set to $1$ for the purposes of deriving a greedy optimal labelling strategy.

\subsection{Quantile estimation-specific greedy optimal labelling strategy}\label{sec:EstimateQuantileSpecificGreedyOpt}

Fix a $k^* \in [K]$ and $\varpi \in \mathbb{R}$. Prior to the start of wave $k^*$ we wish to select a labelling rule for wave $k^*$ that will minimize the asymptotic variance of a version of $\htPTD$ in a greedy manner. Note that $\LimitingLabelRulePi{k}$ for $k < k^*$ and the terms $\LimitingLabelRulePi{k}$ for $k > k^*$ do not appear in the greedy optimal objective (which is the asymptotic variance assuming $k^*$ is the last wave). The only quantity in the formula for $\Sigma^{\PTDSuperScriptAcr}(\omega)$ that the labelling rule controls is $\LimitingLabelRulePi{k^*}$. The greedy optimal objective is thus to find the $\LimitingLabelRulePi{k^*}(\cdot)$ that minimizes $$\Sigma^{\PTDSuperScriptAcr,(1:k^*)}(\omega)= \frac{1}{[f_{\goodX}(\thetaTarg)]^2} \sum_{k=1}^{k^*} \Big(\frac{c_k}{\sum_{j=1}^{k^*} c_{j}} \Big)^2 \e \Bigg[ \frac{ \big( \mathbbm{1} \{ \goodX \leq \thetaTarg \} - q - \varpi ( \mathbbm{1} \{ \proxyX \leq \gammaTarg \} - q) \big)^2}{\bar{\pi}^{(k)}(\proxyX) \prod_{j=1}^{k-1} (1-\bar{\pi}^{(j)}(\proxyX)) } \Bigg] + \Sigma_{c,\varpi},$$ subject to budget and overlap constraints. By ignoring additive terms and positive multiplicative terms that do not depend on $\LimitingLabelRulePi{k^*}$ the optimization problem of interest reduces to

\begin{equation*} 
\begin{split}
    \text{find a} & \quad \quad \bar{\pi}^{(k^*)}: \mathbb{R}^p \to (0,\infty) \\
    \text{minimizing} & \quad \quad \e \Bigg[ \frac{ \big( \mathbbm{1} \{ \goodX \leq \thetaTarg \} - q - \varpi ( \mathbbm{1} \{ \proxyX \leq \gammaTarg \} - q) \big)^2}{\bar{\pi}^{(k^*)}(\proxyX) \prod_{j=1}^{k^*-1} (1-\bar{\pi}^{(j)}(\proxyX)) } \Bigg]  \\  
    \text{subject to} & \quad \quad \e[\bar{\pi}^{(k^*)}(\proxyX)] \leq B^{(k^*)} \quad \text{and} \quad \bar{\pi}^{(k^*)}(\argxProxy) \in [b_l,b_u] \quad \text{for } \argxProxy \in \xProxySpace,
\end{split}
\end{equation*} where $ B^{(k^*)}$ is a budget constraint constant. By Lemma \ref{lemma:ImprovementOnCSOptimalityDeriv}, if $0< b_l < B^{(k^*)} < b_u <1$, the solution to the greedy optimal optimization problem above is given by \begin{equation*}
    \bar{\pi}_{\text{opt}}^{(k^*)} (\argxProxy)  = \min \Biggl\{ b_u, \max \Bigl\{ b_l, \sqrt{\frac{\varrho(\argxProxy)}{ \lambda^{(k^*)} \prod_{k=1}^{k^*-1} \big(1-\LimitingLabelRulePi{k}(\argxProxy) \big)} } \Bigr\} \Biggr\} \quad \text{for each} \quad \argxProxy \in \xProxySpace,
\end{equation*} where $\varrho : \xProxySpace \to [0,\infty)$ is the function defined by $$\varrho(\proxyX) = \e \Big[\big( \mathbbm{1} \{ \goodX \leq \thetaTarg \} - q - \varpi ( \mathbbm{1} \{ \proxyX \leq \gammaTarg \} - q) \big)^2 \giv \proxyX \Big],$$ and where $\lambda^{(k^*)}$ is the existing constant such that $\e[\bar{\pi}_{\text{opt}}^{(k^*)} (\proxyX)] = B^{(k^*)}$. Analogously to the approach in Section \ref{sec:DescribeApproximateGreedyOptimalStrategy} in the main text in our numerical experiments for quantile-specific estimation strategies we use a labelling rule that approximates $\bar{\pi}_{\text{opt}}^{(k^*)}$, given by \begin{equation}\label{eq:EstimatedGreedyOptimalSamplingRuleQuantileSpecific}   \LabelRulePi{k^*}(\argxProxy) =  \min \Biggl\{ b_u, \max \Bigl\{ b_l, \sqrt{\frac{\hat{\varrho}^{(k^*)}(\argxProxy)}{ \hat{\lambda}^{(k^*)} \prod_{k=1}^{k^*-1} \big(1-\LabelRulePi{k}(\argxProxy) \big)} } \Bigr\} \Biggr\} \quad \text{ for each } \argxProxy \in \xProxySpace,\end{equation} where $\hat{\varrho}^{(k^*)}(\cdot)$ is some estimate for the function $\varrho(\cdot)$ that is learned using the data $\cd_{k^*-1}$ and $\hat{\lambda}^{(k^*)}$ is computed numerically to enforce the budget constraint. (We next discuss approaches for estimating $\varrho(\cdot)$, and in Appendix \ref{sec:EnforceBudgetConstraints} we describe how we calculate $\hat{\lambda}^{(k^*)}$ numerically in our experiments). 

To estimate $\varrho(\cdot)$, we first simplify its expression. In particular, observe that if we define $p_{\leq \thetaTarg} (\argxProxy) : \xProxySpace \to [0,1]$ by \begin{equation}\label{eq:ConditionalCDFFunction}
    p_{\leq \thetaTarg} (\argxProxy) \equiv \mathbb{P} ( \goodX \leq \thetaTarg \giv \proxyX =\argxProxy),
\end{equation} then $\varrho$ simplifies as follows $$\begin{aligned} \varrho(\proxyX) & = \var[  \mathbbm{1} \{ \goodX \leq \thetaTarg \} \giv \tilde{X}] + \big( \e \big[ \mathbbm{1} \{ \goodX \leq \thetaTarg \} - q - \varpi ( \mathbbm{1} \{ \proxyX \leq \gammaTarg \} - q) \giv \proxyX \big] \big)^2
\\ & = p_{\leq \thetaTarg}(\proxyX)  \big(1- p_{\leq \thetaTarg}(\proxyX) \big) +  \big( (p_{\leq \thetaTarg}(\proxyX)- q )- \varpi ( \mathbbm{1} \{ \proxyX \leq \gammaTarg \} - q) \big)^2. \end{aligned}$$ Hence, $\varrho : \xProxySpace \to [0,\infty)$ is given by \begin{equation}\label{eq:VarrhoSimplifiedQuantile}
   \varrho(\argxProxy)= p_{\leq \thetaTarg}(\argxProxy)  \big(1- p_{\leq \thetaTarg}(\argxProxy) \big) +  \big( (p_{\leq \thetaTarg}(\argxProxy)- q )- \varpi ( \mathbbm{1} \{ \argxProxy \leq \gammaTarg \} - q) \big)^2 \quad \text{for all } \argxProxy \in \xProxySpace.
\end{equation} Therefore in our numerical experiments to compute $\hat{\varrho}^{(k^*)}(\cdot)$, we \begin{enumerate}
    \item compute $\hga$ using Phase I data,
    \item set $\hat{\varpi}_{\text{opt}}^{(k^*)}=(\hat{\Sigma}_{12}^{(k^*-1)}-\hat{\Sigma}_{13}^{(k^*-1)})/(\hat{\Sigma}_{22}^{(k^*-1)}-\hat{\Sigma}_{33}^{(k^*-1)})$, where $\hat{\Sigma}_{12}^{(k^*-1)}$, $\hat{\Sigma}_{13}^{(k^*-1)}$, $\hat{\Sigma}_{22}^{(k^*-1)}$, and $\hat{\Sigma}_{33}^{(k^*-1)}$ are versions of the covariance component estimators at \eqref{eq:HessianAndSigmaEstimators} that assume that $k^*-1$ is the final wave in Phase II,
    \item and estimate the function $p_{\leq \thetaTarg}(\argxProxy)$ using the Phase I data and Phase II data from waves 1 through $k^*-1$ by either fitting an isotonic regression model (Section \ref{sec:IsotonicRegStrategy}) or a bivariate Gaussian copula model (Section \ref{sec:FittedBivariateCopulaStrategy}). We call the estimated function $\hat{p}_{\leq \thetaTarg}(\cdot)$.
\end{enumerate} We then plug in $\hga$, $\hat{\varpi}_{\text{opt}}^{(k^*)}$, and $\hat{p}_{\leq \thetaTarg}(\cdot)$ into Formula \eqref{eq:VarrhoSimplifiedQuantile}, to get an estimate for $\varrho$ which we call $\hat{\varrho}^{(k^*)}$. $\hat{\varrho}^{(k^*)}$ is then used to determine the wave $k^*$ labelling rule according to \eqref{eq:EstimatedGreedyOptimalSamplingRuleQuantileSpecific}.

Note that steps 2 and 3 for estimating $\varrho(\cdot)$ are not feasible when $k^*=1$ because Phase II data is not yet available. In our numerical experiments we use a uniform rather than greedy optimal first wave by setting $\pi_{\cd_0}^{(1)}(\proxyX_i)=n_{\text{targ}}^{(1)}/N$ for each $i \in [N]$. As an alternative, in Section \ref{sec:FirstWaveStrategyPrior}, we also consider a quantile-specific strategy to steps 2 and 3 above that assumes prior knowledge of the Spearman's rank correlation coefficient between $\goodX$ and $\proxyX$.

\subsubsection{Isotonic regression}\label{sec:IsotonicRegStrategy}

One way to estimate  $p_{\leq \thetaTarg} (\argxProxy) = \mathbb{P} ( \goodX \leq \thetaTarg \giv \proxyX =\argxProxy)$ is with isotonic regression. When $\proxyX$ is a good estimate of $\goodX$, it is reasonable to assume that $p_{\leq \thetaTarg} (\cdot)$ is monotone nonincreasing. For $k^* \in [K] \setminus \{1\}$, $\hat{p}_{\leq \thetaTarg} (\cdot)$ can be set to the monotone nonincreasing function that best fits the available labelled data: $\{(\proxyX_i,\goodX_i) \ : \ I_i^{(k)}=1  \text{ for some } k \in [k^*-1] \}$. This best fit function can be calculated using standard isotonic regression software.

\subsubsection{Fitting a bivariate Gaussian copula model}\label{sec:FittedBivariateCopulaStrategy}

Another way to estimate  $p_{\leq \thetaTarg} (\argxProxy) = \mathbb{P} ( \goodX \leq \thetaTarg \giv \proxyX =\argxProxy)$ is by assuming a single-parameter bivariate Gaussian copula model for $(\goodX,\proxyX)$. Specifically, let $Z \equiv \Phi^{-1}(F_X(X))$ and $\tilde{Z} \equiv \Phi^{-1}(F_{\tilde{X}}(\tilde{X}))$, where recall that $F_X(\cdot)$ and $F_{\tilde{X}}(\cdot)$ are the CDFs of $X$ and $\tilde{X}$ respectively, and where $\Phi$ is the CDF of a standard Gaussian. The bivariate Gaussian copula assumption is that \begin{equation}\label{eq:CopulaModel}
    \begin{bmatrix}
    \Phi^{-1}(F_X(X)) \\ \Phi^{-1}(F_{\tilde{X}}(\tilde{X}))
\end{bmatrix}  =\begin{bmatrix}
    Z \\ \tilde{Z}
\end{bmatrix} \sim \mathcal{N} \Bigg( \begin{bmatrix}
    0 \\ 0
\end{bmatrix}, \begin{bmatrix}
    1 & \rho_0  \\ \rho_0  &  1
\end{bmatrix} \Bigg) \quad \text{for some } \rho_0 \in (-1,1).\end{equation} Under this model $Z \giv \tilde{Z} =\tilde{z} \sim \mathcal{N} (   \rho_0 \tilde{z}, 1-\rho_0^2 ),$ so the formula for $p_{\leq \thetaTarg} (\argxProxy)$ simplifies as follows  $$\begin{aligned} p_{\leq \thetaTarg} (\argxProxy)=\mathbb{P}(\goodX \leq \thetaTarg \giv \tilde{X}=\tilde{x} ) &  = \mathbb{P}\Big(\goodX \leq F_X^{-1}(q)  \giv \tilde{Z} =\Phi^{-1}(F_{\tilde{X}}(\tilde{x})) \Big)
\\ & = \mathbb{P}\Big( F_X(\goodX) \leq  q  \giv \tilde{Z}=\Phi^{-1}(F_{\tilde{X}}(\tilde{x})) \Big) 
\\ & = \mathbb{P}\Big( Z \leq  \Phi^{-1}(q)  \giv \tilde{Z}=\Phi^{-1}(F_{\tilde{X}}(\tilde{x})) \Big)
\\ & = \Phi \Big( \frac{\Phi^{-1}(q)- \rho_0 \cdot \Phi^{-1}(F_{\tilde{X}}(\tilde{x}))}{\sqrt{1-\rho_0^2}} \Big).
\end{aligned}$$ Thus one approach to estimating $p_{\leq \thetaTarg}(\cdot)$ is to use the formula \begin{equation}\label{eq:EstimatedCondCDF_Copula} \hat{p}_{\leq \thetaTarg}(\argxProxy) =  \Phi \Big( \frac{\Phi^{-1}(q)- \hat{\rho}_0 \cdot \Phi^{-1}(\hat{F}_{\tilde{X}}(\tilde{x}))}{\sqrt{1-\hat{\rho}_0^2}} \Big). \end{equation} Above $\hat{F}_{\tilde{X}}$ is the empirical CDF based on the $N$ Phase I samples while $\hat{\rho}_0$ is an estimate of $\rho_0$ based on all previously available data (which can include Phase II data from previous waves). For $k^* \in [K] \setminus \{1\}$, we estimate $\rho_0$ in model \eqref{eq:CopulaModel} prior to the start of wave $k^*$ by maximizing the weighted log-likelihood for this model: $$\hat{\rho}_0 = \argmax_{\rho \in (-1,1)} \Biggl\{ \sum_{i=1}^N \mathcal{W}_i^{(k^*-1)} \Big( \mathbbm{1} \{ \goodX_i \leq \hat{\theta}^{(\text{MPD},k^*-1)} \} \log \big( \check{p}_i(\rho) \big) +  \mathbbm{1} \{ \goodX_i > \hat{\theta}^{(\text{MPD},k^*-1)} \}\log\big(1-\check{p}_i(\rho) \big) \Big)   \Biggr\},$$
 $$\text{where} \quad \check{p}_i(\rho) = \Phi \Big( \frac{\Phi^{-1}(q)- \rho \cdot \Phi^{-1}(\hat{F}_{\tilde{X}}(\proxyX_i))}{\sqrt{1-\rho^2}} \Big) \quad \text{for each } \rho \in (-1,1), i \in [N],$$ and where $\mathcal{W}_i^{(k^*-1)}$ is defined at \eqref{eq:WeightsAggFirstKstar} and $\hat{\theta}^{(\text{MPD},k^*-1)}$ is the Multiwave Predict-Then-Debias estimator of $\thetaTarg$ using the first $k^*-1$ waves of data. For the first wave (i.e., $k^*=1$), $\rho_0$ cannot be estimated directly, so we next present a way to guess $\rho_0$ that leverages prior information.

\subsubsection{First wave strategy using Spearman's rank correlation coefficient}\label{sec:FirstWaveStrategyPrior}

In most motivating applications $\proxyX$ is an estimate of $\goodX$ from a pretrained machine learning model, which may have been trained on a different population. Sometimes the investigator has access to prediction accuracy metrics such as RMSE or correlation coefficients that were released by the individual or group that trained the machine learning model. We suppose that the investigator has access to an estimate of the Spearman's rank correlation coefficient between $\goodX$ and $\proxyX$, which assuming $F_{\goodX}(\cdot)$ and $F_{\proxyX}(\cdot)$ are continuous, is given by $$\rho_S = \corr \big( F_X (X), F_{\tilde{X}}(\tilde{X}) \big).$$ We call $\hat{\rho}_S^{\text{(avail)}}$ the available estimate of $\rho_S$. We note that in many applications, if the target study population is drawn from a different time period or region than the machine learning model was trained and tested on, $\hat{\rho}_S^{\text{(avail)}}$ likely overestimates the Spearman's rank correlation $\rho_S$ due to prediction accuracy reductions that are driven by distribution shifts. Fact \ref{fact:CopulaRhoToSpearman} below gives a formula for $\rho_0$ in terms of $\rho_S$, when $\proxyX$ and $\goodX$ satisfy the bivariate Gaussian copula model at \eqref{eq:CopulaModel}. Using Fact \ref{fact:CopulaRhoToSpearman} and the (presumed) overestimate $\hat{\rho}_S^{\text{(avail)}}$ of $\rho_S$, in our numerical experiments we set \begin{equation}\label{eq:rho0CopulaInitUsingSpearman}
    \hat{\rho}_0^{(\text{init})}= 0.9 \times \Big( 2 \cdot \sin (\pi \hat{\rho}_S^{\text{(avail)}}/6) \Big).
\end{equation} Some of our numerical experiments use $\hat{\rho}_0^{(\text{init})}$ to obtain the copula-based estimate of $\hat{p}_{\leq \thetaTarg} (\cdot)$ using Equation \eqref{eq:EstimatedCondCDF_Copula}. (This is only done for the first wave, and subsequent waves use the approach in Section \ref{sec:FittedBivariateCopulaStrategy} to estimate $\hat{p}_{\leq \thetaTarg} (\cdot)$). 

\begin{fact}\label{fact:CopulaRhoToSpearman}
   If $\proxyX$ and $\goodX$ have continuous CDFs and satisfy the bivariate Gaussian copula model at \eqref{eq:CopulaModel}, then \begin{equation}\label{eq:CopulaRhoToSpearman}
    \rho_0 = 2 \cdot \sin (\pi \rho_S/6)
\end{equation} where $\rho_S = \corr \big( F_X (X), F_{\tilde{X}}(\tilde{X}) \big)$ is the Spearman's rank correlation coefficient.
\end{fact}
\begin{proof}
    Note that under the copula model \eqref{eq:CopulaModel}, $$\rho_0 = \corr \Big( \Phi^{-1} \big( F_X (X) \big), \Phi^{-1} \big( F_{\tilde{X}}(\tilde{X}) \big) \Big)=\corr(Z,\tilde{Z}),$$ while $$\rho_S = \corr \big( F_X (X), F_{\tilde{X}}(\tilde{X}) \big)=\corr \big( \Phi(Z), \Phi(\tilde{Z}) \big) = \frac{\e[\Phi(Z)\Phi(\tilde{Z})]-1/4}{1/12}=12 \cdot \e[ \Phi(Z)\Phi(\tilde{Z})]-3.$$ Now let $\varepsilon_1,\varepsilon_2 \sim \mathcal{N}(0,1)$ be independent of each other and independent of $Z$ and $\tilde{Z}$. Observe that $$\e[ \Phi(Z)\Phi(\tilde{Z})]= \e \Big[ \mathbb{P}( \varepsilon_1 \leq Z \giv Z, \tilde{Z}) \cdot \mathbb{P}( \varepsilon_2 \leq \tilde{Z} \giv Z, \tilde{Z}) \Big]= \e \Big[ \mathbb{P}( \varepsilon_1 \leq Z,\varepsilon_2 \leq \tilde{Z} \giv Z,\tilde{Z} \big) \Big] = \mathbb{P}( \varepsilon_1 \leq Z,\varepsilon_2 \leq \tilde{Z}).$$ Thus 
    $$\e[ \Phi(Z)\Phi(\tilde{Z})]=\mathbb{P}( \varepsilon_1 -Z \leq 0,\varepsilon_2- \tilde{Z} \leq 0) = \mathbb{P}\Big( \frac{\varepsilon_1 -Z}{\sqrt{2}} \leq 0,\frac{\varepsilon_2- \tilde{Z}}{\sqrt{2}} \leq 0 \Big),$$ 
    where $(\varepsilon_1,\varepsilon_2,Z,\tilde{Z})$ are jointly Gaussian with mean 0, variance 1, and the only cross correlations being $\corr(Z,\tilde{Z})=\rho_0$, so $$\begin{bmatrix}
    \varepsilon_1 -Z \\ \varepsilon_2- \tilde{Z}
\end{bmatrix} \sim \mathcal{N} \Big(0, \begin{bmatrix}
    2 & \rho_0 \\ \rho_0 & 2
\end{bmatrix} \Big) \Rightarrow \begin{bmatrix}
    (\varepsilon_1 -Z) \cdot 2^{-1/2} \\ (\varepsilon_2- \tilde{Z}) \cdot  2^{-1/2}
\end{bmatrix} \sim \mathcal{N} \Big(0, \begin{bmatrix}
    1 & \rho_0/2 \\ \rho_0/2 & 1
\end{bmatrix} \Big).$$ Hence by the formula for the probability that a centered and standardized bivariate Gaussian lies in the lower-left quadrant,\footnote{See, for example, \url{https://mathworld.wolfram.com/BivariateNormalDistribution.html}} $$\e[ \Phi(Z)\Phi(\tilde{Z})]=\mathbb{P}\Big( \frac{\varepsilon_1 -Z}{\sqrt{2}} \leq 0,\frac{\varepsilon_2- \tilde{Z}}{\sqrt{2}} \leq 0 \Big)= \frac{1}{4} + \frac{\sin^{-1}(\rho_0/2)}{2 \pi}.$$ Combining this with an earlier formula for $\rho_S$, $$\rho_S = 12 \Big( \frac{1}{4} + \frac{\sin^{-1}(\rho_0/2)}{2 \pi} \Big) -3 = \frac{6 \sin^{-1}(\rho_0/2)}{\pi} \Rightarrow \rho_0 = 2 \cdot \sin (\pi \rho_S/6).$$
\end{proof}

The Spearman's rank correlation coefficient estimate can also be used to strategically select a $\varpi$ to design the first wave of Phase II (in later waves, an optimal $\varpi$ can be estimated empirically). In particular for the first wave the asymptotic variance of the greedy optimal objective is proportional to $$\e \Bigg[ \frac{ \big( \mathbbm{1} \{ \goodX \leq \thetaTarg \} - q - \varpi ( \mathbbm{1} \{ \proxyX \leq \gammaTarg \} - q) \big)^2}{\bar{\pi}^{(1)}(\proxyX)  } \Bigg] =\e \Bigg[ \frac{\varrho(\proxyX)}{\bar{\pi}^{(1)}(\proxyX)  } \Bigg].$$ Following the argument in Section \ref{sec:EstimateQuantileSpecificGreedyOpt}, if one ignores the overlap constraints, the above objective is minimized (for a fixed $\varpi$) when $\bar{\pi}^{(1)}(\proxyX)$ is proportional to $\sqrt{\varrho(\proxyX)}$. Hence ignoring overlap constraints, the asymptotic variance of an estimator using just Phase I and Phase II data from wave 1 is given by $$\kappa_5 \e [\sqrt{\varrho(\proxyX)}]+\kappa_6 = \kappa_5 \e \Bigg[ \sqrt{ p_{\leq \thetaTarg}(\proxyX)  \big(1- p_{\leq \thetaTarg}(\proxyX) \big) +  \big( (p_{\leq \thetaTarg}(\proxyX)- q )- \varpi ( \mathbbm{1} \{ \proxyX \leq \gammaTarg \} - q) \big)^2 } \Bigg]+\kappa_6$$ for some scalars $\kappa_6$ and $\kappa_5 > 0$ that do not depend on $\varpi$. Hence prior to the first wave, we minimize an empirical approximation to the above objective and set $$\hat{\varpi}_{\text{opt}}^{(1)} = \argmin_{\varpi \in \mathbb{R}} \Biggl\{ \frac{1}{N} \sum_{i=1}^N \sqrt{ \hat{p}_{\leq \thetaTarg}(\proxyX_i)  \big(1- \hat{p}_{\leq \thetaTarg}(\proxyX_i) \big) +  \big( (\hat{p}_{\leq \thetaTarg}(\proxyX_i)- q )- \varpi ( \mathbbm{1} \{ \proxyX_i \leq \hga \} - q) \big)^2 } \Biggr\}.$$ Above, $\hat{p}_{\leq \thetaTarg}$ is calculated using Equation \eqref{eq:EstimatedCondCDF_Copula} (which assumes a copula model), with $\hat{\rho}_0=\hat{\rho}_0^{\text{(init)}}$ (which is the estimate at \eqref{eq:rho0CopulaInitUsingSpearman} based on some prior knowledge of Spearman's rank correlation coefficient). Note that in many arms of our simulations, the first wave is a uniform random sample (and does not rely on knowledge of Spearman's rank correlation coefficient). In those arms, $\hat{\varpi}_{\text{opt}}^{(1)}$ need not be estimated.

\subsection{Estimating the Hessians for confidence interval construction}\label{sec:EstimatingHessiansForQuantileEst}

Estimation of the Hessians $H_{\thetaTarg}$ and $H_{\gammaTarg}$ using the formula in \eqref{eq:HessianQuantileFormula} can directly be done using density estimation, but doing so can perform poorly when the sample sizes are too small. Instead, in our numerical quantile estimation experiments, we use the following strategy for estimating the Hessians. This strategy is analogous to the approach used in quantile regression software packages. 

Observe that  $$H_{\thetaTarg}^{-1} =\frac{1}{f_{\goodX}(\thetaTarg) }= \frac{1}{f_{\goodX}(F_{\goodX}^{-1}(q)) }= \frac{\rd}{\rd u} F_{\goodX}^{-1} (u) \mid_{u=q}, \quad \text{and similarly,} \quad H_{\gammaTarg}^{-1}=\frac{\rd}{\rd u} F_{\proxyX}^{-1} (u) \mid_{u=q}.$$ We estimate $H_{\thetaTarg}^{-1}$ and $H_{\gammaTarg}^{-1}$ empirically in our numerical experiments by selecting bandwidths $h, \tilde{h} \in (0,1)$ and by computing empirical estimates of $$\frac{F_{\goodX}^{-1}(q+h)-F_{\goodX}^{-1}(q-h)}{2h} \quad \text{and} \quad \frac{F_{\proxyX}^{-1}(q+\tilde{h})-F_{\proxyX}^{-1}(q-\tilde{h})}{2\tilde{h}},$$ respectively. Above, the terms in the numerator are estimated via weighted quantile estimation and the bandwidths $h$ and $\tilde{h}$ are chosen using the same Hall-Sheather formula \citep{HallSheatherPaper} as implemented in the \texttt{bandwidth.rq()} function in the \texttt{quantreg} R package \citep{QuantregPakcage}. Since the Hessians were estimated using a weighted sample, we used a weight-based effective Phase II sample size formula rather than the actual Phase II sample size to plug into the bandwidth selection formula.

\section{Existing theoretical results in the two-phase multiwave sampling literature}\label{sec:LitSearchTwoPhaseMultiwaveTheory}

In this appendix, we conduct a literature search to provide evidence in support of our claim that our paper is the first paper to provide theoretical guarantees of either asymptotic linearity or asymptotic normality of M-estimators under two-phase multiwave sampling designs. Our literature search considered all papers that cited \cite{AdaptiveSamplingInTwoPhaseDesignsMcIsaacAndCook} as well as papers that cited those papers, and we filtered these papers to only include ones that (i) study two-phase sampling with an adaptively collected Phase II sample and (ii) provide theoretical guarantees (see Figure \ref{fig:LitSearchDiagram} for an overview). This search resulted in 10 selected papers (summarized in Table \ref{tab:filter2}). We next describe the search procedure in more detail. Subsequently we discuss the 10 selected papers in more detail and specify how we concluded that none of the 10 selected papers contradict our novelty claims. 


\begin{figure}
    \centering
    \includegraphics[width=0.8\linewidth]{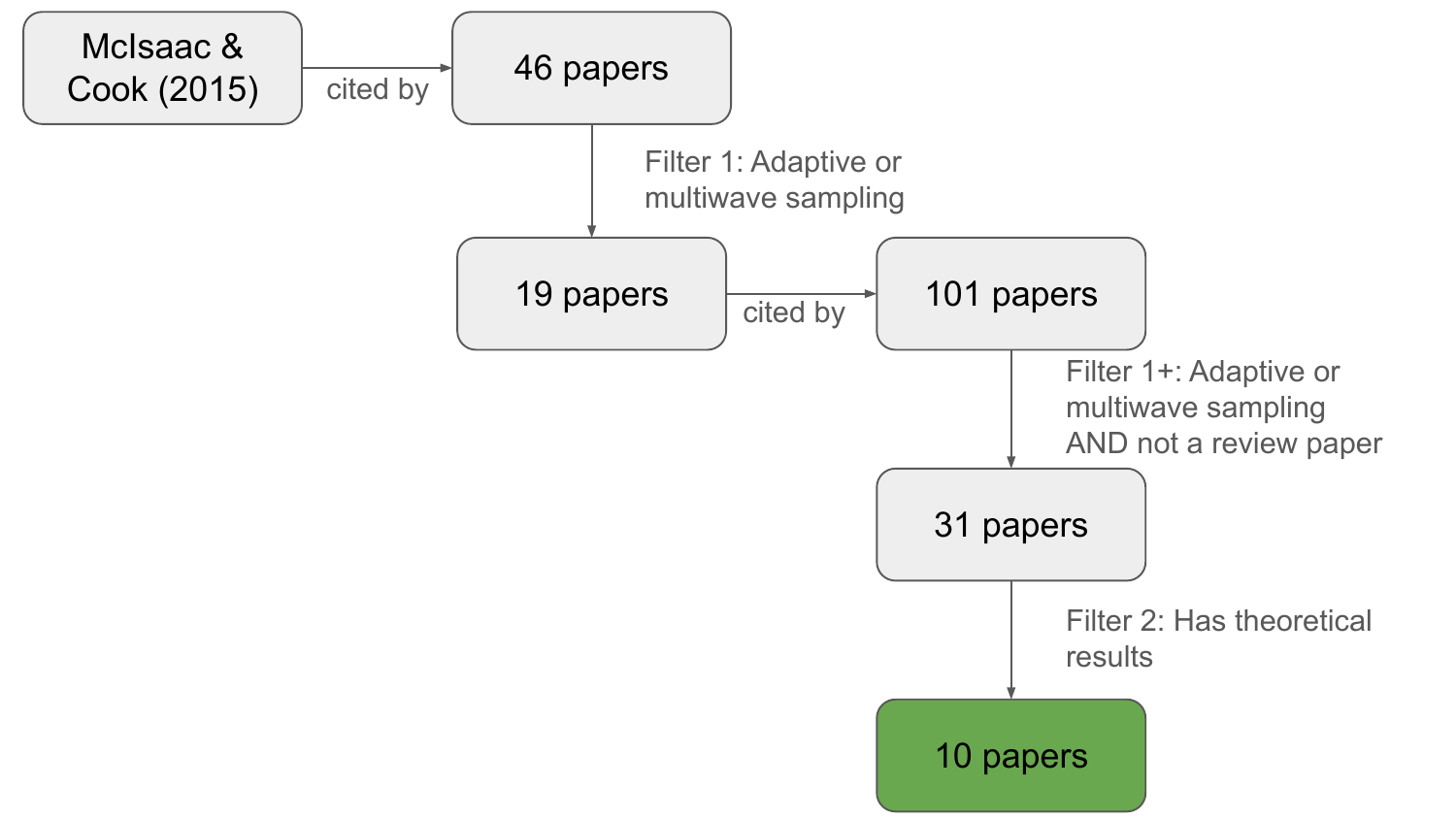}
    \caption{Overview of literature search to identify papers in the two-phase multiwave sampling literature that have theoretical contributions.}\label{fig:LitSearchDiagram}
\end{figure}

\subsubsection*{Two-phase multiwave literature search and filters}

In our literature search we sought to identify papers within the two-phase multiwave sampling literature that present formal theoretical results. To do this we searched through and filtered descendants of \cite{AdaptiveSamplingInTwoPhaseDesignsMcIsaacAndCook} in the citation tree. To our knowledge, \cite{ChenLumleyMultiwaveInTwoPhase} is the first study to use the term multiwave sampling to refer to an adaptive Phase II sampling approach in a two-phase sampling study. They attribute \cite{AdaptiveSamplingInTwoPhaseDesignsMcIsaacAndCook} as having proposed these sampling designs (although \cite{AdaptiveSamplingInTwoPhaseDesignsMcIsaacAndCook} refers to their two-phase two-wave adaptive sampling approach with the terminology of a Phase-IIa internal pilot sample and a subsequent Phase-IIb sample). We therefore conduct a literature search for descendants of \cite{AdaptiveSamplingInTwoPhaseDesignsMcIsaacAndCook} in the citation tree, which includes descendants of \cite{ChenLumleyMultiwaveInTwoPhase}. 

To identify relevant descendants of \cite{AdaptiveSamplingInTwoPhaseDesignsMcIsaacAndCook} in the citation tree, we used the citation chaser Shiny app \citep{citationChaser}. To account for recent arXiv papers not captured by this tool, we also used Google Scholar to identify arXiv papers that were posted in 2023 or later and cite \cite{AdaptiveSamplingInTwoPhaseDesignsMcIsaacAndCook}. We then filtered all of the papers that cited \cite{AdaptiveSamplingInTwoPhaseDesignsMcIsaacAndCook} by reading abstracts and conducting within-paper keyword searches so that only papers that study or propose two-phase sampling with an adaptive Phase II sample passed the filter (Filter 1). Papers that merely mentioned adaptive sampling in the discussion as a future direction were filtered out, and we also filtered out an earlier preprint of our paper. This resulted in 18 papers that cited \cite{AdaptiveSamplingInTwoPhaseDesignsMcIsaacAndCook} and passed our filter (with \cite{AdaptiveSamplingInTwoPhaseDesignsMcIsaacAndCook} being the 19th paper). We used the same methodology of combining the citation chaser Shiny app \citep{citationChaser} with a Google Scholar citation search for recent arXiv papers, to identify 101 unique papers (\cite{AdaptiveSamplingInTwoPhaseDesignsMcIsaacAndCook}, the ones that cite it and pass Filter 1, and the ones that cite those papers). These citation searches were conducted on September 16, 2026.

To narrow down these 101 unique papers to the most relevant ones, we applied a stronger variant of Filter 1 (which we call Filter 1+) in which we removed review papers in addition to papers that do not study two-phase designs with an adaptive Phase II sample because review papers were unlikely to have theoretical contributions. This resulted in a suite of 31 unique papers, 19 of which had their citations searched previously (we did not search the citations of the other 12 papers given that many of them were quite recent and we did not suspect them to have descendants that passed our next filter). 

We further filtered the 31 papers that studied two-phase designs with an adaptive Phase II sample, by screening for papers that presented theoretical contributions (Filter 2). In particular, we searched for the keywords ``proof'', ``prove'', ``theorem'', ``lemma'', and ``proposition'' throughout the main text and supplement of each paper to see if the paper had any formally stated or proven theoretical results. Ten of these papers had formally stated or proven theoretical results.

\subsubsection*{Discussion of 10 selected papers with theoretical results}

We examined each of the 10 remaining papers to see whether it provides theoretical guarantees of asymptotic linearity or asymptotic normality of M-estimators under two-phase multiwave sampling, with theoretical guarantees that account for the statistical dependencies induced by these designs. None of these papers contradicts our theoretical novelty claims for at least one of four reasons:

\begin{enumerate}[(a)]
    \item In many of the 10 selected papers, no theoretical results are presented after the multiwave (or two-wave) Phase II sampling scheme is formally introduced  \citep{HanEtAl_2PhaseSurvival2020,ChenRecentSurrogatePoweredInferenceMultiwave,ZengEtAl2025Causal,LeeJRSSA_2026,OrthognalValHait26}. In particular, these papers are structured as follows. First, they study estimators under two-phase sampling where the Phase II sampling probabilities are determined by the Phase I data only, and present theorems establishing asymptotic linearity \citep{ZengEtAl2025Causal,OrthognalValHait26} or asymptotic normality \citep{HanEtAl_2PhaseSurvival2020,ZengEtAl2025Causal,LeeJRSSA_2026,ChenRecentSurrogatePoweredInferenceMultiwave,OrthognalValHait26}. Second, they use these results to derive optimal Phase II sampling rules, which is sometimes proven in a formal theoretical statement \citep{ZengEtAl2025Causal,OrthognalValHait26}. Third, after noting that the optimal Phase II sampling rule depends on quantities that are unknown in practice but could be estimated with Phase II data, they consider either using a multiwave sample \citep{ChenRecentSurrogatePoweredInferenceMultiwave,OrthognalValHait26} in Phase II or splitting Phase II into a pilot wave and a subsequent wave \citep{HanEtAl_2PhaseSurvival2020,ZengEtAl2025Causal,LeeJRSSA_2026}. Notably after introducing such multiwave or two-wave sampling schemes for Phase II, theoretical guarantees are not provided in subsequent sections of these works and hence do not account for statistical dependencies induced by the adaptive sampling design in Phase II.
     
     \item Some of the 10 selected papers do not present formal theoretical statements that establish asymptotic linearity or normality \citep{ChenAndLumleyPaper2,JasperOptimalSubsamplingViaStratification,Recent2PhaseMultiwaveGeneralizedRaking}. Instead, \cite{ChenAndLumleyPaper2} presents a formal proof (in Appendix A) of the relationship between two conditional variances. Meanwhile, \cite{Recent2PhaseMultiwaveGeneralizedRaking} presents a formal proof of an exact integer-valued solution to an optimization problem. While both of these papers study two-phase multiwave sampling designs and provide formally stated proofs, neither provides a formal theoretical statement or proof establishing asymptotic linearity or normality under these designs. \cite{JasperOptimalSubsamplingViaStratification} proves asymptotic variance formulas for their estimator under the setting where their Phase II sample comes from either a single wave of Poisson sampling or a single wave of stratified sampling (Section S.1.1.1), and they also provide a theoretical comparison of these variance formulas (Section S.1.1.2). The derivations of these asymptotic variance formulas do not explicitly account for the case where the Phase II sample is split into a pilot sample and a subsequent sample considered in different parts of their study. 

    \item Our paper (initially posted as a preprint in February 2026) predates \cite{TSNeymanWithJKKim_2026} (posted June 2026) and \cite{OrthognalValHait26} (posted August 2026). Moreover, \cite{TSNeymanWithJKKim_2026} studies a different adaptive sampling design (analogous to Thompson sampling) than two-phase multiwave sampling and focuses on mean estimation tasks.

     \item At least one of the selected papers explicitly discusses the theoretical challenges associated with conducting statistical inference when Phase II is adaptive \citep{wang2025maximinoptimalapproachsampling}. In particular, \cite{wang2025maximinoptimalapproachsampling} develops theoretical statements about the optimal sampling rules (assuming a single wave in Phase II). Noting that this optimal sampling rule has unknown parameters they then consider the setting where Phase II is split into a pilot wave and subsequent wave whose sampling probabilities depend on observations from the first wave. While they are able to provide theoretical guarantees about the rate of convergence of the estimated sampling rule used for the second wave (Theorems 4 and 5), they are unable to provide theoretical guarantees for their estimator based on this two-phase two-wave (pilot and subsequent wave) adaptive sample. They acknowledge this as a limitation and an area for future work stating that: \begin{quote} \footnotesize
         ``As noted by an anonymous reviewer, the use of the pilot sample breaks the i.i.d. structure across observations because the distribution of the sampling indicator in the second phase depend on the pilot sample. This complicates the theoretical analysis of the proposed one-step estimator $\hat{\theta}$ \dots Although not able to rigorously prove, we conjecture that this deviation does not cause serious issues. Our numerical results indicate that, despite of the non-i.i.d. structure, the proposed estimator $\hat{\theta}$ performs well in practice \dots Based on these observations, we recommend the one-step estimator $\hat{\theta}$ for practical use, while leaving a rigorous investigation of its theoretical properties to future research.''---\cite{wang2025maximinoptimalapproachsampling}.
    \end{quote}

\end{enumerate}

\begin{table}[htbp]
\centering
\footnotesize
\setlength{\tabcolsep}{6pt}
\begin{threeparttable}
\caption{Overview of the ten selected papers and our theoretical contribution to the two-phase multiwave sampling literature.}
\label{tab:filter2}
\begin{tabular}{@{}l p{3.2cm} c c c@{}} 
\toprule
Reference
  & \makecell{Target estimand \\ or type of \\ estimation task}
  & \makecell{Formal proof\\ of asymptotic \\ linearity or\\ normality?}
  & \makecell{Proof is for\\setting of adaptive\\ sampling in \\ Phase II?}
  & \makecell{Predates\\ March 2026?} \\
\midrule
\citet{HanEtAl_2PhaseSurvival2020}                   & Survival reg. coeffs. & \cmark & \xmark & \cmark \\
\citet{ChenAndLumleyPaper2}                          & Regression coeffs. & \xmark & \naentry & \cmark \\
\citet{ZengEtAl2025Causal}                           & Average causal effect & \cmark & \xmark & \cmark \\
\citet{JasperOptimalSubsamplingViaStratification}    & Regression coeffs. & \xmark & \naentry & \cmark \\
\citet{Recent2PhaseMultiwaveGeneralizedRaking}       & General & \xmark & \naentry & \cmark \\
\citet{wang2025maximinoptimalapproachsampling}       & General & \xmark & \naentry & \cmark \\
\citet{ChenRecentSurrogatePoweredInferenceMultiwave} & Z-estimation & \cmark & \xmark & \cmark \\
\citet{LeeJRSSA_2026}                                & Accuracy metrics & \cmark & \xmark & \cmark \\
\citet{TSNeymanWithJKKim_2026}\tnote{a}              & Mean estimation & \cmark & \cmark & \xmark \\
\citet{OrthognalValHait26}                           & Cox model coeffs. & \cmark & \xmark & \xmark \\
\midrule
\textbf{Our paper}                                  & M-estimation & \cmark & \cmark & \cmark \\
\bottomrule
\end{tabular}
\begin{tablenotes}[flushleft]\footnotesize
\item \cmark\ = yes; \xmark\ = no; \naentry\ = not applicable.
\item[a] \citet{TSNeymanWithJKKim_2026} studies mean estimation via an adaptive
Thompson sampling-type scheme. This differs from two-phase multiwave sampling designs or two-phase two-wave (with a pilot wave) sampling designs considered in the other studies listed.
\end{tablenotes}
\end{threeparttable}
\end{table}

\newpage
\section{Additional simulation and dataset details}
\label{sec:SimulationDetails}

In the next subsection, we provide a more detailed description of the Monte Carlo procedure used to study \samplingSchemeName{} on each dataset. In the subsequent subsection, we provide additional dataset and implementation details about each of the 5 experiments from Section \ref{sec:Simulations}. We then discuss in detail our additional experiments, including ones that compare our approach to generalized raking-type estimators and to other sampling schemes that do not induce statistical dependencies. We also give additional details on the attainable crop yield estimation experiment described in Section \ref{sec:AttainableYieldSims}. Code and data that were used in these experiments is available on GitHub (link not shown to preserve anonymity).

\subsection{Monte Carlo procedure}\label{sec:MonteCarloProcedure}

For each experiment, we conduct the following analysis. We start with $\nSuper$ fully observed samples of $(\xobs,\txmiss,\xmiss)$ which are taken from a dataset described in Section \ref{sec:Datasets}. Then for the purposes of testing our methods in a data-driven manner, we set the empirical distribution of these $\nSuper$ samples to be the distribution of the superpopulation $\datvecrawDist$ that we study in the experiments. In particular, we calculate the ``ground truth" value for the parameter of interest $\thetaTarg$ by finding the empirical loss minimizer across the $\nSuper$ fully observed samples. 

For each $K \in \{2,6,26\}$, we performed $1{,}000$ independent Monte Carlo simulations. In each simulation:

\begin{enumerate}
    \item \textbf{Phase I} was run by collecting a sample of size $N$ with replacement from the $\nSuper$ observations of $(\xobs,\txmiss,\xmiss)$. This resulted in a Phase I sample $(\xobs_i,\txmiss_i)_{i=1}^N$, where the corresponding $(\xmiss_i)_{i=1}^N$ values are unobserved but are stored and withheld in order to implement Phase II.
    \item \textbf{The ``explore" wave of Phase II} was run by generating $I_i^{(1)} \stackrel{\text{i.i.d.}}{\sim} \text{Bernoulli}(n_{\text{targ}}^{(1)}/N)$ and collecting $\xmiss_i$ observations for all $i \in [N]$ such that $I_i^{(1)}=1$.
    \item \textbf{The ``exploit" waves of Phase II:} For each $k \in [K] \setminus \{1\}$, the collected data from previous waves and from Phase I was used to estimate an approximate greedy optimal sampling rule, according to the approach described in Section \ref{sec:DescribeApproximateGreedyOptimalStrategy}. The parameter $\hat{\lambda}^{(k)}$ was calculated using the procedure in Appendix \ref{sec:EnforceBudgetConstraints} so that in expectation, $n_{\text{targ}}^{(k)}$ measurements were collected in wave $k$. 
    We set $b_u=0.95$, $$n_{\text{targ}}^{(k)} = \frac{n_{\text{targ}} - n_{\text{targ}}^{(1)}}{K-1}, \quad \text{and} \quad b_{l} =   \frac{n_{\text{targ}}^{(k)}}{4 \cdot N_u^{(k-1)}},$$ where $N_u^{(k-1)} = \vert \{ i \in [N] \ : \ I_i^{(k')}=0 \text{ for each } k' \in [k-1] \} \vert$ denotes the number of samples that have not been labelled prior to wave $k$. These choices spread the labelling budget evenly across the $K-1$ waves and allowed the labelling probabilities for uninformative points to be up to 4 times smaller under adaptive sampling compared to under uniform random sampling.

    \item \textbf{Point estimators and confidence intervals} were computed after Phase II. Using all available data from both phases, $\htc$, $\hgc$, and $\hga$ defined at \eqref{eq:PTDComponentEstimatorsDEF} were computed using standard statistical software for weighted estimators. The optimal diagonal tuning matrix, $\htPTD$, and 90\% confidence intervals were computed using formulas \eqref{eq:EstimatedOptimalDiagTuningMatrix}, \eqref{eq:PTD_estimator}, and \eqref{eq:CIDef}, respectively.
     
    \item \textbf{The uniform sampling baseline} was considered using the same Phase I data. 
    For each $i \in [N]$, $I_{i}^{\text{base}} \stackrel{\text{i.i.d.}}{\sim} \text{Bernoulli}(n_{\text{targ}}/N)$ was generated. The Predict-Then-Debias estimator of $\thetaTarg$ (with an asymptotically optimal diagonal tuning matrix given by \eqref{eq:EstimatedOptimalDiagTuningMatrix}) was evaluated using the samples $(I_{i}^{\text{base}},\xobs_i,\txmiss_i, I_{i}^{\text{base}} \xmiss_i)_{i=1}^N$. This baseline used the same Phase I data and same number of expected Phase II samples.
\end{enumerate}

Step 3 involved estimating the conditional expectation $\varrho_j(\cdot)$ defined in \eqref{eq:varrhoDef}. We conducted two separate sets of the above Monte Carlo simulations each with different approaches for estimating this quantity. In the first set of simulations, we use 20-Nearest-Neighbors (implemented via the \texttt{knn.reg()} function in the \texttt{FNN} package \citep{FNNPackage}) to estimate the conditional expectation in \eqref{eq:varrhoDef}. In the second set of simulations, the space $\xProxySpace$ of Phase I observations was stratified into prespecified strata and \eqref{eq:varrhoDef} was approximated with a function that is constant within each stratum using the approach described in Section \ref{sec:EstimateOptimalLabellingProbabilityInEachStrata}. Details of the prespecified strata used for each experiment can be found in Section \ref{sec:DatasetsAdditionalDetails}. In summary, the strata were selected by taking the Cartesian product of percentile bins of the one to three variables that were expected to be most critical for estimating the quantity of interest.

\subsection{Additional details about each dataset and experiment}\label{sec:DatasetsAdditionalDetails}

In this appendix we give additional details about the datasets used and implementation of each experiment.

\paragraph{Synthetic experiment:} We generated a large synthetic dataset of an outcome variable $Y$, a continuous covariate $Z_{\text{cov}} \in \mathbb{R}$, a binary treatment variable $Z_{\text{trt}} \in \{0,1\}$, and a binary estimate $\tilde{Z}_{\text{trt}} \in \{0,1\}$ of $Z_{\text{trt}}$. The estimand of interest was the population regression coefficient corresponding to the treatment variable in a linear regression of $Y$ on $(Z_{\text{cov}},Z_{\text{trt}})$. The $(Y,Z_{\text{cov}},Z_{\text{trt}},\tilde{Z}_{\text{trt}})$ data were generated so that both \begin{enumerate}[(i)]
    \item the residuals when regressing $Y$ on $(Z_{\text{cov}},Z_{\text{trt}})$ were heteroskedastic, and
    \item the prediction errors $\tilde{Z}_{\text{trt}}-Z_{\text{trt}}$ were differential in the sense that $Z_{\text{trt}} \centernot{\indep} Y \giv (Z_{\text{cov}},\tilde{Z}_{\text{trt}})$.
\end{enumerate}

More specifically, we generated $$Z_{\text{cov}} \sim \mathcal{N}(0,1), \quad Z_{\text{trt}} \sim \text{Bernoulli}(1/2), \quad \varepsilon_1,\varepsilon_2,\varepsilon_3,\varepsilon_4 \stackrel{\text{i.i.d.}}{\sim} \mathcal{N}(0,1), \quad \text{and}  \quad U \sim \text{Unif}(0,1), $$ all independently. We then set the outcome to be given by $$Y=Z_{\text{cov}}+Z_{\text{trt}}+\varepsilon_Y \quad \text{where} \quad \varepsilon_Y=\varepsilon_1 +10 Z_{\text{trt}} \varepsilon_2+ \vert Z_{\text{cov}} \vert \varepsilon_3+ 3 Z_{\text{trt}} \vert Z_{\text{cov}} \vert \varepsilon_4,$$ resulting in a setting with a well-specified linear model with heteroskedastic residuals. To generate a proxy for $Z_{\text{trt}}$ with differential prediction errors, we set the logit scores to be $$\xi= 4 Z_{\text{trt}} +Z_{\text{cov}}+ Y/\sqrt{\var(Y)} -\mu$$ where $\mu$ was chosen so that $\xi$ had approximately mean zero, and we set $$\tilde{Z}_{\text{trt}} = \mathbbm{1} \Bigl\{U \leq \frac{\exp(\xi)}{1+\exp(\xi)} \Bigr\}.$$ 

We generated a large dataset with $\nSuper= 3 \times 10^6$ i.i.d. draws of $(Y,Z_{\text{cov}},Z_{\text{trt}},\tilde{Z}_{\text{trt}})$ from the distribution defined by the above generating process. 
In our experiments, gold-standard measurements of $Z_{\text{trt}}$ were not collected during Phase I (but its proxy $\tilde{Z}_{\text{trt}}$, $Y$, and $Z_{\text{cov}}$ were). For the stratified adaptive approach, $\xProxySpace$ was partitioned into 18 strata defined by the Cartesian product of the terciles of $Y$, the terciles of $Z_{\text{cov}}$, and the binary value of $\tilde{Z}_{\text{trt}}$.

\paragraph{Housing price experiment:} We used a dataset of economic and environmental variables from $\nSuper=46{,}418$ distinct $\sim 1 \text{km} \times 1 \text{km}$ grid cells that was taken from \cite{MOSAIKSPaper,MOSAIKSSourceCode} and has previously been used to study the Predict-Then-Debias method \citep{KerriRSEPaper,PTDBootstrapPaper}. 
The dataset included grid cell-level averages of housing price, income, nightlight intensity, and road length as well as estimates for nightlights and road length based on daytime satellite imagery. For stratified adaptive sampling, $\xProxySpace$ was partitioned into 15 strata defined by the Cartesian product of housing price terciles and predicted nightlight quintiles computed from the superpopulation.

\paragraph{AlphaFold experiment:} The dataset used consisted of $\nSuper=10{,}802$ samples that originated from \cite{bludau2022structural} and was downloaded from Zenodo \citep{PPIZenodo}. Each sample had indicators $Z_{\text{Acet}},Z_{\text{Ubiq}} \in \{0,1\}$ of whether there was acetylation and ubiquitination, and an indicator $Y_{\text{IDR}} \in \{0,1\}$ of whether the protein region was an intrinsically disordered region (IDR) coupled with a prediction of $Y_{\text{IDR}}$ based on AlphaFold \citep{AlphaFoldPaper}. We test the method on an estimation task considered in \cite{PTDBootstrapPaper} where the estimand of interest is the population-level interaction term in the logistic regression of $Y_{\text{IDR}}$ on $(Z_{\text{Acet}},Z_{\text{Ubiq}},Z_{\text{Acet}} \times Z_{\text{Ubiq}})$. In Phase I, $Y_{\text{IDR}}$ was not observed, but $Z_{\text{Acet}}$, $Z_{\text{Ubiq}}$ and predictions of $Y_{\text{IDR}}$ based on AlphaFold were assumed to be collected in Phase I. For
stratified adaptive sampling, $\xProxySpace$ was partitioned into four strata according to the four possible combinations of $Z_{\text{Acet}}$ and $Z_{\text{Ubiq}}$.
We did not consider k-Nearest-Neighbor approaches for estimating the approximate greedy optimal labelling rule given that in this case $\xProxySpace=\{0,1\}^3$ was the corners of the unit cube, so stratification could be done without loss of information while nearest-neighbor approaches would face many instances of arbitrary tie-breaking.

\paragraph{Forest cover experiment:} The dataset used consisted of $\nSuper=67{,}968$ samples of $\sim 1 \text{km} \times 1 \text{km}$ grid cells taken from the previously mentioned data source \citep{MOSAIKSPaper,MOSAIKSSourceCode}. The variables included the percent of tree cover and grid cell-level averages of population and elevation, as well as machine learning-based estimates for tree cover and population. We binarized the tree cover variable and the machine learning-based predictions of tree cover using the 10\% threshold (which is meaningful from a forestry perspective \citep{UDSAForestServiceReport}) to construct a forest cover indicator variable and a cheap-to-measure prediction of it. We test the method on an estimation task considered in \cite{PTDBootstrapPaper} where the estimand of interest was the logistic regression coefficient for population when regressing the forest cover indicator on elevation and population. For stratified adaptive sampling, $\xProxySpace$ was partitioned into 12 strata defined by the Cartesian product of the estimated forest cover indicator and sextiles of predicted population.

\paragraph{Tree cover quantile experiment:} Using the same dataset as in the forest cover experiment, we consider estimating the 0.75-quantile of the percent of tree cover across all of the grid cells. In contrast to the previous forest cover experiment, we do not binarize the percent tree cover as we instead seek to estimate its 0.75-quantile. In our experiments, gold-standard measurements of percent tree cover were not collected during Phase I sampling (but a proxy for it based on satellite imagery was). For stratified adaptive sampling, we stratified the space $\xProxySpace$ of Phase I observations by calculating the $100/15,200/15,\dots, 1400/15$ percentiles of the estimated tree cover on the superpopulation, and used those cutoffs to define 13 distinct strata (the first two percentiles and the minimum estimated tree cover were all $0$, resulting in fewer than 15 distinct strata). Notably, computing
the point estimator $\htPTD$ and its corresponding confidence intervals required a different approach to estimating the Hessians $H_{\thetaTarg}$ and $H_{\gammaTarg}$ than using the formula at \eqref{eq:HessianAndSigmaEstimators}, which is described in Appendix \ref{sec:EstimatingHessiansForQuantileEst}.


\subsection{Comparisons to alternative estimators and sampling schemes}

\subsubsection{Comparison to generalized raking estimators}\label{sec:EmpiricalComparisonToGeneralizedRaking}

For the four experiments involving linear or logistic regression, we compared our estimator to generalized raking estimators. In particular, we repeated the Synthetic, Housing price, AlphaFold and Forest cover regression experiments for the cases where there were $K=6$ waves. In each of the $1{,}000$ Monte Carlo simulations, after collecting the Phase II samples in the same manner described in Sections \ref{sec:MonteCarloProcedure} and \ref{sec:DatasetsAdditionalDetails}, we compute four different estimators using the same Phase I and Phase II data. We describe the four estimators and how corresponding confidence intervals were constructed below and present the results in Figure \ref{fig:ComparisonToRaking}.

\begin{enumerate}
    \item \textbf{The Multiwave Predict-Then-Debias estimator with optimal diagonal tuning matrix (green).} As in the simulations considered in the main text, we computed the estimator $\htPTD$, using the optimal diagonal tuning
matrix,  and 90\% confidence intervals, using formulas \eqref{eq:PTD_estimator}, \eqref{eq:EstimatedOptimalDiagTuningMatrix}, and \eqref{eq:CIDef}, respectively.
  
  \item \textbf{The Multiwave Predict-Then-Debias estimator with optimal full tuning matrix (dark green).} We similarly computed the estimator $\htPTD$, and 90\% confidence intervals, using formulas \eqref{eq:PTD_estimator} and \eqref{eq:CIDef}, respectively. In this case we instead used a tuning matrix given at  \eqref{eq:EstimatedFullOptimalTuningMatrix}. This tuning matrix estimated the optimal tuning matrix among all $d \times d$ (full) matrices.
    
    \item \textbf{Raking estimator using multiwave inverse probability weights (light purple).} We construct the raking-type estimator described in Section \ref{sec:GeneralizedRakingWithMultiwaveIPWWeights}, using the distance function $G(w',w)= w' \log(w'/w) - w' + w$ (Case 2 from Table \ref{tab:distance-functions}) and auxiliary variables for raking given by $A_i=H_{\gammaTarg}^{-1} \dot{l}_{\gammaTarg}(\proxyX_i)$ for each $i \in [N]$. These auxiliary variables can be viewed as oracle ones in the sense that they were constructed using the actual $\gammaTarg$ value rather than using an estimate of $\gammaTarg$ from Phase I data. Confidence intervals were constructed using the asymptotic normal approximation derived in Section \ref{sec:RakingAsymptoticNormality}.

  \item \textbf{Raking estimator using inverse probability weights (purple).} We implement generalized raking and construct confidence intervals using the \texttt{survey} package \citep{SuveryPackage}.
  Here, the distance function was also $G(w',w)= w' \log(w'/w) - w' + w$ and the auxiliary variables were the estimated influence terms from the Phase I data given by  $A_i=\hat{H}_{\gammaTarg}^{-1} \dot{l}_{\hga}(\proxyX_i)$ for each $i \in [N]$. The design weights used for generalized raking were the inverse probability weights, where we use estimates of the probability of labelling the $i$th sample given by the formula $$\hat{\pi}_i^{(\textnormal{II})} = \sum_{k=1}^K \LabelRulePi{k}(\proxyX_i) \prod_{j=1}^{k-1} \big(1-\LabelRulePi{j}(\proxyX_i) \big).$$

\end{enumerate}

\begin{figure}[!hbt]
    \centering
    \includegraphics[width=0.95\linewidth]{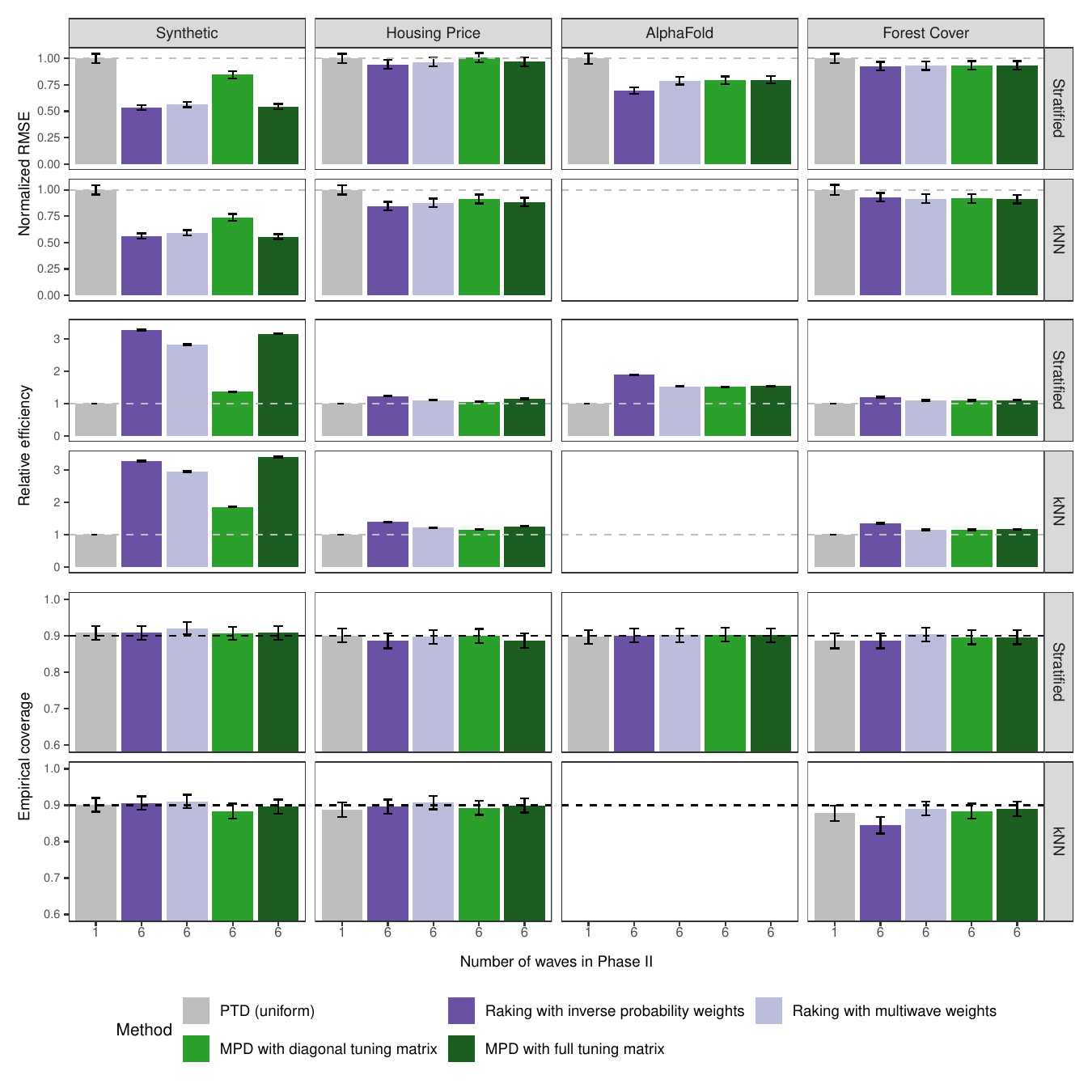}
    \caption{\small Comparison of Predict-Then-Debias and raking-type estimators after adaptive two-phase six-wave sampling. The experimental setups, evaluation metrics, and the approach to computing standard error bars across $1{,}000$ Monte Carlo simulations are the same as in Figure \ref{fig:ResultsAggregated}. In contrast to Figure \ref{fig:ResultsAggregated}, we do not consider quantile estimation tasks or the $K \in \{2,26\}$ settings. We also consider  four different estimators after each simulated two-phase multiwave sample rather than just the Multiwave Predict-Then Debias (MPD) with diagonal tuning (green). The four estimators considered in this figure are described in Section \ref{sec:EmpiricalComparisonToGeneralizedRaking}. The grey bars give the summary statistics for the Predict-Then-Debias (PTD) estimator (with diagonal tuning) under a nonadaptive Phase II sample, and are used as the baseline reference points for relative efficiency calculations and for the normalized RMSE metric.
    }\label{fig:ComparisonToRaking}
\end{figure}

The distance function and auxiliary variables used for raking are the same as those in \cite{ChenLumleyMultiwaveInTwoPhase}. In contrast to \cite{ChenLumleyMultiwaveInTwoPhase}, in our simulations the raking estimators are evaluated on data collected from two-phase multiwave sampling with Poisson rather than stratified sampling in each wave (Section \ref{sec:DescriptionOfSamplingScheme}) and with greedy optimal labelling strategies described in Section \ref{sec:DescribeApproximateGreedyOptimalStrategy}. The labelling strategy in Section \ref{sec:DescribeApproximateGreedyOptimalStrategy} is designed to enhance the efficiency of $\htPTD$ rather than that of raking estimators (see \cite{ChenAndLumleyPaper2} which investigates the additional benefits to raking estimator efficiency when the sample is designed specifically for raking estimator efficiency). 

A notable difference between the two raking estimators described above is that one estimator is constructed using multiwave inverse probability weights (given by Equation \eqref{eq:aggregated_Wk}) and has theoretical guarantees of asymptotic linearity and asymptotic normality under two-phase multiwave sampling designs (Section \ref{sec:GeneralizedRakingTheory}). The other raking estimator uses inverse probability weights given by $1/\hat{\pi}_i^{(\textnormal{II})}$ for each $i \in [N]$. To our knowledge, estimators based on the latter approach do not have existing theoretical guarantees of either asymptotic linearity or asymptotic normality under two-phase multiwave sampling designs.

The results (Figure \ref{fig:ComparisonToRaking}) indicate that generalized raking can lead to further efficiency gains over the estimator $\htPTD$ considered in the main text, most notably in the synthetic experiment. That being said, the efficiency gap is much smaller when considering a version of $\htPTD$ that uses an optimal tuning matrix among all $d \times d$ matrices rather than an optimal diagonal tuning matrix. The raking approach that uses inverse probability weights was found to be the most efficient, although it undercovered in the Forest Cover experiment when kNN was used to estimate the function $\varrho_j(\cdot)$ in the optimal labelling rule. Raking with multiwave inverse probability weights was generally comparable in efficiency to the Multiwave Predict-Then-Debias approach (with the tuning matrix from \eqref{eq:EstimatedFullOptimalTuningMatrix}), and neither undercovered, in line with the theoretical guarantees provided in Section \ref{sec:GeneralizedRakingTheory} and the main text, respectively.

\subsubsection{Comparison to sampling schemes that do not induce dependencies}\label{sec:ComparisonToSchemesThatDoNotInduceDependency}

We next compare the efficiency of \samplingSchemeName{} with two baseline sampling schemes that do not induce statistical dependencies. The comparison was conducted for the Housing Price, Forest Cover, and Tree Cover Quantile experiments in the main text, using the same dataset from the United States \citep{MOSAIKSPaper}, but the experiments were modified by restricting the Phase I and Phase II samples to be taken from the Eastern United States (according to the 100° W longitude threshold that is sometimes used in scientific studies \citep[e.g.,][]{PrecropPayoffsKluger}). We considered the same Phase I and Phase II sample sizes and estimands as in the main text (Table \ref{table:ExperimentSummary}), except the true values of the estimands were calculated using the superpopulation of available data from only the Eastern United States rather than from the whole United States. The two baseline sampling schemes considered in each of $1{,}000$ Monte Carlo simulations are described below.

\begin{itemize}
    \item \textbf{Phase II informed by external data:} We considered a scenario where data from the Western United States had previously been collected and made available, but no data from the Eastern United States is initially available. We estimated $\varrho_j(\cdot)$ by fitting kNN to data from the Western United States (specifically, data that was west of the 100° W meridian), and let $\hat{\varrho}_j^{(\text{west})}(\cdot)$ be such an estimate of $\varrho_j(\cdot)$. We then simulated conducting two-phase proxy-assisted sampling of data from the Eastern United States, using a single wave in Phase II, with labelling probabilities determined by Equation \eqref{eq:EstimatedGreedyOptLabellingRule}. That equation was applied for the case where $k^*=1$, $\hat{\varrho}_j^{(1)}=\hat{\varrho}_j^{(\text{west})}$ , $b_u=0.95$, $b_l=n_{\text{targ}}/(4N)$, and where $\hat{\lambda}^{(1)}$ was set so that $\sum_{i=1}^N \pi_{\cd_0}^{(1)}(\proxyX_i)= n_{\text{targ}}=1{,}000$. 
    
    \item \textbf{Pilot-and-discard:} We also simulated a sampling scheme where $\varrho_j(\cdot)$ was estimated on an internal pilot sample that was ultimately discarded to avoid inducing statistical dependencies. In particular we simulated collecting an internal pilot sample of $(\xobs,\txmiss,\xmiss)$-data of size $n_{\text{pilot}}=250$ from the Eastern United States. We used kNN and this internal pilot sample to estimate $\varrho_j(\cdot)$, and we let $\hat{\varrho}_j^{(\text{pilot})}(\cdot)$ denote the pilot-based estimate of $\varrho_j(\cdot)$. We then simulated conducting two-phase proxy-assisted sampling of data from the Eastern United States, using $N-250$ Phase I samples and a single wave in Phase II with expected size $n_{\text{targ}}-250=750$. The Phase II labelling probabilities were determined by Equation \eqref{eq:EstimatedGreedyOptLabellingRule}, in the case where $k^*=1$, $\hat{\varrho}_j^{(1)}=\hat{\varrho}_j^{(\text{pilot})}$, $b_u=0.95$, $b_l=(n_{\text{targ}}-250)/(4 (N-250))$, and $\hat{\lambda}^{(1)}$ was set so that the labelling probabilities sum to $750$ across the $N-250$ Phase I samples. We then discarded the pilot sample and computed the estimator $\htPTD$ and its corresponding confidence interval using the remaining data. 
    
\end{itemize}

Neither of these sampling schemes induces statistical dependencies between the Phase II sampling indicators. We compared these sampling schemes to the two-phase proxy-assisted six-wave sampling scheme (with a uniform first wave) that was studied in the main text (in the case where kNN was used to estimate $\varrho_j(\cdot)$ after each wave) and used the same parameters as specified in Table \ref{table:ExperimentSummary}, with the one difference being that the Phase I and Phase II samples were collected exclusively from the Eastern United States. For all sampling schemes considered in Figure \ref{fig:ComparisonToBaselineSamplingSchemes}, we use the estimator $\htPTD$ (Equation \eqref{eq:PTD_estimator}) with the tuning matrix $\hat{\Omega}$ given by Equation \eqref{eq:EstimatedOptimalDiagTuningMatrix}, although in the formula for the weights $W_i$ at \eqref{eq:aggregated_Wk} $K$ was equal to $1$ rather than $6$ for all but the multiwave sampling scheme. Confidence intervals were constructed using \eqref{eq:CIDef}.


The results (Figure \ref{fig:ComparisonToBaselineSamplingSchemes}) suggest that \samplingSchemeName{} yields efficiency gains compared to these baselines. The baseline in which Phase II sampling is conducted over a single wave informed by external data can, in principle, be more efficient than \samplingSchemeName{}. However, if there is a large distribution shift between the external population and the target population, this baseline can be even less efficient than simply using a uniform random Phase II sample, as seen in the Housing Price and Tree Cover Quantile experiments. The pilot-and-discard sampling scheme naturally yields efficiency losses given that 25\% of samples are discarded after the labelling rule is selected, but as seen in the Tree Cover Quantile experiment, it can still be more efficient than approaches that take a uniform random Phase II sample. Indeed, in a further modification of the Housing Price experiment that was considered to facilitate comparisons with Batch Active Inference \citep{ActiveInferencePaper}, using already-available external data was more efficient (Figure \ref{fig:ComparisonToBatchActiveInference}).

\begin{figure}[!hbt]
    \centering
    \includegraphics[width=0.95\linewidth]{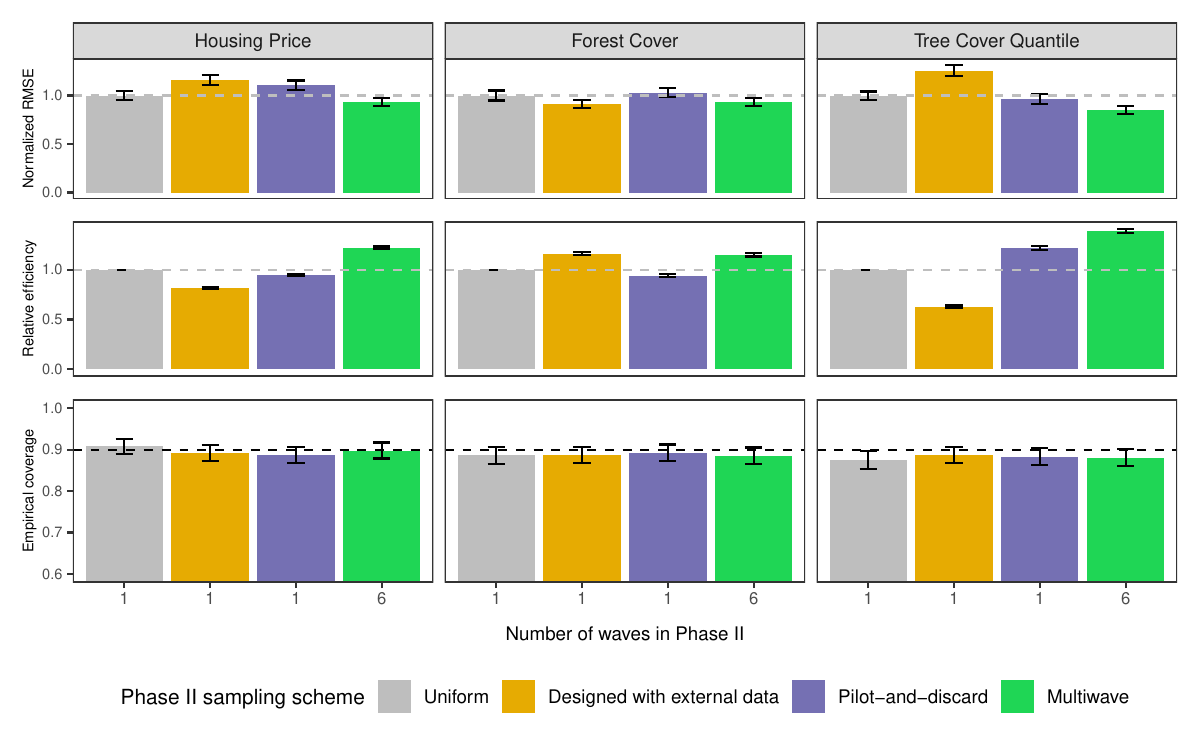}
    \caption{\small Comparison of the Predict-Then-Debias approach under different sampling schemes. Of the sampling schemes depicted, only the two-phase multiwave sampling scheme (green) induces statistical dependencies. The evaluation metrics and the approach to computing standard error bars across $1{,}000$ Monte Carlo simulations are the same as in Figure \ref{fig:ResultsAggregated}, and the uniform and multiwave (kNN, $K=6$) sampling schemes shown here are the same as those in Figure \ref{fig:ResultsAggregated}. The two additional methods considered here are described in Section \ref{sec:ComparisonToSchemesThatDoNotInduceDependency}. In contrast to Figure \ref{fig:ResultsAggregated}, the experiments shown only involve sampling data from the Eastern United States (according to a 100° W longitude threshold) and the estimands of interest are population-level parameters specific to that region. In relative efficiency calculations, the discarded Phase II observations in the pilot-and-discard (purple) method are counted toward the overall number of Phase II samples collected whereas the external data from the Western United States used in the external-data-informed method (orange) are not.
    }\label{fig:ComparisonToBaselineSamplingSchemes}
\end{figure}

\paragraph{Comparison to Batch Active Inference \citep{ActiveInferencePaper}.} We further modify the Housing Price and Forest Cover experiments described above to accommodate existing software implementing Batch Active Inference \citep{ActiveInferencePaper}. Specifically, for regression estimands, Batch Active Inference is designed for settings where all of the gold-standard measurements of all of the covariates are available prior to Phase II while the gold-standard measurements of the response variable are collected in Phase II. Moreover, \cite{ActiveInferencePaper} implement Batch Active Inference for mean estimation, linear regression coefficient estimation, and odds ratio estimation tasks but not for logistic regression or quantile estimation tasks. We thus compare our approach with Batch Active Inference in modifications of the Housing Price and Forest Cover experiments in which the estimand is a linear regression coefficient and the gold-standard covariate measurements are collected in Phase I (this corresponds to the error-in-response setting discussed in \cite{PTDBootstrapPaper}).

Specifically, for the Housing Price (error-in-response) experiment we consider a scenario where the goal is to estimate a regression coefficient in a linear regression of housing price on nightlights, income, and road length, and in Phase I, gold-standard measurements of the covariates, as well as estimates of housing price from \cite{MOSAIKSPaper} are available, while gold-standard measurements of housing price are only collected in Phase II. For the Tree Cover (error-in-response) experiment we consider a scenario where the goal is to estimate a regression coefficient in a linear regression of tree cover (rather than the binarized forest cover outcome) on elevation and population. In this modified experiment, gold-standard measurements of the covariates, as well as estimates of tree cover from \cite{MOSAIKSPaper} are available prior to Phase II, while gold-standard measurements of tree cover are only collected in Phase II. For both experiments, we consider the same setting as described in Section \ref{sec:ComparisonToSchemesThatDoNotInduceDependency} where the population of interest is the Eastern US and external data from the Western US is available to inform the sampling rule. We consider the same baselines of pilot-and-discard and of informing Phase II by external (Western US) data and the same target sample sizes of $N=10{,}000$, $n_{\text{targ}}=1{,}000$, and $n_{\text{targ}}^{(1)}=250$. The Batch Active Inference approach \citep{ActiveInferencePaper} was implemented by using cross-validation to fit an XGBoost model \citep{xgboostPackageR} on the Western US data that predicts the absolute error in the satellite-based estimates from \cite{MOSAIKSPaper} as a function of features that can be measured in Phase I.

\begin{figure}[!hbt]
    \centering
    \includegraphics[width=0.95\linewidth]{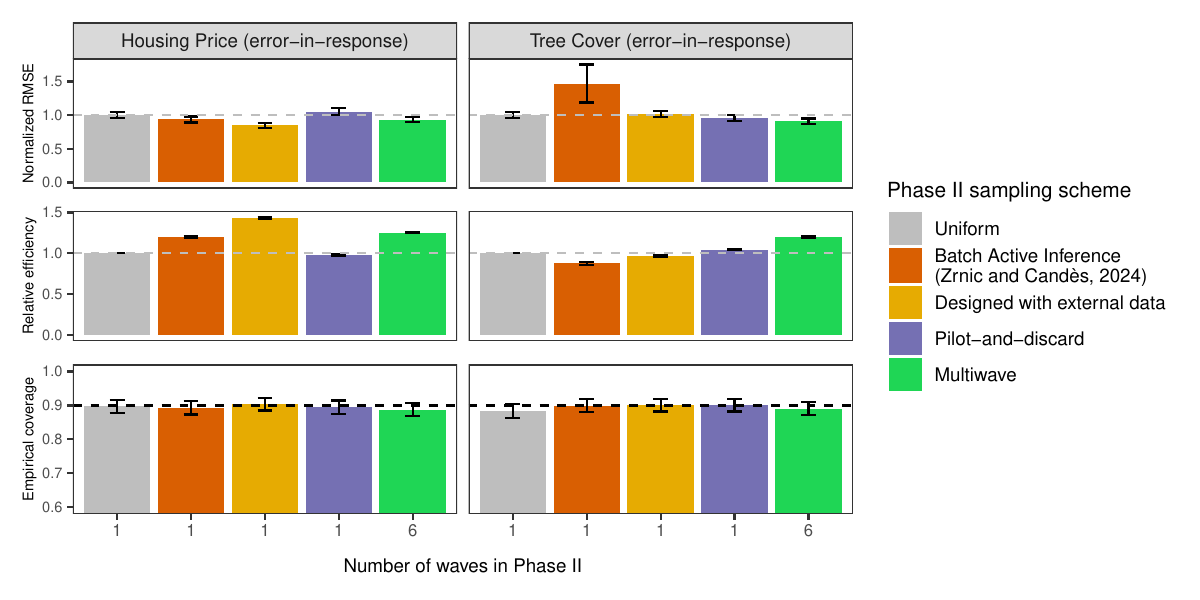}
    \caption{\small Comparison with Batch Active Inference \citep{ActiveInferencePaper} for linear regression coefficient estimation tasks. Both experiments consider the error-in-response regime where the gold-standard measurements of the covariates are available prior to Phase II, but only estimates of the response variable are available prior to Phase II. The evaluation metrics and the approach to computing standard error bars across
1,000 Monte Carlo simulations are the same as in Figure \ref{fig:ComparisonToBaselineSamplingSchemes}. With the exception of Batch Active Inference (dark orange), all approaches use a Predict-Then-Debias type estimator and the same Phase II sampling schemes as those considered in Figure \ref{fig:ComparisonToBaselineSamplingSchemes}. Batch Active Inference uses a sampling scheme informed by model for the absolute error between the estimated response and the actual response (trained on an external dataset) and uses a different point estimator that is computed by minimizing a debiased empirical loss function.
    }\label{fig:ComparisonToBatchActiveInference}
\end{figure}

\subsubsection{Comparison of quantile estimation-specific labelling strategies}\label{sec:ComparisonOfQuantileSpecificLabellingStrategies}

For the tree cover and maize yield quantile estimation experiments (Table \ref{table:ExperimentSummary}, Rows 5 and 6), we conducted additional simulations to test the various quantile estimation-specific labelling strategies developed in Appendix \ref{sec:EstimateQuantileSpecificGreedyOpt}. For both experiments, we conducted $1{,}000$ Monte Carlo simulations. 

In each Monte Carlo simulation, we drew $N$ samples of $(\xobs,\txmiss,\xmiss)$ from the superpopulation being studied, defined by the $\nSuper$ available data points. This resulted in a Phase I sample $(\xobs_i,\txmiss_i)_{i=1}^N$ where the corresponding $(\xmiss_i)_{i=1}^N$ are unobserved but are stored and withheld in order to implement Phase II. We then considered five different Phase II labelling strategies:

\begin{enumerate}
    \item \textbf{Uniform baseline:} Phase II was conducted over a single wave with uniform random sampling. In particular, for each $i \in [N]$, $I_{i}^{\text{base}} \stackrel{\text{i.i.d.}}{\sim} \text{Bernoulli}(n_{\text{targ}}/N)$ was generated, and the resulting dataset was $(I_{i}^{\text{base}},\xobs_i,\txmiss_i, I_{i}^{\text{base}} \xmiss_i)_{i=1}^N$.
    
    \item \textbf{Uniform first wave $\to$ kNN-based subsequent waves:} We conducted the same multiwave sampling scheme described in Section \ref{sec:MonteCarloProcedure} and considered in the main text, for the case where there were $K=6$ waves and k-Nearest-Neighbors was used to estimate $\varrho_j(\cdot)$. (In the crop yield experiment, latitude, longitude, and the image-based estimate of yield were used as features in the kNN model. In the tree cover experiment, elevation, population, and the machine learning-based estimate of tree cover were used as features in the kNN model).

     \item \textbf{Uniform first wave $\to$ isotonic regression-based subsequent waves:} We conducted the same multiwave sampling scheme described in Sections \ref{sec:EstimateQuantileSpecificGreedyOpt} and \ref{sec:IsotonicRegStrategy}, in the case where there were $K=6$ waves. In particular, in the first wave, labels of $\goodX_i$ were collected according to $I_i^{(1)} \stackrel{\text{i.i.d.}}{\sim} \text{Bernoulli}(n_{\text{targ}}^{(1)}/N)$. In the subsequent five waves, we used the quantile-specific approximate greedy optimal sampling strategy described in Section \ref{sec:EstimateQuantileSpecificGreedyOpt}. We used isotonic regression and data from previous waves to estimate $\mathbb{P}(\xmiss \leq \thetaTarg \giv \txmiss = \cdot)$, and plugged this estimate into $\eqref{eq:VarrhoSimplifiedQuantile}$ to estimate $\varrho(\cdot)$. The estimate of $\varrho(\cdot)$ was used to set the sampling rule according to Equation \eqref{eq:EstimatedGreedyOptimalSamplingRuleQuantileSpecific}. (The implementation ignored additional covariates $\xobs$, so we set $\goodX=\xmiss$ and $\proxyX=\txmiss$ for simplicity).
     
     \item \textbf{Uniform first wave $\to$ fitted copula-based subsequent waves:} We repeated the above Phase II sampling scheme with the one exception being that in waves 2--6, the function $\mathbb{P}(\xmiss \leq \thetaTarg \giv \txmiss = \cdot)$ was estimated using the approach described in Section \ref{sec:FittedBivariateCopulaStrategy} rather than by using isotonic regression. This approach involved fitting a bivariate Gaussian copula model to $(\xmiss,\txmiss)$ data.

     \item \textbf{Spearman's correlation-informed first wave $\to$ fitted copula-based subsequent waves:} We repeated the above Phase II sampling scheme with one exception being that in wave 1, we used a nonuniform sampling scheme rather than a uniform one. In wave 1, we instead assumed access to an estimate $\hat{\rho}_S^{\text{avail}}$ of the Spearman's rank correlation coefficient between $\txmiss$ and $\xmiss$. We converted this to a strategic prior guess of the bivariate Gaussian copula model using Equation \eqref{eq:rho0CopulaInitUsingSpearman}, and plugged in this guess of the copula parameter to obtain the first wave labelling rule via Equations \eqref{eq:EstimatedCondCDF_Copula}, \eqref{eq:VarrhoSimplifiedQuantile}, and \eqref{eq:EstimatedGreedyOptimalSamplingRuleQuantileSpecific}. See Sections \ref{sec:EstimateQuantileSpecificGreedyOpt} and \ref{sec:FirstWaveStrategyPrior} for details. Note that in our simulations we set $\hat{\rho}_S^{\text{avail}}$ to be the Spearman's rank correlation coefficient between the $\nSuper$ available data points, but in practice it would need to come from prior information.
     
\end{enumerate}

These independent Monte Carlo simulations were repeated $1{,}000$ times. Within each Monte Carlo simulations, the same Phase I samples were used but different Phase II samples were taken. For all Phase II sampling approaches considered, we evaluated the estimator $\htPTD$ (Equation \eqref{eq:PTD_estimator}) with the same tuning matrix $\hat{\Omega}$ (from \eqref{eq:EstimatedOptimalDiagTuningMatrix}), and constructed the corresponding confidence intervals according to \eqref{eq:CIDef} and using estimates of $H_{\thetaTarg}$ and $H_{\gammaTarg}$ based on the approach described in Section \ref{sec:EstimatingHessiansForQuantileEst}. Across the four multiwave sampling schemes considered (approaches 2--5 above), we used $K=6$ total waves and the same values of $n_{\text{targ}}^{(k)}$ for each $k \in [K]$.

To summarize the $1{,}000$ Monte Carlo simulations (each of which had different Phase I samples), we computed the same evaluation metrics as described in Section \ref{sec:Simulations}: the RMSE (normalized by the RMSE of method 1 above), effective sample size (relative to method 1 described above), and empirical coverage.

The results (Figure \ref{fig:QuantileEstimationStrategiesSims}) suggest that the approach to estimating $\varrho(\cdot)$ that uses the fewest unknown parameters, which is the one based on the single-parameter bivariate Gaussian copula model, was the most efficient in our experiments. A nonuniform first wave, informed by prior knowledge of Spearman's rank correlation coefficient between $\xmiss$ and $\txmiss$, was also found to benefit efficiency.

\begin{figure}[!hbt]
    \centering
    \includegraphics[width=0.95\linewidth]{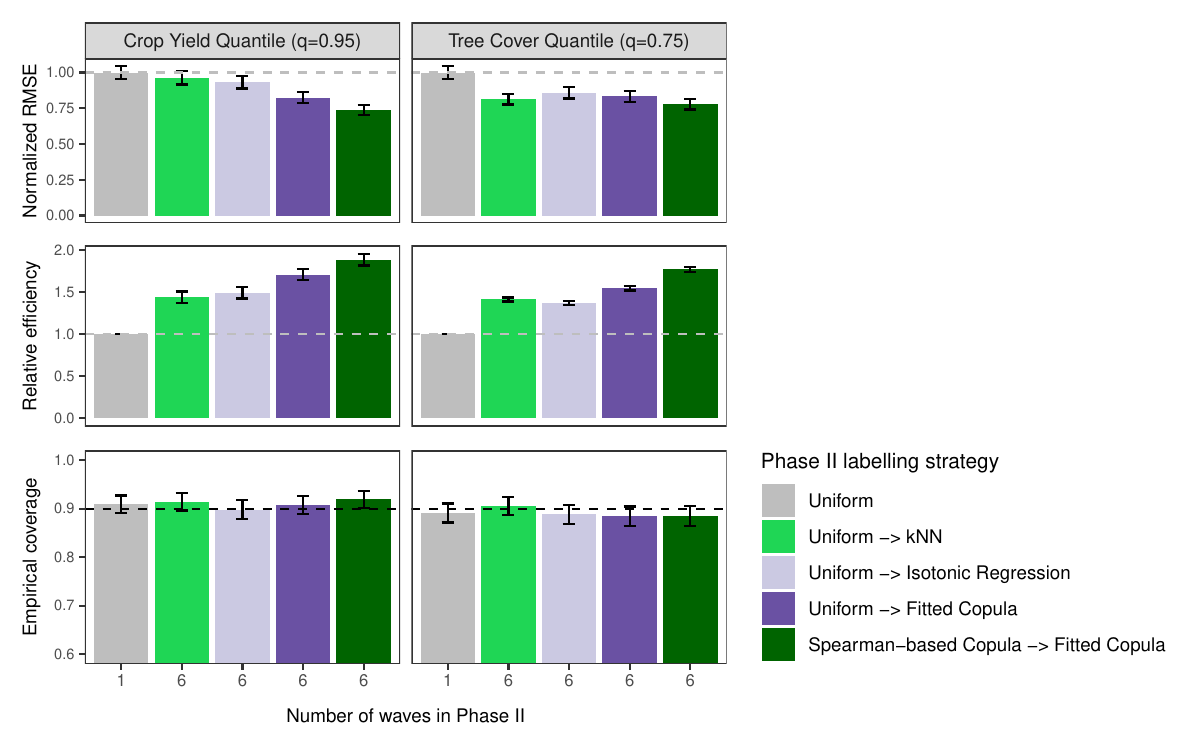}
    \caption{\small Comparison of different adaptive labelling strategies for quantile estimation tasks. In the labelling strategies legend, a label of $A \to B$ means that strategy $A$ was used in wave 1 and strategy $B$ was used in waves 2--6. ``Uniform" refers to i.i.d. Bernoulli sampling. ``kNN", ``Isotonic Regression" and ``Fitted Copula" refer to different methods used to estimate the unknown function $\varrho(\cdot)$ in the greedy optimal labelling rule, described in Methods 2, 3, and 4 in Section \ref{sec:ComparisonOfQuantileSpecificLabellingStrategies}, respectively. ``Spearman-based Copula" uses a nonuniform first wave that is informed by an estimate of Spearman's rank correlation coefficient between $\xmiss$ and $\txmiss$ (See Method 5 in Section \ref{sec:ComparisonOfQuantileSpecificLabellingStrategies} and Sections \ref{sec:EstimateQuantileSpecificGreedyOpt} and \ref{sec:FirstWaveStrategyPrior} for more details). The plotted evaluation metrics for the 90\% confidence intervals are calculated across $1{,}000$ independent Monte Carlo simulations with error bars giving $\pm$2 standard errors.
    }\label{fig:QuantileEstimationStrategiesSims}
\end{figure}

\subsection{Additional details for the attainable yield simulation study}\label{sec:AttainableYieldSimsAdditionalDetails}

 We downloaded $\nSuper=4{,}173$ fully observed samples of $(Z^{\text{(lat)}},Z^{\text{(lon)}},\tilde{Y},Y)$ from \cite{HarrisonCropYieldPPI}, where $\xobs=(Z^{\text{(lat)}},Z^{\text{(lon)}})$ is the latitude and longitude of the field, $\txmiss=\tilde{Y}$ is the estimate of  maize yield from a convolutional neural network that was applied to an image of that field, and $\xmiss=Y$ is the maize yield crop-cut measurement from that field. Then for the purposes of testing our methods in a data-driven manner, we set the empirical distribution of these $\nSuper$ samples to be the distribution of the superpopulation $\datvecrawDist$ that we study in the experiment. In particular, we calculate the ``ground truth" value for the parameter of interest $\thetaTarg$ by finding the empirical 95th percentile of $Y$ across the $\nSuper$ fully observed samples (blue line in Figure \ref{fig:AttainableYieldSims}). Similarly, $\gammaTarg$ is the empirical 95th percentile of $\tilde{Y}$ across the $\nSuper$ fully observed samples (red line in Figure \ref{fig:AttainableYieldSims}).

We then performed $1{,}000$ independent Monte Carlo simulations. In each simulation, we drew $N=10{,}000$ samples with replacement from the $\nSuper$ observations of $(Z^{\text{(lat)}},Z^{\text{(lon)}},\tilde{Y},Y)$. This resulted in a Phase I sample $(\proxyX_i)_{i=1}^N$, where $\proxyX_i =(Z_i^{\text{(lat)}},Z_i^{\text{(lon)}},\tilde{Y}_i)$ for each $i \in [N]$ and where the corresponding $(Y_i)_{i=1}^N$ crop yield values are unobserved but are stored and withheld in order to implement Phase II. We then considered five different approaches to conducting Phase II sampling and constructing estimators using the Phase I and Phase II samples:

\begin{enumerate}
    \item \textbf{Uniform sampling and ignoring image data baseline:} 
    For each $i \in [N]$, $I_{i}^{\text{base}} \stackrel{\text{i.i.d.}}{\sim} \text{Bernoulli}(n_{\text{targ}}/N)$ was generated, where $n_{\text{targ}}=500$. We computed $\htc$ to estimate $\thetaTarg$ by taking the 95th empirical percentile of $\{ Y_i \ : \ I_{i}^{\text{base}}=1 \}$. This approach ignores the image-based estimates of yield $\tilde{Y}_i$ collected in Phase I.

    \item \textbf{Uniform sampling with Predict-Then-Debias estimator baseline:} Using the same indicators generated for the previous baseline, we also considered a Predict-Then-Debias estimator that leverages the $\tilde{Y}$ observations from Phase I. In particular, we estimated $\thetaTarg$ using Equation \eqref{eq:PTD_estimator} (with tuning matrix given by \eqref{eq:EstimatedOptimalDiagTuningMatrix}) using $(I_{i}^{\text{base}},\tilde{Y}_i, I_{i}^{\text{base}} Y_i)_{i=1}^N$ as input data. 

    \item \textbf{Multiwave Phase II sample with no revisiting allowed:} We set $K=6$, $n_{\text{targ}}^{(1)}=125=n_{\text{targ}}/4$, and $n_{\text{targ}}^{(k)}=75$ for each $k>1$. We randomly partitioned the $N$ Phase I samples into 6 batches $B_1,\dots,B_6$ such that $\cup_{k=1}^K B_k= [N]$ and such that for each $k \in [K]$, $ \vert B_k \vert / N =n_{\text{targ}}^{(k)}/n_{\text{targ}}$. In Wave 1 of Phase II, we generated $I_i^{(1)} \stackrel{\text{i.i.d.}}{\sim} \text{Bernoulli}(n_{\text{targ}}^{(1)}/\vert B_1\vert)$ for each $i \in B_1$, and only collected $Y_i$ labels when both $i \in B_1$ and $I_i^{(1)}=1$. For each subsequent wave $k \in [K] \setminus \{1\}$, we first fit a kNN model to estimate $\varrho_j(\cdot)$ defined at Equation \eqref{eq:varrhoDef} using $Y_i$ observations from earlier waves and the $(\proxyX_i)_{i=1}^N=(Z_i^{\text{(lat)}},Z_i^{\text{(lon)}},\tilde{Y}_i)_{i=1}^N$ Phase I data. We called this estimate $\hat{\varrho}_j^{(k)}(\cdot)$. We then set $$\check{\pi}_i^{(k)} = \min \bigl\{b_u, \max \{ b_l, \sqrt{\hat{\varrho}_j^{(k)}(\proxyX_i) / \tilde{\lambda}^{(k)}} \} \bigr\} \quad \text{for each } i \in B_k,$$ where $b_u=0.95$, $b_l = n_{\text{targ}}^{(k)}/(4 \cdot \vert B_k \vert)$, and where $\tilde{\lambda}^{(k)}$ is set to solve $\sum_{i \in B_k} \check{\pi}_i^{(k)} = n_{\text{targ}}^{(k)}$. We then generate $I_i^{(k)} \stackrel{\text{ind.}}{\sim} \text{Bernoulli}(\check{\pi}_i^{(k)})$ for each $i \in B_k$, and collected $Y_i$ labels when both $i \in B_k$ and $I_i^{(k)}=1$. This sampling scheme prevents revisiting fields assigned to batches that correspond to earlier waves, even if there are signs that particularly valuable fields were not labelled in earlier waves. (The labelling probabilities $\check{\pi}_i^{(k)}$ above are the approximately greedy optimal ones according to the strategy described in Section \ref{sec:DescribeApproximateGreedyOptimalStrategy} when the efficiency calculations are adapted to this setting. Notably, under such sampling schemes, $\prod_{k'=1}^{k-1} (1- \LabelRulePi{k'}(\proxyX_i))=1$ for $i \in B_k$). After all six waves were completed, the estimator $\htPTD$ (Equation \eqref{eq:PTD_estimator}) using the tuning matrix $\hat{\Omega}^{\text{opt,diag}}$ (Equation \eqref{eq:EstimatedOptimalDiagTuningMatrix}) and the corresponding confidence intervals (Equation \eqref{eq:CIDef}) were computed, although the formulas were modified to use inverse probability weights rather than the multiwave inverse probability weights.
    
    \item \textbf{Multiwave Phase II sample using a uniform 1st wave and kNN to inform five subsequent waves:} The same multiwave sampling approach as used in the simulations of Section \ref{sec:Simulations} (in the $K=6$, kNN case) was conducted. Specifically, for each $i \in [N]$, $I_{i}^{(1)} \stackrel{\text{i.i.d.}}{\sim} \text{Bernoulli}(n_{\text{targ}}^{(1)}/N)$ was generated, where $n_{\text{targ}}^{(1)}=125=n_{\text{targ}}/4$, and these indicators determined which maize yield measurements $Y_i$ were observed in wave 1. The remaining Phase II waves were conducted according to the recursively defined approximate greedy optimal approach described in Section \ref{sec:DescribeApproximateGreedyOptimalStrategy}. In particular, for each $k \in [K] \setminus \{1\}$, the collected data from previous waves and from Phase I was used to estimate $\varrho_j(\cdot)$ with the kNN-based approach described in Sections \ref{sec:EstimateOptimalStrategyWithML} and \ref{sec:MonteCarloProcedure}, and $\hat{\lambda}^{(k)}$ was calculated using the procedure in Appendix \ref{sec:EnforceBudgetConstraints} so that in expectation, $n_{\text{targ}}^{(k)}=75$ measurements were collected in wave $k$. 
    We set $b_u=0.95$ and $b_l = n_{\text{targ}}^{(k)}/(4 \cdot N_u^{(k-1)})$, where $N_u^{(k-1)}$ denotes the number of samples that had not been labelled prior to wave $k$. $(Z^{\text{(lat)}},Z^{\text{(lon)}},\tilde{Y})$ were all used as features in the kNN model fits. After all six waves were completed, the estimator $\htPTD$ (Equation \eqref{eq:PTD_estimator}) using the tuning matrix $\hat{\Omega}^{\text{opt,diag}}$ (Equation \eqref{eq:EstimatedOptimalDiagTuningMatrix}) and the corresponding confidence intervals (Equation \eqref{eq:CIDef}) were computed.

    \item \textbf{Multiwave Phase II sample using a prior-information-based first wave and a copula model:} Using the same specifications for $K$, $n_{\text{targ}}^{(k)}$, $n_{\text{targ}}$, $b_l$, $b_u$, we conducted a Multiwave Phase II sample with the strategy described in Sections \ref{sec:EstimateQuantileSpecificGreedyOpt}, Section \ref{sec:FirstWaveStrategyPrior} (wave 1 only), and Section \ref{sec:FittedBivariateCopulaStrategy} (wave 2--6). Specifically, we assumed knowledge of Spearman's rank correlation coefficient (in practice this can be reported by the producer of a dataset with machine learning predictions, but in our simulations we calculated it directly from the $\nSuper$ samples). We then converted the rank correlation coefficient to the correlation coefficient in a corresponding bivariate Gaussian copula model, and then we further multiplied this correlation coefficient by a factor of $0.9$ because in practice, $\tilde{Y}$ and $Y$ tend to be more correlated on a training set than on a test or target set (see Equation \eqref{eq:rho0CopulaInitUsingSpearman}). Using this conservative prior estimate for the best copula model, we set $\hat{\varrho}_j^{(1)}(\cdot)$ using Equations \eqref{eq:EstimatedCondCDF_Copula} and \eqref{eq:VarrhoSimplifiedQuantile} to implement the first (nonuniform) wave.  The remaining Phase II waves were conducted according to the recursively defined approximate greedy optimal approach described in Sections \ref{sec:EstimateQuantileSpecificGreedyOpt} and \ref{sec:FittedBivariateCopulaStrategy}. In those subsequent waves $\hat{\varrho}_j^{(k)}(\cdot)$ was estimated empirically by fitting a bivariate Gaussian copula model to the paired $(\tilde{Y},Y)$ data available prior to that wave. After all six waves were completed, the estimator $\htPTD$ (Equation \eqref{eq:PTD_estimator}) using the tuning matrix $\hat{\Omega}^{\text{opt,diag}}$ (Equation \eqref{eq:EstimatedOptimalDiagTuningMatrix}) and the corresponding confidence intervals (Equation \eqref{eq:CIDef}) were computed.

\end{enumerate}

These independent Monte Carlo simulations were repeated $1{,}000$ times. Within each Monte Carlo simulations, the same Phase I samples were used but different Phase II samples were taken. To summarize the $1{,}000$ Monte Carlo simulations (each of which had different Phase I samples), we computed the same evaluation metrics as described in Section \ref{sec:Simulations}: the RMSE (normalized by the RMSE of method 2 above), effective sample size (relative to method 2 described above), and empirical coverage.

\subsection{Supplementary figure}\label{sec:SupplemtaryFigure}

\begin{figure}[!hbt]
    \centering
    \includegraphics[width=0.95\linewidth]{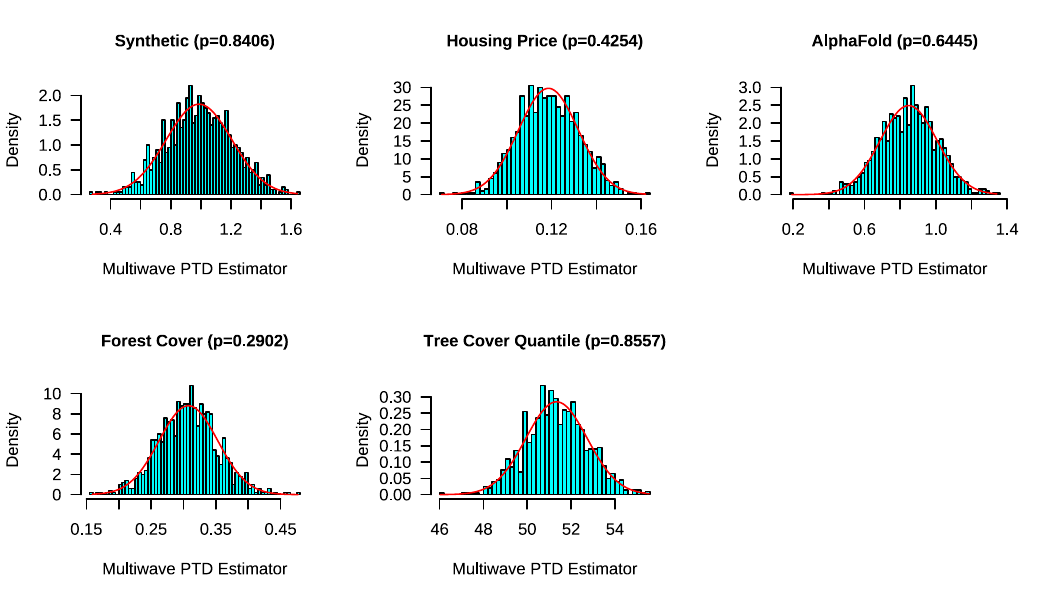}
    \caption{Histograms of the Multiwave Predict-Then-Debias estimator under two-phase multiwave sampling across $1{,}000$ simulations. Each histogram corresponds to a different experiment (Table \ref{table:ExperimentSummary}) and we depict the results when the number of waves $K=6$. For all experiments, with the exception of the AlphaFold one, the kNN-based approach for estimating the greedy optimal labelling rule (Section \ref{sec:EstimateOptimalStrategyWithML}) was used. The histogram y-axis is rescaled to the density scale with the red lines giving a Gaussian distribution with a mean and variance matching the empirical ones from the simulation. The panel titles give the p-value from a Shapiro-Wilk test for normality, rounded to 4 decimal places.}\label{fig:TestNormality}
\end{figure}


\end{document}